\documentclass[12pt, a4paper]{article}

\usepackage[fleqn]{amsmath}      
\numberwithin{equation}{section}
\usepackage{amssymb}
\usepackage{amsthm}
\usepackage[titletoc,title]{appendix}
\usepackage{array}
\usepackage{arydshln}
\usepackage{authblk}
\usepackage[english]{babel}
\usepackage{bbm}
\usepackage{bm}
\usepackage{booktabs}
\usepackage[font=small,labelfont=bf]{caption}
\usepackage{color}
\usepackage{econometrics}
\usepackage{enumerate}
\usepackage{float}
\usepackage[margin=1.9cm]{geometry}
\usepackage{graphicx}
\usepackage[latin1]{inputenc}
\usepackage{rotating}
\usepackage{mathrsfs}
\usepackage{mathtools}
\usepackage{microtype}
\usepackage{multicol}
\usepackage{multirow}
\usepackage[authoryear,round]{natbib}
\newcommand\cites[1]{\citeauthor{#1}'s\ (\citeyear{#1})}

\usepackage{pdflscape}
\usepackage{soul}
\usepackage{subfiles}
\usepackage{subcaption}

\usepackage{times} 
\usepackage{tablefootnote}
\usepackage{titletoc}
\usepackage[flushleft]{threeparttable}
\usepackage{txfonts}
\usepackage{url}
\usepackage{verbatim}
\usepackage[all]{xy}
\usepackage{xr}
\usepackage{csquotes}
\usepackage{adjustbox}
\usepackage{tcolorbox}
\usepackage{thmtools}
\usepackage{dcolumn}
\newcolumntype{d}[1]{D{.}{.}{#1}} 
\usepackage{setspace}

\makeatletter
\let\@addpunct\@gobble
\makeatother

\newtheoremstyle{mystyle}
{0.5cm}
{0.5cm}
{\itshape}
{}
{\bfseries}
{ }
{\newline}
{\thmname{#1}\thmnumber{ #2}\thmnote{ (#3)}}
\theoremstyle{mystyle}
\newtheorem{theorem}{Theorem}
\newtheorem{remark}{Remark}

\newtheorem{example}{Example}
\newtheorem{assumption}{Assumption}
\newtheorem{lemma}{Lemma}

\newtheorem{algorithm}{Algorithm}

\makeatletter
\renewenvironment{proof}[1][\proofname]{\par
	\pushQED{\qed}%
	\normalfont \topsep6\p@\@plus6\p@\relax
	\trivlist
	\item[\hskip\labelsep
	\itshape
	#1]\ignorespaces
}{%
	\popQED\endtrivlist\@endpefalse
}
\makeatother

\makeatletter
\renewcommand\@biblabel[1]{}
\makeatother
\providecommand{\indic}[1]{\mathbbm{1}_{\left\{#1\right\}}}
\newcommand{\argmin}{\operatornamewithlimits{arg\;min}}

\DeclarePairedDelimiter\abs{\lvert}{\rvert}%
\DeclarePairedDelimiter\norm{\lVert}{\rVert}%
\makeatletter
\let\oldabs\abs
\def\abs{\@ifstar{\oldabs}{\oldabs*}}
\let\oldnorm\norm
\def\norm{\@ifstar{\oldnorm}{\oldnorm*}}
\makeatother

\makeatletter
\newcommand{\rmnum}[1]{\romannumeral #1}
\newcommand{\Rmnum}[1]{\expandafter\@slowromancap\romannumeral #1@}
\newcommand{\point}[1]{$(\rmnum{#1})$}
\makeatother

\makeatletter
\newcommand{\biggg}{\bBigg@{3}}
\newcommand{\Biggg}{\bBigg@{4}}
\makeatother

\newcommand{\tran}{'}
\newcommand{\distrequal}{\stackrel{d}{=}}

\newcommand{\brown}{\bm{B}}
\newcommand{\brownnormal}{B}
\newcommand{\dtostar}{\longrightarrow_{d^*}}

\newcommand{\limitdist}{{\bm \calJ}}

\newcommand{\myint}{\int}

\newcommand{\p}{\mathbb{P}}

\newcommand{\serialbias}{\bm{\calB}_{vu}}
\newcommand{\serialbiasSimulated}{\widehat{\bm{\calB}}_{vu}^{-}}
\newcommand{\sumtT}{\sum_{t=1}^T}

\newcommand{\vbias}{{\bm{\mathcal{B}}}}

\newcommand{\romanone}{(\mathbf{\uppercase\expandafter{\romannumeral1}})}
\newcommand{\romantwo}{(\mathbf{\uppercase\expandafter{\romannumeral2}})}
\newcommand{\romanthree}{(\mathbf{\uppercase\expandafter{\romannumeral3}})}
\newcommand{\romanfour}{(\mathbf{\uppercase\expandafter{\romannumeral4}})}

\newcommand{\paraspace}{\bm \Gamma}

\newcommand{\COTWO}{\mathrm{CO}_2}
\newcommand{\Prob}{\mathbb{P}}
\renewcommand{\wto}{\longrightarrow_{d}}

\title{Cointegrating Polynomial Regressions with Power Law Trends: Environmental Kuznets Curve or Omitted Time Effects?}

\author[1]{Yicong Lin}
\author[2]{Hanno Reuvers\thanks{Corresponding author:  Department of Econometrics, Erasmus School of Economics, Erasmus University Rotterdam, Burgemeester Oudlaan 50, 3062 PA Rotterdam, The Netherlands. E-mail address: reuvers@ese.eur.nl. Phone: +31 10 40 81257.}}

\affil[1]{Department of Econometrics and Data Science, Vrije Universiteit Amsterdam, 1081HV Amsterdam, The Netherlands}
\affil[2]{Department of Econometrics, Erasmus University Rotterdam, 3062PA Rotterdam, The Netherlands}

\date{\today}

\begin{document}

\begin{spacing}{1.2}
						
\maketitle

\begin{abstract}

\noindent

The environmental Kuznets curve predicts an inverted U-shaped relationship between environmental pollution and economic growth. Current analyses frequently employ models which restrict nonlinearities in the data to be explained by the economic growth variable only. We propose a Generalized Cointegrating Polynomial Regression (GCPR) to allow for an alternative source of nonlinearity. More specifically, the GCPR is a seemingly unrelated regression with (1) integer powers of deterministic and stochastic trends for the individual units, and (2) a common flexible global trend. We estimate this GCPR by nonlinear least squares and derive its asymptotic distribution. Endogeneity of the regressors will introduce nuisance parameters into the limiting distribution but a simulation-based approach nevertheless enables us to conduct valid inference.  A multivariate subsampling KPSS test is proposed to verify the correct specification of the cointegrating relation. Our simulation study shows good performance of the simulated inference approach and subsampling KPSS test. We illustrate the GCPR approach using data for Austria, Belgium, Finland, the Netherlands, Switzerland, and the UK. A single global trend accurately captures all nonlinearities leading to a \emph{linear} cointegrating relation between GDP and CO\textsubscript{2} for all countries. This suggests that the environmental improvement of the last years is due to economic factors different from GDP. 


	
\end{abstract}
\end{spacing}
\newpage		
\begin{spacing}{2}
\section{Introduction}
On page 370 of their seminal paper, \cite{grossmankrueger1995} conclude:
\begin{displayquote}
 ``\textit{Contrary to the alarmist cries of some environmental groups, we find no evidence that economic growth does unavoidable harm to the natural habitat. Instead we find that while increases in GDP may be associated with worsening environmental conditions in very poor countries, air and water quality appear to benefit from economic growth once some critical level of income has been reached.}''
\end{displayquote}
The quote above suggests an inverted U-shaped relationship between environmental degradation and economic growth. This relationship is currently known as the Environmental Kuznets Curve (EKC) and it forms an active research area. Its relevance becomes clear if we look at some forecasts of long-run economic growth. The projected GDP per capita growth of the world is about 2.1\% per year for the next decades (chapter 3 in \cite{nordhaus2013}; \cite{christensengillinghamnordhaus2018}) and this growth is partially powered by carbon-based energy resources, water usage, and material consumption. In absence of an EKC, economic growth will place more and more stress on the environment. Alternatively, if the EKC exists, then the inverted U-shape eventually implies a turning point after which economic growth and environmental improvement go hand in hand. Due to such considerations, there is now, some 25 years after its first conception, a rich literature that (1) reports on the experimental evidence on the existence/nonexistence of the EKC, (2) provides economic theory to explain the EKC, and/or (3) refines the econometric tools that are used to analyse the EKC.\footnote{Further references to these specific areas of research can be found in the review articles by \cite{dasguptaetal2002}, \cite{stern2004}, and \cite{carson2009} among others.} To quantify the volume of the literature we have entered the search query ``\textit{Environmental Kuznets Curve}'' into the Web of Science: more than 4,200 references are found.\footnote{Web of Science, accessed on December 6, 2021, http://www.webofknowledge.com.}

Driven by contradictory empirical results as well as the variability in estimated turning points, the EKC has been criticised on two main points. First, the income variable was initially treated as a stationary variable whereas later research shows that the unit root hypothesis often cannot be rejected (see \cite{galeottimaneralanza2009}, p. 553; \cite{stern2017}, p. 14--15). Nonstationarity has further implications because EKC regressions include higher integer powers of GDP as well. This combination of nonstationarity and nonlinearity places the EKC in the nonlinear cointegration literature and appropriate econometric techniques should be employed. Such techniques have been developed in \cite{wagner2015} and \cite{wagnerhong2016} under the name of \emph{Cointegrating Polynomial Regressions (CPRs)}. CPRs contain deterministic variables, integrated processes, and their integer powers.\footnote{\cite{stypkawagnergrabarczykkawka2017} reiterate the need to model the income variable as nonstationary. It is less important to use an estimation procedure acknowledging the fact that several integer powers of the same integrated process appear as regressors. That is, \cite{stypkawagnergrabarczykkawka2017} find that the ``standard estimator'' which treats higher order powers of the integrated regressor as additional I(1) variables has the same limiting distribution as the CPR estimator.} Multivariate generalizations of CPRs, \emph{Seemingly Unrelated Cointegrating Polynomial Regressions (SUCPRs)}, are discussed in \cite{wagnergrabarczykhong2019} and \cite{linreuvers2019}.

As a second point of critique, there is an ongoing debate on the model specification. Various functional forms can describe the relationship between national income and the pollution variable. The quadratic specification is widespread but cubic relationships (\cite{harbaughlevinsonwilson2002}; \cite{wagner2015}) and double-nonlinear transformation (\cite{lintuyao2020}) are also in use. Various specification tests are helpful while deciding on the right parametric specification (\cite{hongphillips2010}; \cite{wangphillips2012}; \cite{wangwuzhu2018}). Alternatively, one could resort to nonparametric estimation procedures altogether  (\cite{wangphillips2009}; \cite{lintonwang2016}). Whereas such modelling approaches do allow for a more flexible relationship, they also implicitly assume that nonlinear environmental effects are solely attributable to economic growth. Relevant variables are thus potentially missing from the model specification. Such omitted variables are a valid concern because advances in green technology, pollution policy, and environmental awareness, may all influence pollution levels. However, such data is typically available for short time spans only (and for that reason often excluded from the model). Time effects can control for time-variation in unobserved effects (\cite{volleberghmelenbergdijkgraaf2009}).

Current developments on nonlinear cointegration emphasize the role of the nonstationary regressor yet pay less attention to time effects. Time effects are important. The small simulation exercise in Table \ref{table:sim_intro} illustrates the point. Foreshadowing our proposed model, we consider a multivariate setting with a global nonlinear, smooth time trend. The global trend is omitted by the researcher and a quadratic EKC-specification is estimated: $y_{i,t} = \tau_{1,i} + \tau_{2,i} t + \phi_{1,i} x_{i,t} + \phi_{2,i} x_{i,t}^2 + u_{i,t}$ ($i=1,\ldots,3$), where $x_{i,t}$ and $y_{i,t}$ are unit-specific variables measuring income and environmental pollution, respectively. We test  $H_0: \phi_{2,1}=\phi_{2,2}=\phi_{2,3} =0$ because a significantly negative coefficient in front of $x_{i,t}^2$ is typically interpreted as evidence of an EKC.\footnote{For the moment, we focus on the curvature parameter. Clearly, for an inverted U-shaped relationship the coefficient in front of the linear term should be positive.} Panel (A) reveals exacerbated rejection frequency as curvature caused by the global deterministic trend is mistakenly interpreted as curvature caused by the income variable. In other words, negative and significant coefficients in front of squared GDP are possibly caused by omitted nonlinear deterministic trends rather than being indicative of an EKC. To be on the safe side, we recommend researchers to include a nonlinear trend component in their model specification. If unnecessary, then this is rather innocuous. Indeed, Panel (B) of Table \ref{table:sim_intro} shows that significant results for nonlinear economic growth effects continue to be found with modest losses in statistical power.

\begin{table}[t]
	\centering
	\caption{The rejection rate (in $\%$) when testing $H_0: \phi_{2,1}=\phi_{2,2}=\phi_{2,3}=0$. (A) Falsely inflated rejections of $H_0: \phi_{2,1}=\phi_{2,2}=\phi_{2,3}=0$ when time effects are omitted. (B) Adding an additional global deterministic trend to the model specification hardly influences the power of the test $H_0: \phi_{2,1}=\phi_{2,2}=\phi_{2,3}=0$. That is, significant coefficients in front of $x_{i,t}^2$ remain significant after adding a redundant flexible global trend.}
	\label{table:sim_intro}
	\resizebox{\textwidth}{!}{%
	\begin{threeparttable}
		\begin{tabular}{ccccccccc}
			\toprule
			& \multicolumn{3}{c}{Panel (A): Omitted Global Trend} &  &  & \multicolumn{3}{c}{Panel (B): Redundant Global Trend} \\
			\midrule
			\addlinespace[0.1cm]
			DGP & \multicolumn{3}{c}{$y_{i,t}= \tau_g t^{\theta} +   \tau_{1,i} + \tau_{2,i} t +  \phi_{1,i}   x_{i,t} + u_{i,t}$} &  & \multicolumn{1}{l}{} & \multicolumn{3}{c}{$\text{\small{$y_{i,t}$}}=\text{\small{$\tau_{1,i} + \tau_{2,i} t +  \phi_{1,i} x_{i,t}$}} + \phi_2 x_{i,t}^2 +  \text{\small{$u_{i,t}$}}$} \\
			\addlinespace[0.1cm]
			\cmidrule{2-4}\cmidrule{7-9} 
			\addlinespace[0.1cm]
			Model & \multicolumn{3}{c}{$y_{i,t} = \tau_{1,i} + \tau_{2,i} t + \phi_{1,i}   x_{i,t} + \phi_{2,i} x_{i,t}^2 + u_{i,t}$} &  & \multicolumn{1}{l}{} & \multicolumn{1}{l}{\small Correct Specification} & \multicolumn{1}{l}{} & \multicolumn{1}{l}{$\text{ \small{$y_{i,t}$}} = \tau_g t^{\theta} + \text{ \small{$\tau_{1,i} + \tau_{2,i} t +   \phi_{1,i} x_{i,t} + \phi_{2,i} x_{i,t}^2 + u_{i,t}$}}$} \\
			\addlinespace[0.1cm]
			\midrule
			$\tau_g\,(\times 10^{-5})$ & FM-SOLS &  & FM-SUR &  & $\phi_2$ & SimNLS &  & SimNLS \\
			\midrule
			0 & 6.30 &  & 6.63 &  & 0 & 6.93 &  & 5.70 \\
			\addlinespace[0.1cm]
			-0.5 & 13.27 &  & 12.77 &  & -0.5 & 9.20 &  & 6.97 \\
			-1 & 30.07 &  & 27.50 &  & -1 & 14.97 &  & 9.97 \\
			-1.5 & 46.23 &  & 41.57 &  & -1.5 & 30.47 &  & 20.70 \\
			\addlinespace[0.1cm]
			-2 & 56.60 &  & 50.30 &  & -2 & 55.33 &  & 40.53 \\
			-2.5 & 64.50 &  & 56.00 &  & -2.5 & 81.73 &  & 69.20 \\
			-3 & 68.60 &  & 57.57 &  & -3 & 93.97 &  & 89.50 \\
			\bottomrule
		\end{tabular}%
	    \begin{tablenotes}
	    	\footnotesize
	    	\item Note 1: For illustrative purpose, we consider a stylised example in this introduction. The exact parametrisation is available in Section of the Supplementary Material. More elaborate simulation results based on the empirical application are reported as simulation DGP2 in Section \ref{sec:simulations}. 
	    	\item Note 2: FM-SOLS and FM-SUR are documented in \cite{wagnergrabarczykhong2019}. The results in Panel (B) are based on simulation-based inference, see Section \ref{subsec:sim_inf}.
	    \end{tablenotes} 
    \end{threeparttable}
	}
\end{table}

The contributions of this paper are fourfold. First, we propose the \emph{Generalized Cointegrating Polynomial Regression (GCPR)}. This multivariate model features a global power law trend to capture time effects. Power law trends have been employed to model non-constant growth rates in technology indices (\cite{duggalsaltzmanklein1999}; \cite{duggal2007infrastructure}) and production functions (\cite{klein2004}). Within the EKC context, this flexible trend can capture common time effects that are implicit in omitted variables. Alternatively, as in \cite{lilinton2020}, the reader can view the global flexible trend as an outside option (next to the income variable) to describe nonlinearities in the data. Limiting distributions for estimators in models with purely deterministic power law trends have been reported in \cite{phillips2007}, \cite{robinson2012} and \cite{gaolintonpeng2020}. The presence of integrated variables requires an alternative asymptotic framework. Moreover, due to endogeneity, approaches assuming pre-determined integrated regressors (\cite{parkphillips1999}; \cite{parkphillips2001}; \cite{changparkphillips2001}) are inappropriate and we instead opt for a proof along the lines of \cite{chanwang2015}. The resulting limiting distribution is non-standard because (1) the scaling matrix with convergence rates is non-diagonal and parameter dependent, and (2) second-order bias terms are present. Second, we propose a simulation-based approach to conduct inference. Monte Carlo simulations show clear benefits of this simulation-based approach in terms of size control compared to existing methods. Third, in the spirit of \cite{choisaikkonen2010}, we report a multivariate KPSS-type test to verify the stationarity of the error process thus enabling researchers to avoid spurious results or misspecified cointegrating relations. Fourth and finally, in the empirical application, we investigate the EKC for Austria, Belgium, Finland, the Netherlands, Switzerland and the UK over the period 1870--2014. Nonparametric estimates and tests confirm that the global trend captures all nonlinearities in the data. Nonlinear effects in log per capita GDP (and thus also evidence for an EKC) are absent. Given such findings, we offer a clear recommendation to researchers to check whether their EKC conclusions are robust to the inclusion of power law trends.

This paper is organized as follows. Section \ref{sec:model} introduces the model and the estimation framework. Asymptotic properties of the estimators and parameter inference are discussed in Section \ref{sec:AsympTheory}. The Monte Carlo simulations in Section \ref{sec:simulations} compare asymptotic results to finite sample performance. An in depth discussion of the Environment Kuznets Curve can be found in Section \ref{sec:empapplication}. Section \ref{sec:conclusion} concludes. The proofs of the main theorems are collected in the Appendix and further information is available in the Supplementary Material.

Finally, some words on notation. The integer part of the number $a\in \SR^{+}$ is denoted by $[a]$. For a vector $\vx\in\SR^n$,  its $p$-norm is denoted by $\|\vx\|_p=(\sum_{i=1}^{n}|x_i|^p)^{1/p}$. For a matrix $\mA$, say of dimension ($n\times m$), the induced $p$-norm is defined as $\|\mA\|_p=\sup_{\vx\neq \vzeros} \|\mA\vx\|_p/\|\vx\|_p$. We will omit the subscripts whenever $p=2$. The $(n\times n)$ identity matrix is denoted $\mI_n$ and $\vones_n$ signifies an $n$-dimensional column vector with all entries equal to 1. The block-diagonal matrix $\diag[\mA_1,\ldots,\mA_n]$ stacks the matrices $\mA_1,\ldots,\mA_n$ along its diagonal. We omit the integration bounds whenever the integration interval is $[0,1]$. The symbol ``$\distrequal$'' stands for equality in distribution, and ``$\pto$'' and ``$\dto$'' denote convergence in probability and in distribution. If convergence occurs conditionally on the sample, then we add a superscript ``*'' to the standard notation. The probabilistic Landau symbols are $O_p(\cdot)$ and $o_p(\cdot)$. Finally, the generic constant $C$ can change from line to line. 
\section{The Model and NLS Estimation} \label{sec:model}

Our model specification enriches the Seemingly Unrelated Cointegrating Polynomial Regressions (SUCPRs) from \cite{wagnergrabarczykhong2019} with a flexible deterministic trend. That is, each individual series in the system is affected by specific deterministic variables and integrated regressors (and their integer powers) while a \emph{global} flexible trend describes nonlinear behaviour that is prevalent across all series. The resulting \emph{Generalized Cointegrating Polynomial Regression (GCPR)}  is given by:
\begin{equation}
 y_{i,t} = \tau_g t^{\theta} + \tau_{1,i} + \tau_{2,i} t + \sum_{j=1}^{p_i} \phi_{j,i} x_{i,t}^j + u_{i,t}, \qquad i=1,\ldots,N,\qquad t=1,\ldots,T,
\label{eq:gcp_regmodel}
\end{equation}
where $\theta\in \Theta(\varepsilon)$ with $\Theta(\varepsilon)=\left\{\theta\in [\theta_L,\theta_U]:~ |\theta|>\epsilon,\, |\theta-1|>\epsilon \right\}$ and $-1<\theta_L\leq \theta_U<\infty$. Alternatively, we write $y_{i,t}	= \tau_g t^\theta + \vz_{i,t}\tran \vbeta_i + u_{i,t}$, where $\vz_{i,t}=\big[1,t,x_{i,t},\ldots,x_{i,t}^{p_i}\big]\tran$ and  $\vbeta_i=\big[\tau_{1,i},\tau_{2,i},\phi_{1,i},\ldots,\phi_{p_i,i}\big]\tran$. Finally, we stack all $N$ equations in \eqref{eq:gcp_regmodel} in matrix form to retrieve
\begin{equation}
\vy_t= \tau_{g} t^{\theta} \vones_N+ \mZ_t'\vbeta + \vu_t, \qquad t=1,\ldots,T,
\label{eq:matrixstack}
\end{equation}
with $\vy_t=\big[y_{1,t},\ldots,y_{N,t}\big]\tran$, $\mZ_t=\diag\big[\vz_{1,t},\ldots,\vz_{N,t}\big]$, and the vector $\vbeta=\big[\vbeta_1\tran,\ldots,\vbeta_N\tran\big]\tran$ of length $p=2N+\sum_{i=1}^{N}p_i$ containing all local parameters.

We consider nonlinear least squares (NLS) estimators of the unknown parameters in \eqref{eq:gcp_regmodel}. As such, we define the objective function $Q_T(\theta,\tau_g,\vbeta)= \frac{1}{2} \sum_{t=1}^{T} \big\|  \vy_t - \tau_g t^\theta \vones_N - \mZ_t\tran\vbeta \big\|^2$ and compute
\begin{equation}
 \left(\,\widehat{\theta}_T, \widehat\tau_{g,T},\widehat\vbeta_T \right) = \argmin_{(\theta,\,\tau_g,\,\vbeta)\in \Theta(\varepsilon) \times\SR \times\SR^p} Q_T(\theta,\tau,\vbeta).
\label{eq:NLSestimation}
\end{equation}
The optimization problem in \eqref{eq:NLSestimation} is easy to solve. For any given $\theta$, model \eqref{eq:matrixstack} is linear-in-parameters and the minimizers for $\tau_g$ and $\vbeta$ can be found from an OLS regression by constructing $\mZ_t\tran(\theta)= \big[ t^\theta\vones_N \;\mZ_t\tran\big]$ and computing
$
\left[
\begin{smallmatrix}
  \tau_g(\theta) \\
  \vbeta(\theta)
 \end{smallmatrix}
\right]
 =
  \left( \sumtT \mZ_t^{}(\theta)\mZ_t\tran(\theta)   \right)^{-1} \left( \sumtT \mZ_t(\theta) \vy_t \right)
$. We subsequently minimize the concentrated criterion function $\widetilde{Q}_T(\theta)=Q_T\big(\theta,\tau_g(\theta),\vbeta(\theta)\big)$ to obtain $\widehat{\theta}_T$. At last, we plug in $\widehat{\theta}_T$ and recover $\widehat\tau_{g,T}$ and $\widehat\vbeta_T$ through a final OLS estimation.

\begin{remark}
Keeping the powers of $x_{i,t}$ fixed allows us to test for their significance and thereby distinguish between nonlinearities caused by deterministic and stochastic trends. This is important for our empirical application on the Environmental Kuznets Curve, see Section \ref{sec:empapplication}. \cite{huphillipswang2019} study a model with a flexible power of the integrated regressor. That is, these authors derive the limiting distribution of the NLS estimators for $\beta$ and $\gamma$ when $y_t=\beta |x_t|^\gamma+u_t$ with $\beta\neq 0$.
\end{remark}

\begin{remark}
The GCPR of \eqref{eq:gcp_regmodel} can be extended in several directions. First, integer powers of deterministic trends can be added as long as $\Theta(\varepsilon)$ is adjusted accordingly (to avoid collinearity). Second, multiple explanatory variables can be included. Related to the EKC, the literature suggests examples such as: population density (\cite{seldensong1994}), trade openness (\cite{jalilferdiun2011}), energy prices (\cite{almulaliozturk2016}), and educational level (\cite{maranzanobentomanera2021}). For nonstationary variables, conditions similar to those on $\{x_{i,t}\}$ should be fulfilled (Assumption \ref{assumption:innovations} below). Stationary variables should be strictly exogenous. Avoiding the elaborate notation which would otherwise arise, we focus on the baseline specification in \eqref{eq:gcp_regmodel}.
\end{remark}

\section{Asymptotic Theory} \label{sec:AsympTheory}
We subsequently study the asymptotic properties of the NLS estimators. To this end we first collect all the unknown parameters in the vector $\vgamma= \big[\theta,\tau_g,\vbeta\tran \big]\tran$. This vector is assumed to be an element of the parameter space $\paraspace = \Theta(\varepsilon)\times\SR^{1+p}$. The true parameter vector is $\vgamma_0= \big[\theta_0,\tau_{g,0},\vbeta_0\tran \big]\tran$.

\begin{assumption}
The  global trend is relevant, i.e. $\tau_{g,0}\neq 0$.
\label{assumpt:identifytheta}
\end{assumption}

\begin{assumption}
Let $\vzeta_t=[\eta_t\tran,\vepsi_t\tran]\tran$ be a sequence of i.i.d. random vectors with $\E(\vzeta_t)=\vzeros$, $\mSigma=\E\big(\vzeta_t^{}\vzeta_t\tran)\succ 0$, and $\E\left\|\vzeta_t\right\|^{2q}<\infty$ for some $q>2$.
\begin{enumerate}[(a)]
 \item $u_t = \sum_{k=0}^\infty \psi_k \eta_{t-k}$ with $\sum_{k=1}^\infty k | \psi_k | < \infty$.
 \item $\vx_t = \sum_{s=1}^t \vv_s$, where $\vv_t = \sum_{k=0}^\infty \mPsi_k \vepsi_{t-k}$ with $\sum_{k=0}^\infty \|\mPsi_k\| < \infty$ and $\det\left( \sum_{k=0}^\infty \mPsi_k \right) \neq 0$.
\end{enumerate}
\label{assumption:innovations}
\end{assumption}

The first assumption is needed to avoid identification issues. That is, if $\tau_{g,0}=0$, then $\theta$ is not identified and the Davies problem arrises when testing $H_0: \tau_i = 0$ (see \cite{davies1977,davies1987}). Such complications are not investigated here and this is further reflected in our model specification \eqref{eq:gcp_regmodel}. That is, we consider \emph{flexible} powers of the deterministic trends but \emph{fixed} powers of the stochastic trends, hence allowing us to test zero restrictions on (elements of) $\vbeta$. This is of crucial importance in the EKC application while determining whether nonlinear effects in the economic growth variables ($x_{i,t}$) remain significant after nonlinear time trends have been added to the model. Assumption \ref{assumpt:identifytheta} has been relaxed in the literature albeit for different models. \cite{baekchophillips2015} and \cite{chophillips2018} study the asymptotic behaviour of a quasi-likelihood ratio test when Assumption \ref{assumpt:identifytheta} is violated and the conditional mean of the data contains strictly stationary regressors and a flexible time trend. Alternatively, one can use drifting parameter sequences with different identification strengths as in \cite{andrewscheng2012}.

Assumption \ref{assumption:innovations} excludes cointegration among elements of $\vx_t$ and defines this vector as the partial sum of a short memory process. The latter implies that
$
\frac{1}{\sqrt{T}}\sum_{t=1}^{[rT]}
\left[
\begin{smallmatrix}
\vu_t\\
\vv_t
\end{smallmatrix}
\right]
\wto \brown(r)=
\left[
\begin{smallmatrix}
\brown_u(r)\\
\brown_v(r)
\end{smallmatrix}
\right]
$
where $\brown(r)$ denotes an $2N$-dimensional vector Brownian motion with covariance matrix $\mOmega =
\left[\begin{smallmatrix}
\mOmega_{uu} & \mOmega_{uv}\\
\mOmega_{vu} & \mOmega_{vv}
\end{smallmatrix}\right]$. The one-sided long-run covariance matrix $
\mDelta=
 \sum_{h=0}^\infty \E\left(
 \left[\begin{smallmatrix}
  \vu_t \vu_{t+h}		& \vu_t \vv_{t+h}\tran \\
  \vv_t \vu_{t+h}		& \vv_t^{}\vv_{t+h}\tran
 \end{smallmatrix}
 \right]
  \right)
=
\left[\begin{smallmatrix}
\mDelta_{uu} & \mDelta_{uv}\\
\mDelta_{vu} & \mDelta_{vv}
\end{smallmatrix}\right]$ is partitioned similarly. Subscripts refer to specific elements. For example, $\brown_{v_i}$ and $\mDelta_{v_i u_j}$ denote the $i$\textsuperscript{th} and $(i,j)$\textsuperscript{th} elements of $\brown_v$ and $\mDelta_{vu}$, respectively.

A concise exposition of our results asks for additional notation. An enumeration of various definitions is presented below.
\begin{enumerate}[(1)]
\item Introduce $\mD_{(i),T}=\diag\big[1,T,T^{1/2},T,\ldots,T^{p_i/2}\big]$ to scale the deterministic and stochastic trends within each equation. For the full system of equation, define $\mD_{Z,T}=\diag\left[\mD_{(1),T},\ldots,\mD_{(N),T}\right]$, $\mD_{\theta_{0},T}=\sqrt{T}\left[
\begin{smallmatrix}
T^{\theta_{0}}   & & \\
& T^{\theta_{0}} & \\
&                & \mD_{Z,T}
\end{smallmatrix}\right]$
and $\mL_{\tau_{g,0},T}=\left[
\begin{smallmatrix}
1 & -\tau_{g,0}\ln{T} & \\
0 & 1                & \\
  &                  & \mI_{p}
\end{smallmatrix}\right]$. Finally, set $\mG_{\vgamma_0,T}=\mD_{\theta_0,T}^{}\mL_{\tau_{g,0},T}^{\prime-1}$.
\item Define $\vj_i(r)=\left[1,r,\brownnormal_{v_i}(r),\brownnormal_{v_i}^2(r),\ldots,\brownnormal_{v_i}^{p_i}(r)\right]\tran$, $\mJ_Z(r)=\diag\left[\vj_1(r),\ldots,\vj_N(r)\right]$, and $\mJ(r;\vgamma_0)=\Big[\tau_{g,0} r^{\theta_0}\ln{r}\,\vones_N, r^{\theta_0}\vones_N, \mJ_Z\tran(r)\Big]\tran$.
\item For the second-order bias terms, we define $\vb_i= \Big[\vzeros_{1\times 2},1,2\int \brown_{v_i}(r)dr,\ldots,p_i\int \brown_{v_i}^{p_i-1}(r)dr \Big]\tran$ and $\serialbias=\big[\vzeros_{1\times 2},\vb_1\tran \mDelta_{v_1 u_1} ,\dots,\vb_N \tran \mDelta_{v_N u_N}\big]\tran$. 
\end{enumerate}

\begin{theorem}\label{thm:limiting_dist}
Under Assumptions \ref{assumpt:identifytheta}-\ref{assumption:innovations}, we have
\begin{equation*}
\mG_{\vgamma_0,T}\big(\widehat{\vgamma}_T-\vgamma_0\big)
	\wto \left(\myint \mJ(r;\vgamma_0)\mJ\tran(r;\vgamma_0)~dr\right)^{-1}\left(\myint \mJ(r;\vgamma_0)~d\brown_{u}(r)+\vbias_{vu}\right) =: \limitdist(\vgamma_0),
\label{eq:limiting_distribution}
\end{equation*}
as $T\rightarrow\infty$ and $N$ fixed.
\end{theorem}

The proof of Theorem \ref{thm:limiting_dist} is closely related to the work by \cite{chanwang2015}. These authors provide the asymptotic distribution of NLS estimators under a set of general conditions in univariate, nonstationary time series models (see their theorem 3.1). The results in \cite{chanwang2015} and \cite{wangwuzhu2018} suggest that Assumption \ref{assumption:innovations} can be replaced by a long memory specification for $\diff \vx_t$. However, long memory parameters will enter the limiting distribution and inference will be complicated further. We illustrate Theorem \ref{thm:limiting_dist} with two examples. These examples highlight the two mathematical features that complicate parameter inference.

\begin{example}\label{example:dettrend}
 We consider $y_t= \tau t^\theta+ u_t$ with innovations satisfying Assumption \ref{assumption:innovations}. The limiting distribution of the parameter estimators depends solely on the mean square Riemann-Stieltjes integrals $\myint \tau_0 r^{\theta_0} \ln(r)  dB_u$ and $\myint r^{\theta_0} dB_u$, and is therefore normally distributed (e.g., section 2.3 in \cite{tanaka2017}). We have
 \begin{equation}
   \begin{bmatrix}
   T^{\theta_0+\frac{1}{2}}				& 0\\
   T^{\theta_0+\frac{1}{2}} \tau_0 \ln(T)	& T^{\theta_0+\frac{1}{2}}
 \end{bmatrix}
 \begin{bmatrix}
  \,\widehat{\theta}_T -\theta_0 \\
  \,\widehat{\tau}_T - \tau_0
 \end{bmatrix}
 \wto \rN
 \left( \vzeros,  \Omega_{uu}(2\theta_0+1)^3
 \begin{bmatrix}
 2 \tau_0^2					&  - \tau_0 (2\theta_0+1)  \\
 - \tau_0 (2\theta_0+1)			&  (2\theta_0+1)^2
\end{bmatrix}^{-1} \right).
 \label{eq:powerlawtrend}
 \end{equation}
The scaling matrix in the LHS of \eqref{eq:powerlawtrend},
 $\left[
    \begin{smallmatrix}
   T^{\theta_0+\frac{1}{2}}				& 0\\
   T^{\theta_0+\frac{1}{2}} \tau_0 \ln(T)	& T^{\theta_0+\frac{1}{2}}
 \end{smallmatrix}\right]$, depends on $\theta_0$ and is non-diagonal. The dependence on $\theta_0$ is unavoidable but asymptotic results for the case of a diagonal scaling matrix are obtainable at the expense of a singular joint distribution. That is, noting that
$
\left[ \begin{smallmatrix}
   T^{\theta_0+\frac{1}{2}}			& 0\\
   0							& T^{\theta_0+\frac{1}{2}}/\ln(T)
 \end{smallmatrix} \right]
=
\left[ \begin{smallmatrix}
 1		& 0 \\
 -\tau_0	& 1/\ln(T)
\end{smallmatrix} \right]
\left[\begin{smallmatrix}
   T^{\theta_0+\frac{1}{2}}				& 0\\
   T^{\theta_0+\frac{1}{2}} \tau_0 \ln(T)	& T^{\theta_0+\frac{1}{2}}
 \end{smallmatrix} \right]
$ and since $\lim_{T\to \infty}
\left[ \begin{smallmatrix}
 1		& 0 \\
 -\tau_0	& 1/\ln(T)
\end{smallmatrix} \right]
=
\left[ \begin{smallmatrix}
 1		& 0 \\
 -\tau_0	& 0
\end{smallmatrix} \right]
$, the continuous mapping theorem implies
$
\left[\begin{smallmatrix}
   T^{\theta_0+\frac{1}{2}}			& 0\\
   0							& T^{\theta_0+\frac{1}{2}}/\ln(T)
 \end{smallmatrix} \right]
\left[ \begin{smallmatrix} \widehat{\theta}_T - \theta_0  \\ \widehat{\tau}_T - \tau_0 \end{smallmatrix} \right]
\dto
\left[\begin{smallmatrix}
 1/\tau_0 \\
 -1
\end{smallmatrix} \right]
\times
\rN
 \left( \vzeros,  \Omega_{uu}(2\theta_0+1)^3 \right).
 $
 This limiting distribution coincides with the result in theorem 6.3 of \cite{phillips2007}.
\end{example}

\begin{example}\label{example:stochtrend}
If $y_t = \tau t^\theta+ \phi x_t + u_t$, then the limiting distribution of the NLS estimator is:
 \begin{equation*}
\resizebox{1\hsize}{!}{$
 \begin{bmatrix}
   T^{\theta_0+\frac{1}{2}}  \\
   T^{\theta_0+\frac{1}{2}} \tau_0 \ln(T)  & T^{\theta_0+\frac{1}{2}}			\\
			&		& T
 \end{bmatrix}
 \begin{bmatrix}
  \widehat{\theta}_T -\theta_0 \\
  \widehat{\tau}_T - \tau_0 \\
  \widehat{\phi}_T - \phi_0
 \end{bmatrix}
 \wto
\begin{bmatrix}
\int \big(\tau_0 r^{\theta_0}\ln (r)\big)^2 dr 	& \int \tau_0 r^{2\theta_0}{\ln (r)} dr		& \int  \tau_0 r^{\theta_0} \ln (r) B_v dr \\
\int \tau_0 r^{2\theta_0}{\ln (r)} dr			& \int r^{2\theta_0} dr				& \int r^{\theta_0} B_v dr \\
\int \tau_0 r^{\theta_0} \ln (r) B_v dr		& \int r^{\theta_0} B_v dr					& \int B_v^2 dr 
\end{bmatrix}^{-1} \times \\
\left(
\begin{bmatrix}
\int \tau_0 r^{\theta_0} \ln(r)  dB_u\\
\int r^{\theta_0} dB_u\\
\int B_v dB_u
\end{bmatrix}
 +
 \begin{bmatrix}
 0 \\
 0 \\
 \Delta_{vu}
\end{bmatrix}
 \right).
 $}
 \end{equation*}
 This limiting distribution exhibits second order bias when $ \Delta_{vu}\neq 0$, or when $B_u$ and $B_v$ are correlated.
\end{example}

Two features of the limiting distribution of $\mG_{\vgamma_0,T}\big(\widehat{\vgamma}_T-\vgamma_0\big)$ deserve further comments. First, as emphasised in Examples \ref{example:dettrend}--\ref{example:stochtrend}, the scaling matrix $\mG_{\vgamma_0,T}$ features two less common properties: (1) this matrix depends on the true parameters $\vtau_{g,0}$ and $\theta_0$, and (2) $\mG_{\vgamma_0,T}$ is not diagonal. These peculiarities are caused by the nonlinearity and nonstationarity of the model. More specifically, these features can be traced back to the presence of the global trend. Limiting distributions with a similar mathematical structure can be found in the structural breaks literature, cf. model setting II.b of \cite{perronzhu2005} and its detailed analysis in \cite{beutnerlinsmeekes2019}.

Second, the nonstationary regressor $x_{i,t}$ enters the model \eqref{eq:gcp_regmodel} through a polynomial transformation of the form $g(x_{i,t},\vphi_i)= \phi_{i,1} x_{i,t} + \ldots + \phi_{i ,p_i} x_{i,t}^{p_i}$ ($i=1,2,\ldots,N$). In the terminology of \cite{parkphillips2001} this part of the regression function is a linear combination of $H_0$-regular functions. It is well-documented in the literature, e.g. \cite{changparkphillips2001} and \cite{chanwang2015}, that this leads to second-order bias terms and hence nonstandard inference (except for the special case of strictly exogenous nonstationary regressors).

\subsection{Consistent Long-Run Covariance Matrix Estimation}
Correcting for second-order bias terms typically involves estimating long-run variance (LRV) matrices. This subsection establishes that the NLS residuals can be used to construct consistent kernel estimators for the LRV matrices $\mDelta$ and $\mOmega$. Defining $\bm V_t(\vgamma)=\big[\vu_t(\vgamma)', \diff \vx_t\tran]\tran$ with $\vu_t(\vgamma)=\vy_t - \tau_{g} t^{\theta} \viota_N - \mZ_t'\vbeta $, these LRV estimators are defined as 
\begin{equation}
\widehat{\mDelta}_T = \frac{1}{T} \sum_{t=1}^T \sum_{s=1}^t k\left(\frac{|t-s|}{b_T} \right)  \bm V_t(\,\widehat{\vgamma}_T) \bm V_t(\,\widehat{\vgamma}_T)\tran,
\qquad
\widehat{\mOmega}_T = \frac{1}{T} \sum_{t=1}^T \sum_{s=1}^T k\left(\frac{|t-s|}{b_T} \right)  \bm V_t(\,\widehat{\vgamma}_T) \bm V_t(\,\widehat{\vgamma}_T)\tran,
\label{eq:LRVdefinitions}
\end{equation}
for some kernel function $k(\cdot)$ and bandwidth parameter $b_T$. The first $N$ elements of $\bm V_t(\,\widehat{\vgamma}_T)$ are the elements of the residual vector $\widehat{\vu}_t =\vy_t -\widehat{\tau}_{g,T}\,t^{\widehat{\theta}_T}\vones_N-\mZ_t'\widehat{\vbeta}_T^{}$. The remaining elements are $\diff \vx_t=\vv_t$.

\begin{assumption}
\begin{enumerate}[(a)] 
 \item[]
 \item $k(0)=1$, $k(\cdot)$ is continuous at zero, and $\sup_{x\geq 0}\left| k(x) \right| < \infty$.
 \item $\int_0^\infty \bar{k}(x)dx<\infty$, where $\bar{k}(x)= \sup_{y\geq x} \left| k(y) \right|$.
 \item The bandwidth parameters $\{b_T : T\geq 1 \}$ satisfies $\{b_T\}\subseteq (0,\infty)$ and $\lim_{T\to\infty} \left( b_T^{-1} + T^{-1/2} b_T  \ln T  \right) = 0$.
\end{enumerate}
\label{assumption:LRVestimation}
\end{assumption}

The conditions on the kernel function $k(\cdot)$, Assumptions \ref{assumption:LRVestimation}(a)--(b), are identical to those in \cite{jansson2002}. \cite{jansson2002} remarks that these assumptions ``\emph{would appear to be satisfied by any kernel in actual use}''. Commonly used kernel functions such as the Bartlett, Parzen, and Quadratic Spectral kernels indeed satisfy all these assumptions. Assumption \ref{assumption:LRVestimation}(c) differs from the usual requirement, $\lim_{T\to\infty} \left( b_T^{-1} + T^{-1/2} b_T \right) = 0$, by a factor $\ln T$. The difference is caused by the estimation error in $\widehat{\theta}_T$. This error causes the residuals $\{ \widehat{\vu}_t  \}$ to be less close to the innovations $\{ \vu_t\}$ and we balance this by including autocovariance matrices of higher lags at a slower pace.

\begin{theorem} \label{thm:consistentLRVs}
Under Assumptions \ref{assumpt:identifytheta}-\ref{assumption:LRVestimation}, we have $\widehat{\mDelta}_T\pto \mDelta$ and $\widehat{\mOmega}_T \pto \mOmega$.
\end{theorem}

\subsection{Simulation-Based Inference} \label{subsec:sim_inf}
The limiting distribution in Theorem \ref{thm:limiting_dist} is nonpivotal and thus not directly suited for inference. Some popular solutions for linear-in-parameters cointegration models are: \cites{saikkonen1992} dynamic least squares, the integrated modified OLS and fixed-$b$ approaches by \cite{vogelsangwagner2014}, and the fully modified approach advocated in \cite{phillipshansen1990} and \cite{phillips1995}. For a nonlinear-in-parameter model as in \eqref{eq:gcp_regmodel}, a preliminary Monte Carlo exercise\footnote{The details are available in Section \ref{sec:FMOLS} of the Supplementary Material. The analytical results in that appendix also suggest that the convergence speed of $\widehat{\theta}_T$ to $\theta$ is too slow to recover the standard zero-mean Gaussian limiting distribution.} shows poor performance for fully modified inference but promising results for a simulation based approach. We pursue the latter method for the remainder of this paper.

The main idea behind the simulation based approach is to replace nuisance parameters by consistent estimates and to rely on Monte Carlo (MC) simulations to approximate the limiting distribution. The empirical quantiles of these MC draws allow us to conduct inference. Clearly, this kind of approach will provide exact inference when the limiting distribution is invariant with respect to the nuisance parameters (e.g. \cite{dufourkhalaf2002} and \cite{dufour2006}). In the absence of such invariance, \cite{wangwuzhu2018} and \cite{bergamellibianchikhalafurga2019} show that the simulation approach remains asymptotically justified in several model specification. We adapt the algorithm from \cite{wangwuzhu2018} to the current setting and prove its asymptotic validity.

\begin{algorithm}[Simulation-Based Inference] \
\vspace{-1.2 cm}
\begin{enumerate}[\textsc{Step} 1:]
 \item Estimate $\widehat{\vgamma}_T$ and use the residuals $\{ \widehat{\vu}_t  \}$ to compute the estimators $\widehat{\mDelta}_T$ and $\widehat{\mOmega}_T$ from \eqref{eq:LRVdefinitions}.
 \item Repeat for $j=1,\ldots,J$,
 \end{enumerate}
 \begin{enumerate}[(a)]
  \item Draw random variables $\{ \ve_t \}_{t=1}^{M_T}$ i.i.d. from $\rN(\vzeros,\mI_{2N})$.
  \item Compute $\left[ \begin{smallmatrix} \widehat{\vmu}_t \\ \widehat{\vv}_t \end{smallmatrix} \right]= \widehat{\mOmega}_T^{1/2} \ve_t $ and the partial sum $\widehat{\vchi}_t = \big[\widehat{\chi}_{1,t},\ldots,\widehat{\chi}_{N,t} \big]\tran= \sum_{s=1}^t \widehat{\vupsilon}_s$.
  \item Let $\widehat{\mJ}\big(t;\widehat{\vgamma}_{T}\big)=\Big[\,\widehat{\tau}_{g,T}\,t^{\widehat{\theta}_T}\ln{t}\,\vones_N,t^{\widehat{\theta}_T}\vones_N,\widehat{\calZ}_t'\Big]'$, where $\widehat{\calZ}_t=\diag\big[\widehat{\vz}_{1,t},\ldots,\widehat{\vz}_{N,t}\big]$ with $\widehat{\vz}_{i,t}=\Big[1,t,\widehat{\chi}_{i,t},\ldots,\widehat{\chi}_{i,t}^{\,p_i}\Big]'$ (for $i=1,\ldots,N$). For a given $M_T$, construct the $j$\textsuperscript{th} simulated draw as
		\begin{equation*} 
		\widehat{\limitdist}^{(j)} \left(\widehat{\vgamma}_T,\widehat{\mOmega}_T,\widehat{\mDelta}_{vu}^-\right)
		=\left\{\text{\footnotesize $
			\mG_{\widehat{\vgamma}_T,M_T}^{\prime-1}\left[\sum_{t=1}^{M_T}\widehat{\mJ}\big(t;\widehat{\vgamma}_{T}\big)\widehat{\mJ}\big(t;\widehat{\vgamma}_{T}\big)'\right]\mG_{\widehat{\vgamma}_T,M_T}^{-1}
			$}
		\right\}^{-1}
		\left\{\text{\footnotesize $
			\mG_{\widehat{\vgamma}_T,M_T}^{\prime-1}\left[\sum_{t=1}^{M_T}\widehat{\mJ}\big(t;\widehat{\vgamma}_{T}\big)\,\widehat{\vmu}_t\right]+\serialbiasSimulated
			$}
		\right\},
		\end{equation*}
		where $\widehat{\mDelta}_{vu}^-$ is a consistent estimator of the lower-left subblock of $\mDelta^{-}=\left[\begin{smallmatrix}
		\mDelta_{uu}^{-} & \mDelta_{uv}^{-}\\
		\mDelta_{vu}^{-} & \mDelta_{vv}^{-}
		\end{smallmatrix}\right]=\mSigma-\mDelta\tran$, and $\serialbiasSimulated=\left[\vzeros_{1\times 2},\widehat{\vb}_{1}\tran \widehat{\mDelta}_{v_1u_1}^{-},\dots,\widehat{\vb}_{N}\tran\widehat{\mDelta}_{v_Nu_N}^{-}\right]\tran$ with $\widehat{\vb}_{i}=\Bigg[\vzeros_{1\times 2},1,2\frac{1}{M_T}\sum_{t=1}^{M_T}\Big(\frac{\widehat{\chi}_{i,t}}{\sqrt{M_T}}\Big),\ldots,p_i\frac{1}{M_T}\sum_{t=1}^{M_T}\Big(\frac{\widehat{\chi}_{i,t}}{\sqrt{M_T}}\Big)^{p_i-1}\Bigg]\tran$.
	\end{enumerate}
\begin{enumerate}[\textsc{Step} 1:]
  \setcounter{enumi}{2}
 \item Use the empirical quantiles of elements of $\left\{\widehat{\limitdist}^{(1)},\ldots, \widehat{\limitdist}^{(J)} \right\}$ to conduct inference.
\end{enumerate}
\end{algorithm}

Algorithm 1 uses a discretisation in $M_T$ steps to approximate the limiting distribution of the parameters. In practice, and in accordance with Theorem \ref{thm:subsampling}, we can take $M_T=T$. Remark \ref{remark:simulationbasedtest} details how simulation-based inference can be used to test hypotheses concerning the model's parameters. Discussions on size and power are also presented there.

\begin{theorem} \label{thm:subsampling}
 Suppose Assumptions \ref{assumpt:identifytheta}-\ref{assumption:LRVestimation} hold, let $\{M_T\}\subseteq (0,\infty)$ with $\lim_{T\to\infty} \frac{M_T}{T} \leq \kappa$, for some $\kappa<\infty$, then we have
\begin{equation}
\begin{aligned}
&\left\{\mG_{\widehat{\vgamma}_T,T}^{\prime-1}\left[\sum_{t=1}^{M_T}\widehat{\mJ}\big(t;\widehat{\vgamma}_{T}\big)\widehat{\mJ}\big(t;\widehat{\vgamma}_{T}\big)'\right]\mG_{\widehat{\vgamma}_T,T}^{-1}
	\right\}^{-1}
	\left\{
		\mG_{\widehat{\vgamma}_T,T}^{\prime-1}\left[\sum_{t=1}^{M_T}\widehat{\mJ}\big(t;\widehat{\vgamma}_{T}\big)\,\widehat{\vmu}_t\right]+\serialbiasSimulated	
	\right\}  \\
	&\qquad\qquad\qquad\qquad\dtostar
	\left(\int  \mJ(r;\vgamma_0)\mJ(r;\vgamma_0)' ~dr\right)^{-1}\left(\int  \mJ(r;\vgamma_0)~d\brown_{u}(r)+\vbias_{vu}\right),
\end{aligned}
\end{equation}
 in probability.
\end{theorem}
Theorem \ref{thm:subsampling} establishes the asymptotic validity of the simulation approach. That is, for a large enough $J$, the empirical quantiles of the simulated distribution will coincide with the asymptotic distribution. Two remarks are important. First, even though the simulation algorithm is adapted from \cite{wangwuzhu2018}, the proof of Theorem \ref{thm:subsampling} is not. In particular, the method of proof is similar to Theorem \ref{thm:limiting_dist} and continues to allow for endogeneity of the regressors. Second, the simulation approach mimics the stochastic integrals in the limiting distribution directly. It therefore suffices to draw normally distributed random variables in Step 2(a) and use consistent long-run covariance estimates to replicate the covariance structure of the underlying Brownian motions. Compared to a bootstrap procedure, this simulation approach has the advantage of avoiding tedious NLS re-estimation on bootstrap samples but it comes at the cost of forsaking possible asymptotic refinements.

\begin{remark} \label{remark:simulationbasedtest}
Step 3 in Algorithm 1 has been kept general for notational convenience. An illustrative example is as follows. Assume we are interested in $H_0:\phi_{2,1}=0$ (irrelevance of the regressor $x_{1,t}^2$) when
$$
 y_{i,t} = \tau_g t^{\theta} + \tau_{1,i} + \tau_{2t,i} t + \phi_{1,i} x_{i,t} + \phi_{2,i} x_{i,t}^2 + u_{i,t}, \qquad i=1,\ldots,N,\quad t=1,\ldots,T.
$$
Under $H_0$, we have $T^{3/2}\widehat{\phi}_{2,1} = \ve_6\tran \mG_{\vgamma_0,T}\big(\widehat{\vgamma}_T-\vgamma_0\big) \wto \ve_6\tran \limitdist(\vgamma_0)$ with $\ve_k$ being the k\textsuperscript{th} basis vector in $\SR^{2+p}$. Denoting the empirical $\zeta$-quantiles of $\left\{\ve_6\tran\widehat{\limitdist}^{(1)},\ldots, \ve_6\tran\widehat{\limitdist}^{(J)} \right\}$ by $c_\zeta$, a test of size $\alpha$ will reject for $T^{3/2}\widehat{\phi}_{2,1}<c_{\alpha/2}$ or $T^{3/2}\widehat{\phi}_{2,1}>c_{1-\alpha/2}$. Under the alternative $\phi_{2,1} \neq 0$, we rewrite the test statistic as $T^{3/2}\widehat{\phi}_{2,1}=T^{3/2} \big(\, \widehat{\phi}_{2,1} - \phi_{2,1}\big) + T^{3/2} \phi_{2,1}$. Statistical power is guaranteed because the simulation approach mimics the asymptotic distribution and is thus bounded, whereas the second term diverges.
\end{remark}

\subsection{KPSS-Type Test for the Null of Cointegration} \label{subsec:kpss_test}
The correct specification of the nonlinear cointegrating relation will result in a stationary error process $\{u_t\}_{t\in\SZ}$. We consider a KPSS-type test statistic for the null of stationarity. The candidate statistic is $\widetilde K_T^+ = \frac{1}{T^2}\sum_{t=1}^T \left\|\widehat{\mOmega}_{u.v}^{-1/2}\sum_{i=\ell}^{t}\widehat{\vu}_i^+\right\|^2$, where $\widehat{\vu}_t^+=\vy_t - \widehat{\mOmega}_{uv}\widehat{\mOmega}_{vv}^{-1} \diff \vx_t -\widehat{\tau}_{g,T}\,t^{\widehat{\theta}_T}\vones_N-\mZ_t'\widehat{\vbeta}_T$ and $\widehat{\mOmega}_{u.v}$ is a consistent estimator of $\mOmega_{u.v}=\mOmega_{uu}-\mOmega_{uv}\mOmega_{vv}^{-1} \mOmega_{vu}$. This statistic is stochastically bounded under the null hypothesis but diverges under the alternative. Rejections of the null hypothesis are an indication of a spurious relationship and/or an incorrect functional form of the nonlinear cointegrating relationship. Several authors have reported model settings in which the asymptotic null distribution of $K_T^+$ is known, e.g. \cite{kwiatkowskiphillipsschmidtshin1992} and \cite{wagnerhong2016}.

The estimation of $\vtheta$ contaminates the limiting distribution of $\widetilde K_T^+$ with nuisance parameters.\footnote{Proposition 5 in \cite{wagnerhong2016} shows that the limiting distribution of $K_T^+$ is free of nuisance parameters if $\vtheta_0$ is known and only a single integrated regressor occurs with integer powers greater than one. This result does not carry over to the current setting because of the estimation error in $\widehat{\vtheta}_T$.} \cite{choisaikkonen2010}, \cite{wagnerhong2016}, \cite{jianglupark2019}, and \cite{linreuvers2019}, have shown that subsampling can resolve this issue. We will follow their approach and use subsamples of size $q_T$ to compute the test statistics.

\begin{theorem}\label{thm:KPSS_plus}
	Under Assumptions \ref{assumpt:identifytheta}-\ref{assumption:LRVestimation} and if $\lim_{T\to\infty} \left( q_T^{-1} +  (\ln T) \left(\frac{q_T}{T} \right)^{\theta_L+\frac{1}{2}}  \right) = 0$, then for any $\ell\in\{1,\ldots,T-q_T+1 \}$, we have
	\begin{equation}
	K_{q_T,\ell}^{+}=\frac{1}{q_T}\sum_{t=\ell}^{\ell+q_T-1}\left\|\frac{1}{\sqrt{q_T}}\widehat{\mOmega}_{u.v}^{-1/2}\sum_{i=\ell}^{t}\widehat{\vu}_i^+\right\|^2\wto \int \left\| \mW(r) \right\|^2dr,
	\end{equation}
	where $\bm W(\cdot)$ denotes an $N$-dimensional standard Brownian motion.
\end{theorem}

Theorem \ref{thm:KPSS_plus} does not provide any guidance on the choices for the starting value $\ell$ and the subsample size $q_T$. First, for a given $q_T$, \cite{choisaikkonen2010} argue that the use of a single subsample (instead of all $T$ observations) implies a significant loss of power. We follow their example and combine all $M=[T/q_T]$ subresidual series of length $q_T$ using a Bonferroni procedure. That is, we create subresiduals series by selecting adjacent blocks of $q_T$ residuals while alternating between the start and end of the sample. We calculate the KPSS-type test statistic for each subseries, say $K_1,\ldots, K_M$, and reject the null of stationarity at significance $\alpha$ whenever $\max\{K_1,\ldots,K_M \}$ exceeds $c_{\alpha/M}$ which is defined by $\Prob\left( \myint \big\| \bm W(r) \big\|^2 dr  \geq c_{\alpha/M} \right) = \alpha/M$ . Finally, we select the block size $q_T$ using \cites{romanowolf2001} minimum volatility rule. The approach is now completely data-driven.

\section{Simulations} \label{sec:simulations}
This section lists various Monte Carlo simulations showing that the asymptotic approximations from Section \ref{sec:AsympTheory} provide useful guidance in finite samples. Further details on the implementation are as follows. The long-run covariance matrices in \eqref{eq:LRVdefinitions} are computed using the Barlett kernel, $k(x)= 1 -|x|$ for $|x|\leq 1$ (and zero otherwise), and the bandwidth selection method described in \cite{andrews1991}. Simulated limiting distributions are based on $J=299$ replicates and we set $M_T=T$. We test at 5\% significance and report results based on $3,000$ Monte Carlo replications. Two data generating processes are studied: DGP1 and DGP2.

\subsubsection*{DGP1: Empirical size and power of the coefficient tests}
This DGP is inspired by the simulation study in \cite{wagnergrabarczykhong2019}. It augments their quadratic seemingly unrelated cointegrating polynomial regression model with a global flexible trend. That is, we consider
\begin{equation}
 y_{i,t} = \tau_g t^\theta + \tau_{1,i} + \tau_{2,i} t + \phi_{1,i} x_{i,t} + \phi_{2,i} x_{i,t}^2 + u_{i,t}, \qquad i=1,\ldots,N,\qquad t=1,\ldots,T,
\label{eq:SimDGP1}
\end{equation}
and compute $u_{i,t}$ and $\diff x_{i,t}=v_{i,t}$ recursively as
\begin{equation*}
\label{eq:errorrecursion}
u_{i,t} = \rho_1 u_{i,t-1} + \varepsilon_{i,t} + \rho_2 e_{i,t}, \qquad v_{i,t} = e_{i,t} + 0.5 e_{i,t-1}.
\end{equation*}
All recursions are initialized from zero, i.e. $x_{i,0}=u_{i,0}=e_{i,0}=0$ ($i=1,\ldots,N$). The innovations $\vepsi_t=[\varepsilon_{1,t},\ldots,\varepsilon_{N,t}]\tran$ and $\ve_t=[e_{1,t},\ldots,e_{N,t}]\tran$ are drawn independently as $\vepsi_t\stackrel{i.i.d.}{\sim}\rN(\vzeros,\mSigma_{\varepsilon\varepsilon})$ and $\ve_t\stackrel{i.i.d.}{\sim}\rN(\vzeros,\mSigma_{ee})$, where
$$
 \mSigma_{\varepsilon\varepsilon}
 =
 \renewcommand\arraystretch{0.8}
 \begin{bmatrix}
 1        & \rho_3   & \cdots & \rho_3\\
 \rho_3   &   1      & \cdots & \rho_3\\
 \vdots   & \vdots   & \ddots & \vdots\\ 
 \rho_3   & \rho_3   & \cdots & 1\\
 \end{bmatrix}
 ,\qquad\text{and}\qquad
  \mSigma_{ee}
 =
 \begin{bmatrix}
 1        & \rho_4   & \cdots & \rho_4\\
 \rho_4   &   1      & \cdots & \rho_4\\
 \vdots   & \vdots   & \ddots & \vdots\\ 
 \rho_4   & \rho_4   & \cdots & 1\\
 \end{bmatrix}.
$$
Regarding the global trend in \eqref{eq:SimDGP1}, we set $\tau_g=-0.2$ and consider $\theta\in\{0.8,1.3,1.8\}$. All other coefficient values are inspired by \cite{wagnergrabarczykhong2019}. That is, $\tau_{1,i}=1$, $\tau_{2,i}=1$ and $\phi_{1,i}=5$ are identical across equations.\footnote{This homogenous parametrisation is particularly convenient to study the impact of the cross-sectional dimension. That is, we can vary $N$ without having to provide additional parameter values. DGP2 is directly inspired by the empirical application and thus more realistic.} Also, we let $\rho_1=\rho_2=\rho_3=\rho_4$ and redefine these four parameters as $\rho$. We vary $\rho\in\{0,0.3,0.6,0.8 \}$, $N\in\{3,5,10\}$, and $T\in\{150,300,600\}$. In line with the typical EKC application, we test for the significance of $x_{i,t}^2$. We set $\phi_{2,i}=0$ for $i=1,\ldots,N$ and report the empirical size of the single equation test for $H_0: \phi_{2,1}=0$ and the joint test for $H_0: \phi_{2,1}=\ldots=\phi_{2,N}=0$.

\begin{table}[t]
	\centering
	\caption{The empirical size (in \%) of the single-equation tests $H_0:\phi_{2,1}=0$ and the joint test for $H_0: \phi_{2,1}=\ldots=\phi_{2,N}=0$ with $\phi_{2,i}$ denoting the coefficient in front of $x_{i,t}^2$. The Monte Carlo results are based on: simulated inference with $\theta$ estimated by NLS (SimNLS), simulated inference with known $\theta=1.3$ (SimNLS($\theta_0$)), and two Fully Modified estimators for systems developed by \cite{wagnergrabarczykhong2019} with known $\theta=1.3$ (FM-SOLS($\theta_0$) and FM-SUR($\theta_0$)).}
	\label{tbl:SizeWagnerEtAl1}
	\resizebox{\textwidth}{!}{%
		\begin{tabular}{crrrrlrrrrlrrrr}
			\toprule
			$\theta_0=1.3$ & \multicolumn{4}{c}{$N=3$} & \multicolumn{1}{c}{} & \multicolumn{4}{c}{$N=5$} & \multicolumn{1}{c}{} & \multicolumn{4}{c}{$N=10$} \\
			\midrule
			$\rho$ & \multicolumn{1}{c}{SimNLS} & \multicolumn{1}{c}{SimNLS($\theta_0$)} & \multicolumn{1}{c}{FM-SOLS($\theta_0$)} & \multicolumn{1}{c}{FM-SUR($\theta_0$)} &  & \multicolumn{1}{c}{SimNLS} & \multicolumn{1}{c}{SimNLS($\theta_0$)} & \multicolumn{1}{c}{FM-SOLS($\theta_0$)} & \multicolumn{1}{c}{FM-SUR($\theta_0$)} &  & \multicolumn{1}{c}{SimNLS} & \multicolumn{1}{c}{SimNLS($\theta_0$)} & \multicolumn{1}{c}{FM-SOLS($\theta_0$)} & \multicolumn{1}{c}{FM-SUR($\theta_0$)} \\
			\midrule
			\multicolumn{15}{l}{\textbf{Panel A:   Single-equation test}} \\
			\midrule
			\multicolumn{15}{l}{$T=150$} \\
			\midrule
			0 & 4.77 & 4.63 & 9.53 & 10.77 &  & 4.43 & 4.90 & 10.03 & 12.80 &  & 4.37 & 4.47 & 9.70 & 16.63 \\
			0.3 & 4.80 & 4.77 & 10.47 & 11.97 &  & 4.53 & 4.63 & 10.00 & 13.00 &  & 4.23 & 4.27 & 11.60 & 19.00 \\
			0.6 & 4.83 & 4.53 & 11.90 & 12.93 &  & 3.93 & 4.17 & 11.87 & 16.50 &  & 5.13 & 4.90 & 14.23 & 31.70 \\
			0.8 & 4.50 & 4.67 & 14.00 & 19.10 &  & 4.90 & 4.60 & 15.23 & 26.53 &  & 5.27 & 4.43 & 16.83 & 56.47 \\
			\midrule
			\multicolumn{15}{l}{$T=300$} \\
			\midrule
			0 & 4.43 & 4.10 & 8.03 & 8.63 &  & 4.43 & 4.20 & 7.33 & 8.33 &  & 4.50 & 4.67 & 8.60 & 11.73 \\
			0.3 & 4.20 & 4.60 & 8.13 & 9.20 &  & 4.37 & 4.77 & 8.97 & 9.80 &  & 4.43 & 4.23 & 8.67 & 12.83 \\
			0.6 & 5.80 & 5.70 & 10.23 & 11.80 &  & 4.97 & 4.97 & 10.17 & 12.73 &  & 4.40 & 4.30 & 10.83 & 18.77 \\
			0.8 & 4.97 & 4.43 & 11.17 & 13.40 &  & 4.60 & 4.53 & 12.07 & 19.37 &  & 4.03 & 3.53 & 13.67 & 36.53 \\
			\midrule
			\multicolumn{15}{l}{$T=600$} \\
			\midrule
			0 & 4.67 & 4.67 & 7.00 & 7.20 &  & 4.77 & 4.73 & 6.57 & 7.53 &  & 4.27 & 4.23 & 7.37 & 9.20 \\
			0.3 & 4.77 & 4.37 & 7.67 & 7.77 &  & 4.43 & 4.73 & 7.37 & 7.40 &  & 4.80 & 4.67 & 7.90 & 9.57 \\
			0.6 & 5.60 & 5.07 & 9.43 & 9.50 &  & 5.43 & 5.43 & 8.77 & 10.30 &  & 5.30 & 4.83 & 9.57 & 14.10 \\
			0.8 & 4.70 & 4.50 & 8.73 & 9.20 &  & 5.10 & 5.23 & 9.27 & 13.87 &  & 5.80 & 5.20 & 11.47 & 24.63 \\
			\midrule
			\multicolumn{15}{l}{\textbf{Panel B: Joint test}} \\
			\midrule
			\multicolumn{15}{l}{$T=150$} \\
			\midrule
			0 & 4.23 & 4.03 & 12.70 & 15.10 &  & 4.10 & 4.07 & 14.50 & 21.57 &  & 3.63 & 3.47 & 25.67 & 51.23 \\
			0.3 & 4.80 & 4.57 & 14.37 & 17.33 &  & 3.80 & 3.70 & 19.50 & 26.90 &  & 3.70 & 3.77 & 31.63 & 59.73 \\
			0.6 & 4.27 & 4.03 & 18.03 & 22.07 &  & 3.80 & 3.57 & 23.67 & 37.77 &  & 3.13 & 3.03 & 40.27 & 82.17 \\
			0.8 & 3.13 & 2.93 & 23.37 & 30.47 &  & 3.33 & 2.73 & 31.60 & 57.20 &  & 2.00 & 1.50 & 49.73 & 82.60 \\
			\midrule
			\multicolumn{15}{l}{$T=300$} \\
			\midrule
			0 & 4.80 & 4.73 & 10.07 & 10.87 &  & 4.30 & 4.37 & 10.63 & 14.83 &  & 3.77 & 3.80 & 18.17 & 31.80 \\
			0.3 & 4.77 & 4.97 & 11.60 & 12.93 &  & 4.67 & 4.67 & 14.03 & 17.70 &  & 3.40 & 3.47 & 19.13 & 36.53 \\
			0.6 & 4.93 & 4.40 & 14.33 & 14.87 &  & 3.77 & 3.90 & 17.50 & 25.63 &  & 3.10 & 3.10 & 29.20 & 59.97 \\
			0.8 & 3.87 & 3.40 & 17.40 & 20.97 &  & 3.33 & 2.77 & 22.97 & 38.80 &  & 2.60 & 2.13 & 37.87 & 86.33 \\
			\midrule
			\multicolumn{15}{l}{$T=600$} \\
			\midrule
			0 & 4.27 & 4.53 & 7.37 & 8.03 &  & 4.33 & 4.13 & 8.50 & 10.83 &  & 3.97 & 4.07 & 12.90 & 19.23 \\
			0.3 & 4.93 & 5.17 & 9.23 & 10.00 &  & 4.80 & 4.53 & 10.57 & 12.30 &  & 4.60 & 4.50 & 14.87 & 24.03 \\
			0.6 & 4.07 & 3.80 & 11.63 & 12.43 &  & 4.57 & 4.73 & 12.30 & 16.97 &  & 4.30 & 4.23 & 21.73 & 38.80 \\
			0.8 & 5.00 & 4.53 & 12.77 & 14.37 &  & 3.57 & 3.80 & 15.67 & 25.33 &  & 3.60 & 3.57 & 26.43 & 66.63\\
			\bottomrule
		\end{tabular}%
	}
\end{table}

For $\theta_0=1.3$, the empirical size of various tests are displayed in Table \ref{tbl:SizeWagnerEtAl1}.\footnote{The results for $\theta=0.8$ and $\theta=1.8$ are qualitatively the same. For brevity, we do not include these results in the main paper. The interested reader can find such simulation results in Section \ref{secSup:MoreSimulations} of the Supplementary Material.} These tests are based on four estimators: (1) the NLS estimator with simulated critical values as described in Section \ref{subsec:sim_inf} (SimNLS); (2) the NLS estimator with simulated critical values and the true value for $\theta_0=1.3$ being provided (SimNLS($\theta_0$)); (3) the FM-SOLS estimator based on $\theta_0=1.3$ (FM-SOLS($\theta_0$)); and (4) the FM-SUR estimator based on $\theta_0=1.3$ (FM-SUR($\theta_0$)). The main findings are as follows:
\begin{enumerate}[(a)]
 \item The simulation-based approaches SimNLS and SimNLS($\theta_0$) offer better size control. The size improvements are particularly pronounced when $T=150$ and $\rho=0.8$. The differences in the empirical size of SimNLS and SimNLS($\theta_0$) are small.
 \item Size distortions are more severe when $N$ increases and/or a joint test is performed. The same observation was made in \cite{wagnergrabarczykhong2019}. The behaviour of the simulation-based and fully modified tests is opposite in these cases. SimNLS and SimNLS($\theta_0$) tend to become conservative whereas FM-SOLS and FM-SUR are oversized.
\end{enumerate}

We subsequently simulate power curves.\footnote{Power curves are computationally more intensive. We economize computational time by (1) reducing the number of Monte Carlo replicates to 1,000 and (2) investigating a subset of all possible parameter configurations.} The specification of the single equation test and joint test are as before but we now vary $\phi_{2,1}=\ldots=\phi_{2,N}$ over the set $[-0.008,-0.007,\ldots, 0]$. We take $\rho=0.3$, $\theta=1.3$ and $N=3$ as the baseline scenario and subsequently vary these quantities one-by-one. Figures \ref{fig:SingleEq_Powercurves}--\ref{fig:Joint_Powercurves} show the results. As expected, power increases with increasing sample size, and as $\phi_{2,i}$ moves away from zero.

\begin{figure}[h]
  \centering
\begin{subfigure}{.9\textwidth}
  \includegraphics[width=\linewidth]{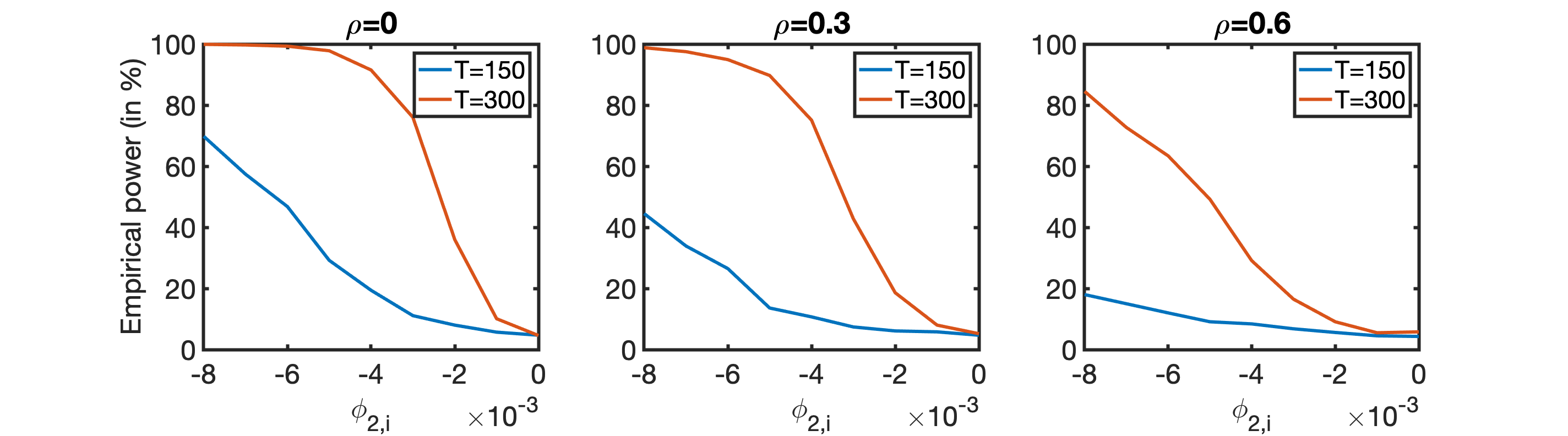}
  \caption{}
\end{subfigure} \\
\centering
\begin{subfigure}{.6\textwidth}
  \includegraphics[width=\linewidth]{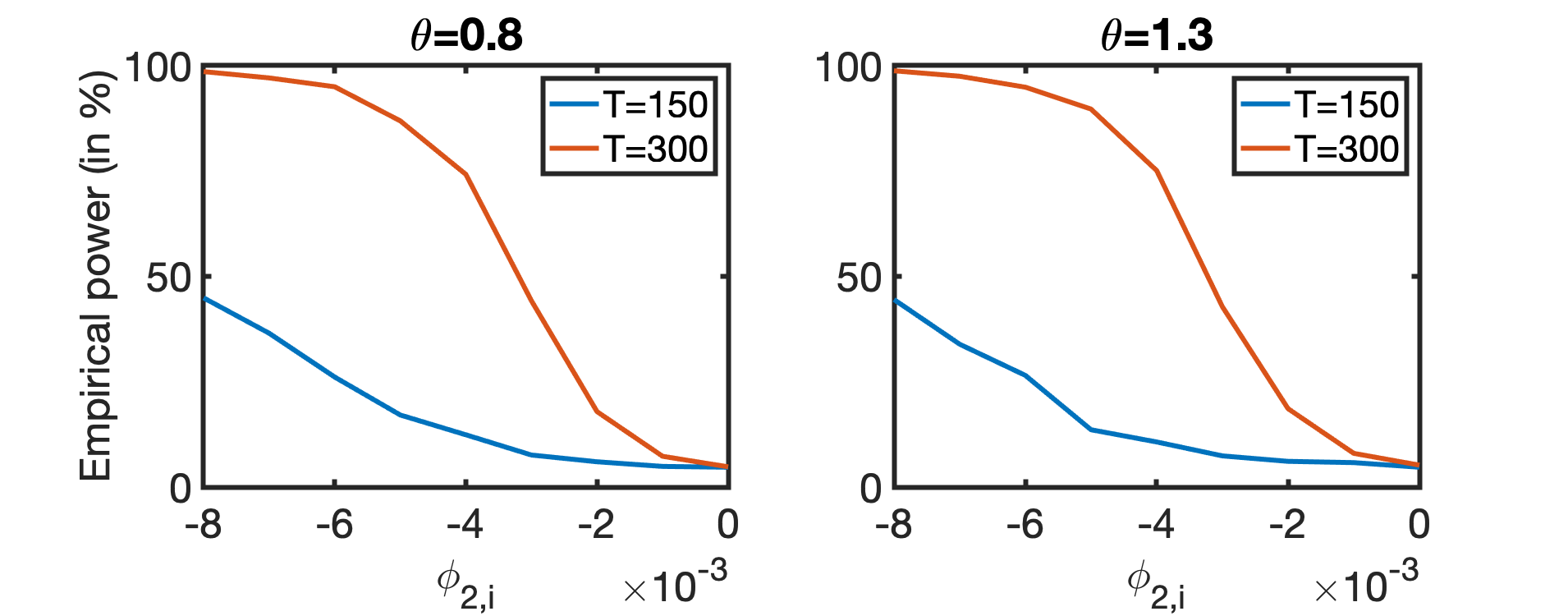}
  \caption{}
\end{subfigure}  \\
\begin{subfigure}{.6\textwidth}
  \centering
  \includegraphics[width=\linewidth]{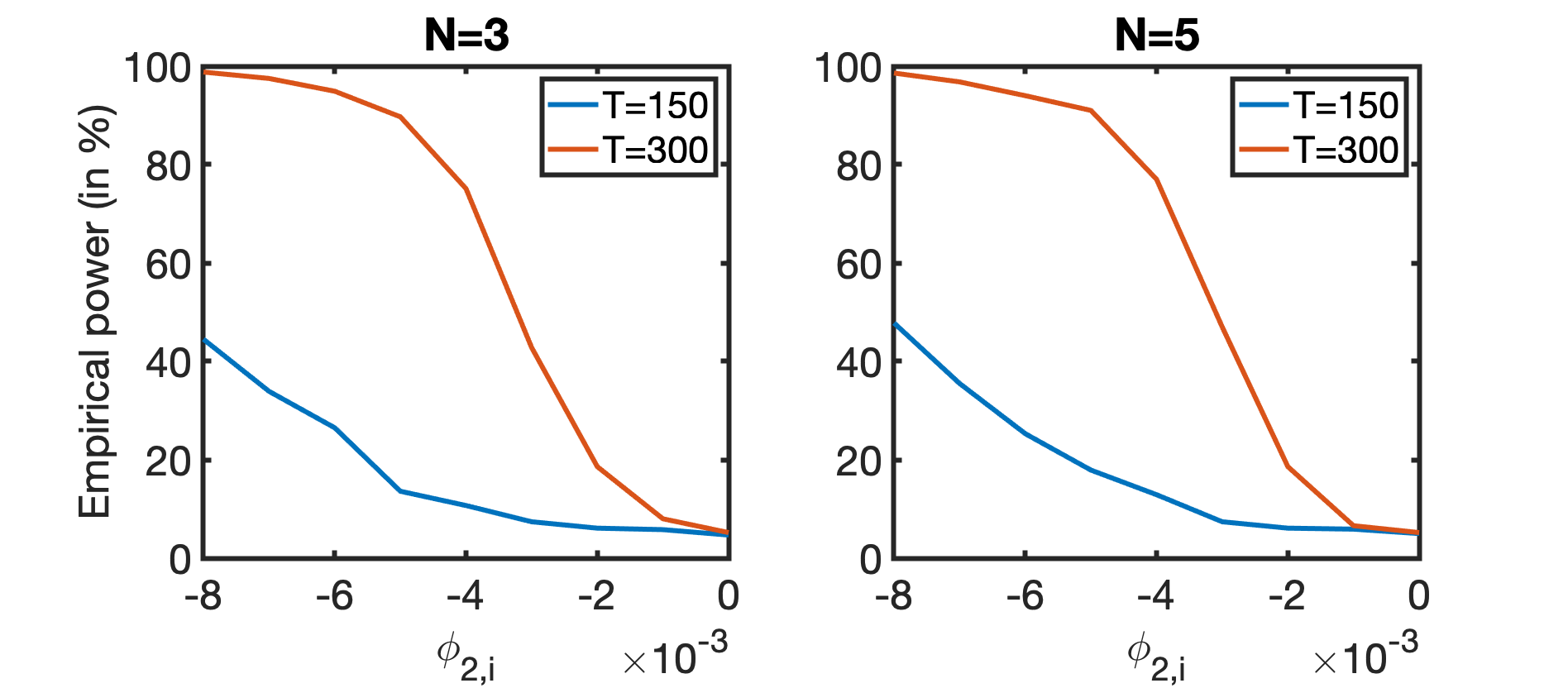}
  \caption{}
\end{subfigure}
\caption{The power curves for the \emph{single equation test} $H_0:\phi_{2,1}=0$. The reference model is DGP1 with $\rho=0.3$, $\theta=1.3$ and $N=3$. We vary the parameters of this reference specification one-by-one while keeping the remaining two parameters fixed at their baseline values. Specifically, we study changes in: \textbf{(a)} the serial correlation and endogeneity parameter $\rho$, \textbf{(b)} the nonlinear deterministic time trend power, and \textbf{(c)} the  cross-sectional dimension.}
\label{fig:SingleEq_Powercurves}
\end{figure}

\begin{figure}[h]
  \centering
\begin{subfigure}{.9\textwidth}
  \includegraphics[width=\linewidth]{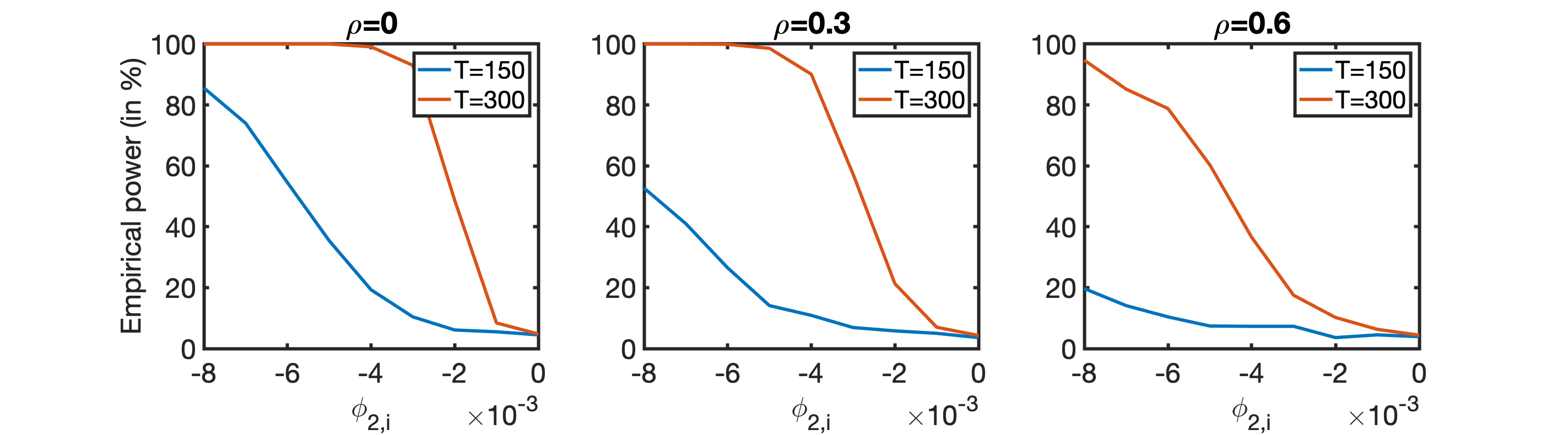}
  \caption{}
\end{subfigure} \\
\centering
\begin{subfigure}{.6\textwidth}
  \includegraphics[width=\linewidth]{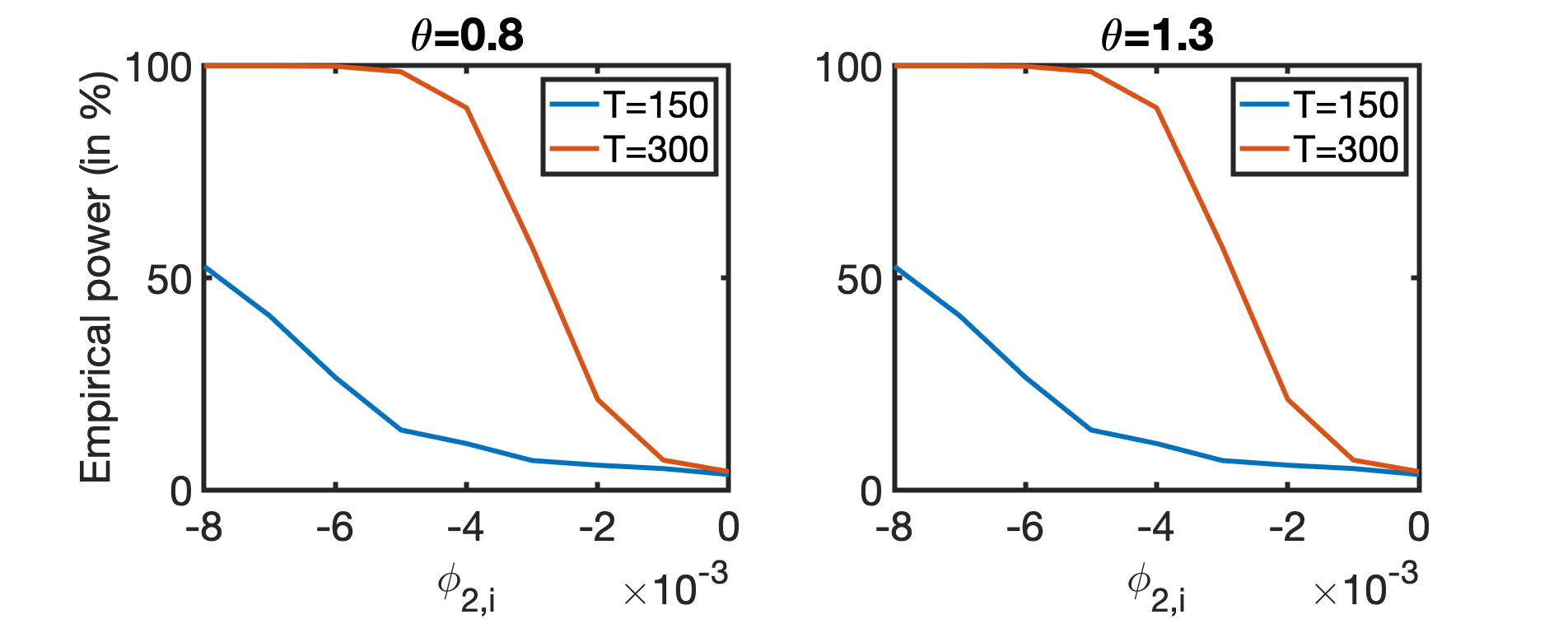}
  \caption{}
\end{subfigure}  \\
\begin{subfigure}{.6\textwidth}
  \centering
  \includegraphics[width=\linewidth]{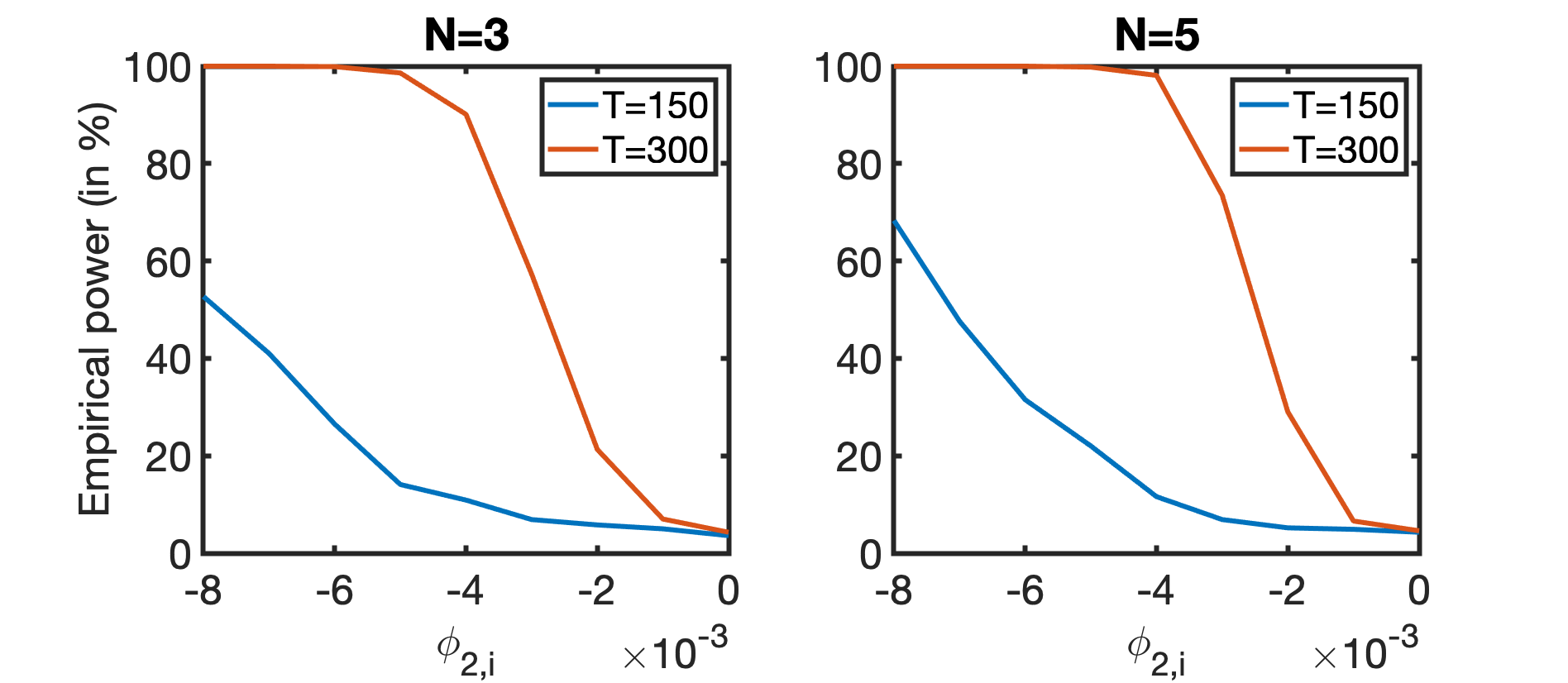}
  \caption{}
\end{subfigure}
\caption{The power curves for the \emph{joint test} for $H_0: \phi_{2,1}=\ldots=\phi_{2,N}=0$. The reference model is DGP1 with $\rho=0.3$, $\theta=1.3$ and $N=3$. We vary the parameters of this reference specification one-by-one while keeping the remaining two parameters fixed at their baseline values. Specifically, we study changes in: \textbf{(a)} the serial correlation and endogeneity parameter $\rho$, \textbf{(b)} the nonlinear deterministic time trend power, and \textbf{(c)} the  cross-sectional dimension.}
\label{fig:Joint_Powercurves}
\end{figure}

\subsubsection*{DGP2: Illustrative simulations in line with the empirical application}
Our second set of simulations is tailored towards the empirical application. That is, we employ parametrizations that mimic the distributional properties of data. Generally speaking, we first estimate the baseline model specification on the data and subsequently fit a VAR(1) specification on the stacked vector of residuals and first-differenced explanatory variables.\footnote{All details on the simulation designs for DGP2(a)--(c) are available in Section \ref{secSup:SimulationDGP2} of the Supplementary Material.} In line with the empirical application, these simulations use $N=6$ and $T=145$. All results are displayed in Figure \ref{fig:DGP2results}. Below, we motivate the simulation settings in view of the EKC application and draw conclusions.
\begin{enumerate}[(a)]
 \item \textbf{Correctly specified model}: The specification $ y_{i,t} = \tau_g t^{\theta} + \tau_{1,i} + \tau_{2,i} t + \phi_{1,i} x_{i,t} + \phi_{2,i} x_{i,t}^2 + u_{i,t}$ with $\phi_{2,i}=0$ is estimated on the data. We subsequently move $\phi_{2,1}=\ldots=\phi_{2,6}$ away from zero in the DGP and check whether we can detect the resulting curvature caused by the integrated variable. Power curves for the individual and joint test for the coefficients in front of $x_{i,t}^2$ are found in Figures \ref{fig:DGP2results}(a) and \ref{fig:DGP2results}(b), respectively. Clearly, nonlinear effects due to $x_{i,t}^2$ are detectable. The statistical power varies across units because (contrary to DGP1) time series properties are now heterogenous across equations.
 \item \textbf{Redundant global trend}: Assumption \ref{assumpt:identifytheta} requires the global trend to be relevant. This simulation DGP investigates how violations of this assumption affect the typical EKC coefficient test. We obtain parameter values by fitting the model $y_{i,t}= \tau_{1,i} + \tau_{2,i} t + \phi_{1,i} x_{i,t} + \phi_{2,i} x_{i,t}^2 + u_{i,t}$ with $\phi_{2,i}=0$. As in (a), we vary  $\phi_{2,1}=\ldots=\phi_{2,6}$ and test for the significance of these parameters. 
 
 The solid lines in Figures \ref{fig:DGP2results}(c) and \ref{fig:DGP2results}(d) are power curves obtained using the correctly specified DGP whereas markers indicate the power when a redundant global trend is estimated as well. The redundant trend has virtually no influence on the statistical power of the coefficient tests of the first five series. There is a power loss for $i=6$. An inspection of the coefficients offers an explanation. The estimated coefficients in front of the global trend are mostly small ($10^{-10}$ to $10^{-9}$) and thus irrelevant. However, in a fraction of cases the flexible trend mimics the curvature in the $6$\textsuperscript{th} series causing the quadratic stochastic trend to become insignificant. As reported in the introduction, the power of the joint test does not suffer from the inclusion of a redundant trend.
 \item \textbf{KPSS test}: Nonstationary residuals are an indication of model misspecification. That is, either the regression is spurious or the functional form of the cointegrating relation is misspecified. We look at the latter situation. The simulation DGP is the quadratic GCPR as in DGP2(a) but the quadratic component is missing in the fitted model. The empirical rejection frequency of the KPSS test (Figure \ref{fig:DGP2results}) is signalling that there are specification issues. However, a comparison with Figures \ref{fig:DGP2results}(a) and \ref{fig:DGP2results}(b) also reveals that if the source of misspecification is known, then a dedicated coefficient test leads to higher power.
\end{enumerate}

\begin{figure}[h]
\begin{subfigure}{.5\textwidth}
  \centering
  \includegraphics[width=.8\linewidth]{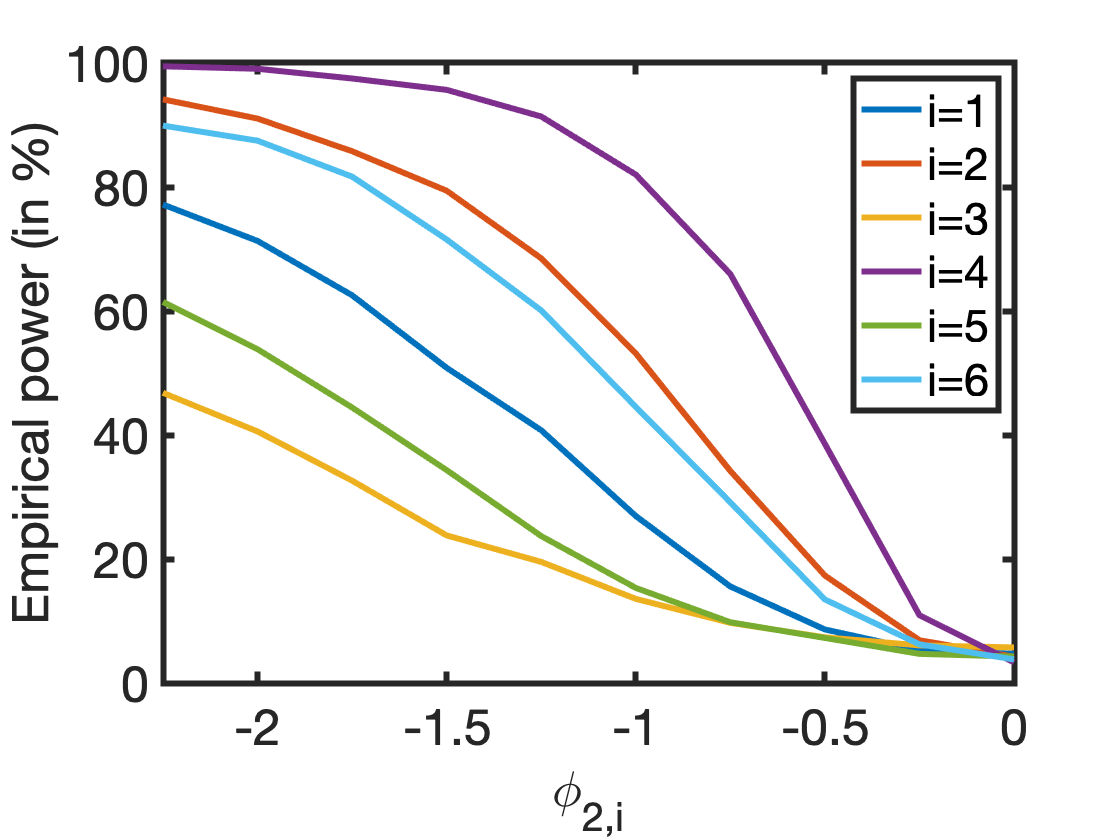}
  \caption{}
\end{subfigure}%
\begin{subfigure}{.5\textwidth}
  \centering
  \includegraphics[width=.8\linewidth]{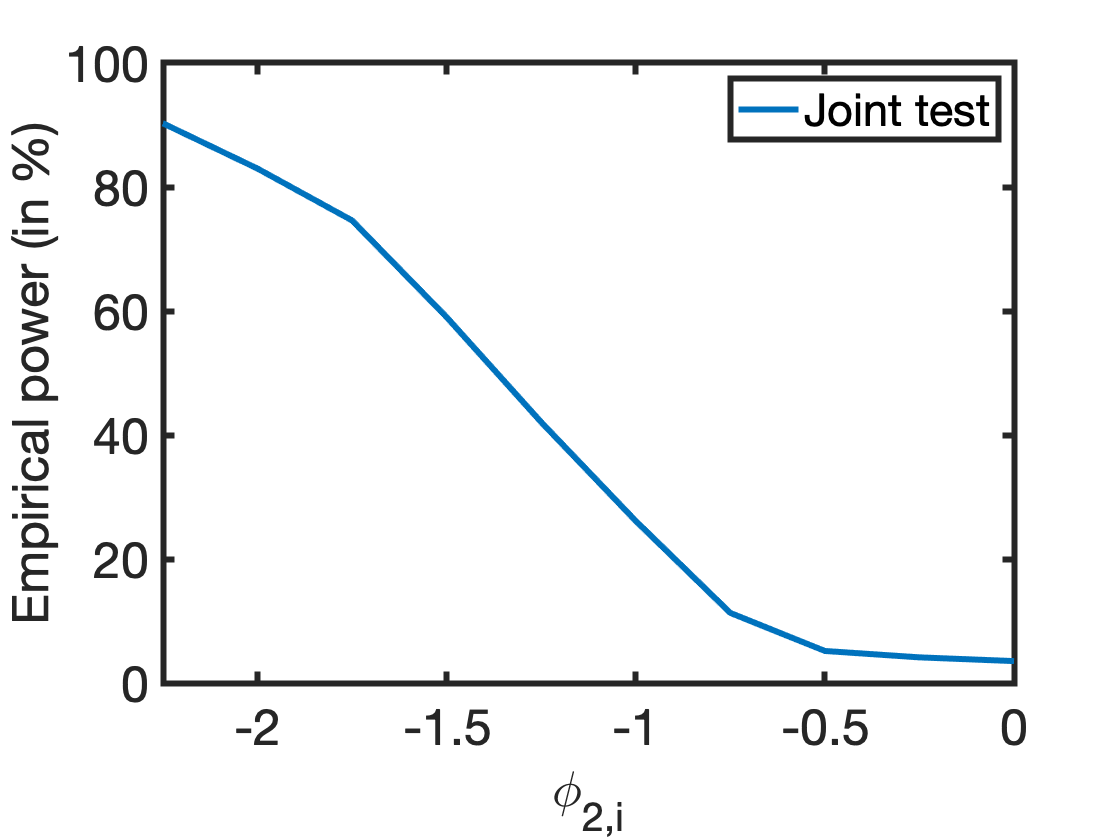}
  \caption{}
\end{subfigure}
\begin{subfigure}{.5\textwidth}
  \centering
  \includegraphics[width=.8\linewidth]{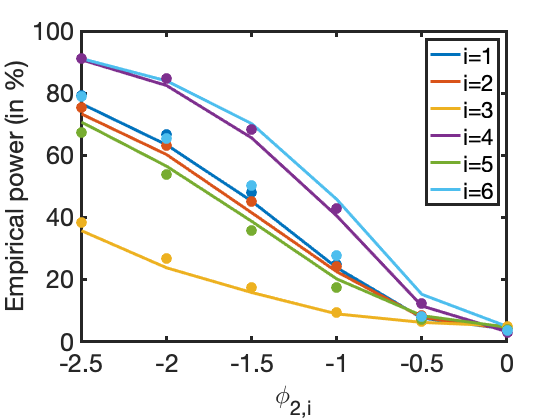}
  \caption{}
\end{subfigure}
\begin{subfigure}{.5\textwidth}
  \centering
  \includegraphics[width=.8\linewidth]{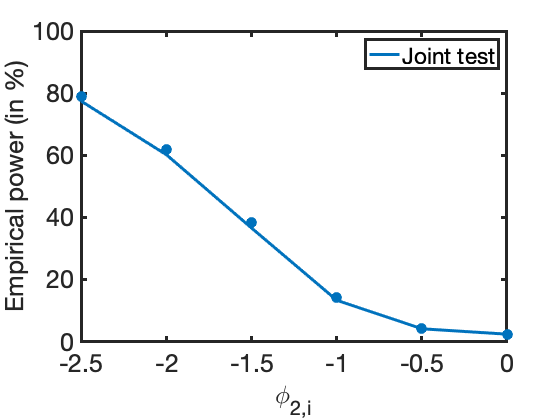}
  \caption{}
\end{subfigure}
\center
\begin{subfigure}{.5\textwidth}
  \centering
  \includegraphics[width=.8\linewidth]{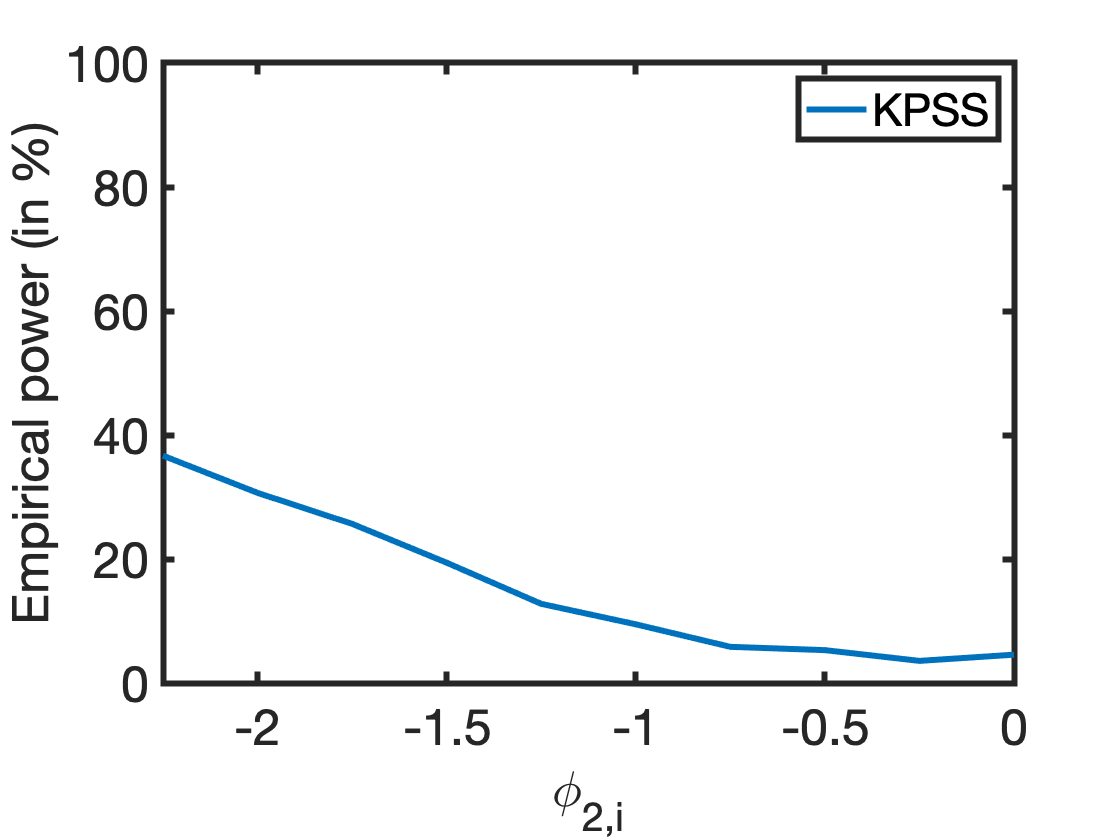}
  \caption{}
\end{subfigure}
\caption{An overview of various power curves. \textbf{(a)} Unit-specific power curves when testing $H_0 : \phi_{2,i}=0$ ($i=1,\ldots,6$) for a correctly specified model. \textbf{(b)} The power curve when testing $H_0: \phi_{2,1}=\ldots=\phi_{2,6}=0$ for a correctly specified model. \textbf{(c)} The empirical rejection frequencies for a correctly specified model (lines) and an estimation with a redundant global time trend (dots). Individual coefficients are tested. \textbf{(d)} As in (c), but now for the joint test $H_0: \phi_{2,1}=\ldots=\phi_{2,6}=0$. \textbf{(e)} The empirical power of the KPSS for a misspecified linear cointegrating relation.}
\label{fig:DGP2results}
\end{figure}

\section{Empirical Application} \label{sec:empapplication}
We examine the evidence for an EKC for a collection of 18 countries over the period 1870--2014 ($T=145$). Economic growth is measured by GDP and we use carbon dioxide (CO\textsubscript{2}) emissions as a proxy for air pollution. The origin of these data is as follows. We use population and GDP data from the Maddison Project (see https://www.rug.nl/ggdc/historicaldevelopment/maddison/). Our carbon dioxide observations are fossil-fuel CO\textsubscript{2} emissions as made available by the Carbon Dioxide Information Analysis Center (CDIAC, see https://cdiac.ess-dive.lbl.gov). The CDIAC database ceased operation in 2017 causing these time series to be available until 2014. Both GDP and CO\textsubscript{2} emissions are expressed per capita and subsequently log-transformed. In accordance with the notation of this paper, we will denote them by $x_{i,t}$ and $y_{i,t}$, respectively. The same data (or subsets thereof) have also been studied by \cite{wagner2015}, \cite{chanwang2015}, \cite{wangwuzhu2018}, \cite{wagnergrabarczykhong2019}, and \cite{linreuvers2019}.\footnote{The stationarity properties of the series have been extensively studied and discussed in these papers. We will not repeat this analysis but refer the interested reader to Section \ref{appendix:furtherempirics} of the Supplement. The exact numbers may show (minor) differences from previously reported results due to differences in: (1) the time span of the data, (2) the implemented long-run covariance estimator, and (3) the scaling of the data. Related to scaling, we follow the official guidelines and multiply by 3.667 and $10^3$ to convert thousand of metric tons of carbon into units of carbon dioxide. Since the data will be expressed in logarithms, this rescaling effectively amounts to a change of intercept.} This conveniently allows us to compare results. All user choices (kernel specification, bandwidth selection, etc.) are kept the same as during the simulation study (see page \pageref{sec:simulations}).

\subsection{An Illustration using Belgian Data}
Prior to the analysis of a multivariate specification, we will first discuss several features of the individual time series (hence omitting subscripts ``$i$''). The example throughout this narrative is Belgium (Figure \ref{fig:overviewBelgium}).\footnote{The data for Austria, Belgium, and Finland are mentioned in both \cite{wagner2015} and \cite{wagnergrabarczykhong2019} to behave in line with the EKC. We discuss Belgium in the main text but the interested reader can find the same figures for Austria and Finland in Section \ref{sec:AuandFi}. Qualitatively, the findings for these other two countries are the same.} An inverted U-shaped relationship between GDP and CO\textsubscript{2} (both in log per capita) is clearly visible in Figure \ref{fig:overviewBelgium}(a) and behavior like this has triggered research on the Environmental Kuznets Curve. However, the time heat map also shows that time is almost monotonically increasing along the curve. Time effects -- e.g. increasing global environmental awareness, worldwide advances in sustainable technologies -- can be valid alternative explanations for these nonlinearities and their omission can (falsely) exaggerate the influence of GDP. It is for this reason that we develop and analyse the Generalized Cointegrating Polynomial Regression (GCPR).

\begin{figure}[h]
\begin{subfigure}{.5\textwidth}
  \centering
  \includegraphics[width=.8\linewidth]{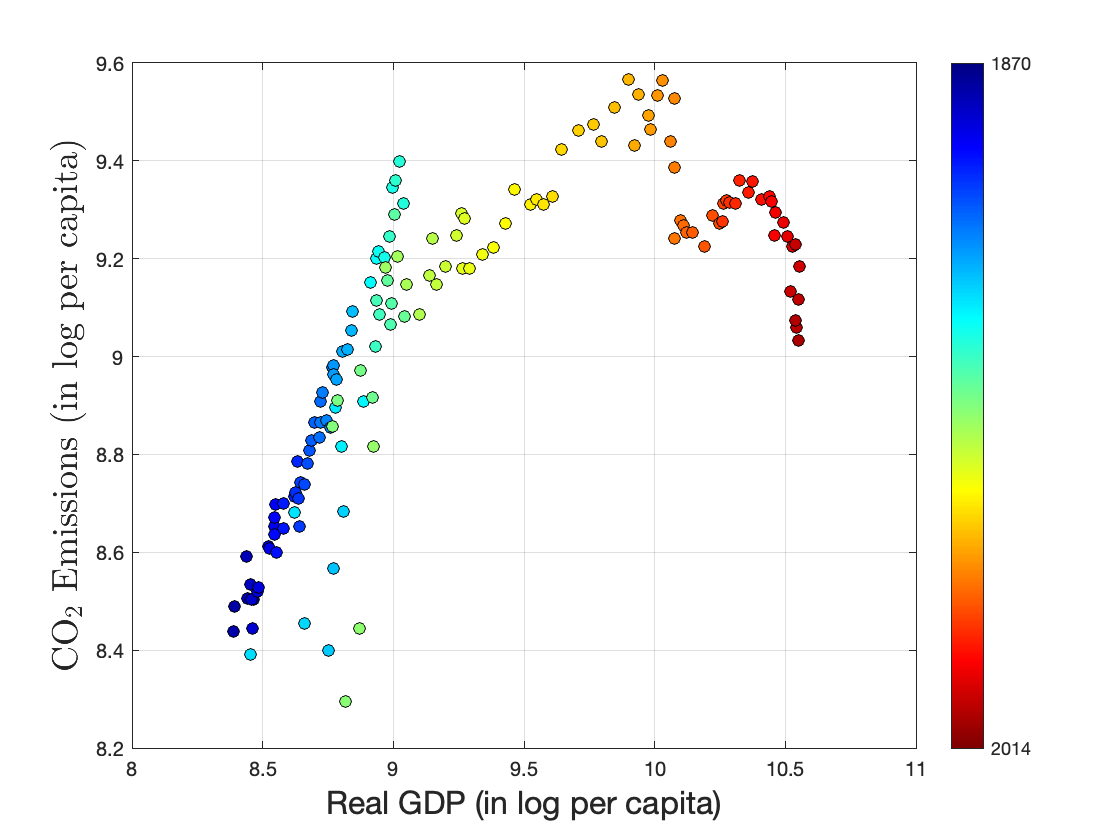}
  \caption{}
  \label{fig:overviewBelgiumA}
\end{subfigure}%
\begin{subfigure}{.5\textwidth}
  \centering
  \includegraphics[width=.8\linewidth]{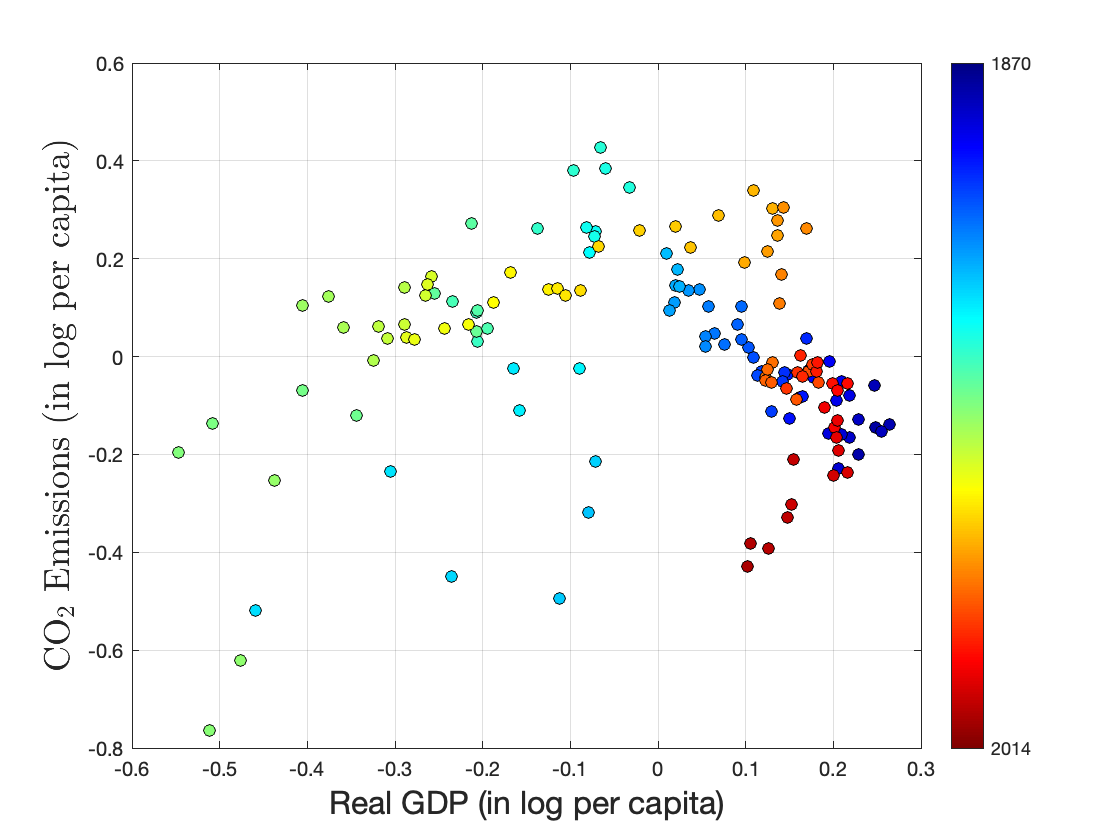}
  \caption{}
\end{subfigure}
\begin{subfigure}{.5\textwidth}
  \centering
  \includegraphics[width=.8\linewidth]{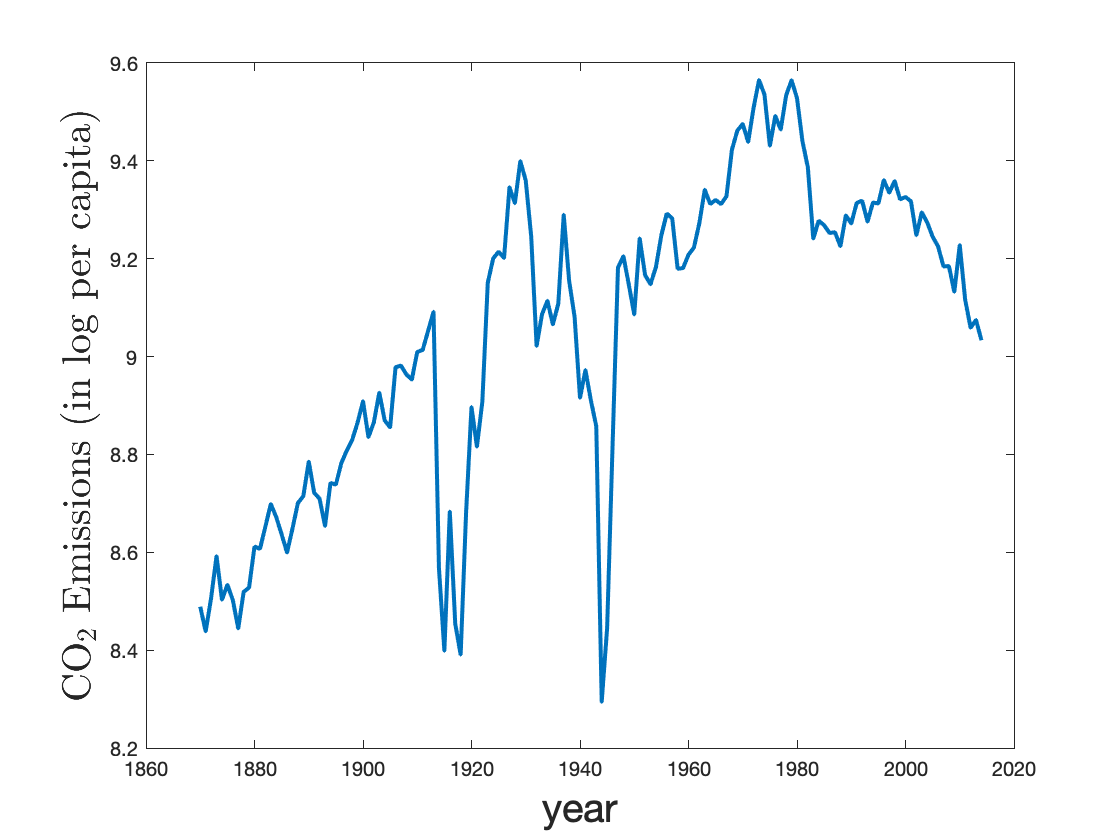}
  \caption{}
\end{subfigure}
\begin{subfigure}{.5\textwidth}
  \centering
  \includegraphics[width=.8\linewidth]{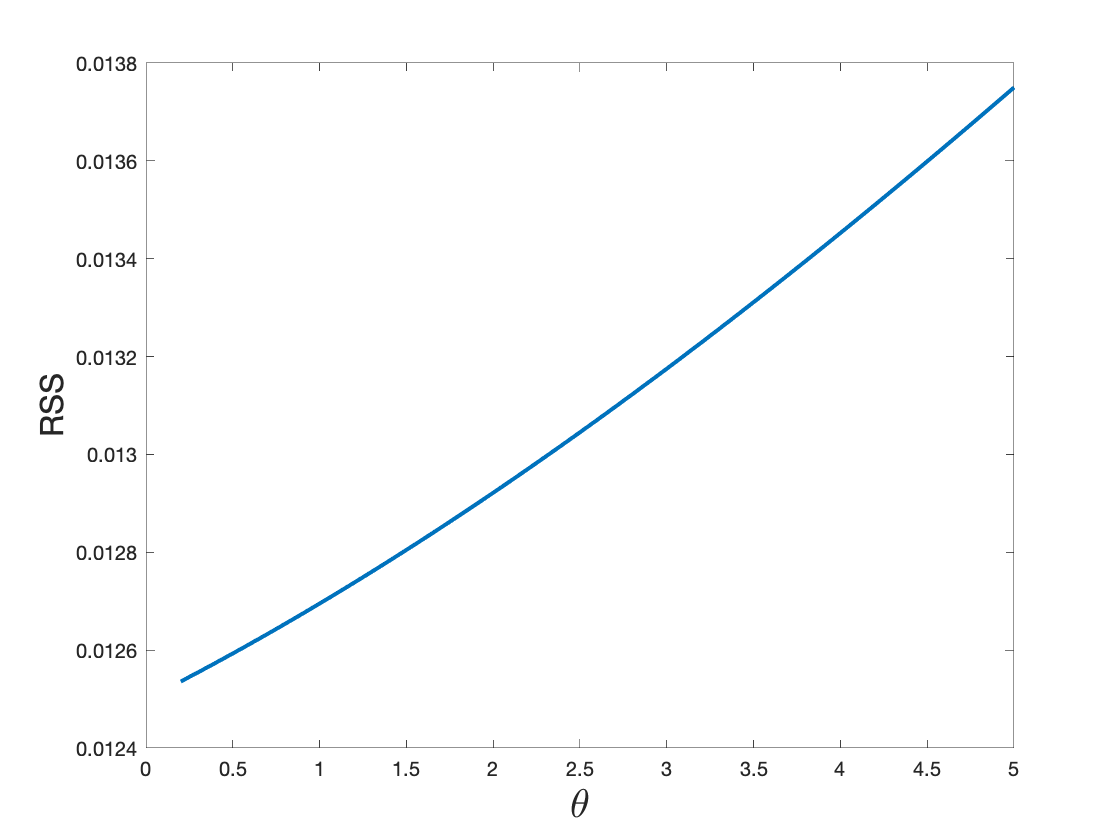}
  \caption{}
\end{subfigure}
\begin{subfigure}{.5\textwidth}
  \centering
  \includegraphics[width=.8\linewidth]{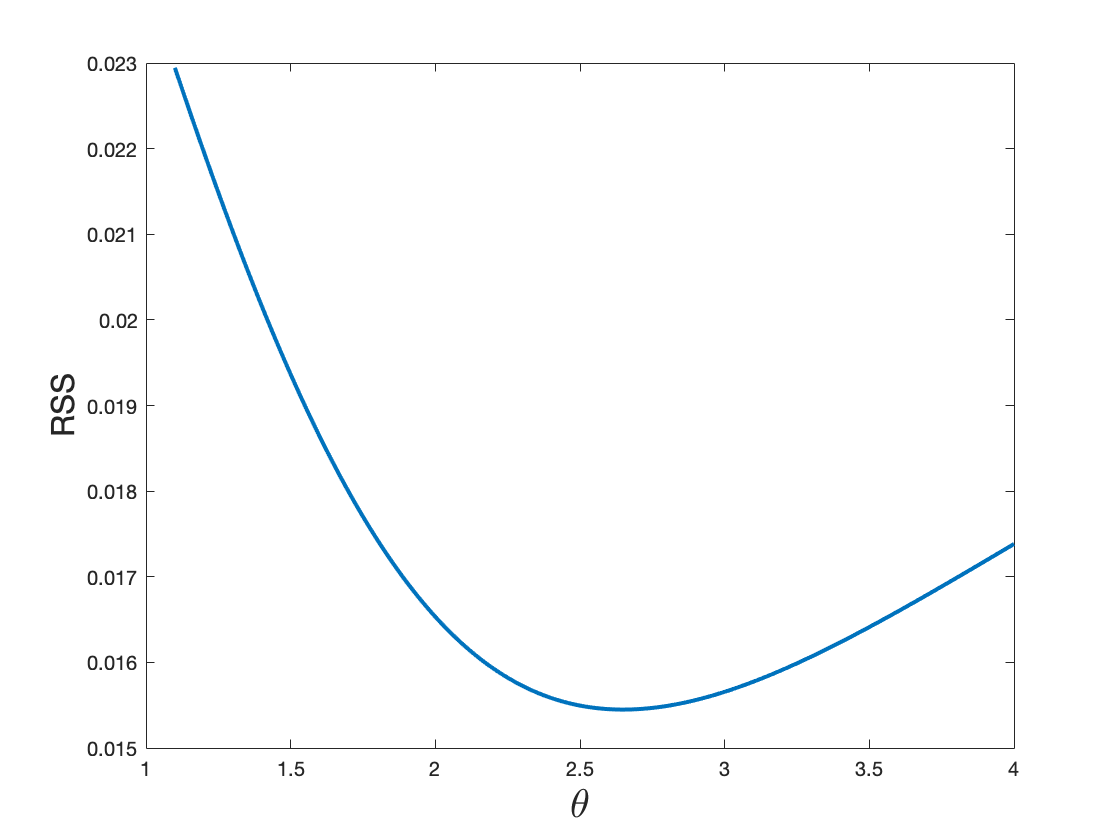}
  \caption{}
\end{subfigure}
\begin{subfigure}{.5\textwidth}
  \centering
  \includegraphics[width=.8\linewidth]{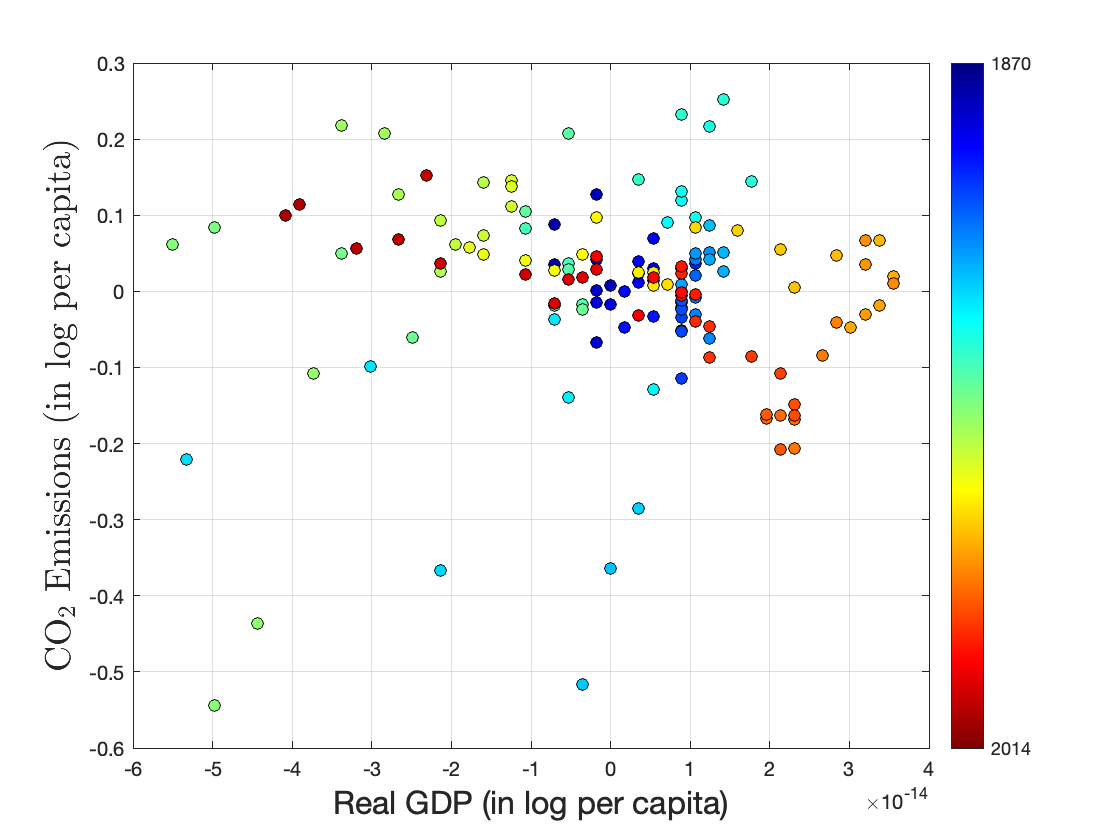}
  \caption{}
\end{subfigure}
\caption{Overview graphs for Belgium over 1870-2014. \textbf{(a)} $\log(\text{GDP})$ versus $\log(\text{CO}_2)$ (both per capita). \textbf{(b)} The same series as in subfigure (a), but now using detrended variables. \textbf{(c)} The log per capita CO\textsubscript{2} emissions time series for Belgium over time. \textbf{(d)} The residual sum of squares (RSS) for the nonlinear model specification $y_t=\tau_1+\tau_2 t + \phi_1 x_t+ \phi_2 x_t^\theta+u_t$ for various values of $\theta$. \textbf{(e)} The RSS as a function of $\theta$ for the flexible nonlinear trend specification $y_t=\tau_1 + \tau_2 t + \tau_3 t^\theta + \phi x_t+u_t$. \textbf{(f)} The relation between $x_t$ and $y_t$ after partialling out the constant, linear trend, and flexible deterministic trend.}
\label{fig:overviewBelgium}
\end{figure}

More evidence for the importance of time effects is available in Figure \ref{fig:overviewBelgium}(b). This figure depicts the same per capita series after detrending.\footnote{The \cite{perronyabu2009} test allows us to test for the presence of a deterministic trend irrespectively of the series being trend-stationary or having an unit root. The results of this test (see supplement) indicate that log per capita GDP is likely to have a deterministic trend component. It is thus recommended to have a deterministic trend in the model for log per capita $\COTWO$ emissions and the visual inspection of the relationship between GDP and $\COTWO$ emissions (in log per capita) should take place after partialling out this deterministic trend.\label{footnoteDetTrend}} The inverted U-shape is now (visually) less pronounced or even absent.

Finally, let us depart from a traditional linear cointegration specification: $y_{t} = \tau_1 + \tau_2 t + \phi_1 x_{t}+u_{t}$. This model cannot incorporate any nonlinear behaviour over time and is therefore ill-suited to fit the data displayed in Figure \ref{fig:overviewBelgium}(c). Cointegrating polynomial regressions use integer powers of $x_{t}$ to describe the curvature over time. More general, as in \cite{huphillipswang2019}, we can allow for an integrated regressor with a flexible power and estimate $y_{t}=\tau_1+\tau_2 t + \phi_1 x_{t}+ \phi_2 x_{t}^\theta+u_{t}$. The residual sum of squares (RSS) of the NLS estimator for this specification is shown in Figure \ref{fig:overviewBelgium}(d). The absence of a minimum at $\theta=2$ casts doubt on the commonly used quadratic specification in $x_{t}$. Additionally, the lack of any minimum might be interpreted as a sign that log per capita GDP is not the source of nonlinearity. This finding is not specific for Belgium. There are no minima in the RSS for 15 out of 18 countries (see Section \ref{sec:thetagraphs}). For the remaining three countries -- Denmark, France and the Netherlands -- minima are found at $\widehat\theta_{DK}=1.46$, $\widehat\theta_{FR}=3.61$ and $\widehat\theta_{NL}=1.28$, respectively. Alternatively, we can describe the nonlinearity in the data using a flexible deterministic trend as in $y_{t}=\tau_1 + \tau_2 t + \tau_3 t^\theta + \phi_1 x_{t}+u_{t}$. The RSS in Figure \ref{fig:overviewBelgium}(e) now exhibits a clear minimum. Further empirical analysis on individual countries (see Section \ref{appendix:univeriateanalysis} of the Supplement) suggests that: (1) the inclusion of a flexible time trend renders all quadratic effects in squared log per capita GDP insignificant, and (2) models remain well-specified after removing quadratic income effects from the model. These results suggest -- albeit in a univariate setting -- that flexible time trends gives a more satisfactory (or at least competing) description of the nonlinearities in the data.

\subsection{Seemingly Unrelated Regression}
The interpretation of a country-specific flexible deterministic trend is complicated because of its high collinearity with GDP per capita. The multivariate analysis of this section allows us to separate country-specific environmental improvements caused by national income growth from global environmental improvements. We study the following six countries ($N=6$): Austria, Belgium, Finland, the Netherlands, Switzerland, and the UK. The motivation behind this choice is as follows. First, based on data series to ours, \cite{piaggiopadilla2012}, \cite{mazzantimusolesi2013}, and \cite{wagnergrabarczykhong2019} report considerable evidence of parameter heterogeneity across countries.\footnote{Parameter heterogeneity is also reported for other data sets. Examples are \cite{listgallet1999},  \cite{cole2005}, and \cite{dijkgraafvollebergh2005}.} The evidence in \cite{mazzantimusolesi2013} is anecdotal in the sense that these authors consider groups of similar countries and find different results for different groups. The lack of overlap among confidence intervals of country-specific parameters has also been interpreted as a sign of heterogeneity (section 4.2 in \cite{piaggiopadilla2012}). \cite{wagnergrabarczykhong2019} explicitly test for various forms of poolability and conclude that pooling is (at most) appropriate for small subgroups of countries. This lack of parameter homogeneity justifies a multivariate approach with small $N$ rather than a panel setting. Admittedly, in the current time-series setting, studying ``large $N$'' is also infeasible since consistent estimators for $(2N\times 2N)$ long-run covariance matrices are required. Second, prior studies already refute the existence of a carbon dioxide EKC for several countries and little seems lost by excluding these countries from the outset.\footnote{Most of the parameters in the Generalized Cointegrating Polynomial Regression are country-specific. The estimation accuracy of these parameters should deteriorate little when focussing attention on a subset of countries. Losses will occur in the precision of the estimators for $\tau_g$ and $\theta$. There is thus a trade-off between accurate global trend estimation (improving with large $N$) and accurate LRV estimation (deteriorating with large $N$). To strike a balance and to connect to the recent literature, we continue the analysis of \cite{wagnergrabarczykhong2019} and take $N=6$.} That is, we consider the same countries as in \cite{wagnergrabarczykhong2019}, who decide on these countries because their prior cointegration analysis ``\emph{leads to evidence for a quadratic cointegrating EKC including a constant and linear trend}''.

\begin{table}[t]
	\centering
	\caption{Parameter estimates and test results for Models \eqref{eq:QuadraticEKC}--\eqref{eq:LinearEKCwithGlobalTrend}. The joint $p$-value refers to the test with null hypothesis $H_0: \phi_{2,1}=\ldots=\phi_{2,6}=0$ and is thus inapplicable for Model \eqref{eq:LinearEKCwithGlobalTrend}.}
	\label{tbl:multivariate_empirical}
	\resizebox{\textwidth}{!}{%
	\begin{threeparttable}
		\begin{tabular}{c d{2.5} d{3.5} d{2.5} d{2.5} d{2.5} d{2.4} c d{2.5} d{2.2} c d{1.3} c}
			\toprule
			& \multicolumn{6}{c}{Omitted Global Trend} & & \multicolumn{5}{c}{Global Trend} \\
			\cmidrule{2-13}
			Model & \multicolumn{6}{c}{(M1)} & & \multicolumn{2}{c}{(M2)} & & \multicolumn{2}{c}{(M3)} \\
			\cmidrule{2-7}\cmidrule{9-10}\cmidrule{12-13}
			& \multicolumn{2}{c}{FM-SOLS} & \multicolumn{2}{c}{FM-SUR} & \multicolumn{2}{c}{SimNLS} & & \multicolumn{2}{c}{SimNLS} & & \multicolumn{2}{c}{SimNLS} \\
			\cmidrule{1-13}
			& \multicolumn{1}{c}{$\phi_{1,i}$} & \multicolumn{1}{c}{$\phi_{2,i}$} &   \multicolumn{1}{c}{$\phi_{1,i}$} & \multicolumn{1}{c}{$\phi_{2,i}$} &  \multicolumn{1}{c}{$\phi_{1,i}$} & \multicolumn{1}{c}{$\phi_{2,i}$} &  & \multicolumn{1}{c}{$\phi_{1,i}$} & \multicolumn{1}{c}{$\phi_{2,i}$} & & \multicolumn{2}{c}{$\phi_{1,i}$} \\
			\midrule
			Austria & 9.37^{***} & -0.43^{***} & 3.96^{*} & -0.16 & 6.42^{***} & -0.28 &  & 3.08^{***} & -0.09 &  & \multicolumn{2}{c}{$1.73^{***}$} \\
			Belgium & 11.78^{***} & -0.59^{***} & 9.92^{***} & -0.50^{***} & 12.36^{***} & -0.62^{**} &  & 7.68^{***} & -0.36 &  & \multicolumn{2}{c}{$1.01^{***} $}\\
			Finland & 16.00^{***} & -0.72^{***} & 15.07^{***} & -0.68^{***} & 17.18^{***} & -0.78^{*} &  & 15.19^{***} & -0.65 &  & \multicolumn{2}{c}{$2.22^{***}$} \\
			Netherlands & 10.68^{***} & -0.51^{***} & 9.58^{***} & -0.46^{***} & 9.27^{***} & -0.44^{*} &  & 4.97^{***} & -0.20 &  &\multicolumn{2}{c}{$ 1.33^{***}$}  \\
			Switzerland & 8.17^{***} & -0.27^{***} & 7.29^{***} & -0.23^{***} & 8.11^{***} & -0.28 &  & 0.58^{*} & 0.10 &  & \multicolumn{2}{c}{$2.55^{***}$} \\
			UK & 9.28^{***} & -0.47^{***} & 7.93^{***} & -0.40^{***} & 9.16^{***} & -0.46^{*} &  & 4.93^{***} & -0.21 &  & \multicolumn{2}{c}{$1.33^{***}$} \\
			\midrule
			Joint $p$-value & \multicolumn{2}{c}{0.00} & \multicolumn{2}{c}{0.00} & \multicolumn{2}{c}{0.16} & & \multicolumn{2}{c}{0.39} & & \multicolumn{2}{c}{---} \\
			\addlinespace[0.1cm]
			KPSS-statistic & \multicolumn{2}{c}{3.45} & \multicolumn{2}{c}{5.10} & \multicolumn{2}{c}{3.46} & & \multicolumn{2}{c}{3.48} & & \multicolumn{2}{c}{3.78}\\
			\midrule
			$\widehat{\tau}\, t^{\widehat{\theta}}$ &  &  &  & &  &  &  & \multicolumn{2}{c}{$-0.012\, t^{1.263}$} & & \multicolumn{2}{c}{$-1.374\cdot 10^{-5} t^{2.450}$} \\
			\bottomrule
		\end{tabular}%
		    \begin{tablenotes}
     \footnotesize
     \item Note: Asterisks denote rejection of the null hypothesis at the $^{***}1\%$, $^{**}5\%$, and $^{*}10\%$ significance level. Depending on the specific table entry, the null hypothesis refers to coefficient(s) being zero or a well-specified cointegrating relation.
     \end{tablenotes}
		    \end{threeparttable}
	}
\end{table}

Having decided on the set of countries, we subsequently study the effect of the global flexible trend on EKC evidence. Table \ref{tbl:multivariate_empirical} shows the estimation results of the quadratic EKC specification
\begin{equation}
 y_{i,t} = \tau_{1,i} + \tau_{2,i} t + \phi_{1,i} x_{i,t} + \phi_{2,i} x_{i,t}^2 + u_{i,t}.
\tag{M1}
\label{eq:QuadraticEKC}
\end{equation}
This setting (possibly with the additional constraint $\tau_{2,i}=0$) has been explored in numerous papers, for example: \cite{seldensong1994}, \cite{piaggiopadilla2012}, \cite{chanwang2015}, \cite{wagner2015}, \cite{wangwuzhu2018}, and \cite{wagnergrabarczykhong2019}. For Model \eqref{eq:QuadraticEKC}, an inverted-U relationship results when $\phi_{1,i}>0$ and $\phi_{2,i}<0$ and empirical evidence hereof is traditionally interpreted as the existence of an EKC. If these coefficients have the correct signs, then the country's turning point -- the level of economic growth at which environmental improvement starts -- can be computed as $\exp\left( -\phi_{1,i}/2\phi_{2,i} \right)$. We assess the parameter values and their significance using FM-SOLS and FM-SUR (repeating the analysis of \cite{wagnergrabarczykhong2019} for ease of comparison) and the simulated approach of Section \ref{subsec:sim_inf}. Regardless of estimation method and country, all coefficient signs are in agreement with the EKC hypothesis. The parameters $\phi_{1,i}$ are generally significantly different from zero but the significance of $\phi_{2,i}$ does vary across estimation methods. FM-SOLS and FM-SUR typically (strongly) reject $H_0: \phi_{2,i}=0$ ($i=1,\ldots,6$) whereas evidence against these null-hypotheses is less pronounced for the simulation-based approach. The same behaviour emerges when testing $\phi_{2,1}=\ldots=\phi_{2,6}=0$ jointly. This pattern reminds of the simulation results in Table \ref{tbl:SizeWagnerEtAl1} where the cross-sectional dimensions $N=5$ and $N=10$ cause over-sized tests for FM-SOLS and FM-SUR and conservative tests for simulation-based inference. The KPSS test does not indicate any signs of misspecification. Overall, Model (M1) leads to considerable evidence in favour of a quadratic cointegrating EKC.
 
 The reported evidence in favour of the EKC should not come as surprise. First, the set of countries was selected based on this criteria. Second, the visualisations of the data clearly suggest nonlinear effects (recall Figures \ref{fig:overviewBelgium}(a) and  \ref{fig:overviewBelgium}(c) for the case of Belgium). With Model (M1) being restrictive in the sense that nonlinearities over time are solely incorporable through $x_{i,t}^2$, we expect this variable to be important. In line with our proposed Generalized Cointegrating Polynomial Regression (GCPR) framework, we subsequently add a global flexible trend and estimate
 \begin{equation}
   y_{i,t} = \tau_g t^{\theta} + \tau_{1,i} + \tau_{2,i} t + \phi_{1,i} x_{i,t} + \phi_{2,i} x_{i,t}^2 + u_{i,t}.
 \tag{M2}
 \label{eq:QuadraticEKCwithGlobalTrend}
 \end{equation}
 From a statistical perspective, the term $\tau_g t^{\theta}$ opens a different channel through which nonlinearities can be described. We refer back to the introduction for a further elaboration on this point. From an economic perspective, $\tau_g t^{\theta}$ captures changes in CO\textsubscript{2} emissions that are common across series and thus unrelated to changes in national GDPs. Parameter inference for Model \eqref{eq:QuadraticEKCwithGlobalTrend} is also reported in Table \ref{tbl:multivariate_empirical}. The contributions of $x_{i,t}^2$ are insignificant for both individual countries and all countries jointly. How about the significance of the global trend? The standard Wald test for $\tau_g=0$ is invalid because $\theta$ is unidentified under the null hypothesis (see Assumption \ref{assumpt:identifytheta} and the related discussion). As a heuristic alternative, we vary $\theta$ over the interval $[0,2.5]$ and compute Wald statistics while assuming $\theta$ to be fixed. Comparing these Wald statistics to the 95\% quantile of a $\chi^2(1)$-distributed random variable (critical value: 3.842), the range of $\theta$-values from about 0.5 to 1.75 implies a significant global trend (Figure \ref{fig:WaldForTheta}). Having estimated $\widehat\theta = 1.263$, our analysis suggests that the global trend and not GDP per capita is the source of nonlinearity. Before interpreting this result, we first verify whether the model with $\phi_{2,1}=\ldots=\phi_{2,6}=0$ shows signs of misspecification.
 
\begin{figure}[t]
\center
  \includegraphics[width=.5\linewidth]{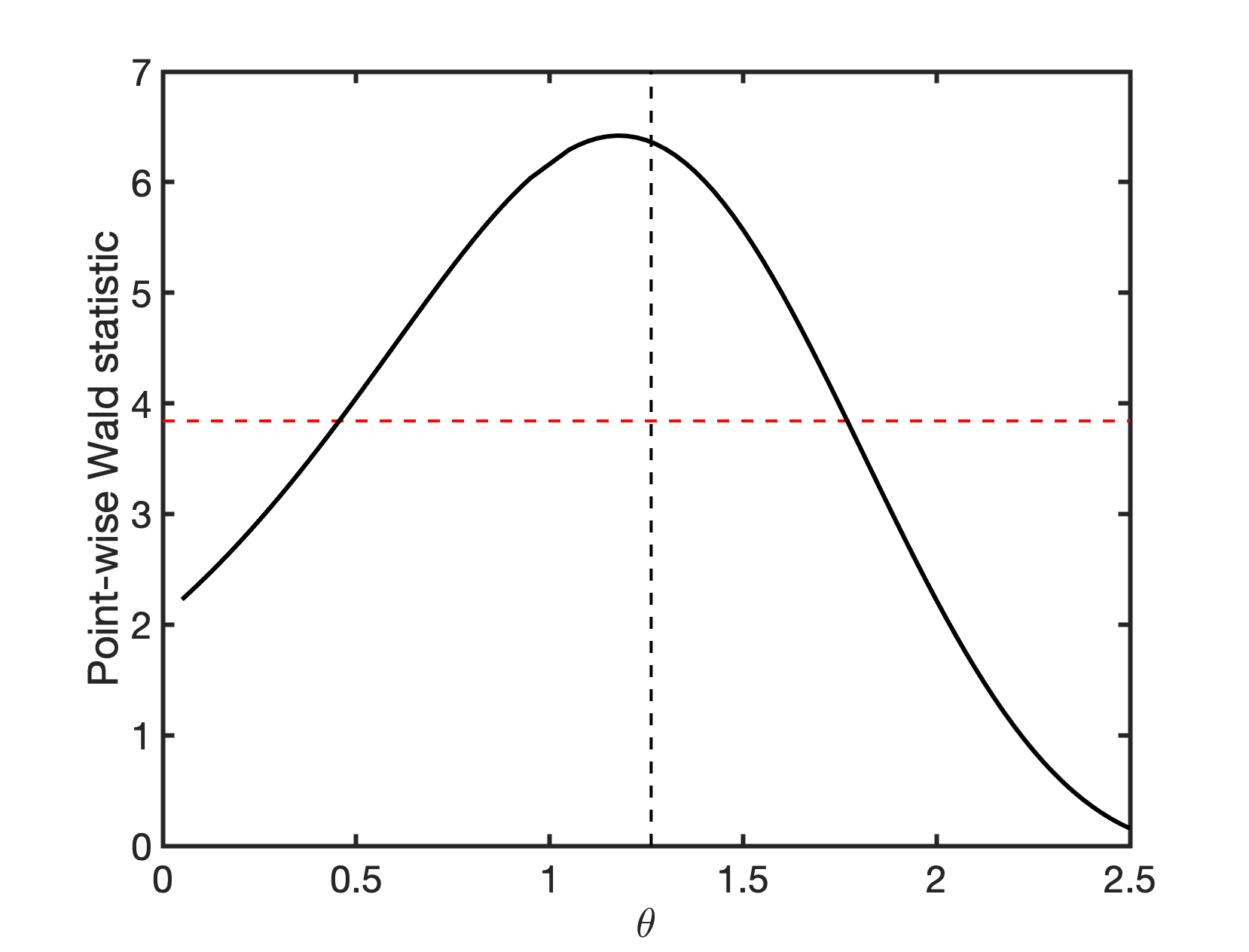}
  \caption{The magnitude of the Wald test for fixed values of $\theta$ when testing $H_0:\tau_g = 0$ under Model \eqref{eq:QuadraticEKCwithGlobalTrend}. Dash lines display the 95\% quantile of a chi-squared distributed random variable with 1 degree of freedom (red) and the NLS estimate $\widehat\theta=1.263$ for specification \eqref{eq:QuadraticEKCwithGlobalTrend}.}
\label{fig:WaldForTheta}
\end{figure}
 
 Omitting insignificant parameters from the previous model specification, we arrive at
  \begin{equation}
   y_{i,t} = \tau_g t^{\theta} + \tau_{1,i} + \tau_{2,i} t + \phi_{1,i} x_{i,t} + u_{i,t}.
 \tag{M3}
 \label{eq:LinearEKCwithGlobalTrend}
 \end{equation}
 Model \eqref{eq:LinearEKCwithGlobalTrend} is linear in log per capita GDP. The positive parameter estimates for $\phi_{1,i}$ imply that \emph{at a given point in time} increases in economic growth imply increases in CO\textsubscript{2} emissions. However, as $\widehat\tau_g= -1.374\times 10^{-5}$ and $\widehat\theta= 2.45$, there will be common emission reductions over time. Also, the omission of the quadratic terms in log per capita GDP do not seem to result in a misspecified model. First, the KPSS test does not reject the null of cointegration. Second, there is no (visual) evidence that the linear functional form of \eqref{eq:LinearEKCwithGlobalTrend} is inappropriate. To arrive at this last conclusion, we compute $\widetilde y_{i,t} = y_{i,t}-\widehat\tau_g t^{\widehat\theta} - \widehat{\tau}_{1,i}-\widehat{\tau}_{2,i} t$ and employ the nonparametric kernel estimator from \cite{wangphillips2009} to estimate $\widetilde y_{i,t}= f(x_{i,t})+\widetilde u_{i,t}$ for each individual country.\footnote{The properties of nonparametric kernel estimators in nonlinear cointegration models have been studied by \cite{wangphillips2009}, \cite{gaokanayalitjostheim2015} and \cite{wangphillips2016}, among others. The latter reference is particularly relevant because it establishes that kernel estimators remain consistent and asymptotically (mixed) normal under serially correlated errors and endogeneity. None of these papers includes deterministic trends in the DGP. However, we conjecture that detrending does not affect the asymptotic properties of the kernel estimator due to the high convergence rates of the trend parameters in comparison to the slow convergence rates of the nonparametric estimator. Our bandwidth choice is $h=T^{-1/3}$.\label{ftnt:nonparametric1}} Figure \ref{fig:NonParametric} shows the nonparametric estimate in blue and the fit of Model \eqref{eq:LinearEKCwithGlobalTrend} in red. After removal of the global trend, there are some temporary departures from linearity but there is little curvature overall and certainly no visual turning point. In Table \ref{tbl:linearitytest}, we formally test the null of linearity using the model specification test as outlined in section 3 of \cite{wangphillips2016}. Based on the full sample, linearity is rejected for Austria only. A comparison with the 95\% confidence intervals of the kernel estimate (Figure \ref{fig:NonParametric}) suggests that this rejection is caused by the sharp decline in CO\textsubscript{2} emissions during World War II. We subsequently repeat the analysis using the $T=69$ observations after 1945. Linearity is never rejected.\footnote{The high $p$-values in Table \ref{tbl:linearitytest} are caused by visually small deviations from the linear trend (see Section \ref{appendix:graphslinearfit} in the Supplement for detailed graphs). Also, the relatively small sample size (for nonparametric settings) might adversely affect power. Matlab functions for nonparametric kernel regression and specification test are available at \texttt{https://github.com/HannoReuvers}. For remarks on bandwidth choice and detrending, we refer to footnote \ref{ftnt:nonparametric1}.} All this align well with our earlier findings of a relevant global trend and irrelevant quadratic effects in log GDP per capita.
 
\begin{figure}[h]
\center
  \includegraphics[width=.8\linewidth]{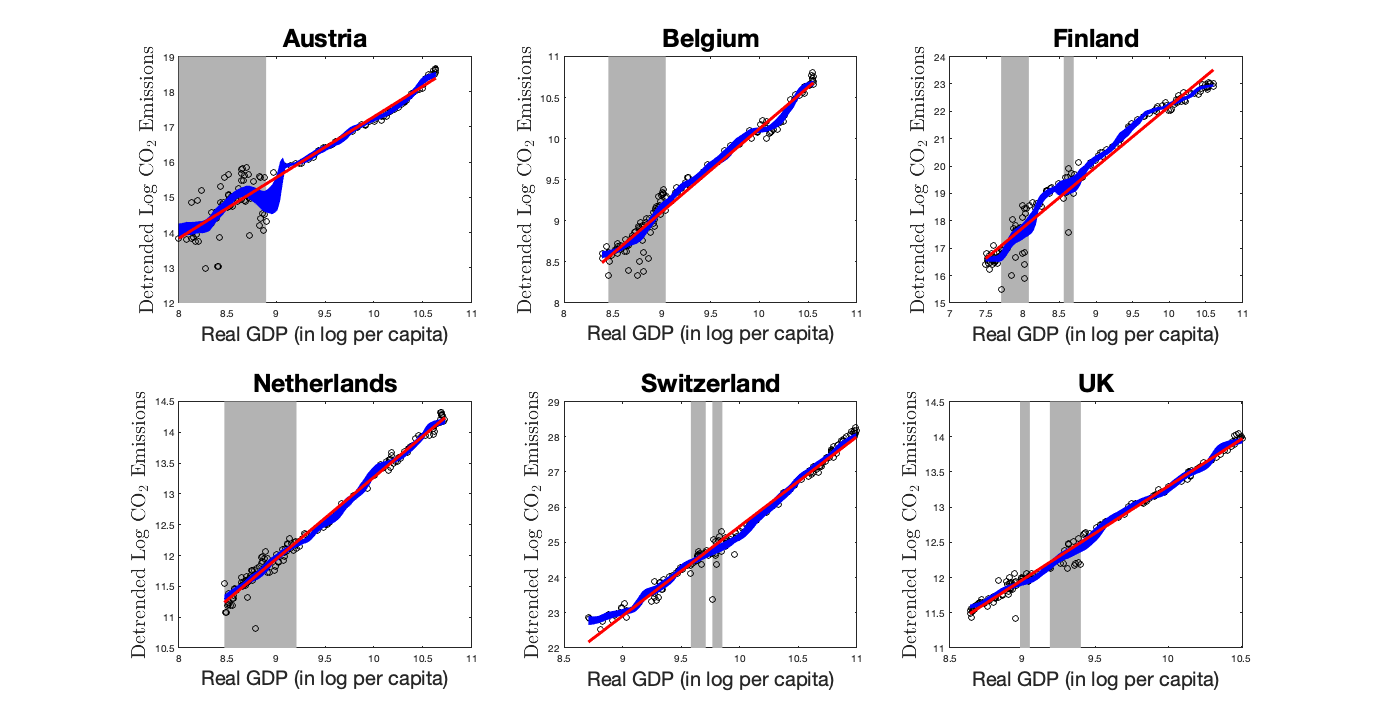}
  \caption{The 95\% (point-wise) confidence intervals of the non-parametric kernel estimate for the relationship between GDP and CO\textsubscript{2} emissions (blue) after removal of the country-specific and joint flexible deterministic trends. The red line indicates the linear fit implied by the estimation results of Model \eqref{eq:LinearEKCwithGlobalTrend}.  As the sample covers the years 1870--2014 there are several observations during World War I and World War II. The affected ranges of GDP are indicated in grey.}
\label{fig:NonParametric}
\end{figure}

\begin{table}[t]
	\centering
	\caption{Linearity test results. The linearity test is based on the model specification test documented in section 3 of \cite{wangphillips2016}. The test is based on the integrated weighted squared deviations between the data and the linear model fit. We report the integration range, the (standardized) test statistic, and the $p$-value. Under $H_0$, the relationship between (detrended) log CO\textsubscript{2} emissions and log GDP per capita is linear.}
	\label{tbl:linearitytest}
	\begin{threeparttable}
		\begin{tabular}{c c c c c c c c c c}
			\toprule
			& \multicolumn{3}{c}{Full Sample} & & \multicolumn{3}{c}{After World War II} \\
			\cmidrule{2-4} \cmidrule{6-8}
						& range 	& $\frac{\phi}{\tau_0 \sqrt{n} h} T_n$  & $p$-value 	&& range & $\frac{\phi}{\tau_0 \sqrt{n} h} T_n$ & $p$-value \\
			\midrule
			Austria		& [8.003,10.635]	& 8.987	& 0.000	&& [8.129,10.635] 	& 0.162& 0.871 \\
			Belgium		& [8.389,10.553]	& 0.299	& 0.765	&& [8.923,10.553]	& 0.076& 0.939 \\
			Finland		& [7.494,10.602]	& 0.105	& 0.916	&& [8.693,10.602]	& 0.020& 0.984  \\
			Netherlands	& [8.469,10.728]	& 0.046	& 0.964	&& [8.990,10.728]	& 0.022& 0.982  \\
			Switzerland	& [8.708,10.993]	& 0.028	& 0.977	&& [9.925,10.993]	& 0.009& 0.993  \\
			UK			& [8.641,10.510]	& 0.022	& 0.982	&& [9.242,10.510]	& 0.021& 0.984  \\
			\bottomrule
		\end{tabular}%
		\begin{tablenotes}
			\footnotesize
			\item Note: The asymptotic properties of $\frac{\phi}{\tau_0 \sqrt{n} h} T_n$ are established in \cite{wangphillips2016}. That is, under suitable conditions, $\frac{\phi}{\tau_0 \sqrt{n} h} T_n\to L_W(1,0)$  as $n\to\infty$ with $L_W(1,0)$ denoting the sojourning time of a standard Brownian motion around zero during the time interval $[0,1]$. The $p$-values are computed using the cumulative distribution function of $L_W(1,0)$, see (2.11) in \cite{donggaotjostheimyin2017}.
		\end{tablenotes}
     \end{threeparttable}
\end{table}
 
 The preceding analysis suggests that the global flexible trend captures omitted determinants of CO\textsubscript{2} emission levels that have been decreasing over time. In their analysis, \cite{grossmankrueger1995} already included a global deterministic trend in their model because they ``\emph{did not want to attribute to national income growth any improvements in local environmental quality that might actually be due to global advances in the technology for environmental preservation or to an increased global awareness of the severity of environmental problems}''. Indeed, since reliable data on green technology adaptation\footnote{\cite{nordhaus2014} discusses the link between climate change and technological changes. As another example, Figure 2 in \cite{gillinghamstock2018} reports a steady decline in the price of solar panels and a steady growth in solar panel sales. Cheaper solar energy can substitute fossil energy thereby reducing pollution.} and global awareness is scarcely available (certainly for time horizons allowing for a cointegration analysis), these variables are likely missing and thus requiring a proxy. Similar remarks are applicable to variables such as pollution control policies\footnote{A policy variable, `Repudiation of Contracts by Government', was included by \cite{panayotou1997} to proxy the quality of environmental policies and institutions.}. In reduced-form models, an EKC finding is typically explained by national income being the proxy for these omitted variables. That is, at higher levels of national income, countries have access to cleaner technologies and its citizens show greater appreciation for the environment and pollution legislation. The current analysis contradicts these income effects and points towards improvements being captured by a global trend.
 
Our final model specification, Model \eqref{eq:LinearEKCwithGlobalTrend}, is linear in log GDP per capita. Moreover, for a given year, the coefficient estimates suggest that increasing national income by 1\% implies an \emph{increase} in carbon-dioxide emissions of about 1\%--2.5\% (depending on the country). This result seems plausible for non-carbon-neutral economies. However, CO\textsubscript{2} emissions in Austria, Belgium, Finland, the Netherlands, Switzerland, and the UK are jointly reducing at the end of the sample. What cause these global emission reductions? \cite{mazzantimusolesi2013} suggest that conglomerates of countries anticipate and respond to international climate agreements such as the Rio convention (1992) and Kyoto protocol (adopted in 1997; operational since 2005). Interestingly, the latter agreement contains emission reduction targets to be reached in 2020 and such ``working-towards-a-common-reduction-deadline'' does point towards a time effect.\footnote{According to the Doha amendment of the Kyoto protocol, the reduction commitments were 92\% (over the period 2008--2012) and 80\% (over the period 2013--2020) of 1990 emission levels for Austria, Belgium, Finland, the Netherlands and the UK. For Switzerland, the reduction target was also 92\% (over the period 2008--2012) but 84.2\% (over the period 2013--2020). (source: https://unfccc.int/files/kyoto\_protocol/application/pdf/kp\_doha\_amendment\_english.pdf).} Alternatively, given our sample of European countries, EU coordinated emission reduction efforts like the EU Emissions Trading System (ETS) can be a driving force behind these common emission decreases.
 
\section{Summary and Conclusion} \label{sec:conclusion}
In this paper we have extended the Seemingly Unrelated Cointegrating Polynomial Regression (SUCPR) model of \cite{wagnergrabarczykhong2019} with a global power law deterministic trend. This multivariate specification allows us to disentangle national income and unobserved time effects. The importance of this separation is well-documented (see, e.g. \cite{volleberghmelenbergdijkgraaf2009} and \cite{mazzantimusolesi2013}) but a methodological approach accounting for nonstationary regressors is currently unavailable. We fill this gap.

The unknown powers of the global trend are estimated jointly with the parameters in the cointegrating relation. The limiting distribution is nonstandard due to a non-diagonal scaling matrix and second order bias terms. We therefore suggest a simulation-based approach to conduct inference. The usual subsampling KPSS-type for stationarity of the innovations of the nonlinear cointegrating relation remains valid. Our results are supported by Monte Carlo simulation. The empirical application on the Environmental Kuznets Curve shows that a flexible trend can fully capture the nonlinearity in the data thereby making higher order powers of log per capita GDP redundant. Our resulting model is linear in log per capita GDP and suggests an alternative explanation in which time effects -- e.g. technological progress, increasing environmental awareness, tightening pollution policy -- rather than economic growth cause the recent slowdown in CO\textsubscript{2} emissions. Contrary to the opening quote in the introduction, our analysis suggests that CO\textsubscript{2} emissions increase with economic growth. Carbon dioxide emissions do decrease due to time effects.

\section{Acknowledgements}
Earlier versions of this paper have been presented at the 2019 CFE meeting in London, the Econometrics Internal Seminar (EIS) at Erasmus University Rotterdam, and the Brownbag Seminar at Vrije Universiteit Amsterdam. We gratefully acknowledge the comments by the participants. We extend our thanks to Eric Beutner, Dick van Dijk, Stephan Smeekes, and Xiaohu Wang for their valuable feedback. All remaining errors are our own.
\end{spacing}

\begin{spacing}{1.52}
\bibliographystyle{chicagoa}
\bibliography{LinReuvers_pwrlaw_Dec21.bbl}

\begin{thebibliography}{}

\bibitem[\protect\citeauthoryear{Adams and Essex}{Adams and
  Essex}{2016}]{adamsessex2016}
Adams, R.~A. and C.~Essex (2016).
\newblock {\em Calculus: A Complete Course}.
\newblock Pearson Canada.


\bibitem[\protect\citeauthoryear{Al-Mulali and Ozturk}{Al-Mulali and
  Ozturk}{2016}]{almulaliozturk2016}
Al-Mulali, U. and I.~Ozturk (2016).
\newblock The investigation of {E}nvironmental {K}uznets {C}urve hypothesis in
  the advanced economies: {T}he role of energy prices.
\newblock {\em Renewable and Sustainable Energy Reviews\/}~{\em 54},
  1622--1631.


\bibitem[\protect\citeauthoryear{Andrews and Cheng}{Andrews and
  Cheng}{2012}]{andrewscheng2012}
Andrews, D.~W. and X.~Cheng (2012).
\newblock Estimation and inference with weak, semi-strong, and strong
  identification.
\newblock {\em Econometrica\/}~{\em 80}, 2153--2211.


\bibitem[\protect\citeauthoryear{Andrews}{Andrews}{1991}]{andrews1991}
Andrews, D. W.~K. (1991).
\newblock Heteroskedasticity and autocorrelation consistent covariance matrix
  estimation.
\newblock {\em Econometrica\/}~{\em 59}, 817--858.


\bibitem[\protect\citeauthoryear{Andrews and Sun}{Andrews and
  Sun}{2004}]{andrewssun2004}
Andrews, D. W.~K. and Y.~Sun (2004).
\newblock Adaptive local polynomial whittle estimation of long-range
  dependence.
\newblock {\em Econometrica\/}~{\em 72}, 569--614.


\bibitem[\protect\citeauthoryear{Baek, Cho, and Phillips}{Baek
  et~al.}{2015}]{baekchophillips2015}
Baek, Y.~I., S.~J. Cho, and P.~C.~B. Phillips (2015).
\newblock Testing linearity using power transforms of regressors.
\newblock {\em Journal of Econometrics\/}~{\em 187}, 376--384.


\bibitem[\protect\citeauthoryear{Bergamelli, Bianchi, Khalaf, and
  Urga}{Bergamelli et~al.}{2019}]{bergamellibianchikhalafurga2019}
Bergamelli, M., A.~Bianchi, L.~Khalaf, and G.~Urga (2019).
\newblock Combining p-values to test for multiple structural breaks in
  cointegrated regressions.
\newblock {\em Journal of Econometrics\/}~{\em 211}, 461--482.


\bibitem[\protect\citeauthoryear{Beutner, Lin, and Smeekes}{Beutner
  et~al.}{2020}]{beutnerlinsmeekes2019}
Beutner, E., Y.~Lin, and S.~Smeekes (2020).
\newblock {GLS} estimation and confidence sets for the date of a single break
  in models with trends.
\newblock Working Paper.

\bibitem[\protect\citeauthoryear{Carson}{Carson}{2009}]{carson2009}
Carson, R.~T. (2009).
\newblock The {E}nvironmental {K}uznets {C}urve: seeking empirical regularity
  and theoretical structure.
\newblock {\em Review of Environmental Economics and Policy\/}~{\em 4}, 3--23.


\bibitem[\protect\citeauthoryear{Chan and Wang}{Chan and
  Wang}{2015}]{chanwang2015}
Chan, N. and Q.~Wang (2015).
\newblock Nonlinear regressions with nonstationary time series.
\newblock {\em Journal of Econometrics\/}~{\em 185}, 182--195.


\bibitem[\protect\citeauthoryear{Chang, Park, and Phillips}{Chang
  et~al.}{2001}]{changparkphillips2001}
Chang, Y., J.~Y. Park, and P.~C.~B. Phillips (2001).
\newblock Nonlinear econometric models with cointegrated and deterministically
  trending regressors.
\newblock {\em The Econometrics Journal\/}~{\em 4}, 1--36.


\bibitem[\protect\citeauthoryear{Cho and Phillips}{Cho and
  Phillips}{2018}]{chophillips2018}
Cho, J.~S. and P.~C.~B. Phillips (2018).
\newblock Sequentially testing polynomial model hypotheses using power
  transforms of regressors.
\newblock {\em Journal of Applied Econometrics\/}~{\em 33}, 141--159.


\bibitem[\protect\citeauthoryear{Choi and Saikkonen}{Choi and
  Saikkonen}{2010}]{choisaikkonen2010}
Choi, I. and P.~Saikkonen (2010).
\newblock Tests for nonlinear cointegration.
\newblock {\em Econometric Theory\/}~{\em 26}, 682--709.


\bibitem[\protect\citeauthoryear{Cole}{Cole}{2005}]{cole2005}
Cole, M. (2005).
\newblock Re-examining the pollution-income relationship: {A} random
  coefficients approach.
\newblock {\em Economics Bulletin\/}~{\em 14}, 1--7.


\bibitem[\protect\citeauthoryear{Dasgupta, Wang, and Wheeler}{Dasgupta
  et~al.}{2002}]{dasguptaetal2002}
Dasgupta, S., B.~K.~H. Wang, and D.~Wheeler (2002).
\newblock Confronting the {E}nvironmental {K}uznets {C}urve.
\newblock {\em Journal of Economic Perspectives\/}~{\em 16}, 147--168.


\bibitem[\protect\citeauthoryear{Davies}{Davies}{1977}]{davies1977}
Davies, R.~B. (1977).
\newblock Hypothesis testing when a nuisance parameter is present only under
  the alternative.
\newblock {\em Biometrika\/}~{\em 64}, 247--254.


\bibitem[\protect\citeauthoryear{Davies}{Davies}{1987}]{davies1987}
Davies, R.~B. (1987).
\newblock Hypothesis testing when a nuisance parameter is present only under
  the alternative.
\newblock {\em Biometrika\/}~{\em 74}, 33--43.


\bibitem[\protect\citeauthoryear{Dijkgraaf and Vollebergh}{Dijkgraaf and
  Vollebergh}{2005}]{dijkgraafvollebergh2005}
Dijkgraaf, E. and H.~R.~J. Vollebergh (2005).
\newblock A test for parameter homogeneity in {CO}\textsubscript{2} panel {EKC}
  estimations.
\newblock {\em Environmental and Resource Economics\/}~{\em 32}, 229--239.


\bibitem[\protect\citeauthoryear{Dong, Gao, Tj{\o}stheim, and Yin}{Dong
  et~al.}{2017}]{donggaotjostheimyin2017}
Dong, C., J.~Gao, D.~Tj{\o}stheim, and J.~Yin (2017).
\newblock Specification testing for nonlinear multivariate cointegrating
  regressions.
\newblock {\em Journal of Econometrics\/}~{\em 200}, 104--117.


\bibitem[\protect\citeauthoryear{Dufour}{Dufour}{2006}]{dufour2006}
Dufour, J.-M. (2006).
\newblock Monte {C}arlo tests with nuisance parameters: A general approach to
  finite-sample inference and nonstandard asymptotics.
\newblock {\em Journal of Econometrics\/}~{\em 133}, 443--477.


\bibitem[\protect\citeauthoryear{Dufour and Khalaf}{Dufour and
  Khalaf}{2002}]{dufourkhalaf2002}
Dufour, J.-M. and L.~Khalaf (2002).
\newblock Simulation based finite and large sample tests in multivariate
  regressions.
\newblock {\em Journal of Econometrics\/}~{\em 111}, 303--322.


\bibitem[\protect\citeauthoryear{Duggal, Saltzman, and Klein}{Duggal
  et~al.}{1999}]{duggalsaltzmanklein1999}
Duggal, V.~G., C.~Saltzman, and L.~R. Klein (1999).
\newblock Infrastructure and productivity: a nonlinear approach.
\newblock {\em Journal of Econometrics\/}~{\em 92}, 47--74.


\bibitem[\protect\citeauthoryear{Duggal, Saltzman, and Klein}{Duggal
  et~al.}{2007}]{duggal2007infrastructure}
Duggal, V.~G., C.~Saltzman, and L.~R. Klein (2007).
\newblock Infrastructure and productivity: An extension to private
  infrastructure and it productivity.
\newblock {\em Journal of Econometrics\/}~{\em 140}, 485--502.


\bibitem[\protect\citeauthoryear{Galeotti, Manera, and Lanza}{Galeotti
  et~al.}{2009}]{galeottimaneralanza2009}
Galeotti, M., M.~Manera, and A.~Lanza (2009).
\newblock On the robustness of robustness checks of the environmental {K}uznets
  curve hypothesis.
\newblock {\em Environmental and Resource Economics\/}~{\em 42}, 551.


\bibitem[\protect\citeauthoryear{Gao, Kanaya, Li, and Tj{\o}stheim}{Gao
  et~al.}{2015}]{gaokanayalitjostheim2015}
Gao, J., S.~Kanaya, D.~Li, and D.~Tj{\o}stheim (2015).
\newblock Uniform consistency for nonparametric estimators in null recurrent
  time series.
\newblock {\em Econometric Theory\/}~{\em 31}, 911--952.


\bibitem[\protect\citeauthoryear{Gao, Linton, and Peng}{Gao
  et~al.}{2020}]{gaolintonpeng2020}
Gao, J., O.~Linton, and B.~Peng (2020).
\newblock Inference on a semiparametric model with global power law and local
  nonparametric trends.
\newblock {\em Econometric Theory\/}~{\em 36}, 223--249.


\bibitem[\protect\citeauthoryear{Gillingham and Stock}{Gillingham and
  Stock}{2018}]{gillinghamstock2018}
Gillingham, K. and J.~H. Stock (2018).
\newblock The cost of reducing greenhouse gas emissions.
\newblock {\em Journal of Economic Perspectives\/}~{\em 32}, 53--72.


\bibitem[\protect\citeauthoryear{Gillingham and Nordhaus}{Gillingham and
  Nordhaus}{2018}]{christensengillinghamnordhaus2018}
Gillingham, P. C.~K. and W.~Nordhaus (2018).
\newblock Uncertainty in forecasts of long-run economic growth.
\newblock {\em Proceedings of the National Academy of Sciences\/}~{\em 115},
  5409--5414.


\bibitem[\protect\citeauthoryear{Grossman and Krueger}{Grossman and
  Krueger}{1995}]{grossmankrueger1995}
Grossman, G.~M. and A.~B. Krueger (1995).
\newblock Economic growth and the environment.
\newblock {\em The Quarterly Journal of Economics\/}~{\em 110}, 353--377.


\bibitem[\protect\citeauthoryear{Hamilton}{Hamilton}{1994}]{hamilton1994}
Hamilton, J.~D. (1994).
\newblock {\em Time Series Analysis}.
\newblock Princeton University Press.


\bibitem[\protect\citeauthoryear{Harbaugh, Levinson, and Wilson}{Harbaugh
  et~al.}{2002}]{harbaughlevinsonwilson2002}
Harbaugh, W.~T., A.~Levinson, and D.~M. Wilson (2002).
\newblock Reexamining the empirical evidence for an environmental {K}uznets
  curve.
\newblock {\em Review of Economics and Statistics\/}~{\em 84}, 541--551.


\bibitem[\protect\citeauthoryear{Hong and Phillips}{Hong and
  Phillips}{2010}]{hongphillips2010}
Hong, S.~H. and P.~C.~B. Phillips (2010).
\newblock Testing linearity in cointegrating relations with an application to
  purchasing power parity.
\newblock {\em Journal of Business \& Economic Statistics\/}~{\em 28}, 96--114.


\bibitem[\protect\citeauthoryear{Hu, Phillips, and Wang}{Hu
  et~al.}{2021}]{huphillipswang2019}
Hu, Z., P.~C.~B. Phillips, and Q.~Wang (2021).
\newblock Nonlinear cointegrating power function regression with endogeneity.
\newblock {\em Econometric Theory\/}~{\em 37}, 1173--1213.


\bibitem[\protect\citeauthoryear{Jalil and Feridun}{Jalil and
  Feridun}{2011}]{jalilferdiun2011}
Jalil, A. and M.~Feridun (2011).
\newblock The impact of growth, energy and financial development on the
  environment in {C}hina: {A} cointegration analysis.
\newblock {\em Energy Economics\/}~{\em 33}, 284--291.


\bibitem[\protect\citeauthoryear{Jansson}{Jansson}{2002}]{jansson2002}
Jansson, M. (2002).
\newblock Consistent covariance matrix estimation for linear processes.
\newblock {\em Econometric Theory\/}~{\em 18}, 1449--1459.


\bibitem[\protect\citeauthoryear{Jiang, Lu, and Park}{Jiang
  et~al.}{2019}]{jianglupark2019}
Jiang, B., Y.~Lu, and J.~Y. Park (2019).
\newblock Testing for stationary at high frequency.
\newblock {\em Journal of Econometrics\/}~{\em 215}, 341--374.


\bibitem[\protect\citeauthoryear{Klein, Duggal, and Saltzman}{Klein
  et~al.}{2004}]{klein2004}
Klein, L.~R., V.~G. Duggal, and C.~Saltzman (2004).
\newblock Contributions of input-output analysis to the understanding of
  technological change: the information sector in the {U}nited {S}tates.
\newblock In E.~Dietzenbacher and M.~L. Lahr (Eds.), {\em Wassily Leontief and
  Input-output Economics}, Chapter~17, pp.\  311--336. Cambridge University
  Press.

\bibitem[\protect\citeauthoryear{Kwiatkowski, Phillips, Schmidt, and
  Shin}{Kwiatkowski et~al.}{1992}]{kwiatkowskiphillipsschmidtshin1992}
Kwiatkowski, D., P.~C.~B. Phillips, P.~Schmidt, and Y.~Shin (1992).
\newblock Testing the null hypothesis of stationarity against the alternative
  of a unit root: How sure are we that economic time series have a unit root?
\newblock {\em Journal of Econometrics\/}~{\em 54}, 159--178.


\bibitem[\protect\citeauthoryear{Li and Linton}{Li and
  Linton}{2020}]{lilinton2020}
Li, S. and O.~Linton (2020).
\newblock When will the covid-19 pandemic peak?
\newblock {\em Journal of Econometrics\/}~{\em 220}, 130--157.


\bibitem[\protect\citeauthoryear{Lin and Reuvers}{Lin and
  Reuvers}{2020}]{linreuvers2019}
Lin, Y. and H.~Reuvers (2020).
\newblock Efficient estimation by fully modified \uppercase{GLS} with an
  application to the {E}nvironmental {K}uznets {C}urve.
\newblock Working paper.

\bibitem[\protect\citeauthoryear{Lin, Tu, and Yao}{Lin
  et~al.}{2020}]{lintuyao2020}
Lin, Y., Y.~Tu, and Q.~Yao (2020).
\newblock Estimation for double-nonlinear cointegration.
\newblock {\em Journal of Econometrics\/}~{\em 216}, 175--191.


\bibitem[\protect\citeauthoryear{Linton and Wang}{Linton and
  Wang}{2016}]{lintonwang2016}
Linton, O. and Q.~Wang (2016).
\newblock Nonparametric transformation regression with nonstationary data.
\newblock {\em Econometric Theory\/}~{\em 32}, 1--29.


\bibitem[\protect\citeauthoryear{List and Gallet}{List and
  Gallet}{1999}]{listgallet1999}
List, J.~A. and C.~A. Gallet (1999).
\newblock The environmental {K}uznets curve: {D}oes one size fit all?
\newblock {\em Ecological Economics\/}~{\em 31}, 409--423.


\bibitem[\protect\citeauthoryear{Maranzano, Bento, and Manera}{Maranzano
  et~al.}{2021}]{maranzanobentomanera2021}
Maranzano, P., J.~P.~C. Bento, and M.~Manera (2021).
\newblock The role of education and income inequality on environmental quality.
  {A} panel data analysis of the {EKC} hypothesis on {OECD}.
\newblock Nota di Lavoro 8.2021, Milano, Italy: Fondazione Eni Enrico Mattei.

\bibitem[\protect\citeauthoryear{Mazzanti and Musolesi}{Mazzanti and
  Musolesi}{2013}]{mazzantimusolesi2013}
Mazzanti, M. and A.~Musolesi (2013).
\newblock The heterogeneity of carbon kuznets curves for advanced countries:
  Comparing homogeneous, heterogeneous and shrinkage/bayesian estimators.
\newblock {\em Applied Economics\/}~{\em 45}, 3827--3842.


\bibitem[\protect\citeauthoryear{Nordhaus}{Nordhaus}{2013}]{nordhaus2013}
Nordhaus, W.~D. (2013).
\newblock {\em The Climate Casino: Risk, Uncertainty, and Economics for a
  Warming World}.
\newblock Yale University Press.


\bibitem[\protect\citeauthoryear{Nordhaus}{Nordhaus}{2014}]{nordhaus2014}
Nordhaus, W.~D. (2014).
\newblock The perils of the learning model for modeling endogenous
  technological change.
\newblock {\em The Energy Journal\/}~{\em 35}, 1--13.


\bibitem[\protect\citeauthoryear{Panayotou}{Panayotou}{1997}]{panayotou1997}
Panayotou, T. (1997).
\newblock Demystifying the environmental kuznets curve: Turning a black box
  into a policy tool.
\newblock {\em Environment and Development Economics\/}~{\em 2}, 465--484.


\bibitem[\protect\citeauthoryear{Park}{Park}{2002}]{park2002}
Park, J.~Y. (2002).
\newblock An invariance principle for sieve bootstrap in time series.
\newblock {\em Econometric Theory\/}~{\em 18}, 469--490.


\bibitem[\protect\citeauthoryear{Park and Phillips}{Park and
  Phillips}{1999}]{parkphillips1999}
Park, J.~Y. and P.~C.~B. Phillips (1999).
\newblock Asymptotics for nonlinear transformations of integrated time series.
\newblock {\em Econometric Theory\/}~{\em 15}, 269--298.


\bibitem[\protect\citeauthoryear{Park and Phillips}{Park and
  Phillips}{2001}]{parkphillips2001}
Park, J.~Y. and P.~C.~B. Phillips (2001).
\newblock Nonlinear regressions with integrated time series.
\newblock {\em Econometrica\/}~{\em 69}, 117--161.


\bibitem[\protect\citeauthoryear{Perron and Yabu}{Perron and
  Yabu}{2009}]{perronyabu2009}
Perron, P. and T.~Yabu (2009).
\newblock Estimating deterministic trends with an integrated or stationary
  noise component.
\newblock {\em Journal of Econometrics\/}~{\em 151}, 56--69.


\bibitem[\protect\citeauthoryear{Perron and Zhu}{Perron and
  Zhu}{2005}]{perronzhu2005}
Perron, P. and X.~Zhu (2005).
\newblock Structural breaks with deterministic and stochastic trends.
\newblock {\em Journal of Econometrics\/}~{\em 129}, 65--119.


\bibitem[\protect\citeauthoryear{Phillips}{Phillips}{1995}]{phillips1995}
Phillips, P. C.~B. (1995).
\newblock Fully modified least squares and vector autoregression.
\newblock {\em Econometrica\/}~{\em 63}, 1023--1078.


\bibitem[\protect\citeauthoryear{Phillips}{Phillips}{2007}]{phillips2007}
Phillips, P. C.~B. (2007).
\newblock Regression with slowly varying regressors and nonlinear trends.
\newblock {\em Econometric Theory\/}~{\em 23}, 557--614.


\bibitem[\protect\citeauthoryear{Phillips and Hansen}{Phillips and
  Hansen}{1990}]{phillipshansen1990}
Phillips, P. C.~B. and B.~E. Hansen (1990).
\newblock Statistical inference in instrumental variables regression with
  {I}(1) processes.
\newblock {\em The Review of Economic Studies\/}~{\em 57}, 99--125.


\bibitem[\protect\citeauthoryear{Piaggio and Padilla}{Piaggio and
  Padilla}{2012}]{piaggiopadilla2012}
Piaggio, M. and E.~Padilla (2012).
\newblock {CO}\textsubscript{2} emissions and economic activity: Heterogeneity
  across countries and non-stationary series.
\newblock {\em Energy Policy\/}~{\em 46}, 370--381.


\bibitem[\protect\citeauthoryear{Robinson}{Robinson}{2012}]{robinson2012}
Robinson, P.~M. (2012).
\newblock Inference on power law spatial trends.
\newblock {\em Bernoulli\/}~{\em 18}, 644--677.


\bibitem[\protect\citeauthoryear{Romano and Wolf}{Romano and
  Wolf}{2001}]{romanowolf2001}
Romano, J.~P. and M.~Wolf (2001).
\newblock Subsampling intervals in autoregressive models with linear time
  trend.
\newblock {\em Econometrica\/}~{\em 69}, 1283--1314.


\bibitem[\protect\citeauthoryear{Saikkonen}{Saikkonen}{1992}]{saikkonen1992}
Saikkonen, P. (1992).
\newblock Estimation and testing of cointegrated systems by an autoregressive
  approximation.
\newblock {\em Econometric Theory\/}~{\em 8}, 1--27.


\bibitem[\protect\citeauthoryear{Selden and Song}{Selden and
  Song}{1994}]{seldensong1994}
Selden, T.~M. and D.~Song (1994).
\newblock Environmental quality and development: {I}s there a {K}uznets curve
  for air pollution emissions?
\newblock {\em Journal of Environmental Economics and Management\/}~{\em 27},
  147--162.


\bibitem[\protect\citeauthoryear{Soong}{Soong}{1973}]{soong1973}
Soong, T.~T. (1973).
\newblock {\em Random Differential Equations in Science and Engineering}.
\newblock Academic Press, Inc.


\bibitem[\protect\citeauthoryear{Stern}{Stern}{2004}]{stern2004}
Stern, D.~I. (2004).
\newblock The rise and fall of the environmental kuznets curve.
\newblock {\em World Development\/}~{\em 32}, 1419--1439.


\bibitem[\protect\citeauthoryear{Stern}{Stern}{2017}]{stern2017}
Stern, D.~I. (2017).
\newblock The environmental {K}uznets curve after 25 years.
\newblock {\em Journal of Bioeconomics\/}~{\em 19}, 7--28.


\bibitem[\protect\citeauthoryear{Stypka, Wagner, Grabarczyk, and Kawka}{Stypka
  et~al.}{2017}]{stypkawagnergrabarczykkawka2017}
Stypka, O., M.~Wagner, P.~Grabarczyk, and R.~Kawka (2017).
\newblock The asymptotic validity of ``standard'' fully modified {OLS}
  estimation and inference in cointegrating polynomial regressions.
\newblock Working Paper.

\bibitem[\protect\citeauthoryear{Tanaka}{Tanaka}{2017}]{tanaka2017}
Tanaka, K. (2017).
\newblock {\em Time Series Analysis: Nonstationary and Noninvertible
  Distribution Theory}.
\newblock John Wiley \& Sons.


\bibitem[\protect\citeauthoryear{Vogelsang and Wagner}{Vogelsang and
  Wagner}{2014}]{vogelsangwagner2014}
Vogelsang, T.~J. and M.~Wagner (2014).
\newblock Integrated modified {OLS} estimation and fixed-b inference for
  cointegrating regressions.
\newblock {\em Journal of Econometrics\/}~{\em 178}, 741--760.


\bibitem[\protect\citeauthoryear{Vollebergh, Melenberg, and
  Dijkgraaf}{Vollebergh et~al.}{2009}]{volleberghmelenbergdijkgraaf2009}
Vollebergh, H. R.~J., B.~Melenberg, and E.~Dijkgraaf (2009).
\newblock Identifying reduced-form relations with panel data: {T}he case of
  pollution and income.
\newblock {\em Journal of Environmental Economics and Management\/}~{\em 58},
  27--42.


\bibitem[\protect\citeauthoryear{Wagner}{Wagner}{2015}]{wagner2015}
Wagner, M. (2015).
\newblock The {E}nvironmental {K}uznets {C}urve, cointegration and
  nonlinearity.
\newblock {\em Journal of Applied Econometrics\/}~{\em 30}, 948--967.


\bibitem[\protect\citeauthoryear{Wagner, Grabarczyk, and Hong}{Wagner
  et~al.}{2020}]{wagnergrabarczykhong2019}
Wagner, M., P.~Grabarczyk, and S.~H. Hong (2020).
\newblock Fully modified {OLS} estimation and inference for seemingly unrelated
  cointegrating polynomial regressions and the {E}nvironmental {K}uznets
  {C}urve for carbon dioxide emissions.
\newblock {\em Journal of Econometrics\/}~{\em 214}, 216--255.


\bibitem[\protect\citeauthoryear{Wagner and Hong}{Wagner and
  Hong}{2016}]{wagnerhong2016}
Wagner, M. and S.~H. Hong (2016).
\newblock Cointegrating polynomial regressions: Fully modified {OLS} estimation
  and inference.
\newblock {\em Econometric Theory\/}~{\em 32}, 1289--1315.


\bibitem[\protect\citeauthoryear{Wang and Phillips}{Wang and
  Phillips}{2009}]{wangphillips2009}
Wang, Q. and P.~C.~B. Phillips (2009).
\newblock Structural nonparametric cointegrating regression.
\newblock {\em Econometrica\/}~{\em 77}, 1901--1948.


\bibitem[\protect\citeauthoryear{Wang and Phillips}{Wang and
  Phillips}{2012}]{wangphillips2012}
Wang, Q. and P.~C.~B. Phillips (2012).
\newblock A specification test for nonlinear nonstationary models.
\newblock {\em The Annals of Statistics\/}~{\em 40}, 727--758.


\bibitem[\protect\citeauthoryear{Wang and Phillips}{Wang and
  Phillips}{2016}]{wangphillips2016}
Wang, Q. and P.~C.~B. Phillips (2016).
\newblock Nonparametric cointegrating regression with endogeneity and long
  memory.
\newblock {\em Econometric Theory\/}~{\em 32}, 359--401.


\bibitem[\protect\citeauthoryear{Wang, Wu, and Zhu}{Wang
  et~al.}{2018}]{wangwuzhu2018}
Wang, Q., D.~Wu, and K.~Zhu (2018).
\newblock Model checks for nonlinear cointegrating regression.
\newblock {\em Journal of Econometrics\/}~{\em 207}, 261--284.


\end{thebibliography}
\end{spacing}

\clearpage
\begin{appendices}
 \section{Proofs for Main Theorems}
 
 \textbf{Proof of Theorem \ref{thm:limiting_dist}} In view of the identity $\norm{\va + \vb}^2=\norm{\va}^2+\norm{\vb}^2+2\va\tran\vb$, we also have
 \begin{equation*}
Q_T(\vgamma)=\frac{1}{2}\sum_{t=1}^{T}\left\|\vy_t-\mZ_t'\vbeta\right\|^2-\tau_g\sum_{t=1}^{T}t^{\theta}\left(\vy_t-\mZ_t'\vbeta\right)'\vones_N+\frac{1}{2}N\tau_g^2\sum_{t=1}^{T}t^{2\theta}.
\end{equation*}
The proof proceeds along the lines of Lemma 1 of \cite{andrewssun2004} and Theorem 3.1 of \cite{chanwang2015}. The proofs separate into two parts. The first part  uses a Taylor expansion of $Q_T(\vgamma)$ around $Q_T(\vgamma_0)$ to recover a quadratic approximation for $Q_T(\vgamma)$ on the set $\mGamma_{\delta,k_T}\subseteq \mGamma$. In the second part, we obtain the limiting distribution from this quadratic approximation.

\medskip
\noindent \emph{Part 1}: Let $\{k_T,T\geq 1\}$ denote a deterministic sequence such that $k_T\to\infty$ as $T\to\infty$.  Define $\mGamma_{\delta,k_T}=\left\{\vgamma\in\mGamma:~ \left\|\mG_{\vgamma_0,T}\left(\vgamma-\vgamma_0\right)\right\|\leq k_T,~\left\|\vgamma-\vgamma_0\right\|\leq \delta\right\}$ and select a $\delta>0$ such that $Q_T(\cdot)$ is twice differentiable on $\{\vgamma\in\SR^{p+2}:\,\left\|\vgamma-\vgamma_0\right\|\leq \delta\}\subset\mGamma$. For any $\vgamma\in\mGamma_{\delta,k_T}$, the Taylor expansion of $Q_T(\vgamma)$ around $\vgamma_0$ reads
\begin{equation}
\begin{aligned}
Q_T(\vgamma) &-Q_T(\vgamma_0)
 = \dot{Q}_T'(\vgamma_0)(\vgamma-\vgamma_0)+\frac{1}{2}(\vgamma-\vgamma_0)\tran\ddot{Q}_T(\bar{\vgamma})(\vgamma-\vgamma_0) \\
 &= \dot{Q}_T'(\vgamma_0)(\vgamma-\vgamma_0)+\frac{1}{2}(\vgamma-\vgamma_0)\tran\left[ \ddot{Q}_T(\bar{\vgamma}) - \ddot{Q}_T(\vgamma_0) - \ddot{Q}_{T,2}(\vgamma_0) + \ddot{Q}_{T,1}(\vgamma_0) \right](\vgamma-\vgamma_0),
\end{aligned}
\label{eq:taylorexpansionQT}
\end{equation}
where $\bar{\vgamma}$ is a point on the line segment connecting $\vgamma$ and $\vgamma_0$, and the various derivatives of $Q_T$ are
\begin{equation*}
\begin{aligned}
\dot{Q}_T(\vgamma_0)&=-\sum_{t=1}^{T}\left[
\begin{smallmatrix}
\vones_N'\tau_{g0}\,t^{\theta_0}{\ln t}\\
\vones_N'\,t^{\theta_0}\\
\mZ_t
\end{smallmatrix}
\right]\vu_t,\\
\ddot{Q}_T(\vgamma)&=\sum_{t=1}^{T}\left[
\begin{smallmatrix}
-\tau_gt^{\theta}(\ln{t})^2\left(\vy_t-\mZ_t'\vbeta\right)'\vones_N+2N\tau_g^2t^{2\theta}(\ln{t})^2\; & -t^{\theta}\ln{t}\left(\vy_t-\mZ_t'\vbeta\right)'\vones_N+2N\tau_gt^{2\theta}\ln{t}\; & \tau_gt^{\theta}\ln{t}\,\vones_N'\mZ_t' \\
-t^{\theta}\ln{t}\left(\vy_t-\mZ_t'\vbeta\right)'\vones_N+2N\tau_gt^{2\theta}\ln{t} & Nt^{2\theta} & t^{\theta}\vones_N'\mZ_t'\\	
\tau_gt^{\theta}\ln{t}\,\mZ_t\vones_N & t^{\theta}\mZ_t\vones_N & \mZ_t\mZ_t'
\end{smallmatrix}
\right], \\
\ddot{Q}_T(\vgamma_0)&=\sum_{t=1}^{T}\left[
\begin{smallmatrix}
N\tau_{g0}^2t^{2\theta_0}(\ln{t})^2-\tau_{g0}t^{\theta_0}(\ln{t})^2\vu_t'\vones_N\; & N\tau_{g0}t^{2\theta_0}\ln{t}-t^{\theta_0}\ln{t}\,\vu_t'\vones_N\; & \tau_{g0}t^{\theta_0}\ln{t}\vones_N'\mZ_t'\\
N\tau_{g0}t^{2\theta_0}\ln{t}-t^{\theta_0}\ln{t}\,\vu_t'\vones_N & Nt^{2\theta_0} & t^{\theta_0}\vones_N'\mZ_t'\\
\tau_{g0}t^{\theta_0}\ln{t}\,\mZ_t\vones_N & t^{\theta_0}\mZ_t\vones_N & \mZ_t\mZ_t'
\end{smallmatrix}\right]\\
&=\sum_{t=1}^{T}
\begin{bmatrix}
N\tau_{g0}^2t^{2\theta_0}(\ln{t})^2 & N\tau_{g0}t^{2\theta_0}\ln{t} & \tau_{g0}t^{\theta_0}\ln{t}\,\vones_N'\mZ_t'\\
N\tau_{g0}t^{2\theta_0}\ln{t} & N t^{2\theta_0} & t^{\theta_0}\vones_N'\mZ_t'\\ 
\tau_{g0} t^{\theta_0}\ln{t}\,\mZ_t\vones_N & t^{\theta_0}\mZ_t\vones_N & \mZ_t\mZ_t'
\end{bmatrix}
-\sum_{t=1}^{T}t^{\theta_0}\ln{t}
\begin{bmatrix}
\tau_{g0}\ln{t} & 1 &  \\
1 & 0 &  \\
 &  & \mZeros
\end{bmatrix}\vu_t'\vones_N\\
&=:\ddot{Q}_{T,1}(\vgamma_0)-\ddot{Q}_{T,2}(\vgamma_0).
\end{aligned}
\end{equation*}
Defining $
R_T(\bar{\vgamma},\vgamma_0)=\frac{1}{2}(\vgamma-\vgamma_0)'\left[\ddot{Q}_T(\bar{\vgamma})-\ddot{Q}_T(\vgamma_0)-\ddot{Q}_{T,2}(\vgamma_0)\right](\vgamma-\vgamma_0)
$, $\mA_T:=\mG_{\vgamma_0,T}^{\prime-1}\ddot{Q}_{T,1}^{}(\vgamma_0)\mG_{\vgamma_0,T}^{-1}$, and $\vb_T:=-\mG_{\vgamma_0,T}^{\prime-1}\dot{Q}_T^{}(\vgamma_0)$, we finally arrive at
\begin{equation}
\begin{aligned}
Q_T&(\vgamma) -Q_T(\vgamma_0) =-\vb_T'\big[\mG_{\vgamma_0,T}(\vgamma-\vgamma_0)\big]+\frac{1}{2}\big[\mG_{\vgamma_0,T}(\vgamma-\vgamma_0)\big]'\mA_T\big[\mG_{\vgamma_0,T}(\vgamma-\vgamma_0)\big]+R_T(\bar{\vgamma},\vgamma_0) \\
&=\frac{1}{2}\left[\mG_{\vgamma_0,T}(\vgamma-\vgamma_0)-\mA_T^{-1}\vb_T^{}\right]'\mA_T^{}\left[\mG_{\vgamma_0,T}(\vgamma-\vgamma_0)-\mA_T^{-1}\vb_T^{}\right]-\frac{1}{2}\vb_T'\mA_T^{-1}\vb_T^{}+R_T(\bar{\vgamma},\vgamma_0).
\end{aligned}
\label{eq:taylorexpansionQTfull}
\end{equation}
\emph{Part 2}: For any $\varepsilon>0$, let $\mGamma_{T}(\varepsilon)=\big\{\vgamma\in\mGamma:~\big\|\mG_{\vgamma_0,T}(\vgamma-\vgamma_0)-\mA_T^{-1}\vb_T^{}\big\|\leq \varepsilon\big\}$. We shall show that the minimum of $Q_T(\cdot)$ over $\vgamma\in \mGamma_{T}(\varepsilon)$ is attained in the interior of $\mGamma_{T}(\varepsilon)$. The next two statements are proven later:
\begin{enumerate}[(a)]
	\item $\sup_{\vgamma\in\mGamma_{\delta,k_T}}\left\|\mG_{\vgamma_0,T}^{\prime-1}\big[\ddot{Q}_T(\vgamma)-\ddot{Q}_T(\vgamma_0)\big]\mG_{\vgamma_0,T}^{-1}\right\|=o_p(1)$;
	\item $\mA_T^{-1}\vb_T^{}=O_p(1)$, where $\mA_T\overset{T\rightarrow\infty}{\wto} \mA_{\infty}$ with $\Prob\left(\mA_{\infty}>0\right)=1$.
\end{enumerate}
Given claim (b), for any $\varepsilon>0$, we have $\Prob\left(\mGamma_{T}(\varepsilon)\subset\mGamma_{\delta,k_T}\right)\rightarrow 1$ as $T\rightarrow\infty$ because $\|\mG_{\vgamma_0,T}^{-1}\|\to 0$. Define $\vgamma_T^*=\vgamma_0+\mG_{\vgamma_0,T}^{-1}\mA_T^{-1}\vb_T^{}$. Clearly, $\vgamma_T^*$ is an interior point of $\mGamma_{T}(\varepsilon)$ as long as $\varepsilon>0$. Subsequently select a $\vgamma_{\varepsilon}\in \partial\mGamma_{T}(\varepsilon)$, i.e. $\vgamma_{\varepsilon}$ is a boundary point of $\mGamma_{T}(\varepsilon)$. From \eqref{eq:taylorexpansionQTfull}, we have
\begin{equation*}
\begin{aligned}
Q_T(\vgamma_{\varepsilon})-Q_T(\vgamma_T^*)
&=\left[Q_T(\vgamma_{\varepsilon})-Q_T(\vgamma_0)\right]-\left[Q_T(\vgamma_T^*)-Q_T(\vgamma_0)\right]\\
&=\left[\frac{1}{2}\vmu_T'\mA_T^{}\vmu_T^{}-\frac{1}{2}\vb_T'\mA_T^{-1}\vb_T^{}+R_T(\bar{\vgamma}_{\varepsilon},\vgamma_0)\right]-\left[-\frac{1}{2}\vb_T'\mA_T^{-1}\vb_T^{}+R_T(\bar{\vgamma}_{T}^*,\vgamma_0)\right]\\
&=\frac{1}{2}\vmu_T'\mA_T^{}\vmu_T^{}+R_T(\bar{\vgamma}_{\varepsilon},\vgamma_0)-R_T(\bar{\vgamma}_{T}^*,\vgamma_0)=\frac{1}{2}\vmu_T'\mA_T^{}\vmu_T^{}+o_p(1),
\end{aligned}
\end{equation*}
where $\vmu_T$ a random vector with $\|\vmu_T\|=\varepsilon$, and $\bar{\vgamma}_{\varepsilon}$ is a point on the line segment connecting $\vgamma_{\varepsilon}$ and $\vgamma_0$. The point $\bar{\vgamma}_{T}^*$ is defined similarly. Moreover, the final equality follows from
\begin{equation*}
R_T(\bar{\vgamma},\vgamma_0)\leq \frac{1}{2}\left\|\mG_{\vgamma_0,T}\left(\vgamma-\vgamma_0\right)\right\|^2\left\{\sup_{\vgamma\in\mGamma_{\delta,k_T}}\left\|\mG_{\vgamma_0,T}^{\prime-1}\big[\ddot{Q}_T(\vgamma)-\ddot{Q}_T(\vgamma_0)\big]\mG_{\vgamma_0,T}^{-1}\right\|+\left\|\mG_{\vgamma_0,T}^{\prime-1}\ddot{Q}_{T,2}(\vgamma_0)\mG_{\vgamma_0,T}^{-1}\right\|\right\},
\end{equation*}
claim (a), and $\big\|\mG_{\vgamma_0,T}^{\prime-1}\ddot{Q}_{T,2}(\vgamma_0)\mG_{\vgamma_0,T}^{-1}\big\|\leq o_p(1) \big| T^{-1}\sum_{t=1}^{T}\left(\frac{t}{T}\right)^{\theta_{0}+1/2}\big|=o_p(1)$. The second part of claim (b), $\p\left(\mA_{\infty}>0\right)=1$, implies that $\p\big(\frac{1}{2}\vmu_T'\mA_T^{}\vmu_T^{}>0\big)\rightarrow 1$ as $T\rightarrow\infty$. Therefore, $\Prob\big(Q_T(\vgamma_{\varepsilon})>Q_T(\vgamma_T^*)\big)\rightarrow 1$ for any boundary point $\vgamma_{\varepsilon}$ and the minimum of $Q_T$ must be attained at an interior point of $\mGamma_{T}(\varepsilon)$, say $\widehat{\vgamma}_T(\varepsilon)$. As in \cite{andrewssun2004}, we can now construct a sequence $\{\widehat{\vgamma}_T\}$ such that $\widehat{\vgamma}_T^{}=\widehat{\vgamma}_T^{}(J_T^{-1})\in\mGamma_{T}(J_T^{-1})$, where $J_T\rightarrow\infty$, satisfying the first-order conditions $\p\Big(\dot{Q}_T\big(\widehat{\vgamma}_T\big)=0\Big)\rightarrow 1$ as $T\rightarrow\infty$. As a result, we obtain
\begin{equation}\label{eq:asymp_dist_classic}
\mG_{\vgamma_0,T}(\widehat{\vgamma}_T-\vgamma_0)=\mA_T^{-1}\vb_T^{}+o_p(1). 
\end{equation}

It remains to verify claims (a) and (b). We consider the sequence $\mGamma_{\delta,k_T}$ for $k_T=\tilde{\kappa}\ln{T}$ and $\tilde{\kappa}>0$. There exists $T^*>0$ such that whenever $T>T^*$,
\begin{align*} 
\mGamma_{\delta,k_T}&\subset\left\{\text{\footnotesize 
$\vgamma\in\mGamma:~ \left\|\mG_{\vgamma_0,T}\left(\vgamma-\vgamma_0\right)\right\|\leq \tilde{\kappa} \ln{T}$
}\right\}\\
&\subset\left\{\text{\footnotesize 
$\vgamma\in\mGamma:\, T^{\theta_{0}+1/2}\left|\theta-\theta_0\right|\leq \tilde{\kappa}\ln{T},\,T^{\theta_{0}+1/2}\big|\big(\tau_g-\tau_{g0}\big)+\tau_{g0}\left(\theta-\theta_0\right)\ln{T}\big|\leq \tilde{\kappa}\ln{T},\,T^{1/2}\left\|\mD_{Z,T}\big(\vbeta-\vbeta_0\big)\right\| \leq \tilde{\kappa}\ln{T}
$}\right\}\\
&\subset \calN_{\kappa,T}(\vgamma_0),
\end{align*}
where $\calN_{\kappa,T}(\vgamma_0)$ is given in \eqref{eq:neighborhood_gamma0}, and $\kappa=C\tilde{\kappa}$ with some constant $C>0$. Claim (a) thus holds if $\sup_{\vgamma\in\calN_{\kappa,T}(\vgamma_0)}\left\|\mG_{\vgamma_0,T}^{\prime-1}\big[\ddot{Q}_T(\vgamma)-\ddot{Q}_T(\vgamma_0)\big]\mG_{\vgamma_0,T}^{-1}\right\|=o_p(1)$. Since $N$ is fixed, we can bound $\mG_{\vgamma_0,T}^{\prime-1}\big[\ddot{Q}_T(\vgamma)-\ddot{Q}_T(\vgamma_0)\big]\mG_{\vgamma_0,T}^{-1}$ element-wise. Using the identity $\left(\vy_t-\mZ_t'\vbeta\right)'\vones_N=N\tau_{g0}t^{\theta_0}-\big(\vbeta-\vbeta_0\big)'\mZ_t\vones_N+\vu_t\tran\vones_N$ and Lemma \ref{lem:supremums}, it is easily shown that the supremum of each element is indeed $o_p(1)$. Claim (b) follows directly from the weak convergence results in Lemma \ref{lem:Brownianconv}. That is, $\mA_T \wto \myint \mJ(r;\vgamma_0)\mJ(r;\vgamma_0)' ~dr$ and $\vb_T \wto \myint \mJ(r;\vgamma_0)~d\brown_{u}(r)+\vbias_{vu}$ as $T\rightarrow\infty$. Theorem \ref{thm:limiting_dist} now follows from \eqref{eq:asymp_dist_classic}. \qed 

\bigskip
\noindent
\textbf{Proof of Theorem \ref{thm:consistentLRVs}} The proof is to a large extent an application of theorem 2 in \cite{jansson2002}. We provide the details in the Supplement. \qed 

\bigskip
\noindent
\textbf{Proof of Theorem \ref{thm:subsampling}}
We abbreviate  $M=M_T$. By simple rearrangements, we obtain
\begin{equation}
\begin{aligned}
	&\left\{\text{\footnotesize $
		\mG_{\widehat{\vgamma}_T,M}^{\prime-1}\left[\sum_{m=1}^{M}\widehat{\mJ}\big(m;\widehat{\vgamma}_{T}\big)\widehat{\mJ}\big(m;\widehat{\vgamma}_{T}\big)'\right]\mG_{\widehat{\vgamma}_T,M}^{-1}
		$}
	\right\}^{-1}
	\left\{\text{\footnotesize $
		\mG_{\widehat{\vgamma}_T,M}^{\prime-1}\left[\sum_{m=1}^{M}\widehat{\mJ}\big(m;\widehat{\vgamma}_{T}\big)\,\widehat{\vmu}_m\right]+\serialbiasSimulated
		$}
	\right\}\\
	&\,\,=\mS_M^{-1}
	\left\{\text{\footnotesize $
		\mG_{\vgamma_0,M}^{\prime-1}\left[\sum_{m=1}^{M}\widehat{\mJ}\big(m;\widehat{\vgamma}_{T}\big)\widehat{\mJ}\big(m;\widehat{\vgamma}_{T}\big)'\right]\mG_{\vgamma_0,M}^{-1}
		$}
	\right\}^{-1}
	\left\{\text{\footnotesize $
		\mG_{\vgamma_0,M}^{\prime-1}\left[\sum_{m=1}^{M}\widehat{\mJ}\big(m;\widehat{\vgamma}_{T}\big)\,\widehat{\vmu}_m\right]+\mS_M^{\prime-1}\serialbiasSimulated
		$}
	\right\}\\
	&\,\,=\mS_M^{-1}
	\left\{\text{\footnotesize $
		\mG_{\vgamma_0,M}^{\prime-1}\left[\sum_{m=1}^{M}\widehat{\mJ}\big(m;\vgamma_0\big)\widehat{\mJ}\big(m;\vgamma_0\big)'\right]\mG_{\vgamma_0,M}^{-1}+\mR_{1,M}
		$}
	\right\}^{-1}
	\left\{\text{\footnotesize $
		\mG_{\vgamma_0,M}^{\prime-1}\left[\sum_{m=1}^{M}\widehat{\mJ}\big(m;\vgamma_0\big)\,\widehat{\vmu}_m\right]+\mS_M^{\prime-1}\serialbiasSimulated+\mR_{2,M}
		$}
	\right\},
\end{aligned}
\label{eq:simulatedApproachDecomp}
\end{equation}
while having defined $\widehat{\mJ}\big(m;\vgamma_0\big)=\Big[\tau_{g0}\,m^{\theta_{0}}\ln{m}\,\vones_N,m^{\theta_0}\vones_N,\widehat{\calZ}_m\tran \Big]\tran$ and the quantities
\begin{enumerate}[(a)]
 \item $\mS_M:=\mG_{\vgamma_0,M}^{}\mG_{\widehat{\vgamma}_T,M}^{-1}
	= \begin{bmatrix}
		M^{\theta_0-\widehat{\theta}_T}  & 0  & 0\\
		(\tau_{g0}-\widehat{\tau}_{g,T})M^{\theta_0-\widehat{\theta}_T}{\ln M} & M^{\theta_0-\widehat{\theta}_T}  & \vzeros_{1\times p} \\
		\vzeros_{p\times 1}              & \vzeros_{p\times 1} & \mI_p
	\end{bmatrix}$,
 \item $\mR_{1,M}=\mG_{\vgamma_0,M}^{\prime-1}\sum_{m=1}^{M}\left[\widehat{\mJ}\big(m;\widehat{\vgamma}_{T}\big)\widehat{\mJ}\big(m;\widehat{\vgamma}_{T}\big)'-\widehat{\mJ}\big(m;\vgamma_0\big)\widehat{\mJ}\big(m;\vgamma_0\big)'\right]\mG_{\vgamma_0,M}^{-1}$,
 \item $\mR_{2,M}=\mG_{\vgamma_0,M}^{\prime-1}\sum_{m=1}^{M}\left[\widehat{\mJ}\big(m;\widehat{\vgamma}_{T}\big)-\widehat{\mJ}\big(m;\vgamma_0\big)\right]\widehat{\vmu}_m$.
\end{enumerate}
\textbf{(a)} We have $\mS_M\pto \mI_{p+2}$. To see this, note that  $M^{|\widehat{\theta}_T-\theta_0|}= \exp\left((\ln M)T^{-(\theta_0+1/2)}\big|T^{\theta_{0}+1/2}\big(~\widehat{\theta}_T-\theta_{0}\big)\big|\right)\pto 1$ and $\big|\,\widehat{\tau}_{g,T}-\tau_{g0}\big|\ln M = \big|\frac{T^{\theta_{0}+1/2}}{\ln{T}}\big(\widehat{\tau}_{g,T}-\tau_{g0}\big)\big| \frac{\ln T \ln M}{T^{\theta_{0}+1/2}}=o_p(1)$. \textbf{(b)} Looking at the elements of $\mR_{1,M}$, we conclude that $\mR_{1,M}=o_p^*(1)$ if results similar to those in Lemma \ref{lem:supremums}(\emph{i})-(\emph{iii}) continue to hold. Two conditions need to be verified:
\begin{enumerate}
 \item [(b1)] $\p\left(\widehat{\vgamma}_T\in \calN_{\kappa,M}(\vgamma_0)\right) \to 1$ with $\calN_{\kappa,M}(\vgamma_0)$ similarly defined to \eqref{eq:neighborhood_gamma0},
 \item [(b2)] the stochastic order of terms remains the same when replacing $\mZ_m$ by $\widehat{\calZ}_m$, conditional on the sample $(\vx_1,\vy_1),\ldots,(\vx_T,\vy_T)$.
\end{enumerate}
 For condition (b1), using set inclusions similar to those below \eqref{eq:asymp_dist_classic}, it suffices to show $\widehat{\vgamma}_T\in\big\{\vgamma\in\mGamma:~ \|\mG_{\vgamma_0,M}\left(\vgamma-\vgamma_0\right)\|\leq \tilde{\kappa} \ln{M}\big\}$ with large probability for some $\tilde{\kappa}>0$. This is trivial because
\begin{equation*}
	\mG_{\vgamma_0,M}^{}\mG_{\vgamma_0,T}^{-1}=\sqrt{\frac{M}{T}}\begin{bmatrix}
	\left(\frac{M}{T}\right)^{\theta_0}  &                     & \\
	\tau_{g0}\left(\frac{M}{T}\right)^{\theta_0}{\ln \left(\frac{M}{T}\right)} & \left(\frac{M}{T}\right)^{\theta_0} & \\
	\vzeros_{p\times 1}                 & \vzeros_{p\times 1} & \mD_{Z,M}^{}\mD_{Z,T}^{-1}
	\end{bmatrix}=O(1),
\end{equation*}
and thus $\left\|\mG_{\vgamma_0,M}\big(\widehat{\vgamma}_T-\vgamma_0\big)\right\|\leq \left\|\mG_{\vgamma_0,M}^{}\mG_{\vgamma_0,T}^{-1}\right\|\left\|\mG_{\vgamma_0,T}\big(\widehat{\vgamma}_T-\vgamma_0\big)\right\|=O_p(1)$. Continuing with (b2), by independence between $\{\ve_m\}$ and $\big\{\widehat{\mOmega}_T,\widehat{\mDelta}_{vu}^-\big\}$, the consistency of $\widehat{\mOmega}_T$, and a FCLT for an $i.i.d.$ sequence, we may have
\begin{equation}
\frac{1}{\sqrt{M}}\sum_{m=1}^{[rM]}\begin{bmatrix} \widehat{\vmu}_m\\
	\widehat{\vupsilon}_m
	\end{bmatrix}=\widehat{\mOmega}_T^{1/2}\frac{1}{\sqrt{M}}\sum_{m=1}^{[rM]}\ve_n \dtostar \brown(r),
\label{eq:fclt_simerrors}
\end{equation}
in probability, c.f. Section 2 of \cite{park2002}. Since $\big\{\widehat{\calZ}_m\big\}$ contains partial sum processes of $\{\widehat{\vupsilon}_m\}$, its integer powers and deterministic terms, (b2) is satisfied. \textbf{(c)} We have $\big\|\mR_{2,M}\big\|\leq C\sum_{j=1}^{4}\big|\mR_{2,M,j}\big|$ where
\begin{equation}
\begin{aligned}
	\big|\mR_{2,M,1}\big|
	&=M^{-1/2}\left|\sum_{m=1}^{M}M^{-\theta_{0}}\left(m^{\widehat{\theta}_T}-m^{\theta_{0}}\right)\widehat{\vmu}_m'\vones_N^{}\right|\\
	&=O\left(\ln M\right)\left|\,\widehat{\theta}_T-\theta_{0}\right|\,\left\|M^{-1/2}\sum_{m=1}^{M}\left(\frac{m}{M}\right)^{\theta_{0}}\widehat{\vmu}_m\right\|=O_p^*\left(\frac{\ln M}{T^{\theta_{0}+1/2}}\right)=o_p^*(1),\\
	\big|\mR_{2,M,2}\big|&=M^{-1/2}\left|\sum_{m=1}^{M}M^{-\theta_{0}}\left(m^{\widehat{\theta}_T}-m^{\theta_{0}}\right)\ln\frac{m}{M}\widehat{\vmu}_m'\vones_N^{}\right|=T^{-(\theta_{0}+1/2)}O_p^*(1)=o_p^*(1),
\end{aligned}
\end{equation}
$\big|\mR_{2,M,3}\big|=M^{-1/2}\left|\,\widehat{\tau}_{g,T}-\tau_{g0}\right|\,\left|\sum_{m=1}^{M}M^{-\theta_{0}}\left(m^{\widehat{\theta}_T}-m^{\theta_{0}}\right)\ln{m}\,\widehat{\vmu}_m'\vones_N^{}\right|=O_p^*\left(\frac{\ln T (\ln M)^2}{T^{2\theta_{0}+1}}\right)=o_p^*(1)$, and $\big|\mR_{2,M,4}\big|=M^{-1/2}\left|\,\widehat{\tau}_{g,T}-\tau_{g0}\right|\,\left|\sum_{m=1}^{M}\left(\frac{m}{M}\right)^{\theta_{0}}\ln{m}\,\widehat{\vmu}_m'\vones_N^{}\right|=O_p^*\left(\frac{\ln T \ln M}{T^{\theta_{0}+1/2}}\right)=o_p^*(1)$. All these stochastic orders are a consequence of \eqref{eq:fclt_simerrors} and a straightforward modification of Lemma \ref{lem:Brownianconv}. Overall, we have $\mR_{2,M}=o_p^*(1)$.

It remains to look at the leading terms in \eqref{eq:simulatedApproachDecomp}. The elements of $\widehat{\mOmega}$ and $\widehat{\mDelta}$ are always multiplicative in the construction. From $\mS_M\pto \mI_{p+2}$, \eqref{eq:fclt_simerrors}, and Lemma \ref{lem:Brownianconv}, we have
	\begin{equation}\label{eq:simdist_1}
	\mG_{\vgamma_0,M}^{\prime-1}\left[\sum_{m=1}^{M}\widehat{\mJ}\big(m;\vgamma_0\big)\,\widehat{\vmu}_m\right]+\mS_M^{\prime-1}\serialbiasSimulated\dtostar\int_{0}^{1} \mJ(r;\vgamma_0)~d\brown_{u}(r)+
	\begin{bmatrix}
	\vzeros_{2\times 1}\\
	\mOmega_{v_1 u_1}\vb_1\\
	\vdots\\
	\mOmega_{v_N u_N}\vb_N
	\end{bmatrix}+\begin{bmatrix}
	\vzeros_{2\times 1}\\
	\mDelta^{-}_{v_1 u_1}\vb_1\\
	\vdots\\
	\mDelta^{-}_{v_N u_N}\vb_N
	\end{bmatrix},
	\end{equation}
	in probability. The last two terms in \eqref{eq:simdist_1} equal $\vbias_{vu}$, because $\mOmega+\mDelta^{-}=(\mDelta+\mDelta'-\mSigma)+(\mSigma-\mDelta')=\mDelta$. Similarly, 
	\begin{equation}
	\mG_{\vgamma_0,M}^{\prime-1}\left[\sum_{m=1}^{M}\widehat{\mJ}\big(m;\vgamma_0\big)\widehat{\mJ}\big(m;\vgamma_0\big)'\right]\mG_{\vgamma_0,M}^{-1}\dtostar
	\int_{0}^{1} \mJ(r;\vgamma_0)\mJ(r;\vgamma_0)' ~dr,\qquad\text{in probability.}
	\label{eq:simdist_2}
	\end{equation}
The theorem follows after combining the limiting distribution of these leading terms  through \eqref{eq:simulatedApproachDecomp}. \qed 
 
\bigskip
\noindent
\textbf{Proof of Theorem \ref{thm:KPSS_plus}}
Without loss of generality, we set $\ell=1$. Subsequently, note that
\begin{equation}
 q_T^{-1/2}\sum_{t=1}^{[rq_T]}\widehat{\vu}_t^+=q_T^{-1/2}\sum_{t=1}^{[rq_T]}\left(\vu_t-\mOmega_{uv}^{}\mOmega_{vv}^{-1}\vv_t\right)-\sum_{j=1}^{3}\widetilde{\mR}_{q_T,j}, \qquad r\in[0,1],
\label{eq:KPSS_decomp}
\end{equation}
where the stochastic order of the remainder terms $\widetilde{\mR}_{q_T,1}$--$\widetilde{\mR}_{q_T,3}$ follows from Lemma \ref{lem:supremums} and Theorem \ref{thm:limiting_dist}:
\begin{enumerate}[(a)]
 \item $\widetilde{\mR}_{q_T,1}=\left(\widehat{\mOmega}_{uv}^{} \widehat{\mOmega}_{vv}^{-1}-\mOmega_{uv}^{}\mOmega_{vv}^{-1}\right)q_T^{-1/2}\sum_{t=1}^{[rq_T]}\vv_t=o_p(1)$,
 \item $\widetilde{\mR}_{q_T,2}=q_T^{-1/2}\sum_{t=1}^{[rq_T]}\left(\,\widehat{\tau}_{g,T}\,t^{\widehat{\theta}_T}-\tau_{g0}\,t^{\theta_{0}}\right)\vones_N= O_p\left(\ln{T}\left(\frac{q_T}{T}\right)^{\theta_{0}+1/2}\right)$,
 \item $\widetilde{\mR}_{q_T,3}=q_T^{-1/2}\sum_{t=1}^{[rq_T]}\mZ_t'\left(\widehat{\vbeta}_T-\vbeta_0\right)=O_p\left(\left(\frac{q_T}{T}\right)^{1/2}\right)q_T^{-1}\sum_{t=1}^{[rq_T]}\left(\mD_{Z,q_T}^{-1}\mZ_t^{}\right)'\mD_{Z,q_T}^{}\mD_{Z,T}^{-1}=O_p\left(\left(\frac{q_T}{T}\right)^{1/2}\right)$.
\end{enumerate}
The theorem follows from \eqref{eq:KPSS_decomp}, a functional central limit theorem for linear processes, the continuous mapping theorem, and the rate requirements.  \qed 
\end{appendices}

\clearpage

\newcommand{\hbAppendixPrefix}{S}
\renewcommand{\thelemma}{\hbAppendixPrefix\arabic{lemma}} 
\setcounter{lemma}{0}   
\renewcommand{\thecorollary}{\hbAppendixPrefix\arabic{corollary}}  
\setcounter{corollary}{0}   
\renewcommand{\thetheorem}{\hbAppendixPrefix\arabic{theorem}}  
\setcounter{theorem}{0}   
\renewcommand{\theassumption}{\hbAppendixPrefix\arabic{assumption}}  
\setcounter{assumption}{0}
\renewcommand{\theremark}{\hbAppendixPrefix\arabic{remark}}  
\setcounter{remark}{0}
\renewcommand{\thetable}{\hbAppendixPrefix\arabic{table}}  
\setcounter{table}{0}
\renewcommand{\thefigure}{\hbAppendixPrefix\arabic{figure}}  
\setcounter{figure}{0}
\renewcommand{\thepage}{\hbAppendixPrefix\arabic{page}}  
\setcounter{page}{1}

\counterwithout{equation}{section} 
\renewcommand{\theequation}{\hbAppendixPrefix\arabic{equation}} 
\setcounter{equation}{0}

\makeatletter
\renewcommand\appendix{\par
	\setcounter{section}{0}%
	\def\@chapapp{\appendixname}%
	\def\thesection{S\arabic{section}}}
\makeatother


%
%




\begin{titlepage}
	\centering
	{\scshape \huge Supplement to ``Cointegrating Polynomial Regressions with Power Law Trends: Environmental Kuznets Curve or Omitted Time Effects?'' \par}
	
	\bigskip
	\startlist{toc}
	\printlist{toc}{}{}
\end{titlepage}

\clearpage
\appendix

\renewcommand{\arraystretch}{1.1}
\section{Simulation DGP used for Introduction} \label{sec:ParametersSimulationDGP}
The simulation DGPs of the introduction are based on the data for Austria, Belgium and Finland. Parameter values are (nonlinear) least squares estimates and innovations are mean-zero normally distributed random variables with a covariance matrix estimated from the residuals and $\diff\vx_t$. The specific parametrization for the \emph{model with global trend} is
\begin{equation}
 \begin{bmatrix}
  y_{1,t} \\
  y_{2,t} \\
  y_{3,t}
 \end{bmatrix}
 = \tau_g t^{2.21} \vones_3
 -
 \begin{bmatrix}
   8.89 \\
   4.16 \\
   16.39
 \end{bmatrix}
 +
 \begin{bmatrix}
 0.0017\\
 0.0122 \\
 0.0108
 \end{bmatrix}
 t
 +
 \begin{bmatrix}
 2.015 x_{1,t} \\
 1.477 x_{2,t} \\
 2.703 x_{3,t}
 \end{bmatrix}
 +\vu_t,
\end{equation}
where $[\vu_t\tran \; \diff \vx_t\tran]\tran \sim \rN(\vzeros,\widehat\mSigma)$ with
$$
 \widehat\mSigma
 =
 \begin{bmatrix}
  18.86	& *		& *		& *		& *		& *\\
  1.35 	& 2.02 	& *		& *		& *		& *\\
  3.68	& 2.65	& 1.88 	& *		& *		& *\\
  0.10	& 0.38	& 1.46	& 0.83 	& *		& *\\
  0.45	& 0.09 	& 0.22 	& 0.07	& 0.18	& *\\
  0.26	& 0.14	& 0.56	& 0.15 	& 0.14	& 0.24
 \end{bmatrix}
 \times 10^{-2}.
$$
The simulations investigating the influence of the redundant trend follow
$$
\begin{bmatrix}
y_{1,t} \\
y_{2,t} \\
y_{3,t}
\end{bmatrix}
=
\begin{bmatrix}
 -1.01 \\
 8.64 \\
 -5.15
\end{bmatrix}
+
\begin{bmatrix}
 -0.0111 \\
 0.0058 \\
 0.0163
\end{bmatrix}
t
+
\begin{bmatrix}
 1.103 x_{1,t} \\
 -0.001 x_{2,t} \\
 1.232 x_{3,t}
\end{bmatrix}
+
\phi_2
\begin{bmatrix}
 x_{1,t}^2 \\
 x_{2,t}^2 \\
 x_{3,t}^2
\end{bmatrix}
+ \vu_t,
$$
where $[\vu_t\tran \; \diff \vx_t\tran]\tran \sim \rN(\vzeros,\widehat\mSigma)$ with
$$
 \widehat\mSigma
 =
 \begin{bmatrix}
  18.57	& *		& *		& *		& *		& *\\
  4.03 	& 4.02 	& *		& *		& *		& *\\
  11.82	& 9.90	& 35.88 	& *		& *		& *\\
  0.56	& 0.64	& 1.79	& 0.83 	& *		& *\\
  0.47	& 0.20 	& 0.45 	& 0.07	& 0.18	& *\\
  0.43	& 0.35	& 0.91	& 0.15 	& 0.14	& 0.24
 \end{bmatrix}
 \times 10^{-2}.
$$
Compared to the empirical application (and thus also DGP2 in Section \ref{sec:simulations}), the main differences are the smaller $N$ and the omission of serial correlation in both the innovations and the increments of the integrated variables. These modifications allow us to showcase the influence of the omitted global trend while not having to worry about the effects of long-run covariance estimation on statistical size.

\section{Auxiliary Lemmas}
\begin{lemma} \ \vspace{-0.8cm}
	\begin{enumerate}[(i)]
		\item For $a_L> -1$, we have $\sup_{a\geq a_L} \left| \frac{1}{T} \sumtT \left(\frac{t}{T}\right)^a \right| \leq C$.
		\item Under Assumption \ref{assumption:innovations}, for any $a\geq a_L>-\frac{1}{2}$, and any $k\geq 0$,
		$
		\E\left( \frac{1}{\sqrt{T}}\sum_{t=1}^{T}\left( \frac{t}{T} \right)^a ({\ln t})^k u_{i,t}\right)^2\leq C({\ln T})^{2k} 
		$, $i=1,\ldots,N.$
		\item Under Assumption \ref{assumption:innovations}, for some $a_L$ and $a_U$ such that $-\frac{1}{2}<a_L<a_U<\infty$, and any $k\geq 0$,
		$
		\E \left( \sup_{a\in[a_L,a_U]} \left| \frac{1}{\sqrt{T}} \sumtT \left( \frac{t}{T} \right)^a (\ln t)^k u_{i,t} \right| \right) \leq C(\ln T)^k
		$, $i=1,\ldots,N.$
		\item If $a_L$ and $a_U$ satisfy $-1<a_L<a_U<\infty$, and if $k=0,1,2,\ldots$, then
		$$
		\sup_{a\in[a_L,a_U]} \left| \frac{1}{T} \sumtT \left( \frac{t}{T} \right)^a \left( \ln \frac{t}{T} \right)^k - \int_0^1 r^a (\ln r)^k dr \right| \leq C \frac{(\ln T)^{k+1}}{T^{1+\min(a_L,0)}}.
		$$
	\end{enumerate}
	\label{lem:boundsforfixed}
\end{lemma}
\begin{proof}
	\textbf{\textit{(i)}} This is shown in lemma 4 of \cite{robinson2012}. \textbf{\textit{(ii)}} Note that
	\begin{equation}
	\begin{aligned}
	\E\,\Bigg( \frac{1}{\sqrt{T}} &\sum_{t=1}^{T}\left( \frac{t}{T} \Bigg)^a ({\ln t})^k u_{i,t}\right)^2
	= \frac{1}{T} \sum_{t=1}^T \sum_{s=1}^T \left( \frac{t}{T} \right)^a \left( \frac{s}{T} \right)^a (\ln t)^k (\ln s)^k \, \E(u_{i,t} u_{i,s}) \\
	&\leq \frac{(\ln T)^{2k} }{T} \sum_{t=1}^T \sum_{s=1}^T \left( \frac{t}{T} \right)^a \left( \frac{s}{T} \right)^a \big| \E(u_{i,t} u_{i,s}) \big|
	\leq 2 \frac{(\ln T)^{2k} }{T} \sum_{t=1}^T \sum_{s=0}^{t-1} \left( \frac{t}{T} \right)^a \left( \frac{t-s}{T} \right)^a \big| \gamma_{i,s} \big|,
	\end{aligned}
	\end{equation}
	where we define $\gamma_{i,s} = \E(u_{i,t} u_{i,t-s})$. For the given index ranges, we also have $|t-s|\leq t$ such that
	\begin{equation}
	\E\,\Bigg( \frac{1}{\sqrt{T}} \sum_{t=1}^{T}\left( \frac{t}{T} \Bigg)^a ({\ln t})^k u_{i,t}\right)^2
	\leq 2 (\ln T)^{2k} \frac{1}{T}  \sum_{t=1}^T \left(\frac{t}{T}\right)^{2a}\sum_{s=0}^\infty | \gamma_{i,s} |.
	\label{eq:lemma1ii}
	\end{equation}
	The first summation in the RHS of \eqref{eq:lemma1ii} is bounded in view of Lemma \ref{lem:boundsforfixed}\textit{(i)} and $\sum_{s=0}^\infty |\gamma_{i,s}|<\infty$ due to Assumption \ref{assumption:innovations}(a) (cf. Appendix 3.A. in \cite{hamilton1994}). \textbf{\textit{(iii)}} Using the equality $\frac{t}{T}=\sum_{s=0}^{t-1}\left[ \left( \frac{s+1}{T} \right)^a - \left( \frac{s}{T} \right)^a \right]$ and a change in the order of summation, we find
	$$
	\begin{aligned}
	\sumtT &\left( \frac{t}{T} \right)^a (\ln t)^k u_{i,t}
	= \sumtT \sum_{s=0}^{t-1} \left[ \left( \frac{s+1}{T} \right)^a - \left( \frac{s}{T} \right)^a \right] (\ln t)^k u_{i,t}
	= \sum_{s=0}^{T-1} \left[ \left( \frac{s+1}{T} \right)^a - \left( \frac{s}{T} \right)^a \right] \sum_{t=s+1}^T (\ln t)^k u_{i,t} \\
	&= \left( \frac{1}{T} \right)^a \sumtT (\ln t)^k u_{i,t} + \sum_{s=1}^{T-1} \left[ \left( \frac{s+1}{T} \right)^a - \left( \frac{s}{T} \right)^a \right] \left( \sumtT (\ln t)^k u_{i,t} - \sum_{t=1}^s (\ln t)^k u_{i,t}\right) \\
	&= \left( \frac{1}{T} \right)^a \sumtT (\ln t)^k u_{i,t} + \sumtT (\ln t)^k u_{i,t} - \left( \frac{1}{T} \right)^a \sumtT (\ln t)^k u_{i,t}  - \sum_{s=1}^{T-1} \left[ \left( \frac{s+1}{T} \right)^a  - \left( \frac{s}{T} \right)^a \right]\sum_{t=1}^s (\ln t)^k u_{i,t},
	\end{aligned}
	$$
	and hence
	\begin{equation}
	\begin{aligned}
	\E \Bigg( &\sup_{a\in[a_L,a_U]} \left| \frac{1}{\sqrt{T}} \sumtT \left( \frac{t}{T} \right)^a (\ln t)^k u_{i,t} \right| \Bigg)
	\leq 
	\E \left| \frac{1}{\sqrt{T}} \sumtT (\ln t)^k u_{i,t}  \right| \\
	&\qquad\qquad +\E \left( \sup_{a\in[a_L,a_U]} \left|  \frac{1}{\sqrt{T}}\sum_{s=1}^{T-1} \left[ \left( \frac{s+1}{T} \right)^a - \left( \frac{s}{T} \right)^a \right]\sum_{t=1}^s (\ln t)^k u_{i,t}\right| \right).
	\end{aligned}
	\label{eq:lemma1.3}
	\end{equation}
	For the first term in the RHS of \eqref{eq:lemma1.3}, we have $\E \left| \frac{1}{\sqrt{T}} \sumtT (\ln t)^k u_{i,t}  \right| \leq \left(\E \left( \frac{1}{\sqrt{T}} \sumtT (\ln t)^k u_{i,t}  \right)^2 \right)^{1/2}\leq C (\ln T)^k$ by Lemma \ref{lem:boundsforfixed}(ii) with $a=0$. For the second term, note that
	\begin{equation}
	\left| \frac{1}{\sqrt{T}}\sum_{s=1}^{T-1} \left[ \left( \frac{s+1}{T} \right)^a - \left( \frac{s}{T} \right)^a \right]\sum_{t=1}^s (\ln t)^k u_{i,t} \right|
	\leq \frac{1}{\sqrt{T}} \sum_{s=1}^{T-1} \left( \frac{s}{T} \right)^a \left| \left(1 +\frac{1}{s} \right)^a-1 \right| \left| \sum_{t=1}^s (\ln t)^k u_{i,t} \right|.
	\label{eq:nextbound}
	\end{equation}
	To deal with the supremum of $\left|\left(1+\frac{1}{s} \right)^a -1\right|$ over $[a_L,a_U]$, we define $g_a(x)=(1+x)^a-1$ for $0\leq x \leq 1$. If $-\frac{1}{2} < a\leq 1$, then $|g_a(x)|\leq |a| x$ by Bernoulli's inequality. If $a\geq 1$, then convexity of $g_a(x)$ implies
	$$
	g_a(x) \leq (1-x)g_a(0) + x g_a(1) \leq \left(2^a-1\right) x.
	$$
	We conclude that $|g_a(x)|\leq C x$ for all $a_L\leq a \leq a_U$ and $x\in[0,1]$. Combining this result with  \eqref{eq:nextbound}, we have
	$$
	\begin{aligned}
	\E &\left( \sup_{a\in[a_L,a_U]} \left|  \frac{1}{\sqrt{T}}\sum_{s=1}^{T-1} \left[ \left( \frac{s+1}{T} \right)^a - \left( \frac{s}{T} \right)^a \right]\sum_{t=1}^s (\ln t)^k u_{i,t}\right| \right) \\
	&\leq \E \left( \frac{1}{\sqrt{T}} \sum_{s=1}^{T-1} \left( \frac{s}{T} \right)^{a_L} \sup_{a\in[a_L,a_U]}\left| \left(1 +\frac{1}{s} \right)^a-1 \right|\, \left| \sum_{t=1}^s (\ln t)^k u_{i,t} \right| \right)
	\leq C \frac{1}{\sqrt{T}} \sum_{s=1}^{T-1}\left( \frac{s}{T} \right)^{a_L} \frac{1}{s} \, \E \left| \sum_{t=1}^s (\ln t)^k u_{i,t} \right| \\
	&\leq C T^{-(a_L+1/2)} \sum_{s=1}^{T-1} s^{a_L-1/2} (\ln s)^k
	\leq C (\ln T)^k \left[ \frac{1}{T}\sum_{s=1}^T \left(\frac{s}{T} \right)^{a_L-1/2} \right]
	\leq C (\ln T)^k,
	\end{aligned}
	$$
	where we used $\E \left| \sum_{t=1}^s (\ln t)^k u_{i,t} \right|\leq \left(\E \left( \sum_{t=1}^s (\ln t)^k u_{i,t} \right)^2 \right)^{1/2} \leq C s^{1/2} (\ln s)^k$ (the steps in the proof of (ii) require a small modification to establish this) to go to the last line, and (i) to obtain the final inequality. The proof is complete since we have bounded both terms in the RHS of \eqref{eq:lemma1.3}. \textbf{\textit{(iv)}} If we divide the integral into integration intervals of width $\frac{1}{T}$, then we find
	\begin{align}
	&\sup_{a\in[a_L,a_U]} \left| \frac{1}{T} \sumtT \left( \frac{t}{T} \right)^a \left( \ln \frac{t}{T} \right)^k - \int_0^1 r^a (\ln r)^k dr \right| \nonumber\\
	&= \sup_{a\in[a_L,a_U]} \left|  \sumtT \int_{(t-1)/T}^{t/T} \left( \frac{t}{T} \right)^a \left( \ln \frac{t}{T} \right)^k dr- \sum_{t=1}^T \int_{(t-1)/T}^{t/T} r^a (\ln r)^k dr \right| \nonumber\\
	&= \sup_{a\in[a_L,a_U]} \left| \frac{1}{T} \left( \frac{1}{T} \right)^a \left( \ln \frac{1}{T} \right)^k - \int_0^{1/T} r^a (\ln r)^k dr + \sum_{t=2}^T \int_{(t-1)/T}^{t/T} \left[\left( \frac{t}{T} \right)^a \left( \ln \frac{t}{T} \right)^k- r^a(\ln r)^k\right] dr \right| \nonumber\\
	&\leq \sup_{a\in[a_L,a_U]} \left| \left( \frac{1}{T} \right)^{a+1} \left( \ln \frac{1}{T} \right)^k  \right|
	+ \sup_{a\in[a_L,a_U]} \left|  \int_0^{1/T} r^a (\ln r)^k dr \right| + \sup_{a\in[a_L,a_U]}\sum_{t=2}^T  \int_{(t-1)/T}^{t/T} \left| \left( \frac{t}{T} \right)^a \left( \ln \frac{t}{T} \right)^k- r^a(\ln r)^k\right| dr \nonumber\\
	&=: Ia+Ib+Ic, \label{eq:threeparts}
	\end{align}
	using the triangle inequality. Clearly, $Ia$ is bounded by $T^{-(a_L+1)} (\ln T)^k$. For $Ib$ we can use the standard integral (cf. \cite{adamsessex2016}), namely $\int_0^{1/T} r^a (\ln r)^k dr = \frac{(-1)^k}{a+1}\left( \frac{1}{T} \right)^{a+1}(\ln T)^k-\frac{k}{a+1} \int_0^{1/T} r^a (\ln r)^{k-1} dr$ for $k\neq -1$, to obtain
	$$
	\begin{aligned}
	\int_0^{1/T} r^a (\ln r)^k dr &=
	(-1)^k \left(\frac{1}{T} \right)^{a+1} \sum_{j=0}^{k-1} \frac{k!}{(k-j)!} \frac{1}{(a+1)^{1+j}} (\ln T)^{k-j} + (-1)^k \frac{k!}{(a+1)^k} \int_0^{1/T} r^a dr \\
	&=   (-1)^k \left(\frac{1}{T} \right)^{a+1} \sum_{j=0}^{k} \frac{k!}{(k-j)!} \frac{1}{(a+1)^{1+j}} (\ln T)^{k-j}.
	\end{aligned}
	$$
	We therefore conclude that
	$$
	\begin{aligned}
	Ib &\leq  \sum_{j=1}^k \frac{k!}{(k-j)!} \sup_{a\in[a_L,a_U]} \frac{1}{(a+1)^{1+j}}  \left(\frac{1}{T} \right)^{a+1}
	(\ln T)^{k-j}\leq \sum_{j=1}^k \frac{k!}{(k-j)!}  \frac{1}{(a_L+1)^{1+j}} \left( \frac{1}{T} \right)^{a_L+1}  (\ln T)^{k-j} \\
	& \leq C T^{-(a_L+1)} (\ln T)^k.
	\end{aligned}
	$$
	It remains to bound the term $Ic$. Changing the integration variable to $r=\frac{t}{T}-s$ yields
	\begin{equation}
	Ic=\sup_{a\in[a_L,a_U]}\sum_{t=2}^{T}\int_{0}^{1/T}\left|\left(\frac{t}{T}\right)^a\left({\ln \frac{t}{T}}\right)^k-\left(\frac{t}{T}-s\right)^a\left[{\ln\left(\frac{t}{T}-s\right)}\right]^k\right|ds.
	\label{eq:Ic}
	\end{equation}
	We subsequently derive an upper bound for the integrand using an approach which mimics the derivations in (D.14) and (D.15) in \cite{robinson2012}. For any $\frac{2}{T}\leq \ell \leq 1$ (such that $0<s/\ell\leq \frac{1}{2}$), we have
	\begin{equation}\label{eq:Ic_rewriting}
	\begin{aligned}
	\Big| \ell^a(\ln \ell)^k&-(\ell-s)^a\big(\ln(\ell-s)\big)^k\Big|
	=\Big| \big[\ell^a-(\ell-s)^a \big] (\ln \ell)^k + (\ell-s)^a \big[(\ln \ell)^k-(\ln(\ell-s))^k\big] \Big| \\
	&\leq \Big| \big[ \ell^a-(\ell-s)^a\big](\ln \ell)^k\Big|+\Big|(\ell-s)^a\Big[(\ln \ell)^k-(\ln(\ell-s))^k\Big] \Big| \\
	&=\ell^a \big|\ln \ell \big|^k \big|1-\left(1-s/\ell \right)^a\big|+\ell^a(1-s/\ell)^a\left|(\ln \ell)^k-(\ln(\ell-s))^k\right|=: IIa + IIb,
	\end{aligned}
	\end{equation}
	by the triangle inequality and the fact that $\big| (\ell-s)^a \big| = (\ell-s)^a$. For $IIa$ similar arguments as those found below \eqref{eq:nextbound} give $\big|1- (1-x)^a \big| \leq C x$, and hence
	\begin{equation}
	IIa \leq C \ell^{a_L} \big|\ln \ell \big|^k \frac{s}{\ell} \leq C \ell^{a_L-1} \big|\ln \ell \big|^k s \leq  C \ell^{a_L-1} \big|\ln \ell \big|^k \frac{1}{T} \leq C \ell^{a_L-1} \big|\ln \ell \big|^k \frac{1}{T} \leq C \ell^{a_L-1} \big( \ln T \big)^k \frac{1}{T},
	\label{eq:Ic_upperbound1}
	\end{equation}
	since $| \ln \ell| \leq | \ln T |$ for all $\frac{2}{T}\leq \ell \leq 1$. For $IIb$ we first note that $\frac{1}{2}\leq 1-s/\ell <1$ and therefore $\left(1-s/\ell \right)^a < \left(1-s/\ell \right)^{-1} \leq 2$. Moreover, we use the factorization $p^n-q^n = (p-q) \sum_{j=0}^{n-1} p^{n-1-j} q^j$ to obtain\footnote{For any $x>-1$, we have the inequality $\frac{x}{1+x} \leq \ln(1+x) \leq x$. This implies that $|\ln (1- s/\ell) |=- \ln \big(1 - s/\ell \big)\leq \frac{s/\ell}{1-s/\ell}\leq 2 \frac{s}{\ell}$.}
	\begin{equation}
	\begin{aligned}
	\Big|(\ln \ell)^k&-(\ln(\ell-s))^k\Big|=\big|\ln \ell-\ln(\ell-s) \big| \left| \sum_{j=0}^{k-1}(\ln \ell)^{k-1-j} \big(\ln(\ell-s)\big)^j\right|  \\
	&= \big|\ln(1-s/\ell)\big| \left| \sum_{j=0}^{k-1} (\ln \ell)^{k-1-j} \big(\ln(\ell-s) \big)^j \right| 
	\leq \big|\ln(1-s/\ell)\big| \sum_{j=0}^{k-1} \left|\ln \ell\right|^{k-1-j} \left|\ln(\ell-s)\right|^j \\
	&\leq k\left|{\ln(1-s/\ell)}\right|({\ln T})^{k-1} \leq  2 k \frac{s}{\ell} \big(\ln T \big)^{k-1},
	\end{aligned}
	\label{eq:Ic_upperbound2}
	\end{equation}
	because $1/T\leq \ell-s<1$ and thus $|\ln (\ell-s) |\leq \ln T$. Combining all previous results for $IIb$ gives
	$$
	IIb \leq C \ell^a  \frac{s}{\ell} \big(\ln T \big)^{k-1} \leq C \ell^{a_L-1} \big(\ln T \big)^{k-1} \frac{1}{T}.
	$$
	Since $\frac{2}{T}\leq \ell \leq 1$, we use the bounds on $IIa$ and $IIb$ to bound the integrand of \eqref{eq:Ic} as follows:
	\begin{equation*}
	Ic\leq C\sup_{a\in[a_L,a_U]}\sum_{t=2}^{T}\int_{0}^{1/T}\left(\frac{t}{T}\right)^{a_L-1}\frac{1}{T}({\ln T})^kds\leq C\frac{({\ln T})^k}{T^2}\sum_{t=1}^{T}\left(\frac{t}{T}\right)^{a_L-1}.
	\end{equation*}
	The asymptotic order of $\sum_{t=1}^{T}\big(\frac{t}{T}\big)^{a_L-1}$ relies on the values of $a_L$. We distinguish three cases: (1) if $a_L<0$, then $\sum_{t=1}^{T}\big(\frac{t}{T}\big)^{a_L-1}=T^{1-a_L}\sum_{t=1}^{T}\frac{1}{t^{1-a_L}}=T^{1-a_L}O(1)$, (2) if $a_L=0$, then $\sum_{t=1}^{T}\big(\frac{t}{T}\big)^{a_L-1}=T\sum_{t=1}^{T}t^{-1}=TO({\ln T})$, and (3) if $a_L>0$, $\sum_{t=1}^{T}\big(\frac{t}{T}\big)^{a_L-1}=O(T)$ by Lemma \ref{lem:boundsforfixed}(i). Overall, we have
	\begin{equation}\label{eq:limit_normtrend_firstterm_uppbound}
	Ic \leq C\frac{({\ln T})^k}{T^2}\sum_{t=1}^{T}\left(\frac{t}{T}\right)^{a_L-1}=O\left(\frac{({\ln T})^k}{T^{a_L+1}}\indic{a_L<0}+\frac{({\ln T})^{k+1}}{T}\indic{a_L=0}+\frac{({\ln T})^k}{T}\indic{a_L>0}\right).
	\end{equation}
	It is seen that $Ia$, $Ib$, and $Ic$ converge to zero as $T\to \infty$. The proof follows from \eqref{eq:threeparts}.
\end{proof}

\newpage
\begin{lemma}\label{lem:Brownianconv}
	Let Assumption \ref{assumption:innovations} hold. For any $a$ such that $ - \frac{1}{2} < a_L \leq a \leq a_U < \infty$, any $j\in\{1,2,\ldots,p_i\}$, $i\in\{1,2,\ldots,N\}$, and $k\in \{0,1,2,\ldots\}$, as $T\rightarrow\infty$, we have:
	\begin{enumerate}[(i)]
		\item $\frac{1}{\sqrt{T}}\sum_{t=1}^{T}\bigg(\frac{x_{i,t}}{\sqrt{T}}\bigg)^ju_{i,t}\wto\int_{0}^{1}\brownnormal_{v_i}^j(r)d\brownnormal_{u_i}(r)+j\mDelta_{v_iu_i}\int_{0}^{1}\brownnormal_{v_i}^{j-1}(r)dr$,
		\item $\frac{1}{\sqrt{T}}\sum_{t=1}^{T}\left(\frac{t}{T}\right)^a\left({\ln \frac{t}{T}}\right)^k u_{i,t} \wto  \int_{0}^{1}r^a({\ln r})^k d\brownnormal_{u_i}(r)$,
		\item $\frac{1}{T}\sum_{t=1}^{T}\left(\frac{t}{T}\right)^a\left({\ln \frac{t}{T}}\right)^k\bigg(\frac{x_{i,t}}{\sqrt{T}}\bigg)^j\wto\int_{0}^{1}r^a({\ln r})^k\brownnormal_{v_i}^j(r)dr$.
	\end{enumerate}
\end{lemma}

\begin{proof}
	
	For $r\in (0,1]$, we define $f(r)=r^a({\ln r})^k$. Two partial sum processes are defined as $S_{i,T}(r)=\frac{1}{\sqrt{T}}\sum_{s=1}^{[rT]}u_{i,s}$, and $X_{i,T}(r)=\frac{1}{\sqrt{T}}x_{i,[rT]}=\frac{1}{\sqrt{T}}\sum_{t=1}^{[rT]} v_{i,t}$. Finally, set $f_{T}(r)=\left(\frac{[rT]}{T}\right)^a\left({\ln \frac{[rT]}{T}}\right)^k$ for $r\in\left[\frac{1}{T},1 \right]$. \textbf{\textit{(i)}} This result follows from lemma 1 of \cite{hongphillips2010}. \textbf{\textit{(ii)}} We have
	\begin{equation}
	\begin{aligned}
	\frac{1}{\sqrt{T}} &\left(\frac{t}{T}\right)^a\left({\ln \frac{t}{T}}\right)^k u_{i,t}
	= f_T\left( \frac{t}{T}\right)  \frac{u_{i,t}}{\sqrt{T}}
	= f_T\left( \frac{t}{T} \right) \left[S_{i,T}\left( \frac{t}{T}\right)- S_{i,T}\left( \frac{t-1}{T}\right) \right] \\
	&= \left[  f_T\left(\frac{t}{T} \right) S_{i,T}\left( \frac{t}{T}\right) - f_T\left(\frac{t-1}{T} \right) S_{i,T}\left( \frac{t-1}{T}\right)  \right]
	- \left[ f_T\left( \frac{t}{T} \right) - f_T\left( \frac{t-1}{T} \right)  \right] S_{i,T}\left( \frac{t-1}{T}\right) 
	\end{aligned}
	\label{eq:telescope}
	\end{equation}
	and hence
	\begin{align}
	& \frac{1}{\sqrt{T}} \sumtT \left(\frac{t}{T}\right)^a\left({\ln \frac{t}{T}}\right)^k u_{i,t} \nonumber
	= \left(\frac{1}{T}\right)^a\left({\ln \frac{1}{T}}\right)^k \frac{u_{i,1}}{\sqrt{T}}+\frac{1}{\sqrt{T}}\sum_{t=2}^{T}\left(\frac{t}{T}\right)^a\left({\ln \frac{t}{T}}\right)^k u_{i,t} \\ 
	\nonumber
	&\stackrel{\eqref{eq:telescope}}{=} f_T\left(\frac{1}{T}\right) S_{i,T}\left(\frac{1}{T}\right) + \left[ f_T(1) S_{i,T}(1) - f_T\left(\frac{1}{T}\right) S_{i,T}\left(\frac{1}{T}\right) \right]
	- \sum_{t=2}^T \left[ f_T\left( \frac{t}{T} \right) - f_T\left( \frac{t-1}{T} \right)  \right] S_{i,T}\left( \frac{t-1}{T}\right) \\
	& \stackrel{f_T(1)=0}{=} - \sum_{t=2}^{T}\int_{(t-1)/T}^{t/T}S_{i,T}(r)df_T(r) \label{eq:sumofints}
	\end{align}
	where we used the fact that $S_{i,T}(\cdot)$ is piecewise constant. In view of Assumption \ref{assumption:innovations}, we can extend suitably extend the probability space and have the following uniformly strong approximation of the partial sum process $S_{i,T}$ (see for example page 562 of \cite{phillips2007}):
	\begin{equation}
	\sup_{1\leq t\leq T}\left|S_{i,T}\left(\frac{t-1}{T}\right)-\brownnormal_{u_i}\left(\frac{t-1}{T}\right)\right|=o_{a.s.}\left(\frac{1}{T^{(1/2)-(1/q)}}\right),
	\end{equation}
	for $q>2$. Continuing from \eqref{eq:sumofints}, this uniformly strong approximation gives
	\begin{equation}
	\begin{aligned}
	\frac{1}{\sqrt{T}} \sumtT \left(\frac{t}{T}\right)^a\left({\ln \frac{t}{T}}\right)^k u_{i,t}
	&=-\sum_{t=2}^{T}\int_{(t-1)/T}^{t/T}\brownnormal_{u_i}(r)df_T(r)+o_{a.s.}\left(\frac{1}{T^{(1/2)-(1/q)}}\right) \\
	& = -\int_{1/T}^{1}\brownnormal_{u_i}(r)df_T(r)+o_{a.s.}\left(\frac{1}{T^{(1/2)-(1/q)}}\right) \\
	& = \brownnormal_{u_i}\left(\frac{1}{T}\right)f_T\left(\frac{1}{T}\right) + \int_{1/T}^{1}f_T(r)d\brownnormal_{u_i}(r)+o_{a.s.}\left(\frac{1}{T^{(1/2)-(1/p)}}\right) \\
	& =  \int_0^1 f(r)d\brownnormal_{u_i}(r) - \int_0^{1/T} f(r)d\brownnormal_{u_i}(r)  + \brownnormal_{u_i}\left(\frac{1}{T}\right)f_T\left(\frac{1}{T}\right) \\
	& \qquad + \int_{1/T}^1 \Big[ f_T(r) - f(r)  \Big] d\brownnormal_{u_i}(r)   +o_{a.s.}\left(\frac{1}{T^{(1/2)-(1/p)}}\right),
	\end{aligned}
	\end{equation}
	where the third line is obtained using integration by parts of the mean square Riemann-Stieltjes integral, c.f. theorem 2.7 in \cite{tanaka2017}. It remains to show that $\int_0^{1/T} f(r)d\brownnormal_{u_i}(r)$, $\brownnormal_{u_i}\left(\frac{1}{T}\right)f_T\left(\frac{1}{T}\right)$, and $\int_{1/T}^1 \Big[ f_T(r) - f(r)  \Big] d\brownnormal_{u_i}(r)$ are asymptotically negligible. These quantities are zero mean so it suffices to show that their variances vanish as $T\to \infty$. By the isometry property and steps similar to those above \eqref{eq:Ic}, we have
	\begin{equation}
	\var\left(\int_0^{1/T} f(r)d\brownnormal_{u_i}(r)\right) = \Omega_{u_iu_i} \int_0^{1/T} \big[ f(r) \big]^2 dr \leq CT^{-(2a_L+1)}(\ln T)^{2k}\rightarrow 0,
	\end{equation}
	as $T\to \infty$. Also, $\var\left( \brownnormal_{u_i}\left(\frac{1}{T}\right)f_T\left(\frac{1}{T}\right) \right)=\frac{1}{T} \Omega_{u_iu_i}\left[f_T\left(\frac{1}{T}\right)  \right]^2=\Omega_{u_iu_i} \left( \frac{1}{T} \right)^{2a_L +1} \left(\ln \frac{1}{T} \right)^{2k}\to 0$. To control the variance of $\int_{1/T}^1 \Big[ f_T(r) - f(r)  \Big] d\brownnormal_{u_i}(r)$, we look at 
	\begin{equation}
	\begin{aligned}
	\int_{1/T}^{1}\left|f(r)-f_T(r)\right|^2dr&=\sum_{t=2}^{T}\int_{(t-1)/T}^{t/T}\left|f(r)-\left(\frac{t-1}{T}\right)^a\left(\ln \frac{t-1}{T}\right)^k\right|^2dr \\
	&=\sum_{t=1}^{T-1}\int_{t/T}^{(t+1)/T}\left| r^a \big( \ln r \big)^k - \left(\frac{t}{T}\right)^a\left(\ln \frac{t}{T}\right)^k\right|^2dr \\
	&=\sum_{t=1}^{T-1}\int_{0}^{1/T}\left|\left(\frac{t}{T}+s\right)^a\left[\ln\left(\frac{t}{T}+s\right)\right]^k-\left(\frac{t}{T}\right)^a\left(\ln \frac{t}{T}\right)^k\right|^2 ds.
	\end{aligned}
	\label{eq:splittedintegral}
	\end{equation}
	Now let $\ell\in\left\{\frac{1}{T},\frac{2}{T},\ldots,1\right\}$ and recall that $0 \leq s \leq \frac{1}{T}$ (hence also $0\leq \frac{s}{\ell} \leq 1$). Using the triangle inequality, the expression in absolute values can be bounded as
	\begin{equation}
	\begin{aligned}
	\Big|(\ell+s)^a (\ln (\ell+s))^k-\ell^a &(\ln\ell)^k\Big|=\left|\left[(\ell+s)^a-\ell^a\right](\ln (\ell+s))^k+\ell^a\left[(\ln (\ell+s))^k-(\ln\ell)^k\right]\right|\\
	&\hspace{-1.5cm}\leq \left|\left[(\ell+s)^a-\ell^a\right](\ln (\ell+s))^k\right|+\left|\ell^a\left[(\ln (\ell+s))^k-(\ln\ell)^k\right]\right|\\
	&\hspace{-1.5cm}=\ell^a \Big| \left(1+s/\ell\right)^a-1\Big| \left|\ln (\ell+s)\right|^k+\ell^a\left|(\ln (\ell+s))^k-(\ln\ell)^k\right|=\Rmnum{2}c+\Rmnum{2}d.
	\end{aligned}
	\end{equation}
	By the inequality $|g_a(x)|\leq Cx$ below (\ref{eq:nextbound}) and the fact that $\left|\ln (\ell+s)\right|\leq |\ln \ell|+|\ln(1+s/\ell)|\leq \ln T+s/\ell$, we obtain 
	$\Rmnum{2}c\leq C\ell^{a_L}\frac{s}{\ell}\left|\ln T+\frac{s}{\ell}\right|^k\leq C\ell^{a_L-1}(\ln T)^k\frac{1}{T}$. Moreover, the factorisation $p^n-q^n=(p-q)\sum_{j=0}^{n-1}p^{n-1-j}q^j$ yields
	\begin{equation}
	\Rmnum{2}d=\ell^a\left|\ln \left(1+s/\ell\right)\right|\left|\sum_{j=0}^{k-1}\left(\ln(\ell+s)\right)^{k-1-j}\left(\ln \ell\right)^j\right|\leq k\ell^{a_L}\frac{s}{\ell}\left|(\ln T)+1\right|^{k-1}\leq C\ell^{a_L-1}(\ln T)^{k-1}\frac{1}{T}.
	\end{equation} 
	By combination of the bounds on $\Rmnum{2}c$ and $\Rmnum{2}d$, we conclude that $\left|(\ell+s)^a(\ln (\ell+s))^k-\ell^a(\ln\ell)^k\right|\leq C\ell^{a_L-1}(\ln T)^k\frac{1}{T}$ and arrive at the following upper bound on the RHS of \eqref{eq:splittedintegral}:
	\begin{equation}
	\begin{aligned}
	\int_{1/T}^{1}\left|f(r)-f_T(r)\right|^2dr &\leq C(\ln T)^{2k}\frac{1}{T^3}\sum_{t=1}^{T}\left(\frac{t}{T}\right)^{2(a_L-1)}\\
	&=O\left(\frac{(\ln T)^{2k}}{T^{2(a_L+\frac{1}{2})}}\indic{a_L<\frac{1}{2}}+\frac{(\ln T)^{2k+1}}{T^{2}}\indic{a_L=\frac{1}{2}}+\frac{(\ln T)^{2k}}{T^2}\indic{a_L>\frac{1}{2}}\right).
	\end{aligned}
	\label{eq:differentcases}
	\end{equation}
	The RHS of \eqref{eq:differentcases} will go to zero as $T\to \infty$, thereby establishing that $\int_{1/T}^1 \Big[ f_T(r) - f(r)  \Big] d\brownnormal_{u_i}(r)$ is also asymptotically negligible. The proof of part (ii) is now complete. \textbf{\textit{(iii)}} We have
	\begin{equation}
	\begin{aligned}
	\frac{1}{T} &\sum_{t=1}^{T}\left(\frac{t}{T}\right)^a\left({\ln \frac{t}{T}}\right)^k\bigg(\frac{x_{i,t}}{\sqrt{T}}\bigg)^j=\sum_{t=2}^{T}\int_{(t-1)/T}^{t/T}f_T(r)X_{i,T}^j(r) dt \\
	&\qquad=\int_0^1 f(r) X_{i,T}^j(r) dr + \int_{1/T}^{1}\left[f_T(r)-f(r)\right]X_{i,T}^j(r) dr=: IIIa + IIIb.
	\end{aligned}
	\end{equation}
	Given the CMT and $X_{i,T}\wto \brownnormal_{v_i}$, term $IIIa$ will converge weakly to $\int_{0}^{1}f(r)\brownnormal_{v_i}^j(r)dr$ if we can show that $x\mapsto\int_{0}^{1}f(r)x^j(r)dr$ is a continuous functional. Let  $x,y\in D[0,1]$. H\"older's inequality implies
	\begin{multline}
	\left| \int_0^1 f(r)x^j(r)dr-\int_{0}^{1}f(r) y^j(r)dr\right| =
	\left| \int_0^1 f(r) \left( x^j(r)-y^j(r) \right) dr \right| \\
	\leq \int_{0}^{1}|f(r)|dr\sup_{r\in[0,1]}|x^j(r)-y^j(r)|\leq C\sup_{r\in[0,1]}|x(r)-y(r)|\rightarrow 0, 
	\label{eq:contFunc}
	\end{multline}
	because $\int_{0}^{1}|f(r)|dr=\frac{k!}{(1+a)^{k+1}}$ is bounded. Continuity of the functional now follows from \eqref{eq:contFunc}. If we apply the Cauchy-Schwartz inequality to $IIIb$, then we find
	\begin{equation*}
	IIIb\leq \left[\int_{1/T}^{1}\left|f(r)-f_T(r)\right|^2dr\right]^{1/2}\left[\int_{1/T}^{1}X_{i,T}^{2j}(r) dr\right]^{1/2}.
	\end{equation*}
	Since $\int_{1/T}^{1}\left|f(r)-f_T(r)\right|^2=o(1)$ by (\ref{eq:differentcases}) and $\int_{1/T}^{1}X_{i,T}^{2j}(r) dr=\myint X_{i,T}^{2j}(r) dr\wto \myint \brownnormal_{v_i}^{2j}(r) dr$. We conclude that $IIIb=o_p(1)$. Now combine the limiting results for $IIIa$ and $IIIb$ to complete the argument.
\end{proof}

\begin{lemma}\label{lem:supremums}
	For any $\kappa>0$, define 
	\begin{equation}\label{eq:neighborhood_gamma0}
		\calN_{\kappa,T}(\vgamma_0)=\Big\{\text{\footnotesize $\vgamma\in\mGamma:\, T^{\theta_{0}+1/2}\left|\theta-\theta_0\right|\leq \kappa\ln{T},\, T^{\theta_{0}+1/2}|\tau_g-\tau_{g0}|\leq \kappa(\ln{T})^2,\,T^{1/2}\left\|\mD_{Z,T}\big(\vbeta-\vbeta_0\big)\right\| \leq \kappa\ln{T}$}\Big\}.
	\end{equation}	
	Assume $N$ is fixed and $T\rightarrow\infty$. Let $k_1$, $k_2$ be any nonnegative integers. Let Assumption \ref{assumption:innovations} hold.
	\begin{enumerate}[(i)]
		\item $\sup_{\vgamma\in \calN_{\kappa,T}(\vgamma_0)}(\ln{T})^{k_2}T^{-1}\sum_{t=1}^{T} \left|T^{-\theta_0}\big(\tau_{g}t^{\theta}-\tau_{g0}t^{\theta_0}\big)\left(\ln{t}\right)^{k_1}\right|=o(1)$.
		\item $\sup_{\vgamma\in \calN_{\kappa,T}(\vgamma_0)} N (\ln{T})^{k_2} T^{-1}\sum_{t=1}^{T}\left|T^{-2\theta_0}\big(\tau_gt^{2\theta}-\tau_{g0}t^{2\theta_0}\big)\left(\ln{t}\right)^{k_1}\right|=o(1)$.
		\item $\sup_{\vgamma\in \calN_{\kappa,T}(\vgamma_0)}(\ln{T})^{k_2}T^{-1}\left|\sum_{t=1}^{T}T^{-\theta_{0}}\big(t^{\theta}-t^{\theta_0}\big)\mD_{Z,T}^{-1}\mZ_t^{}\viota_N^{}\right|=o_p(1)$.
		\item $\sup_{\vgamma\in \calN_{\kappa,T}(\vgamma_0)}T^{-1}\left|\sum_{t=1}^{T}T^{-\theta_{0}}\left[\big(\tau_gt^{\theta}-\tau_{g0}t^{\theta_0}\big)\ln{t}-\tau_{g0}\big(t^{\theta}-t^{\theta_0}\big)\ln{T}\right]\mD_{Z,T}^{-1}\mZ_t^{}\viota_N^{}\right|=o_p(1)$.
		\item $\sup_{\vgamma\in \calN_{\kappa,T}(\vgamma_0)}(\ln{T})^{k_2}T^{-1}\left|\sum_{t=1}^{T}T^{-2\theta_{0}}t^{\theta}\left(\ln{t}\right)^{k_1}\big(\vbeta-\vbeta_0\big)'\mZ_t^{}\viota_N^{}\right|=o_p(1)$.
		\item $\sup_{\vgamma\in \calN_{\kappa,T}(\vgamma_0)}(\ln{T})^{k_2}T^{-1}\left|\sum_{t=1}^{T}T^{-2\theta_{0}}\big(\tau_g t^{\theta}-\tau_{g0}t^{\theta_0}\big)\left(\ln{t}\right)^{k_1}\vu_t'\viota_N\right|=o_p(1)$.
	\end{enumerate}
\end{lemma}

\begin{proof}
	We only show \textbf{\textit{(i)}}, \textbf{\textit{(iii)}}, \textbf{\textit{(v)}} and \textbf{\textit{(vi)}}. The proof of the remaining results is similar and thus omitted.
	
	\medskip
	\noindent \textbf{\textit{(i)}} For any $\vgamma\in\calN_{\kappa,T}(\vgamma_0)$, by the triangular inequality and the mean-value theorem (MVT), 
	\begin{align*}
	\left|T^{-\theta_0}\big(\tau_{g}t^{\theta}-\tau_{g0}t^{\theta_0}\big)\left(\ln{t}\right)^{k_1}\left(\ln{T}\right)^{k_2}\right|&=\left(\frac{t}{T}\right)^{\theta_0}\left|\tau_{g}\big(t^{\theta-\theta_0}-1\big)+\big(\tau_{g}-\tau_{g0}\big)\right|\left(\ln{t}\right)^{k_1}\left(\ln{T}\right)^{k_2}\\
	&\leq \left(\frac{t}{T}\right)^{\theta_0}\left[\big|\tau_{g}\big|\left|\,t^{\tilde{\theta}}(\theta-\theta_{0})\ln{t}\,\right|+\big|\tau_{g}-\tau_{g0}\big|\right]\left(\ln{T}\right)^{k_1+k_2}\\
	&\leq C\left(\frac{t}{T}\right)^{\theta_0} \frac{(\ln{T})^{k_1+k_2+2}}{T^{\theta_{0}+1/2}},
	\end{align*}
	where $t^{|\tilde{\theta}|}\leq T^{|\theta-\theta_0|}=\exp\left(|\theta-\theta_0|\ln {T}\right)\leq C$ whenever $T$ is sufficiently large. We obtain the first result due to Lemma \ref{lem:boundsforfixed}\point{1} and $\frac{(\ln{T})^{k}}{T^{\theta_{0}+1/2}}=o(1)$ for any $k\geq 0$. 
	
	\medskip
	\noindent \textbf{\textit{(iii)}} By Lemma \ref{lem:boundsforfixed}\point{1}, Lemma \ref{lem:Brownianconv}\point{3}, and the MVT,
	\begin{equation*}
	(\ln{T})^{k_2}T^{-1}\left|\sum_{t=1}^{T}T^{-\theta_{0}}\big(t^{\theta}-t^{\theta_0}\big)\mD_{Z,T}^{-1}\mZ_t^{}\viota_N^{}\right|\leq \sqrt{N}~O\left(\frac{(\ln T)^{k_2+2}}{T^{\theta_{0}+1/2}}\right)\left\|T^{-1}\sum_{t=1}^{T}\left(\frac{t}{T}\right)^{\theta_{0}}\mD_{Z,T}^{-1}\mZ_t^{}\,\right\|=o_p(1),
	\end{equation*}
	where the term $o_p(1)$ is uniform over $\vgamma\in\calN_{\kappa,T}(\vgamma_0)$. 
	
	\medskip
	\noindent \textbf{\textit{(v)}} By Part \textbf{\textit{(i)}} and Lemma \ref{lem:Brownianconv}\point{3},
	\begin{multline*}
	(\ln{T})^{k_2}T^{-1}\left|\sum_{t=1}^{T}T^{-2\theta_{0}}t^{\theta}\left(\ln{t}\right)^{k_1}\big(\vbeta-\vbeta_0\big)'\mZ_t^{}\viota_N^{}\right|
	\leq \sqrt{N}(\ln{T})^{k_2}T^{1/2}\left\|\mD_{Z,T}\big(\vbeta-\vbeta_0\big)\right\|\\
	\times \left\|T^{-(\theta_0+1/2)}T^{-1}\sum_{t=1}^{T}\left(\frac{t}{T}\right)^{\theta_{0}}\left(\ln{t}\right)^{k_1}\mD_{Z,T}^{-1}\mZ_t^{}+T^{-(\theta_0+1/2)}o_p(1)\right\|=o_p(1)
	\end{multline*}
	
	\medskip
	\noindent \textbf{\textit{(vi)}} Using the MVT and Lemma \ref{lem:Brownianconv}\point{2}, we obtain
	\begin{equation*}
	(\ln{T})^{k_2}T^{-1}\left|\sum_{t=1}^{T}T^{-2\theta_{0}}\big(\tau_g t^{\theta}-\tau_{g0}t^{\theta_0}\big)\left(\ln{t}\right)^{k_1}\vu_t'\viota_N\right|\leq \sqrt{N}\,o\left(\frac{(\ln{T})^{k_2}}{T^{\theta_{0}+1/2}}\right)\left\|T^{-1/2}\sum_{t=1}^{T}\left(\frac{t}{T}\right)^{\theta_0}\vu_t\right\|=o_p(1).
	\end{equation*}
	The proof is completed.
\end{proof}

\section{Proof of Theorem \ref{thm:consistentLRVs}}

\begin{proof}
We write $\widehat{\mDelta}_T\equiv\widehat{\mDelta}_T(\widehat{\vgamma}_T,b_T)$ and $\widehat{\mOmega}_T\equiv\widehat{\mOmega}_T(\widehat{\vgamma}_T,b_T)$ to make their dependence on the parameter estimator $\widehat{\vgamma}_T$ and bandwidth $b_T$ explicit. Changing the summation indices, we can express the one-sided long-run covariance estimator as
\begin{equation*}
	\widehat{\mDelta}_T(\widehat{\vgamma}_T,b_T)
	= \sum_{i=0}^{T-1} k\left(\frac{i}{b_T}\right)\left[ \frac{1}{T} \sum_{t=1}^{T-i} \bm V_{t+i}(\widehat{\vgamma}_T) \bm V_t(\widehat{\vgamma}_T)\tran \right]=:\widehat{\mSigma}_T(\widehat{\vgamma}_T)+\widehat{\mGamma}_T(\widehat{\vgamma}_T,b_T),
\end{equation*}
where $\widehat{\mSigma}_T(\widehat{\vgamma}_T)= T^{-1} \sum_{t=1}^T  \bm V_t(\widehat{\vgamma}_T) \bm V_t(\widehat{\vgamma}_T)\tran$ and $\widehat{\mGamma}_T(\widehat{\vgamma}_T,b_T)= \sum_{i=1}^{T-1} k\left(\frac{i}{b_T}\right)\left[ T^{-1} \sum_{t=1}^{T-i} \bm V_{t+i}(\widehat{\vgamma}_T) \bm V_t(\widehat{\vgamma}_T)\tran \right]$. Similarly, we have
\begin{equation*}
	\widehat{\mOmega}_T(\widehat{\vgamma}_T,b_T) = \widehat{\mSigma}_T(\widehat{\vgamma}_T)+\widehat{\mGamma}_T(\widehat{\vgamma}_T,b_T)+\widehat{\mGamma}_T(\widehat{\vgamma}_T,b_T)\tran.
\end{equation*}
Clearly, it suffices to study the asymptotic behavior of $\widehat{\mSigma}_T(\widehat{\vgamma}_T)$ and $\widehat{\mGamma}_T(\widehat{\vgamma}_T,b_T)$. As the lower right subblock of $\bm V_{t+i}(\vgamma)\bm V_t(\vgamma)\tran$ equals $\vv_{t+i}^{}\vv_t\tran$ (no parameter estimation uncertainty here), the consistency result for this subblock follows from the properties of $\{\vv_t\}$ in Assumption \ref{assumption:innovations}, the kernel requirements in Assumption \ref{assumption:LRVestimation}, and an application in Theorem 2 of \cite{jansson2002}. 

We proceed to the upper left subblocks of $\widehat{\mSigma}_T(\widehat{\vgamma}_T)$ and $\widehat{\mGamma}_T(\widehat{\vgamma}_T,b_T)$. If the residuals are close enough to the true innovations, then the results from Theorem 2 of \cite{jansson2002} again applies. It suffices to show
$$
T^{-1} \sumtT \big[ \widehat{\vu}_t\widehat{\vu}_t\tran - \vu_t \vu_t\tran \big] \pto 0 \quad\text{and}\quad \sum_{i=1}^{T-1} k\left(i/b_T\right)\left( T^{-1} \sum_{t=1}^{T-i}  \big[ \widehat{\vu}_{t+i} \widehat{\vu}_t\tran - \vu_{t+i} \vu_t\tran  \big] \right) \pto 0.
$$
Using Lemmas \ref{lem:boundsforfixed}--\ref{lem:Brownianconv} and Theorem \ref{thm:limiting_dist}, the following key result follows immediately
\begin{equation}
	\begin{aligned}
		\widehat{\vu}_t 
		= \vu_t - \left(\,\widehat{\tau}_{g,T}\,t^{\widehat{\theta}_T}-\tau_{g0}\,t^{\theta_{0}}\right)\vones_N - \mZ_t'\left(\widehat{\vbeta}_T-\vbeta_0\right)
		= \vu_t + O_p\big(T^{-1/2}\ln T\big)\left(\frac{t}{T}\right)^{\theta_{0}} + O_p\left(T^{-1/2}\right)\left(\mD_{Z,T}^{-1}\mZ_t\right)'.  \label{eq:resiDecomp}\end{aligned}
\end{equation}
This implies $T^{-1} \sumtT \big[ \widehat{\vu}_t\widehat{\vu}_t\tran - \vu_t \vu_t\tran \big]=O_p\big(T^{-1/2}\ln T\big)$ and 
\begin{equation}
	\left\|\sum_{i=1}^{T-1} k\left(i/b_T\right) T^{-1}\sum_{t=1}^{T-i} \big( \widehat{\vu}_{t+i}^{} \widehat{\vu}_t' - \vu_{t+i}^{} \vu_t'  \big)  \right\|\leq O_p\big(T^{-1/2}\ln T\big)\sum_{i=1}^{T-1} \left|k\left(i/b_T\right)\right|=O_p\big(T^{-1/2}b_T\ln T\big),
\end{equation}
where the final step is due to lemma 1 of \cite{jansson2002}.  Clearly, both terms are asymptotically negligible under the assumption $T^{-1/2}b_T\ln T\rightarrow 0$ as $T\to \infty$. The limits of the two remaining subblocks of $\widehat{\mSigma}_T(\widehat{\vgamma}_T)$ and $\widehat{\mGamma}_T(\widehat{\vgamma}_T,b_T)$ are derived similarly.  
\end{proof}

\section{Limiting Distribution for Example \ref{example:dettrend}}
Invoking Theorem \ref{thm:limiting_dist}, we have
\begin{equation*}
  \begin{bmatrix}
   T^{\theta_0+\frac{1}{2}}  \\
   T^{\theta_0+\frac{1}{2}} \tau_0 \ln(T)  & T^{\theta_0+\frac{1}{2}}
 \end{bmatrix}
 \begin{bmatrix}
  \,\widehat{\theta}_T -\theta_0 \\
  \,\widehat{\tau}_T - \tau_0
 \end{bmatrix}
 \wto
 \begin{bmatrix}
\myint \big(\tau_0 r^{\theta_0}\ln (r)\big)^2 dr 	& \myint \tau_0 r^{2\theta_0}{\ln (r)} dr  \\
\myint \tau_0 r^{2\theta_0}{\ln (r)} dr			& \myint r^{2\theta_0} dr 
\end{bmatrix}^{-1}
\begin{bmatrix}
\myint \tau_0 r^{\theta_0} \ln(r)  dB_u\\
\myint r^{\theta_0} dB_u
\end{bmatrix}.
\end{equation*}
It remains to show that the quantity in the RHS is normally distributed with a mean and variance as in \eqref{eq:powerlawtrend} of the main paper. Consider an arbitrary vector $\vc=[c_1,c_2]\tran$ and define
\begin{equation*}
 A_{\vc} = \vc\tran 
 \begin{bmatrix}
\myint \tau_0 r^{\theta_0} \ln(r)  dB_u \\
\myint r^{\theta_0} dB_u
\end{bmatrix}
=
\myint \left[c_1 \tau_0 r^{\theta_0} \ln(r) + c_2 r^{\theta_0}  \right] dB_u
\distrequal
\Omega_{uu}^{1/2} \myint \left[c_1 \tau_0 r^{\theta_0} \ln(r) + c_2 r^{\theta_0}  \right] dW_u.
\end{equation*}
Gaussianity is preserved under mean square integration (see, e.g., section 4.6 in \cite{soong1973}) and we proceed to the mean and variance of $A_{\vc}$. From (4.190) in the same reference, we get $\E(A_{\vc})=\Omega_{uu}^{1/2} \myint \left[c_1 \tau_0 r^{\theta_0} \ln(r) + c_2 r^{\theta_0}  \right] d\E\big(W_u\big)=0$. Moreover, (2.16) in \cite{tanaka2017} yields
$$
 \var\big(  A_{\vc} \big) = \Omega_{uu}^{} \myint \left[c_1 \tau_0 r^{\theta_0} \ln(r) + c_2 r^{\theta_0}  \right]^2 dr
  = \Omega_{uu}^{} \vc\tran
   \begin{bmatrix}
\myint \big(\tau_0 r^{\theta_0}\ln (r)\big)^2 dr 	& \myint \tau_0 r^{2\theta_0}{\ln (r)} dr  \\
\myint \tau_0 r^{2\theta_0}{\ln (r)} dr			& \myint r^{2\theta_0} dr 
\end{bmatrix}
\vc.
$$
Our choice of $\vc$ was arbitrary and thus $\left[\begin{smallmatrix}
\myint \tau_0 r^{\theta_0} \ln(r)  dB_u\\
\myint r^{\theta_0} dB_u
\end{smallmatrix}\right]\sim \rN \left(\vzeros,  \Omega_{uu}^{}
   \left[\begin{smallmatrix}
\myint \big(\tau_0 r^{\theta_0}\ln (r)\big)^2 dr 	& \myint \tau_0 r^{2\theta_0}{\ln (r)} dr  \\
\myint \tau_0 r^{2\theta_0}{\ln (r)} dr			& \myint r^{2\theta_0} dr 
\end{smallmatrix}\right] \right)$. Use $\myint \big(r^{\theta_0} \ln(r) \big)^2 dr=\frac{2}{(2\theta_0+1)^3}$, $\myint r^{2\theta_0} \ln(r) dr= -\frac{1}{(2\theta_0+1)^2}$, and basic linear algebra to recover the result.

\section{Simulation and Calculations related to FMOLS} \label{sec:FMOLS}

\subsection{Preliminary simulations}
We consider $N=1$ and test $H_0:\phi_2=0$ versus $H_a:\phi_2\neq0$ with FMOLS. Specifically, we generate the data according to
\begin{equation}
  y_t = \tau_1 + \tau_2 t  + \tau_g t^{\,\theta} + \phi_1 x_t + \phi_2 x_t^2 + u_t,
\label{eq:simDGP}
\end{equation}
where $x_t = \sum_{s=1}^t v_s$. The chosen parameter values are $\theta=2$, $\vtau=[\tau_1,\tau_2,\tau_g]\tran=[7,0.05,-5\times 10^{-4}]\tran$, and $\vphi=[\phi_1,\phi_2]\tran=[5,0]\tran$. These parameter values are representative. The disturbance vector $[u_t,v_t]\tran$ is generated from the VAR($1$) specification\footnote{We start the VAR recursions from $\left[\begin{smallmatrix} u_0 \\ v_0 \end{smallmatrix} \right]=\vzeros$ and use a presample of 50 observations to reduce the influence of these initial values.}
\begin{equation}
\begin{bmatrix}
 u_t \\
 v_t
\end{bmatrix}
=
\mA
\begin{bmatrix}
 u_{t-1} \\
 v_{t-1}
\end{bmatrix}
+
\begin{bmatrix}
 \eta_t \\
 \epsilon_t
\end{bmatrix},
\qquad\qquad\qquad 
\begin{bmatrix}
 \eta_t \\
 \epsilon_t
\end{bmatrix}
 \stackrel{i.i.d.}{\sim} \rN\
 \left( \vzeros,
 \left[
\begin{smallmatrix}
 1_{\phantom 1}^{\phantom 2}		 	& 0.5 \\
 0.5								& 1_{\phantom 1}^{\phantom 2}
\end{smallmatrix}
\right]
\right).
\label{eq:simdesign}
\end{equation}
We construct the autoregressive matrix $\mA$ along the following two steps: (1) generate a $(2\times 2)$ random matrix $\mU$ from $\rU[0,1]$ to construct the orthogonal matrix $\mH = \mU\left( \mU\tran \mU \right)^{-1/2}$, and (2) compute $\mA = \mH \mL \mH\tran$ with $\mL = \diag[0.9,0.7]$.

As shown in Figure \ref{fig:fmgls_comparisons}, for sample sizes as large as $15,000$, the empirical size of the feasible FMOLS estimator seems to stabilisze at $11\%$  whereas the infeasible estimator FMOLS($\theta_0$) yields an empirical size close to $5\%$. These results indicate poor finite sample performance of FMOLS or possible even a lack of asymptotic validity.

\begin{figure}[H]
	\centering
	\includegraphics[width=0.5\textwidth]{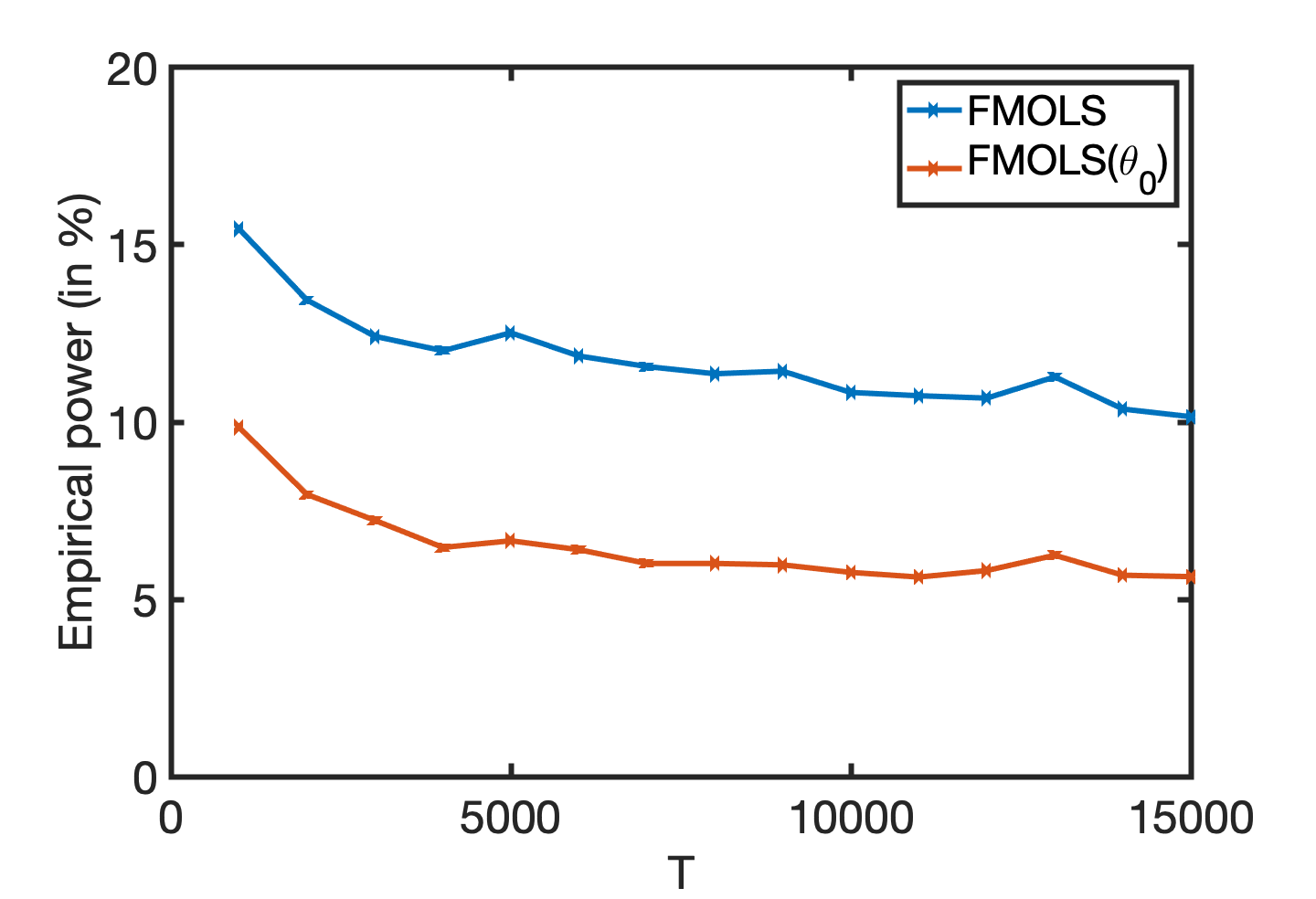}
	\caption{\footnotesize The empirical size of feasible and infeasible FMOLS estimators for a large range of sample sizes.}
	\label{fig:fmgls_comparisons}
\end{figure}

\subsection{Asymptotic properties of FMOLS}
We comment on the asymptotic properties of the FMOLS estimator when $N=1$. To shorten notation, the subscript `$i$' in $x_{i,t}$, $y_{i,t}$ and $p_i$ will be omitted. We analyse the asymptotic properties of $
 \widetilde{\mD}_{\theta_0,T}
 \left[ \begin{smallmatrix}
  \widehat{\tau}_{g,T}^+ - \vtau_{g,0} \\
  \widehat{\vbeta}_T^+ - \vbeta_0
 \end{smallmatrix} \right],
$
with $\widetilde{\mD}_{\theta_0,T}=\sqrt{T}
\left[\begin{smallmatrix}
 T^{\theta_0}		& \vzeros_{1 \times (p+2)}\\
 \vzeros_{(p+2)\times 1}			& \mD_{(1),T}
\end{smallmatrix}\right]
$
and 
$$
 \begin{bmatrix}
  \widehat{\tau}_{g,T}^+ \\
  \widehat{\vbeta}_T^+
 \end{bmatrix}
 =
 \left(
  \sumtT \vz_t^{}(\widehat{\theta}_T) \vz_t^{}(\widehat{\theta}_T)\tran
 \right)^{-1}
 \left(
 \sumtT \vz_t^{}(\widehat{\theta}_T) y_t^+
  - \mA^*
 \right),
$$
where $\vz_t(\theta) = [t^\theta, 1,t,x_{1,t},\ldots, x_{1,t}^{p_1}]\tran$, and $y_t^+$ and $\mA^*$ are second-order bias corrections. That is, $y_t^+= y_t - \widehat{\mOmega}_{uv} ^{} \widehat{\mOmega}_{vv}^{-1} \diff \vx_t$ and $\mA^*= [\vzeros_{3\times 1}',\mA_1^{*\prime}]\tran$ with $\mA_1^*=\widehat{\mDelta}_{vu}^{+}\left[T,2\sum_{t=1}^{T}x_{t},\dots,p\sum_{t=1}^{T}x_{t}^{p-1}\right]\tran$ and $\widehat{\mDelta}_{vu}^{+}$ equals $\widehat{\mDelta}_{vu}^+=\widehat{\mDelta}_{vu}^{}-\widehat{\mDelta}_{vv}^{}\widehat{\mOmega}_{vv}^{-1}\widehat{\mOmega}_{vu}^{}$.

We now investigate how the estimation of $\theta$ affects the limiting distribution of the FMOLS estimator. By straightforward linear algebra manipulations, we find
\begin{equation}\label{eq:fmols}
\widetilde{\mD}_{\theta_0,T}
\begin{bmatrix}
	\widehat{\vtau}_{g,T}^+ - \vtau_{g,0}\\
	\widehat{\vbeta}_T^+ - \vbeta_0
\end{bmatrix}
= \left(\widetilde{\mD}_{\theta_0,T}^{-1} \sumtT \vz_t^{}\big(\widehat{\theta}_T\big)\vz_t^{}\big(\widehat{\theta}_T\big)\tran \widetilde{\mD}_{\theta_0,T}^{-1}\right)^{-1} \widetilde{\mD}_{\theta_0,T}^{-1} \left[\sumtT \vz_t^{}\big(\widehat{\theta}_T\big) \tilde{u}_t^+-\mA^*\right],
\end{equation}
where $\tilde{u}_t^+ = \Big(\vz_t^{}\big( \theta_0 \big)- \vz_t^{}\big(\widehat{\theta}_T\big) \Big)\tran 
\left[\begin{smallmatrix}
	\vtau_{g,0}\\
	\vbeta_0
\end{smallmatrix}\right] +u_t
-\widehat{\mOmega}_{uv}^{}\widehat{\mOmega}_{vv}^{-1}\diff x_t^{}$. We will discuss $\widetilde{\mD}_{\theta_0,T}^{-1} \sumtT \vz_t^{}\big(\widehat{\theta}_T\big)\vz_t^{}\big(\widehat{\theta}_T\big)\tran \widetilde{\mD}_{\theta_0,T}^{-1}$ and $ \widetilde{\mD}_{\theta_0,T}^{-1} \left[\sumtT \vz_t^{}\big(\widehat{\theta}_T\big) \tilde{u}_t^+-\mA^*\right]$ separately after having enumerate several intermediate results.

\begin{lemma}\label{lem:fmolsresults}
Define $\widetilde{\vj}(r;\theta_0)=\big[r^{\theta_0}, 1, r ,B_v(r),\dots,B_v^p(r)\big]\tran$ and $\brownnormal_{u.v}^{}=\brownnormal_{u}^{}-\mOmega_{uv}^{}\mOmega_{vv}^{-1}B_{v}^{}$. Then, under Assumptions \ref{assumpt:identifytheta}-\ref{assumption:LRVestimation}, we have
 \begin{enumerate}[(i)]
 \item  $\widetilde{\mD}_{\theta_0,T}^{-1} \sumtT \vz_t^{}\big(\widehat{\theta}_T\big)\vz_t^{}\big(\widehat{\theta}_T\big)\tran \widetilde{\mD}_{\theta_0,T}^{-1} \wto \myint \widetilde{\vj}(r;\theta_0)\widetilde{\vj}(r;\theta_0)\tran dr$,
 \item $\widetilde{\mD}_{\theta_0,T}^{-1} \left[\sumtT \vz_t^{}\big(\theta_0\big) \left(u_t^{}-\widehat{\mOmega}_{uv}^{}\widehat{\mOmega}_{vv}^{-1}v_t^{}\right)-\mA^*\right] \wto \myint \widetilde{\vj}(r;\theta_0)d\brownnormal_{u.v}(r)$,
 \item $\widetilde{\mD}_{\theta_0,T}^{-1} \sumtT \vz_t^{}\big(\theta_0\big)\left(\vz_t^{}\big(\widehat{\theta}_T\big)-\vz_t^{}\big(\theta_0\big)\right)\tran\left[\begin{smallmatrix}
	\vtau_{g,0}\\
	\vphi_0
	\end{smallmatrix}\right]=O_p\Big(  \ln T \Big)$\,,
 \item $\sum_{t=1}^{T}\widetilde{\mD}_{\theta_0,b_T}^{-1}\left(\vz_t^{}\big(\widehat{\theta}_T\big)-\vz_t^{}\big(\theta_0\big)\right)\left(\vz_t^{}\big(\widehat{\theta}_T\big)-\vz_t^{}\big(\theta_0\big)\right)\tran \left[\begin{smallmatrix}
	\vtau_{g,0} \\ \vphi_0 \end{smallmatrix} \right]=O_p\left( (\ln T)^2 T^{-(\theta_L+\frac{1}{2})}\right)\,$,
 \item $\widetilde{\mD}_{\theta_0,T}^{-1} \sumtT \left(\vz_t^{}\big(\widehat{\theta}_T\big)-\vz_t^{}\big(\theta_0\big)\right)\left(u_t^{}-\widehat{\mOmega}_{uv}^{}\widehat{\mOmega}_{vv}^{-1}v_t^{}\right)=o_p(1)$.
 \end{enumerate}
\end{lemma}

\begin{proof}
\textbf{\textit{(i)}} We can always add and subtract such that the LHS of $(i)$ reads
\begin{equation}
\begin{aligned}
 \widetilde{\mD}_{\theta_0,T}^{-1} \sumtT & \vz_t^{}\big(\widehat{\theta}_T\big)\vz_t^{}\big(\widehat{\theta}_T\big)\tran \widetilde{\mD}_{\theta_0,T} = 
 \widetilde{\mD}_{\theta_0,T}^{-1} \sumtT \vz_t^{}\big(\theta_0\big)\vz_t^{}\big(\theta_0\big)\tran \widetilde{\mD}_{\theta_0,T}^{-1} \\
 &\qquad+ \left( \widetilde{\mD}_{\theta_0,T}^{-1}  \sumtT \vz_t^{}\big(\widehat{\theta}_T\big)\vz_t^{}\big(\widehat{\theta}_T\big)\tran \widetilde{\mD}_{\theta_0,T}^{-1} - \widetilde{\mD}_{\theta_0,T}^{-1} \sumtT  \vz_t^{}\big(\theta_0 \big)\vz_t^{}\big(\theta_0\big)\tran\widetilde{\mD}_{\theta_0,T}^{-1} \right).
\end{aligned}
\label{eq:myterms1}
\end{equation}
Lemma \ref{lem:Brownianconv}(iii) implies that the first term in the RHS of \eqref{eq:myterms1} converges to $\myint \widetilde{\vj}(r;\theta_0)\widetilde{\vj}(r;\theta_0)\tran dr$. It remains to show that the term in parenthesis vanishes. By $\sum_{t}\va_t^{}\va_t\tran-\sum_{t}\vb_t^{}\vb_t\tran=\sum_{t}(\va_t^{}-\vb_t^{})(\va_t^{}-\vb_t^{})\tran+\sum_{t}(\va_t^{}-\vb_t^{})\vb_t\tran+\sum_{t}\vb_t^{}(\va_t^{}-\vb_t^{})\tran$ and the Cauchy-Schwarz inequality, we have
\begin{equation*}
\begin{aligned}
	\Bigg\| & \widetilde{\mD}_{\theta_0,T}^{-1} \sumtT \vz_t\big(\widehat{\theta}_T\big)\vz_t\big(\widehat{\theta}_T\big)'\widetilde{\mD}_{\theta_0,T}^{-1} -\widetilde{\mD}_{\theta_0,T}^{-1} \sumtT \vz_t\big(\theta_0\big)\vz_t\big(\theta_0\big)\tran \widetilde{\mD}_{\theta_0,T}^{-1}\Bigg\| \\
	&\leq \sumtT \left\|\widetilde{\mD}_{\theta_0,T}^{-1}\left(\vz_t\big(\widehat{\theta}_T\big)-\vz_t\big(\theta_0\big)\right)\right\|^2+2\sum_{t=1}^{T}\left\|\widetilde{\mD}_{\theta_0,T}^{-1}\vz_t\big(\theta_0\big)\right\|\left\|\widetilde{\mD}_{\theta_0,T}^{-1}\left(\vz_t\big(\widehat{\theta}_T\big)-\vz_t\big(\theta_0\big)\right)\right\| \\
	&\leq \sumtT \left\|\widetilde{\mD}_{\theta_0,T}^{-1}\left(\vz_t\big(\widehat{\theta}_T\big)-\vz_t\big(\theta_0\big)\right)\right\|^2+2\sqrt{\sumtT \left\|\widetilde{\mD}_{\theta_0,T}^{-1}\vz_t\big(\theta_0\big)\right\|^2}\sqrt{\sumtT \left\|\widetilde{\mD}_{\theta_0,T}^{-1}\left(\vz_t\big(\widehat{\theta}_T\big)-\vz_t\big(\theta_0\big)\right)\right\|^2}.
\end{aligned}
\end{equation*}
We have $\sumtT \big\|\widetilde{\mD}_{\theta_0,T}^{-1}\vz_t\big(\theta_0\big)\big\|^2 = \tr\left( \sumtT \widetilde{\mD}_{\theta_0,T}^{-1}\vz_t\big(\theta_0\big) \vz_t\big(\theta_0\big)\tran \widetilde{\mD}_{\theta_0,T}^{-1} \right) \wto \tr\Big( \myint \widetilde{\vj}(r;\theta_0)\widetilde{\vj}(r;\theta_0)\tran dr \Big)$. Next note that $\sumtT \big\|\widetilde{\mD}_{\theta_0,T}^{-1}\big(\vz_t\big(\widehat{\theta}_T\big)-\vz_t\big(\theta_0\big)\big)\big\|^2=\frac{1}{T}\sumtT [T^{-\theta_0}(t^{\widehat\theta_T} - t^{\theta_0})]^2$. We have
\begin{equation}
\begin{aligned}
 \frac{1}{T}\sumtT \big[T^{-\theta_{0}} &\big(t^{\widehat{\theta}_T}-t^{\theta_{0}}\big)\big]^2
  \leq C \left(\widehat{\theta}_T - \theta_{0} \right) \frac{1}{T} \sumtT \left(\frac{t}{T}\right)^{2 \theta_{0}} (\ln t)^2 \\
  & \leq C T^{-2(\theta_{0}+\frac{1}{2})} (\ln T)^2 \left[T^{\theta_{0}+\frac{1}{2}} \left(\, \widehat{\theta}_T - \theta_{0} \right)\right]^2  \sup_{\theta_L\leq \theta \leq \theta_U} \left| \frac{1}{T} \sumtT \left(\frac{t}{T} \right)^{2\theta} \right|
  = o_p(1),
\end{aligned}
\end{equation}
where we used the mean-value theorem and Lemma \ref{lem:boundsforfixed}(i). The claim follows. \textbf{\textit{(ii)}} $\widehat{\mOmega}_{uv}$ and $\widehat{\mOmega}_{vv}$ consistently estimate $\mOmega_{uv}$ and $\mOmega_{vv}$, respectively (Theorem \ref{thm:consistentLRVs}). It therefore suffices to look at $\widetilde{\mD}_{\theta_0,T}^{-1}\sumtT \vz_t^{}\big(\theta_0\big) \left(u_t^{} - \mOmega_{uv}^{}\mOmega_{vv}^{-1}\vv_t^{}\right)$ and $\widetilde{\mD}_{\theta_0,T}^{-1} \mA^*$. Lemma \ref{lem:Brownianconv}(ii) with $u_t^+=u_t^{} - \mOmega_{uv}^{}\mOmega_{vv}^{-1}v_t^{}$ instead of $u_t$ gives the limiting result $\frac{1}{\sqrt{T}}\sum_{t=1}^{T}\big(x_{t}/\sqrt{T}\big)^j u_t^+\wto\int_{0}^{1}\brown_{v}^j(r)d\brownnormal_{u.v}(r)+j\mDelta_{vu}^+\int_{0}^{1}\brown_{v}^{j-1}(r)dr$, which implies
\begin{equation}
 \widetilde{\mD}_{\theta_0,T}^{-1}\sumtT \vz_t^{}\big(\theta_0\big) \left(u_t^{} - \mOmega_{uv}^{}\mOmega_{vv}^{-1}v_t^{}\right) \wto \myint \widetilde{\vj}(r;\theta_0)d\brownnormal_{u.v}(r)+\widetilde{\bm{\calB}}_{vu}^+,
\label{eq:FMOLScov}
\end{equation}
where $\widetilde{\bm{\calB}}_{vu}^+=\big[\vzeros_{3\times 1}\tran,\vb\tran\mDelta_{vu}^+\big]\tran$. The term $-\widetilde{\mD}_{\vtheta_0,T}^{-1} \mA^*$ is constructed to asymptotically cancel out the term $\widetilde{\bm{\calB}}_{vu}^+$ in the RHS of \eqref{eq:FMOLScov}. \textbf{\textit{(iii)}} Using $\vz_t(\widehat{\theta}_T)-\vz_t(\theta_0)=\begin{bmatrix} t^{\widehat{\theta}_T}-t^{\vtheta_0} & \vzeros\tran \end{bmatrix}\tran$, we have
\begin{equation*}
\widetilde{\mD}_{\theta_0,T}^{-1} \sumtT \vz_t^{}\big(\theta_0\big)\left(\vz_t^{}\big(\widehat{\theta}_T\big)-\vz_t^{}\big(\theta_0\big)\right)\tran\left[\begin{smallmatrix}
	\tau_{g,0}\\
	\beta_0
	\end{smallmatrix}\right]
=
\widetilde{\mD}_{\theta_0,T}^{-1} \sumtT \vz_t^{}\big(\theta_0\big)\left( t^{\widehat{\theta}_T}-t^{\vtheta_0}  \right) \tau_{g,0}.
\end{equation*}
The typical elements in the vector on the RHS are of the form $\frac{1}{\sqrt{T}}\sumtT \left(\frac{t}{T}\right)^{\theta_{0}} \tau_{g,0}\big(t^{\widehat{\theta}_{T}}-t^{\theta_{0}}\big)$ or $\frac{1}{\sqrt{T}}\sumtT \left(\frac{x_{it}}{\sqrt{T}}\right)^{j} \tau_{g,0}\big(t^{\widehat{\theta}_{T}}-t^{\theta_{0}}\big)$. We show that both contributions are $O_p\left( \ln T \right)$. By the mean-value theorem and Lemma \ref{lem:boundsforfixed}(i),
\begin{equation}
\begin{aligned}
 \Bigg|\frac{1}{\sqrt{T}}\sumtT &\left(\frac{t}{T}\right)^{\theta_{0}} \tau_{g,0}\big(t^{\widehat{\theta}_{T}}-t^{\theta_{0}}\big) \Bigg|
  \leq \Bigg| \frac{1}{\sqrt{T}}  \tau_{g,0} \sumtT \left( \frac{t}{T} \right)^{\theta_{0}} t^{\theta_{0}} \left( t^{\widehat{\theta}_{T}-\theta_{0}}-1\right) \Bigg|\\
   & \leq C |\tau_{g,0}| \, \Big| T^{\theta_{0}+\frac{1}{2}}(\widehat{\theta}_{T}-\theta_{0})\Big| \, \frac{1}{T} \sumtT \left(\frac{t}{T}\right)^{2\theta_{0}} \ln t \\
   &\leq  C (\ln T) |\tau_{g,0}| \, \Big| T^{\theta_{0}+\frac{1}{2}}(\widehat{\theta}_{T}-\theta_{0})\Big| \, \left[ \frac{1}{T} \sumtT \left(\frac{t}{T}\right)^{2\theta_{0}} \right] =  O_p( \ln T).
\end{aligned}
\label{eq:fm_subsample_I1}
\end{equation}
Similarly, from the mean-value theorem and Cauchy-Schwartz inequality, we see that
\begin{equation}
\begin{aligned}
 \Bigg| &\frac{1}{\sqrt{T}}\sumtT \left(\frac{x_{it}}{\sqrt{T}}\right)^{j} \tau_{g,0}\big(t^{\widehat{\theta}_{T}}-t^{\theta_{0}}\big) \Bigg|
  \leq \Bigg| \frac{1}{\sqrt{T}} \tau_{g,0} \sumtT \left(\frac{x_{it}}{\sqrt{T}}\right)^{j}   t^{\theta_{0}} \left( t^{\widehat{\theta}_{T}-\theta_{0}}-1\right) \Bigg| \\
  & \leq C  |\tau_{0k}| \, \Big| T^{\theta_{0}+\frac{1}{2}}(\widehat{\theta}_{T}-\theta_{0})\Big| \; \frac{1}{T} \sumtT \left| \frac{x_{it}}{\sqrt{T}}\right|^{j} \left( \frac{t}{T} \right)^{\theta_{0}} \ln t \\
  & \leq  C (\ln T) |\tau_{g,0}| \, \Big| T^{\theta_{0}+\frac{1}{2}}(\widehat{\theta}_{T}-\theta_{0})\Big| \; \sqrt{ \frac{1}{T} \sumtT \left(\frac{x_{it}}{\sqrt{T}} \right)^{2j}} \sqrt{\frac{1}{T} \sumtT \left(\frac{t}{T}\right)^{2\theta_L} }.
\end{aligned}
\label{eq:fm_subsample_I2}
\end{equation}
From \eqref{eq:fm_subsample_I1} and \eqref{eq:fm_subsample_I2} we conclude that
$
\widetilde{\mD}_{\theta_0,T}^{-1} \sumtT \vz_t^{}\big(\theta_0\big)\left(\vz_t^{}\big(\widehat{\theta}_T\big)-\vz_t^{}\big(\theta_0\big)\right)\tran\left[\begin{smallmatrix}
	\tau_{g,0}\\
	\vbeta_0
	\end{smallmatrix}\right]
=
O_p(\ln T)$. \textbf{\textit{(iv)}} Use $\vz_t(\widehat{\theta}_T)-\vz_t(\theta_0)=\begin{bmatrix} t^{\widehat{\theta}_T}-t^{\vtheta_0} & \vzeros\tran \end{bmatrix}\tran$ to obtain $\widetilde{\mD}_{\theta_0,T}^{-1} \sumtT  \left(\vz_t^{}\big(\widehat{\theta}_T\big)-\vz_t^{}\big(\theta_0\big)\right)\left(\vz_t^{}\big(\widehat{\theta}_T\big)-\vz_t^{}\big(\theta_0\big)\right)\tran
\left[\begin{smallmatrix}
	\tau_{g,0} \\ \vbeta_0
\end{smallmatrix} \right]\\
	=\left[\begin{smallmatrix}
T^{-\theta_0} \frac{1}{\sqrt{T}}\sumtT \left(t^{\widehat{\theta}_T}-t^{\vtheta_0}\right)^2 \tau_{g,0} \\
	\vzeros
	\end{smallmatrix}\right]$. The absolute value of the nonzero element can be bounded as follows
\begin{equation*}
\begin{aligned}
 \Bigg| &\tau_{g,0} \frac{1}{T^{\theta_{0}+1/2}} \sumtT \left(t^{\widehat{\theta}_{T}}-t^{\theta_{0}} \right)^2  \Bigg|
  \leq|\tau_{g,0} | \, \frac{1}{T^{\theta_{0}+1/2}} \sumtT t^{ 2\theta_{0}} \left| t^{\widehat{\theta}_{T}-\theta_{0}}-1 \right| \left| t^{\widehat{\theta}_{T}-\theta_{0}}-1 \right| \\
  & \leq C |\tau_{g,0} | \, \big| \widehat{\theta}_{T}-\theta_{0}^2 \big| \frac{1}{T^{\theta_{0}+1/2}} \sumtT t^{ 2\theta_{0i}} (\ln t)^2 \\
  & \leq C (\ln T)^2  T^{-(\theta_{L}+\frac{1}{2})} |\tau_{g,0} | \,  \big| T^{\theta_{0}+\frac{1}{2}} (\widehat{\theta}_{T}-\theta_{0} ) \big|^2  \left[ \frac{1}{T} \sumtT \left( \frac{t}{T} \right)^{2\theta_{0}} \right] 
   = O_p\left( \frac{ (\ln T)^2 }{T^{\theta_{L}+\frac{1}{2}}} \right).
\end{aligned}
\end{equation*}
\textbf{\textit{(v)}} By similar steps as before, and invoking Theorem \ref{thm:consistentLRVs}, it is easy to show that it suffices to bound $T^{-(\theta_{0}+\frac{1}{2})} \sumtT \big(t^{\widehat{\theta}_T}-t^{\theta_{0}} \big)\big(u_t -\mOmega_{uv}^{} \mOmega_{vv}^{-1} v_t^{} \big)$. Writing $u_t^+=u_t -\mOmega_{uv}^{} \mOmega_{vv}^{-1} v_t^{}$, we have
\begin{equation*}
\begin{aligned}
  &T^{-(\theta_{0}+\frac{1}{2})} \sumtT \big(t^{\widehat{\theta}_T}-t^{\theta_{0}} \big) u_t^+
  = \frac{1}{\sqrt{T}} \sumtT \left( \frac{t}{T} \right)^{\theta_{0}}
   \big( t^{\widehat{\theta}_{T}-\theta_{0}} -1 \big) u_t^+ 
   = \big(\, \widehat{\theta}_T -\theta_{0} \big) \frac{1}{\sqrt{T}} \sumtT (\ln t) \left( \frac{t}{T} \right)^{\theta_{0}} u_t^+ +o_p(1) \\
   &= T^{-(\theta_{0}+\frac{1}{2})} \left[T^{\theta_{0}+\frac{1}{2}}  \big( \widehat{\theta}_T -\theta_{0} \big) \right] \frac{1}{\sqrt{T}} \sumtT \left(\ln \frac{t}{T} \right) \left( \frac{t}{T} \right)^{\theta_{0}} u_t^+ \\
   &\qquad
   + T^{-(\theta_{0}+\frac{1}{2})} \left[T^{\theta_{0}+\frac{1}{2}}  \big(\, \widehat{\theta}_T -\theta_{0} \big) \right] (\ln T)\frac{1}{\sqrt{T}}  \sumtT \left( \frac{t}{T} \right)^{\theta_{0}} u_t^+ + o_p(1) = \frac{1}{T^{\theta_{0}+\frac{1}{2} }} O_p(1) + \frac{\ln T}{T^{\theta_{0}+\frac{1}{2}}} O_p(1).
\end{aligned}
\end{equation*}
This establishes (v).
\end{proof}

The currents upper bounds in the lemma above suggest that the RHS of \eqref{eq:fmols} does not converge to a Gaussian mixture limiting distribution. The problematic expression is Lemma \ref{lem:fmolsresults}(iii). That is, if $\theta$ is estimated, then $\vz_t^{}\big(\widehat{\theta}_T\big)-\vz_t^{}\big(\theta_0\big)$ does not convergence sufficiently fast to zero to obtain the standard stochastic integral.

\newpage
\section{Additional Simulation Results}\label{secSup:MoreSimulations}
\subsection{Empirical size for DGP1 when $\theta=0.8$}
\begin{table}[h!]
	\centering
	\caption{The empirical size (in \%) of the single-equation $t$-tests $H_0:\beta_{2,1}=0$ and the joint Wald tests for $H_0: \beta_{2,1}=\ldots=\beta_{2,N}=0$ with $ \beta_{2,i}$ denoting the coefficient in front of $x_{i,t}^2$. The Monte Carlo results are based on: simulated inference with $\theta$ estimated by NLS (SimNLS), simulated inference with known $\theta=0.8$ (SimNLS($\theta_0$)), and two Fully Modified estimators for systems developed by \cite{wagnergrabarczykhong2019} with known $\theta=0.8$ (FM-SOLS($\theta_0$) and FM-SUR($\theta_0$)).}
	\label{tbl:SizeWagnerEtAl2}
	\resizebox{\textwidth}{!}{%
		\begin{tabular}{crrrrrrrrrrrrrr}
			\toprule
			$\theta_0=0.8$ & \multicolumn{4}{c}{$N=3$} & \multicolumn{1}{c}{} & \multicolumn{4}{c}{$N=5$} & \multicolumn{1}{c}{} & \multicolumn{4}{c}{$N=10$} \\
			\midrule
			$\rho$ & \multicolumn{1}{c}{SimNLS} & \multicolumn{1}{c}{SimNLS($\theta_0$)} & \multicolumn{1}{c}{FM-SOLS($\theta_0$)} & \multicolumn{1}{c}{FM-SUR($\theta_0$)} & \multicolumn{1}{c}{} & \multicolumn{1}{c}{SimNLS} & \multicolumn{1}{c}{SimNLS($\theta_0$)} & \multicolumn{1}{c}{FM-SOLS($\theta_0$)} & \multicolumn{1}{c}{FM-SUR($\theta_0$)} & \multicolumn{1}{c}{} & \multicolumn{1}{c}{SimNLS} & \multicolumn{1}{c}{SimNLS($\theta_0$)} & \multicolumn{1}{c}{FM-SOLS($\theta_0$)} & \multicolumn{1}{c}{FM-SUR($\theta_0$)} \\
			\midrule
\multicolumn{15}{l}{$T=150$} \\
\midrule
0 & 4.03 & 3.93 & 9.10 & 10.47 &  & 4.67 & 4.67 & 10.03 & 12.90 &  & 4.40 & 4.50 & 10.80 & 16.67 \\
0.3 & 4.60 & 4.50 & 9.77 & 11.07 &  & 4.53 & 4.50 & 10.37 & 13.07 &  & 5.07 & 4.90 & 11.90 & 19.50 \\
0.6 & 4.53 & 4.47 & 10.57 & 12.60 &  & 4.30 & 4.23 & 11.87 & 16.27 &  & 4.60 & 4.40 & 13.57 & 29.83 \\
0.8 & 4.33 & 4.50 & 13.80 & 18.47 &  & 4.70 & 4.63 & 15.60 & 27.07 &  & 4.33 & 4.30 & 16.30 & 56.73 \\
\midrule
\multicolumn{15}{l}{$T=300$} \\
\midrule
0 & 4.20 & 4.23 & 7.87 & 8.67 &  & 4.87 & 4.87 & 7.87 & 9.40 &  & 4.37 & 4.40 & 8.27 & 10.97 \\
0.3 & 5.27 & 5.23 & 8.47 & 9.50 &  & 4.47 & 4.50 & 9.23 & 10.83 &  & 4.47 & 4.50 & 8.87 & 12.93 \\
0.6 & 4.50 & 4.63 & 9.47 & 11.00 &  & 5.10 & 4.83 & 10.10 & 12.90 &  & 4.27 & 4.47 & 10.63 & 18.57 \\
0.8 & 4.60 & 4.40 & 12.07 & 14.47 &  & 4.43 & 4.30 & 12.17 & 18.13 &  & 5.47 & 5.23 & 14.00 & 35.00 \\
\midrule
\multicolumn{15}{l}{$T=600$} \\
\midrule
0 & 4.43 & 4.47 & 6.83 & 7.47 &  & 4.23 & 4.37 & 6.53 & 6.93 &  & 4.23 & 4.17 & 7.27 & 9.10 \\
0.3 & 4.97 & 5.03 & 7.57 & 8.60 &  & 5.13 & 4.93 & 7.37 & 8.10 &  & 4.70 & 4.93 & 8.10 & 9.60 \\
0.6 & 5.27 & 5.40 & 8.27 & 9.50 &  & 5.17 & 4.93 & 8.43 & 9.33 &  & 5.03 & 4.97 & 9.60 & 14.57 \\
0.8 & 4.13 & 4.30 & 8.83 & 10.10 &  & 4.93 & 4.77 & 9.73 & 14.10 &  & 5.20 & 4.93 & 10.90 & 23.70 \\
\midrule
			\multicolumn{15}{l}{\textbf{Panel B: Joint   test}} \\
\midrule
\multicolumn{15}{l}{$T=150$} \\
\midrule
0 & 3.57 & 3.63 & 12.03 & 15.23 &  & 4.00 & 4.23 & 14.30 & 21.57 &  & 4.03 & 3.93 & 26.03 & 50.13 \\
0.3 & 4.07 & 3.90 & 13.83 & 16.43 &  & 3.77 & 3.47 & 19.47 & 26.93 &  & 3.23 & 3.40 & 29.67 & 60.20 \\
0.6 & 3.73 & 3.70 & 17.03 & 21.43 &  & 3.60 & 3.67 & 23.60 & 38.07 &  & 2.33 & 2.03 & 39.73 & 83.70 \\
0.8 & 3.20 & 2.80 & 23.43 & 31.03 &  & 2.87 & 2.87 & 32.13 & 58.27 &  & 1.53 & 1.37 & 50.30 & 82.73 \\
\midrule
\multicolumn{15}{l}{$T=300$} \\
\midrule
0 & 5.13 & 5.13 & 10.47 & 11.43 &  & 3.43 & 3.60 & 12.00 & 15.17 &  & 3.67 & 3.57 & 17.93 & 30.33 \\
0.3 & 4.40 & 4.30 & 9.90 & 11.63 &  & 4.07 & 3.87 & 13.43 & 17.83 &  & 3.83 & 4.00 & 19.30 & 36.63 \\
0.6 & 4.20 & 4.37 & 13.40 & 15.57 &  & 3.97 & 4.00 & 17.47 & 24.33 &  & 3.30 & 3.20 & 28.60 & 59.80 \\
0.8 & 4.07 & 3.63 & 16.27 & 20.87 &  & 3.47 & 3.17 & 22.20 & 38.97 &  & 2.40 & 2.40 & 37.77 & 86.00 \\
\midrule
\multicolumn{15}{l}{$T=600$} \\
\midrule
0 & 3.50 & 3.53 & 7.03 & 7.97 &  & 4.43 & 4.63 & 8.83 & 11.07 &  & 3.90 & 4.10 & 12.63 & 18.93 \\
0.3 & 4.57 & 4.53 & 8.90 & 9.53 &  & 4.17 & 4.13 & 10.70 & 12.67 &  & 4.53 & 4.53 & 15.07 & 23.83 \\
0.6 & 5.37 & 4.87 & 10.40 & 12.23 &  & 4.73 & 4.33 & 13.30 & 16.97 &  & 4.03 & 4.07 & 21.47 & 39.70 \\
0.8 & 3.70 & 3.83 & 11.63 & 14.30 &  & 3.50 & 3.70 & 15.20 & 24.83 &  & 3.60 & 3.60 & 26.37 & 66.77\\
\bottomrule
		\end{tabular}%
	}
\end{table}

\newpage
\subsection{Empirical size for DGP1 when $\theta=1.8$}
\begin{table}[h!]
	\centering
	\caption{The empirical size (in \%) of the single-equation $t$-tests $H_0:\beta_{2,1}=0$ and the joint Wald tests for $H_0: \beta_{2,1}=\ldots=\beta_{2,N}=0$ with $ \beta_{2,i}$ denoting the coefficient in front of $x_{i,t}^2$. The Monte Carlo results are based on: simulated inference with $\theta$ estimated by NLS (SimNLS), simulated inference with known $\theta=1.8$ (SimNLS($\theta_0$)), and two Fully Modified estimators for systems developed by \cite{wagnergrabarczykhong2019} with known $\theta=1.8$ (FM-SOLS($\theta_0$) and FM-SUR($\theta_0$)).}
	\label{tbl:SizeWagnerEtAl3}
	\resizebox{\textwidth}{!}{%
		\begin{tabular}{crrrrrrrrrrrrrr}
			\toprule
			$\theta_0=1.8$ & \multicolumn{4}{c}{$N=3$} & \multicolumn{1}{c}{} & \multicolumn{4}{c}{$N=5$} & \multicolumn{1}{c}{} & \multicolumn{4}{c}{$N=10$} \\
			\midrule
			$\rho$ & \multicolumn{1}{c}{SimNLS} & \multicolumn{1}{c}{SimNLS($\theta_0$)} & \multicolumn{1}{c}{FM-SOLS($\theta_0$)} & \multicolumn{1}{c}{FM-SUR($\theta_0$)} & \multicolumn{1}{c}{} & \multicolumn{1}{c}{SimNLS} & \multicolumn{1}{c}{SimNLS($\theta_0$)} & \multicolumn{1}{c}{FM-SOLS($\theta_0$)} & \multicolumn{1}{c}{FM-SUR($\theta_0$)} & \multicolumn{1}{c}{} & \multicolumn{1}{c}{SimNLS} & \multicolumn{1}{c}{SimNLS($\theta_0$)} & \multicolumn{1}{c}{FM-SOLS($\theta_0$)} & \multicolumn{1}{c}{FM-SUR($\theta_0$)} \\
\midrule
\multicolumn{15}{l}{$T=150$} \\
\midrule
0 & 4.70 & 4.80 & 9.87 & 11.37 &  & 4.13 & 4.23 & 10.27 & 12.30 &  & 4.43 & 4.37 & 9.80 & 15.60 \\
0.3 & 3.97 & 3.80 & 9.67 & 10.97 &  & 4.43 & 4.50 & 9.83 & 13.13 &  & 4.77 & 4.50 & 11.40 & 18.37 \\
0.6 & 5.03 & 4.63 & 12.53 & 14.77 &  & 4.70 & 4.83 & 12.27 & 17.07 &  & 3.63 & 3.37 & 12.27 & 29.63 \\
0.8 & 5.43 & 5.47 & 14.93 & 18.93 &  & 4.97 & 4.60 & 14.80 & 27.27 &  & 5.23 & 4.90 & 16.30 & 56.23 \\
\midrule
\multicolumn{15}{l}{$T=300$} \\
\midrule
0 & 4.53 & 4.80 & 7.23 & 8.10 &  & 4.40 & 4.43 & 7.80 & 9.63 &  & 4.47 & 4.73 & 8.53 & 11.77 \\
0.3 & 4.27 & 4.50 & 8.10 & 9.37 &  & 5.07 & 5.00 & 9.23 & 9.97 &  & 4.37 & 4.20 & 8.90 & 12.80 \\
0.6 & 6.23 & 5.87 & 10.17 & 12.23 &  & 4.73 & 4.77 & 10.03 & 13.63 &  & 4.77 & 4.30 & 10.50 & 18.83 \\
0.8 & 4.50 & 4.43 & 11.00 & 13.80 &  & 4.57 & 4.43 & 11.97 & 18.70 &  & 4.10 & 3.73 & 13.70 & 36.73 \\
\midrule
\multicolumn{15}{l}{$T=600$} \\
\midrule
0 & 4.13 & 4.37 & 7.33 & 7.83 &  & 4.07 & 4.13 & 6.90 & 7.90 &  & 4.97 & 4.97 & 6.87 & 8.17 \\
0.3 & 4.70 & 4.90 & 8.33 & 8.90 &  & 4.77 & 4.83 & 6.93 & 8.13 &  & 4.80 & 4.67 & 8.10 & 10.57 \\
0.6 & 5.77 & 6.10 & 8.70 & 9.03 &  & 5.17 & 5.23 & 8.23 & 9.37 &  & 5.53 & 5.53 & 9.27 & 13.80 \\
0.8 & 4.73 & 4.80 & 9.07 & 10.33 &  & 4.57 & 4.77 & 9.63 & 12.80 &  & 4.93 & 4.73 & 12.03 & 25.87 \\
\midrule
			\multicolumn{15}{l}{\textbf{Panel B: Joint   test}} \\
\midrule
\multicolumn{15}{l}{$T=150$} \\
\midrule
0 & 4.23 & 4.03 & 12.13 & 15.07 &  & 3.53 & 3.50 & 16.50 & 23.07 &  & 3.40 & 3.33 & 25.57 & 50.30 \\
0.3 & 4.13 & 3.80 & 14.50 & 16.00 &  & 3.63 & 3.77 & 19.90 & 27.00 &  & 3.77 & 3.63 & 30.80 & 58.17 \\
0.6 & 3.23 & 2.83 & 17.13 & 22.07 &  & 3.93 & 3.73 & 24.67 & 38.53 &  & 2.13 & 2.43 & 40.87 & 82.83 \\
0.8 & 3.57 & 3.03 & 23.77 & 30.83 &  & 2.93 & 2.30 & 32.43 & 57.57 &  & 2.37 & 1.77 & 49.77 & 83.23 \\
\midrule
\multicolumn{15}{l}{$T=300$} \\
\midrule
0 & 4.63 & 4.83 & 9.40 & 10.97 &  & 4.10 & 4.20 & 11.47 & 15.63 &  & 3.80 & 3.83 & 18.50 & 31.87 \\
0.3 & 4.60 & 4.80 & 11.70 & 13.00 &  & 4.67 & 4.67 & 14.90 & 18.43 &  & 3.67 & 3.57 & 19.17 & 36.73 \\
0.6 & 5.00 & 4.53 & 14.00 & 15.17 &  & 4.03 & 3.33 & 17.00 & 24.80 &  & 3.17 & 2.83 & 29.00 & 59.57 \\
0.8 & 3.43 & 3.10 & 18.00 & 21.27 &  & 3.33 & 2.80 & 22.77 & 39.80 &  & 2.50 & 1.93 & 37.93 & 86.07 \\
\midrule
\multicolumn{15}{l}{$T=600$} \\
\midrule
0 & 4.23 & 4.47 & 7.97 & 8.63 &  & 3.20 & 3.43 & 9.10 & 11.43 &  & 3.87 & 3.83 & 10.47 & 17.40 \\
0.3 & 4.20 & 4.37 & 9.37 & 10.80 &  & 4.53 & 4.60 & 10.37 & 12.17 &  & 3.87 & 3.77 & 15.47 & 23.80 \\
0.6 & 5.43 & 5.33 & 11.90 & 12.30 &  & 4.80 & 4.70 & 12.27 & 17.03 &  & 5.00 & 4.57 & 19.63 & 39.03 \\
0.8 & 3.73 & 3.63 & 11.67 & 14.40 &  & 3.97 & 3.77 & 16.37 & 24.77 &  & 3.30 & 3.07 & 27.87 & 68.27 \\
\bottomrule
		\end{tabular}%
	}
\end{table}

\section{Details on Simulation DGPs 2(a)--2(c)} \label{secSup:SimulationDGP2}
The parameters of simulation DGPs have been selected according to the following general procedure.
\begin{enumerate}[\textsc{Step} 1:]
 \item Load the data and estimate the model corresponding to the specification under $H_0$. The resulting coefficients $\widehat\vgamma_T$ and residual series $\{\widehat{\vu}_t \}$ are stored.
 \item Estimate a VAR($1$) on the residuals, $\widehat{\vu}_t = \mA^{(1)} \widehat{\vu}_{t-1}  + \vxi_t^{(1)}$, and compute $\widehat\vxi_t^{ (1)} = \widehat{\vu}_t - \widehat{\mA}^{(1)} \widehat{\vu}_{t-1}$ for $t=2,\ldots, T$.
 \item Repeat \textsc{Step} 2 for $\diff \vx_t$. Using obvious notation, the resulting filtered residuals are $\widehat\vxi_t^{ (2)} = \diff \vx_t- \widehat{\mA}^{(2)} \diff \vx_{t-1}$ for $t=3,\ldots, T$.
 \item Set $\widehat \vxi_t = [\widehat\vxi_t^{ (1)\prime}, \widehat\vxi_t^{ (2)\prime}]\tran $ and compute the $(2N\times 2N)$ covariance matrix estimate $\widehat\mSigma = \frac{1}{T-2}\sum_{t=3}^T \widehat \vxi_t \widehat \vxi_t\tran$.
 \item The simulated data is based on the parameters from \textsc{Step} 1 -- \textsc{Step} 4. First, generate $\vxi_t =[\vxi_t^{ (1)\prime}, \vxi_t^{ (2)\prime}]\tran \stackrel{i.i.d.}{\sim} \rN(\vzeros, \widehat\mSigma)$. Subsequently, we use the results from the VAR($1$) models:
 \begin{enumerate}
  \item Set $\vu_0=\vzeros$ and construct innovations according to $\vu_t = \widehat \mA^{(1)} \vu_{t-1}+\vxi_t^{ (1)}$.
  \item Set $\diff \vx_0 = \vzeros$, construct the increments of the integrated explanatory variables through the recursion $\diff \vx_t = \widehat{\mA}^{(2)} \diff \vx_{t-1}+ \vxi_t^{ (2)}$, and compute the partial sums $\vx_t = \sum_{s=1}^t \diff \vx_{s}$.
 \end{enumerate}
 Given the simulated innovations and simulated integrated regressors, it remains to use $\widehat\vgamma_T$ to obtain the simulated dependent variables.
\end{enumerate}

Three remarks follow. First, we explicitly choose individual VAR($1$) models for $\{\widehat\vu_t\}$ and $\{\widehat \diff \vx_t \}$ rather than a single $2N$-dimensional VAR(1) for the joint vector. Otherwise, with $N=6$ and $T=145$ in the data, the number of parameters in the autoregressive matrix would be $12^2=144$ which is rather close to the length of the data  series. Similarly, the selection of the VAR($1$) specification results from a trade-off between model parsimony and a simulation DGP with serial correlation. Second, we follow the literature and rebuild the integrated explanatory variables as random walks \emph{without} drift. Third, the specific values of all parameters (rounded to 2 decimals) are reported in the next subsections.

\subsection{Parameter values for DGP2(a)}
The estimated parameters of the model are
$$
 \vy_t
 =
 -1.37\times 10^{-5} t^{2.45} \vones_6 
 +
 \left[
 \begin{smallmatrix}
  -6.39 \\
  -0.12 \\
  -12.66 \\
  -3.99 \\
  -1.67 \times 10^1 \\
  -2.55
 \end{smallmatrix}
 \right]
 +
 \left[
 \begin{smallmatrix}
  -4.9 \times 10^{-3} \\
  7.6 \times 10^{-3} \\
  1.05 \times 10^{-3} \\
  1.09 \times 10^{-3} \\
  -2.7 \times 10^{-3} \\
  2.14 \times 10^{-4}
 \end{smallmatrix} \right]t
 +
 \left[
  \begin{smallmatrix}
 1.73 x_{1,t} \\
 1.01 x_{2,t}\\
 2.22 x_{3,t}\\
 1.33 x_{4,t}\\
 2.55 x_{5,t}\\
 1.33 x_{6,t}
 \end{smallmatrix} \right] 
 + \widehat{\vu}_t.
$$
The results for the VAR(1) specifications follow
$$
 \widehat\mA^{(1)}
 =
 \left[
 \begin{smallmatrix}
  0.74	& -0.04	& 0.02	& 0.32	& -0.34	& 0.41 \\
  0.05	& 0.68	& 0.01	& -0.30	& -0.05	& -0.01 \\
  0.16	& 1.05	& 0.62	& -0.54	& -0.29	& -0.55 \\
  0.04	& 0.12	& 0.03	& 0.41	& 0.00	& -0.21 \\
  0.10	& 0.31	& -0.08	& -0.29	& 0.67	& 0.01 \\
  0.05	& 0.08	& -0.03	& -0.12	& 0.02	& 0.44
 \end{smallmatrix}\right], \qquad
  \widehat\mA^{(2)}
  =
  \left[
  \begin{smallmatrix}
  -0.09	& -0.32	& 0.21	& 0.96	& -0.22	& 0.55 \\
  -0.04	& 0.02	& 0.39	& 0.10	& -0.15	& -0.00 \\
  -0.08	&-0.10	& 0.38	& 0.15	& -0.07	& 0.35 \\
  -0.37	& 0.38	& -0.13	& 0.25	& -0.00	& 0.26 \\
  -0.18	& 0.19	& 0.05	& 0.18	& -0.07	& 0.41 \\
  0.03	& 0.18	& -0.05	& 0.02	& -0.06	& 0.47
  \end{smallmatrix}\right],
$$
and
$$
 \widehat\mSigma
 =
 \left[
 \begin{smallmatrix}
 6.13 	& *		& *		& *		&*		&*		&*		&*		&*		&*		&*		&* \\
 0.39		& 0.70 	&*		&*		&*		&*		&*		&*		&*		&*		&*		&* \\
 0.22		& 1.00	& 7.19 	&*		&*		&*		&*		&*		&*		&*		&*		&*\\
 0.59		& 0.06	& 0.29	& 0.91 	&*		&*		&*		&*		&*		&*		&*		&*\\
 0.57		& 0.15	& 0.68	& 0.38 	& 0.01 	&*		&*		&*		&*		&*		&*		&*\\
 0.36		& 0.13	& 0.37	& 0.01	& 0.13	& 0.46 	&*		&*		&*		&*		&*		&* \\
 -0.29	& -0.02	& 0.32	& 0.03	& 0.18	& -0.02	& 0.43 	&*		&*		&*		&*		&*\\
 0.13		& 0.07	& 0.22	& 0.04	& -0.06 	& 0.06	& -0.00	& 0.14	&*		&*		&*		&*\\
 0.21		& 0.11 	&0.21	& 0.09 	& -0.03	& 0.05	& 0.04	& 0.10	& 0.19	&*		&*		&*\\
 0.23 	& 0.18	& 0.28	& -0.07	& 0.01	& 0.04	& 0.03	& 0.11	& 0.10	& 0.33	&*		&*\\
 0.10 	& 0.07	& 0.25	& -0.01	& -0.34	& 0.07	& 0.01	& 0.06	& 0.05	& 0.05	& 0.20	&*\\
 0.02		& 0.02	& -0.03	& -0.00	& 0.03	& 0.02	& 0.05	& -0.01	& 0.02	& 0.01	& 0.02	& 0.08 \\
 \end{smallmatrix}\right]\times 10^{-2}.
$$

\subsection{Parameter values for DGP2(b)}
The estimated model specification is
$$
\vy_t 
=
\left[
\begin{smallmatrix}
  -1.01 \\
  8.64 \\
  -5.15 \\
  3.78 \\
  -17.84 \\
  11.84
\end{smallmatrix}
\right]
+
\left[ \begin{smallmatrix}
 -1.11\times 10^{-2} \\
 5.78\times 10^{-3} \\
 1.63 \times 10^{-2} \\
 7.69 \times 10^{-3} \\
 -2.36 \times 10^{-2} \\
 4.69\times 10^{-3}
\end{smallmatrix}
\right] t+
+
\left[
\begin{smallmatrix}
 1.10 x_{1,t} \\
 -1.44\times 10^{-3} x_{2,t} \\
 1.23 x_{3,t} \\
 4.42\times 10^{-1} x_{4,t} \\
 2.73 x_{5,t} \\
 -3.21 \times 10^{-1} x_{6,t}
\end{smallmatrix}
\right]
+\widehat\vu_t.
$$
The VAR($1$) dynamics in the innovations and increments are governed by
$$
 \widehat\mA^{(1)}
 =
 \left[ \begin{smallmatrix}
 0.78		& 0.01	& 0.15	& -0.07	& -0.39	& -0.07 \\
 0.03		& 0.87	& -0.01	& 0.02	& -0.02	& -0.07 \\
 0.06		& 1.29	& 0.60	& -0.37	& -0.01	& -0.35 \\
 0.04		& 0.36	& 0.01	& 0.50	& -0.05	& 0.06 \\
 0.03		& 0.38	& -0.03	& -0.07	& 0.60	& -0.06 \\
 0.03		& -0.05	& 0.01	& 0.08	& 0.02	& 0.59
 \end{smallmatrix}\right], \qquad
  \widehat\mA^{(2)}
  =
  \left[
  \begin{smallmatrix}
  -0.09	& -0.32	& 0.21	& 0.96	& -0.22	& 0.55 \\
  -0.04	& 0.02	& 0.39	& 0.10	& -0.15	& -0.00 \\
  -0.08	& -0.10	& 0.38	& 0.15	& -0.07	& 0.35 \\
  -0.37	& 0.38	& -0.13	& 0.25	& -0.00	& 0.26 \\
  -0.18	& 0.19	& 0.05	& 0.18	& -0.07	& 0.41 \\
  0.03	& 0.18	& -0.05	& 0.02	& -0.06	& 0.47
  \end{smallmatrix}
  \right],
$$
and
$$
\widehat\mSigma =
\left[\begin{smallmatrix}
5.65 		& *		&*		&*		&*		&*		&*		&*		&*		&*		&*		&*\\
0.68 		& 1.13	& *		&*		&*		&*		&*		&*		&*		&*		&*		&*\\
1.29		& 1.79 	& 8.56	&*		&*		&*		&*		&*		&*		&*		&*		&*\\
0.84 		& 0.58	& 1.33	& 1.23	&*		&*		&*		&*		&*		&*		&*		&*\\
0.67 		& 0.30	& 1.16	& 0.71	& 2.32	&*		&*		&*		&*		&*		&*		&*\\
0.09 		& 0.25	& 0.60	& 0.17	& 0.16	&0.72	&*		&*		&*		&*		&*		&*\\
-0.09 	& -0.01	& 0.24	& 0.06	& 0.19	&0.04	& 0.43	&*		&*		&*		&*		&*\\
0.15 		& 0.19	& 0.39	& 0.13	& -0.04	& 0.01	& -0.00	& 0.14	&*		&*		&*		&*\\
0.24 		& 0.20	& 0.40	& 0.20	& 0.01	&0.04	& 0.04	& 0.10	&0.19	&*		&*		&*\\
0.24 		& 0.32	& 0.43	& 0.19	& 0.09	& 0.01	& 0.03	& 0.11	& 0.10	& 0.33	&*		&*\\
0.09 		& 0.12	& 0.28	& 0.04	& -0.34	& 0.08	& 0.01	& 0.06	& 0.05	& 0.05	& 0.20	&*\\
0.00		& 0.01	& -0.07	& 0.02	& 0.02	& 0.09	& 0.05	& -0.01	& 0.02	& 0.01	& 0.02 	& 0.08
\end{smallmatrix}\right]\times 10^{-2}.
$$

\subsection{Parameter values for DGP2(c)}
The simulation experiments regarding the performance of the KPSS test are based on the same specification as DGP2(a).

\section{Further Empirical Results}\label{appendix:furtherempirics}
\subsection{Unit root tests}
\begin{table}[H]
	\centering
	\caption{\footnotesize The t-statistics for the ADF and DF-GLS unit root tests.  The columns with header `const' and `const \& trend' refer to the inclusion of only an intercept or both intercept and linear trend. Rejection of the unit root hypothesis at a $10\%$ and $5\%$ level are indicated with one and two stars, respectively.}
	\label{table:unitroottests}
	\begin{threeparttable}
	\begin{tabular}{cc cc cc cc cc}
		\toprule
		& \multicolumn{4}{c}{ADF}       &	                                        & \multicolumn{4}{c}{DF-GLS}                                         \\
		\cline{2-5} \cline{7-10}
		& \multicolumn{2}{c}{const} 	& \multicolumn{2}{c}{const \& trend} 		&	&\multicolumn{2}{c}{const} 	& \multicolumn{2}{c}{const \& trend} \\ 
		\cline{2-5} \cline{7-10}
		& GDP      & $\text{CO}_2$	& GDP          & $\text{CO}_2$      		&	& GDP     & $\text{CO}_2$ 	& GDP 		&$\text{CO}_2$ 	\\ 
		\midrule
		Australia   &  0.287   & -2.549   & -2.050   & -1.986      &   & 2.046  &  1.379  & -1.577      & -0.732            \\
		Austria     & -0.055   & -2.118   & -1.943   & -2.738      &   & 1.478  & -1.143  & -1.655      & $-2.718^{*}$      \\
		Belgium     &  0.153   & -2.336   & -1.705   & -2.818      &   & 2.041	& -0.794  & -1.287		& -2.644			\\
		Canada      & -0.575   & -1.133   & -2.020   & -1.120      &   & 1.117	&  0.874  & -1.894		& -0.387			\\
		Denmark    	& -0.235   & -2.446   & -2.326   & -0.136      &   & 1.393	&  0.410  & -1.505 		&  0.084			\\
		Finland    	& -0.362   & -1.327   & -2.315   &$-3.248^{*}$ &   & 0.420	& -0.076  & -1.155 		& $-3.217^{**}$		\\
		France      & -0.557   & -2.438   & -1.823   & -1.858      &   & 1.087	& -0.267  & -1.470		& -1.212			\\
		Germany     & -0.374   &$-3.099^{**}$& -2.767&$-3.971^{**}$&   & 1.195	& -0.726  & -2.474		& -2.080			\\
		Italy       & -0.252   & -1.546   & -1.759   & -1.987      &   & 1.213	&  0.354  & -1.240		& -1.860			\\
		Japan       &  0.010   & -0.862   & -1.733   & -0.941      &   & 1.382	&  0.504  & -1.272		& -0.878			\\
		Netherlands & -0.106   & -1.629   & -2.247   & -3.106      &   & 1.378	&  0.213  & -1.679 		& $-2.818^{*}$		\\
		Norway      & -0.680   & -2.044   & -2.064   & -2.318      &   & 0.749  &  0.331  & -1.017      & -1.292      	    \\
		Portugal    & -1.432   & -0.455   & -1.697   & -1.676      &   & -0.708 &  0.593  & -0.741      & -1.923      	    \\
		Spain       &  0.402   & -1.243   & -1.354   & -1.994      &   & 1.487  &  0.959  & -1.077      & -2.014       	    \\
		Sweden      & -0.789   & -2.075   & -2.289   & -1.625      &   & 0.258  &  0.180  & -1.513      & -0.968       	    \\
		Switzerland & -1.093   & -1.963   & -2.785   & -1.989      &   & 2.272  &  0.368  & -2.447      & -1.237        	\\
		UK          & -0.179   & -0.721   & -1.262   & -0.402      &   & 2.446	& -0.622  & -0.608		& -0.013 			\\
		USA         & -0.349   & -2.055   & -2.871   & -1.322      &   & 2.409  & -0.101  &$-2.708^{*}$	& -0.812 			\\
		\bottomrule   
	\end{tabular}
		    \begin{tablenotes}
	    	\footnotesize
	    	\item Note: Asterisks denote rejection of the null hypothesis at the $^{***}1\%$, $^{**}5\%$, and $^{*}10\%$ significance level.
    	\end{tablenotes}
\end{threeparttable}
\end{table}

\newpage
\subsection{\cite{perronyabu2009} Test for deterministic trend coefficient}
The \cite{perronyabu2009} test is used to test for the presence of a deterministic trend function in the log per capita GDP series, see Table \ref{tbl:pytests}. The test allows for integrated or stationary errors. The details of the procedure can be found on page 61 of \cite{perronyabu2009}. The asymptotic distribution of this test statistic is standard normal (quantiles are $z_{0.95}=1.645$, $z_{0.975}=1.96$, and $z_{0.995}=2.58$).

\begin{table}[h]
\centering
\caption{\cite{perronyabu2009} test statistic for each of the 18 countries.}
\label{tbl:pytests}
\begin{subtable}{.32\textwidth}
\centering
	\begin{tabular}{c  | c  c c c c c c c c}
		\toprule
			& $\widehat{PY}$ 		\\ \hline
Australia		& 3.17 \\
Austria		& 2.19 \\
Belgium		& 3.52 \\
Canada		& 3.33 \\
Denmark		& 5.58 \\
Finland		& 4.27 \\
        \bottomrule
	\end{tabular}
\end{subtable}
\begin{subtable}{.32\textwidth}
\centering
	\begin{tabular}{c  | c  c c c c c c c c}
		\toprule
			& $\widehat{PY}$ \\ \hline
France		& 2.41 \\
Germany		& 1.91 \\
Italy			& 2.11 \\
Japan		& 2.93 \\
Netherlands	& 2.27 \\
Norway		& 5.85 \\
        \bottomrule
	\end{tabular}
\end{subtable}
\begin{subtable}{.32\textwidth}
\centering
	\begin{tabular}{c  | c  c c c c c c c c}
		\toprule
			& $\widehat{PY}$ \\ \hline
Portugal		& 2.16 \\
Spain		& 2.31 \\
Sweden		& 7.12 \\
Switzerland	& 3.91 \\
UK			& 3.60 \\
USA			& 4.12 \\
    \bottomrule
	\end{tabular}
\end{subtable}
\end{table}

\newpage
\subsection{Overviews for Austria and Finland} \label{sec:AuandFi}
\begin{figure}[H]
\begin{subfigure}{.5\textwidth}
  \centering
  \includegraphics[width=.8\linewidth]{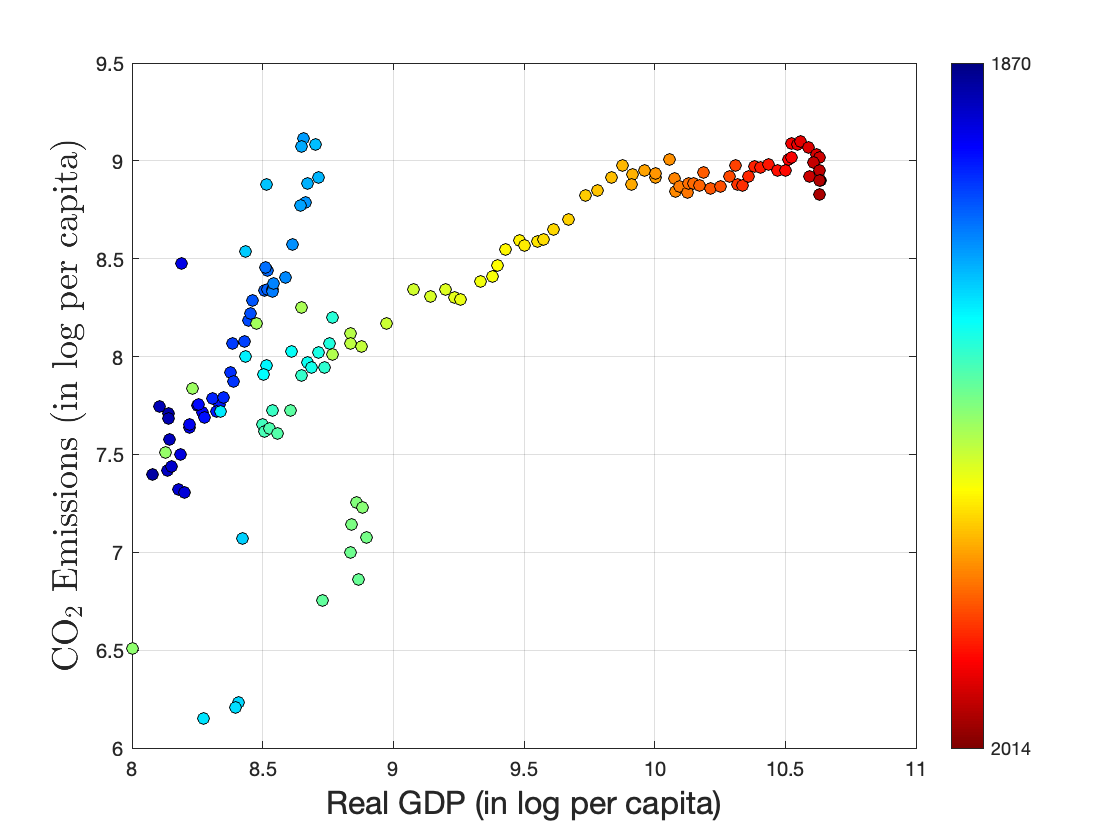}
  \caption{}
\end{subfigure}%
\begin{subfigure}{.5\textwidth}
  \centering
  \includegraphics[width=.8\linewidth]{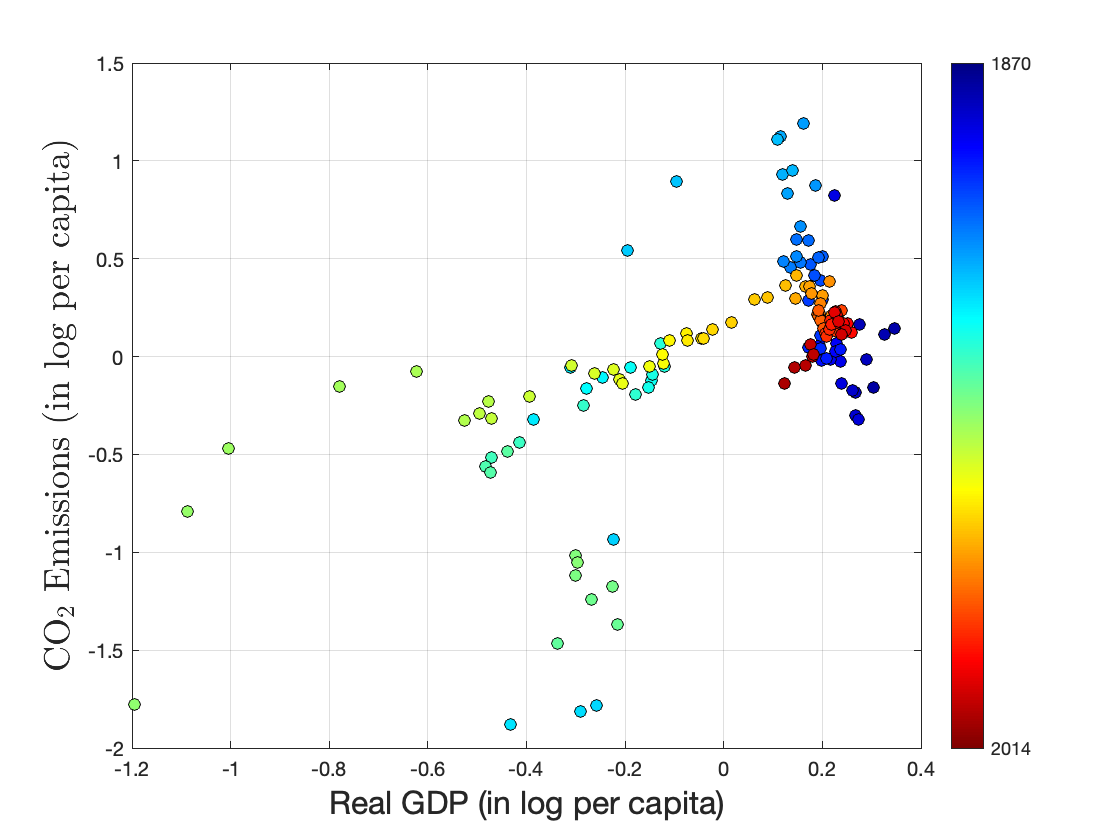}
  \caption{}
\end{subfigure}
\begin{subfigure}{.5\textwidth}
  \centering
  \includegraphics[width=.8\linewidth]{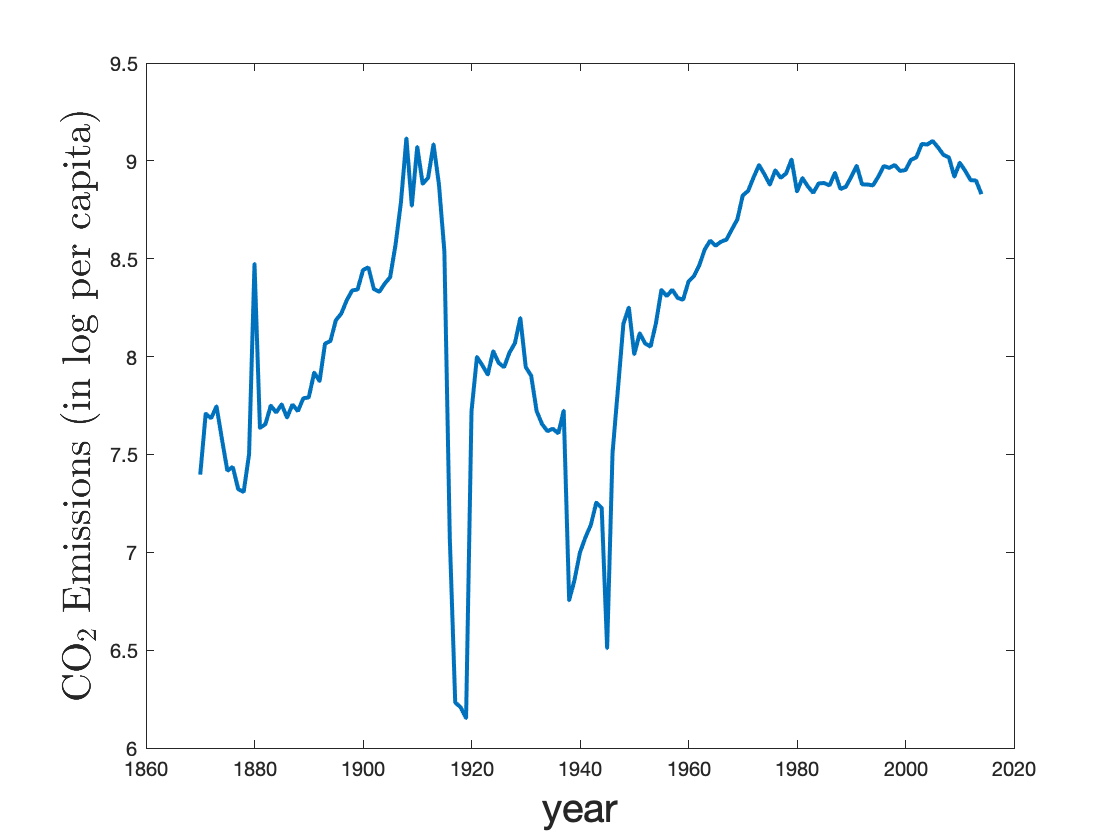}
  \caption{}
\end{subfigure}
\begin{subfigure}{.5\textwidth}
  \centering
  \includegraphics[width=.8\linewidth]{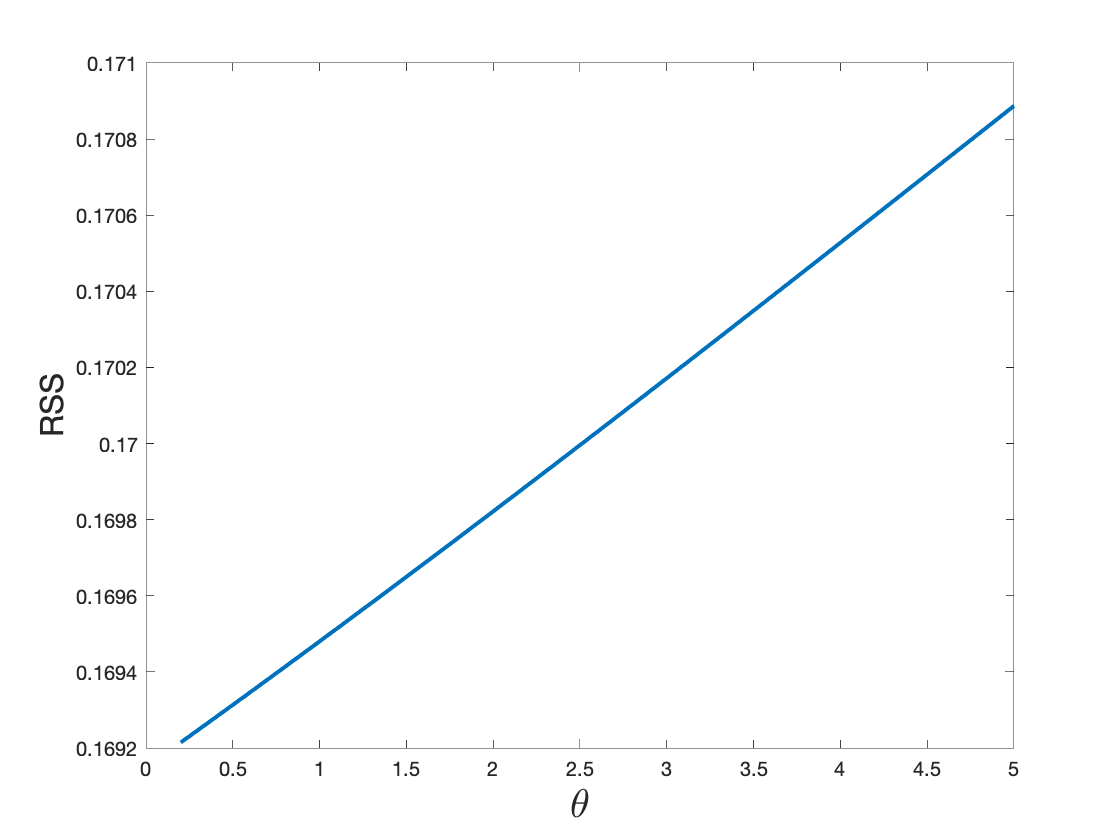}
  \caption{}
\end{subfigure}
\begin{subfigure}{.5\textwidth}
  \centering
  \includegraphics[width=.8\linewidth]{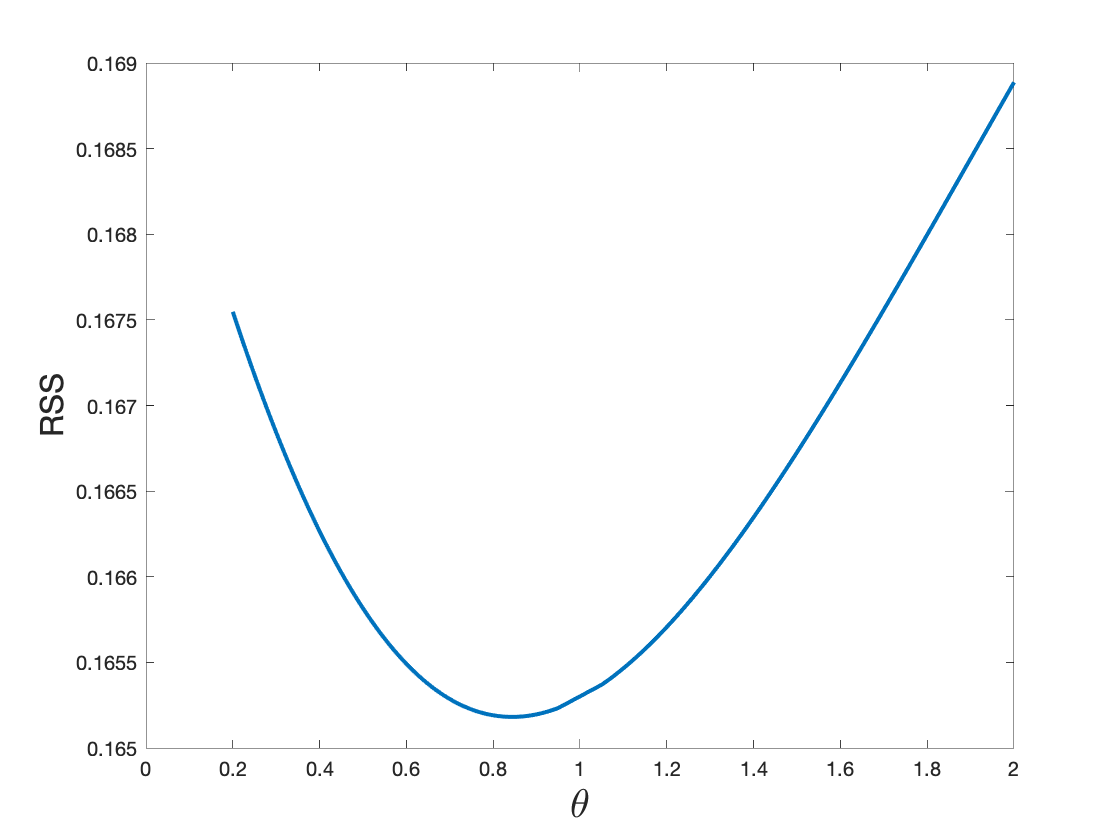}
  \caption{}
\end{subfigure}
\begin{subfigure}{.5\textwidth}
  \centering
  \includegraphics[width=.8\linewidth]{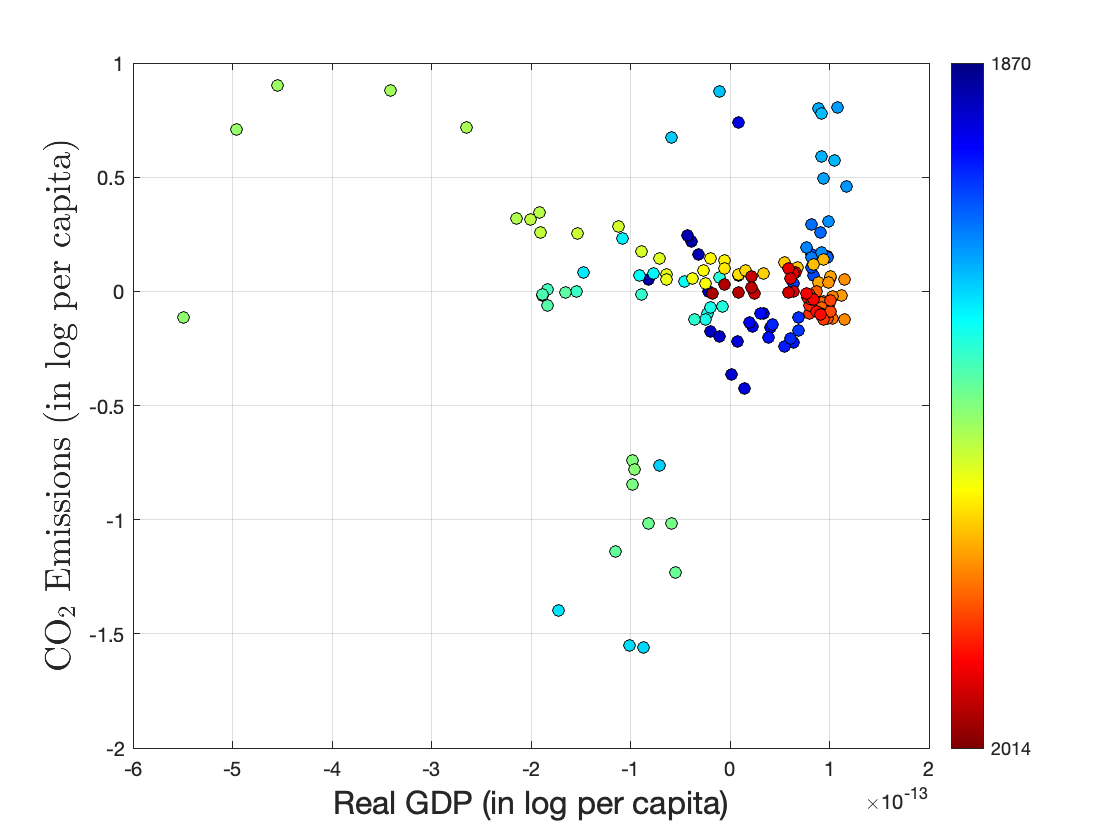}
  \caption{}
\end{subfigure}
\caption{Overview graphs for Austria over 1870-2014. \textbf{(a)} $\log(\text{GDP})$ versus $\log(\text{CO}_2)$ (both per capita). \textbf{(b)} As subfigure (a) but using detrended variables. \textbf{(c)} The log per capita CO\textsubscript{2} emissions time series for Austria. \textbf{(d)} The residual sum of squares (RSS) for the nonlinear model specification $y_t=\tau_1+\tau_2 t + \phi_1 x_t+ \phi_2 x_t^\theta+u_t$ for various values of $\theta$. \textbf{(e)} The RSS as a function of $\theta$ for the flexible nonlinear trend specification $y_t=\tau_1 + \tau_2 t + \tau_3 t^\theta + \phi x_t+u_t$. \textbf{(f)} The relation between $x_t$ and $y_t$ after partialling out the constant, linear trend, and flexible deterministic trend.}
\label{fig:overviewAustria}
\end{figure}

\begin{figure}[H]
\begin{subfigure}{.5\textwidth}
  \centering
  \includegraphics[width=.8\linewidth]{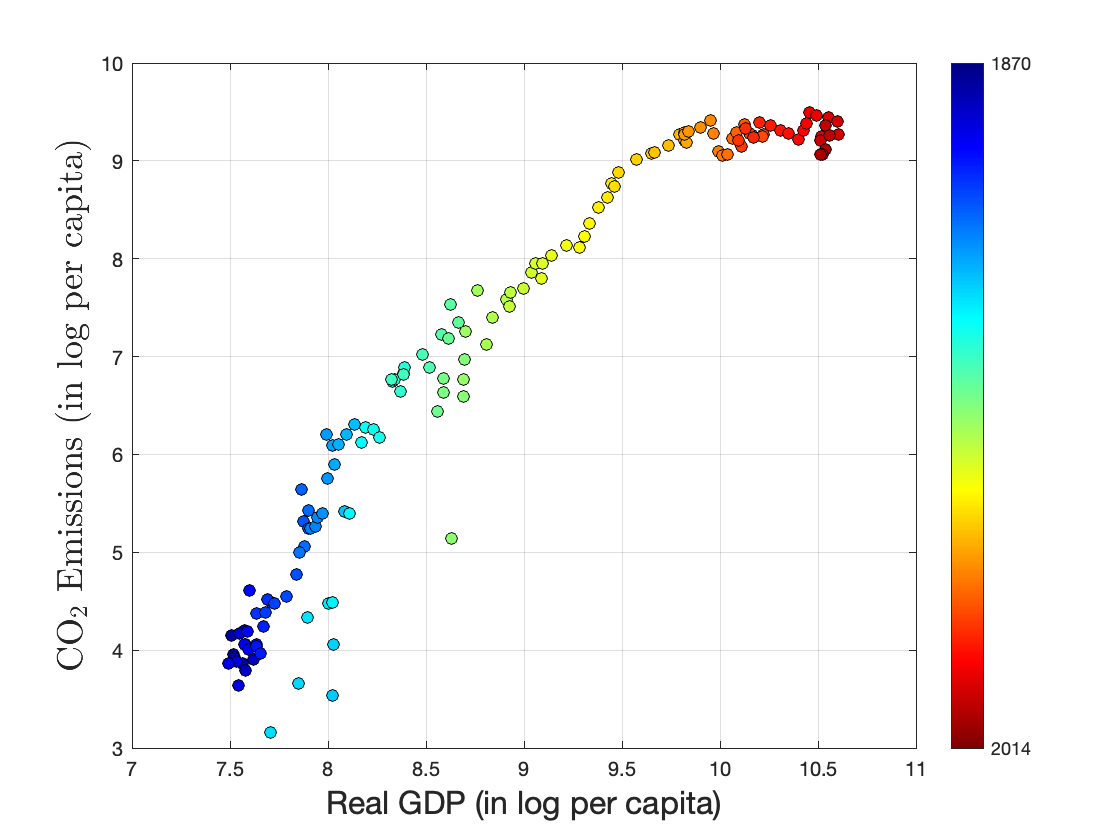}
  \caption{}
\end{subfigure}%
\begin{subfigure}{.5\textwidth}
  \centering
  \includegraphics[width=.8\linewidth]{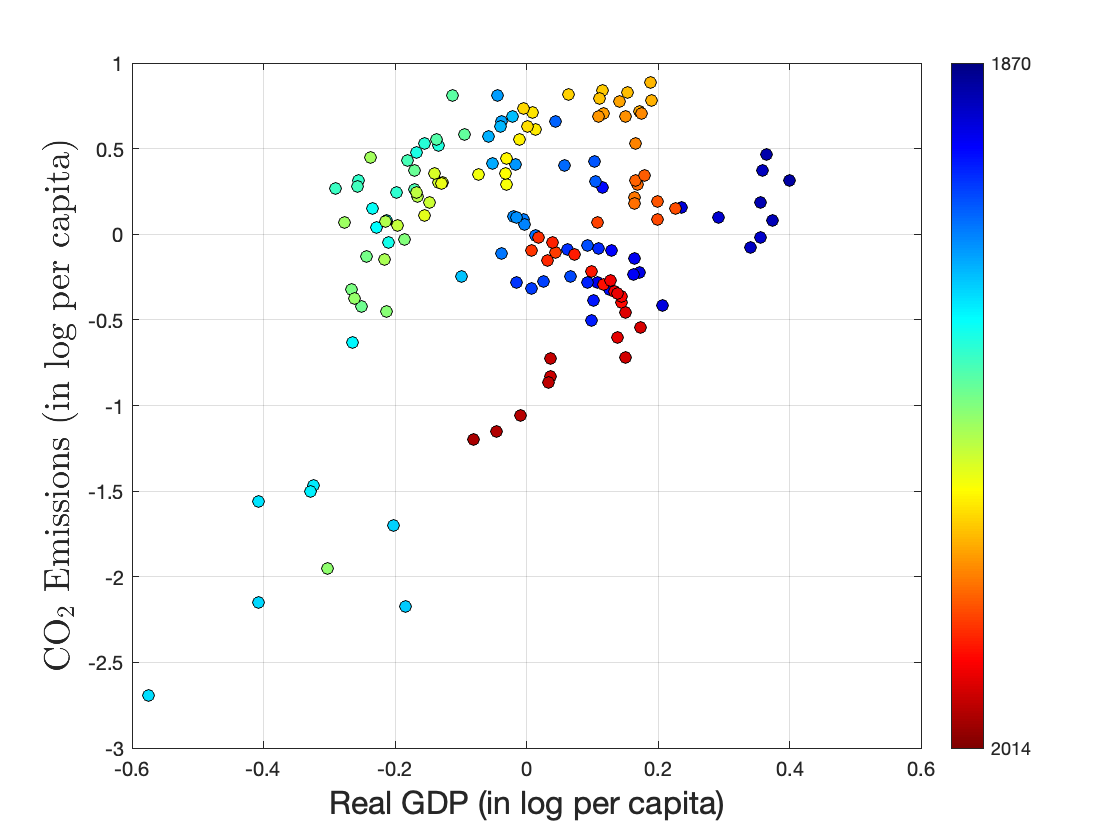}
  \caption{}
\end{subfigure}
\begin{subfigure}{.5\textwidth}
  \centering
  \includegraphics[width=.8\linewidth]{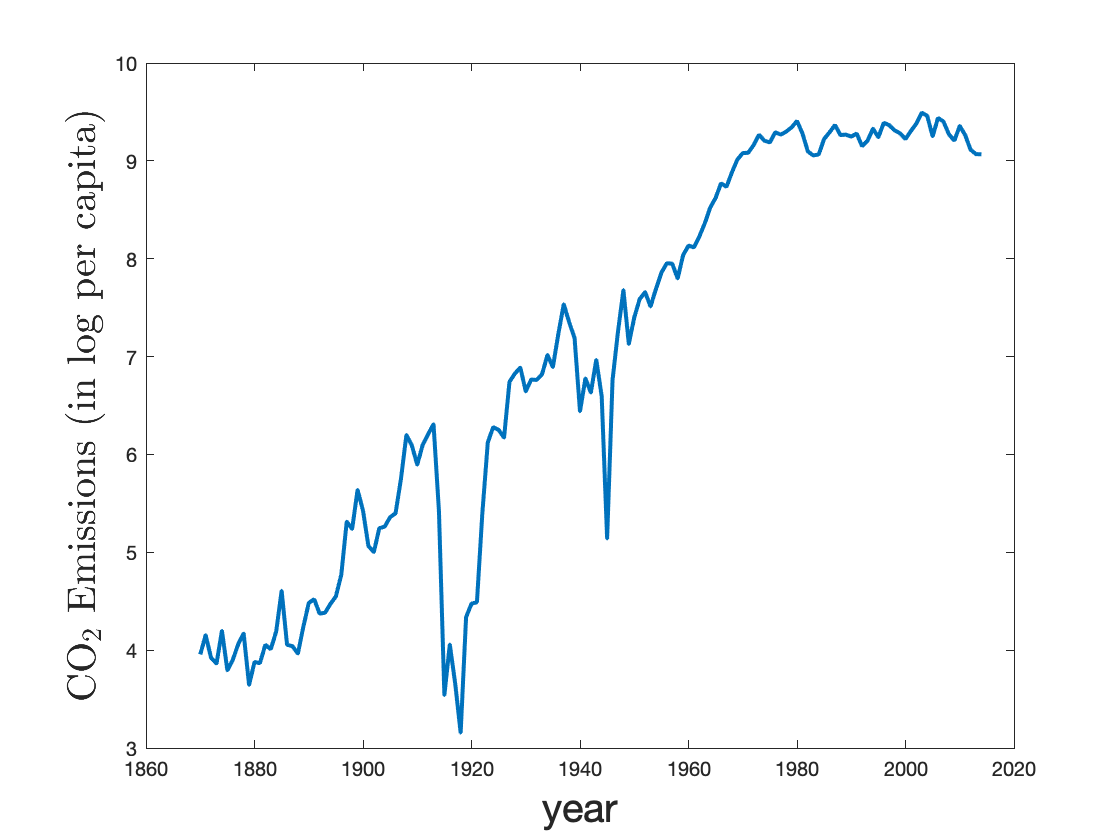}
  \caption{}
\end{subfigure}
\begin{subfigure}{.5\textwidth}
  \centering
  \includegraphics[width=.8\linewidth]{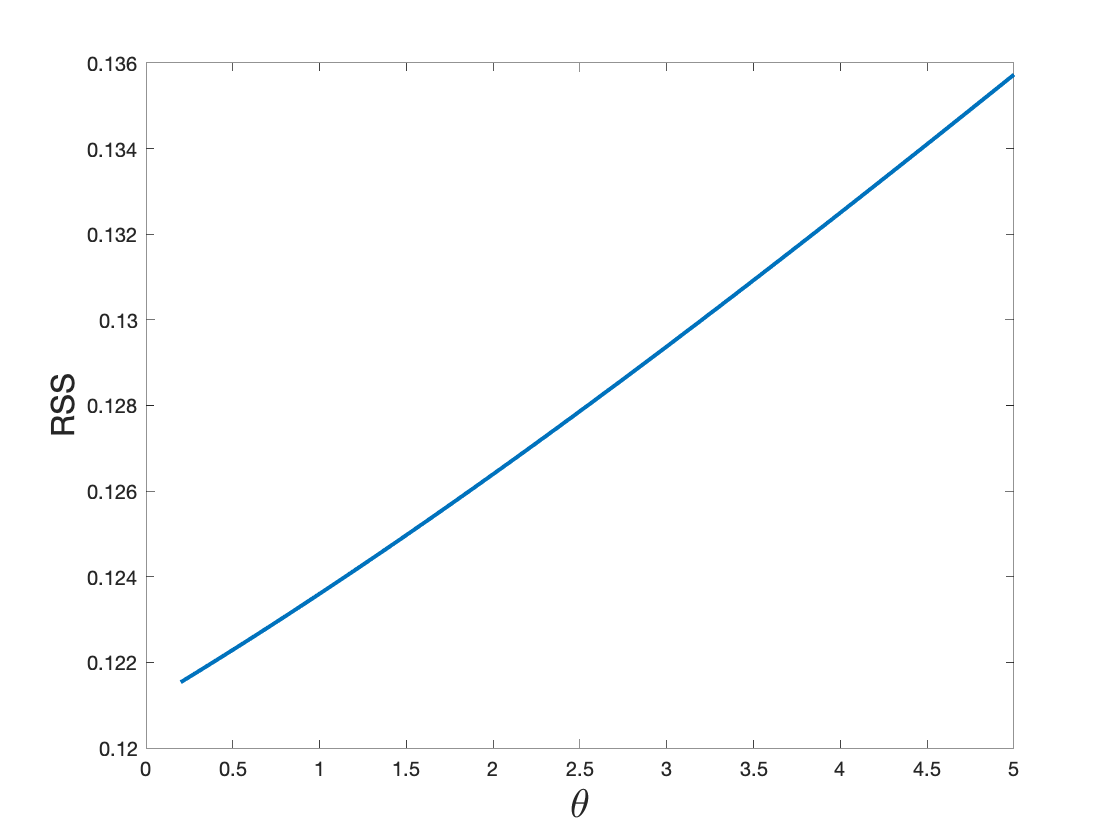}
  \caption{}
\end{subfigure}
\begin{subfigure}{.5\textwidth}
  \centering
  \includegraphics[width=.8\linewidth]{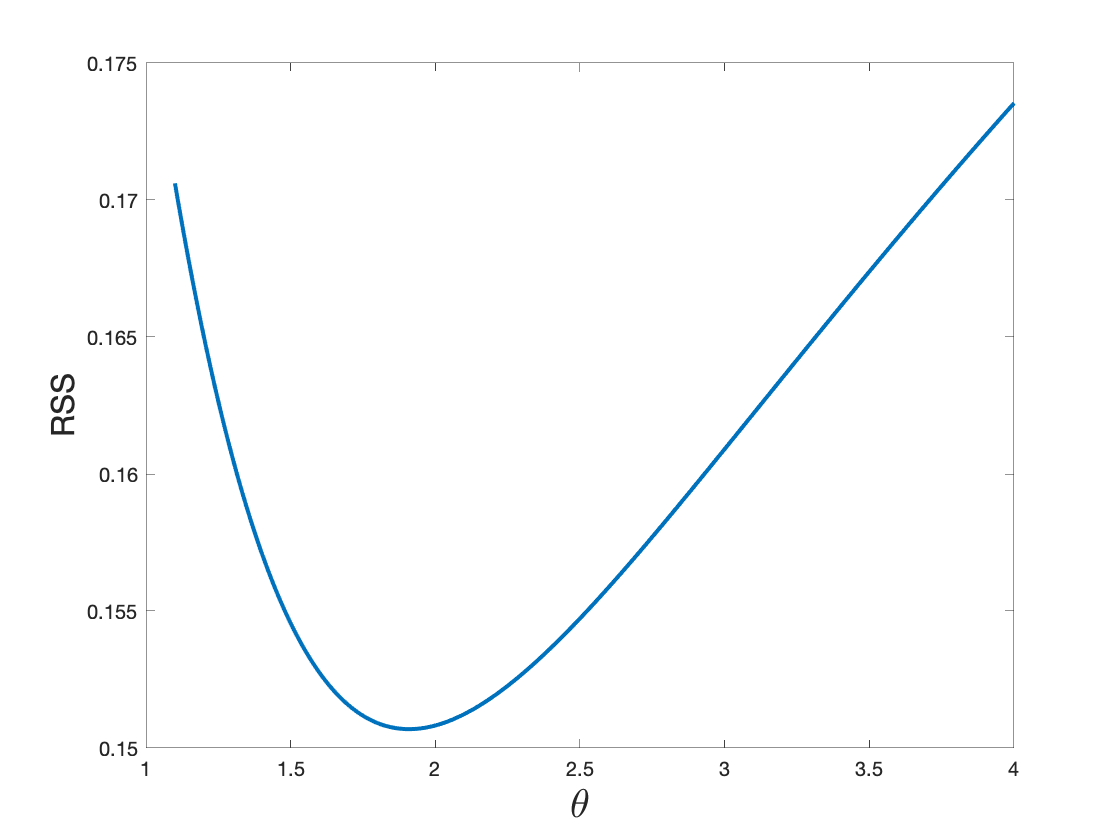}
  \caption{}
\end{subfigure}
\begin{subfigure}{.5\textwidth}
  \centering
  \includegraphics[width=.8\linewidth]{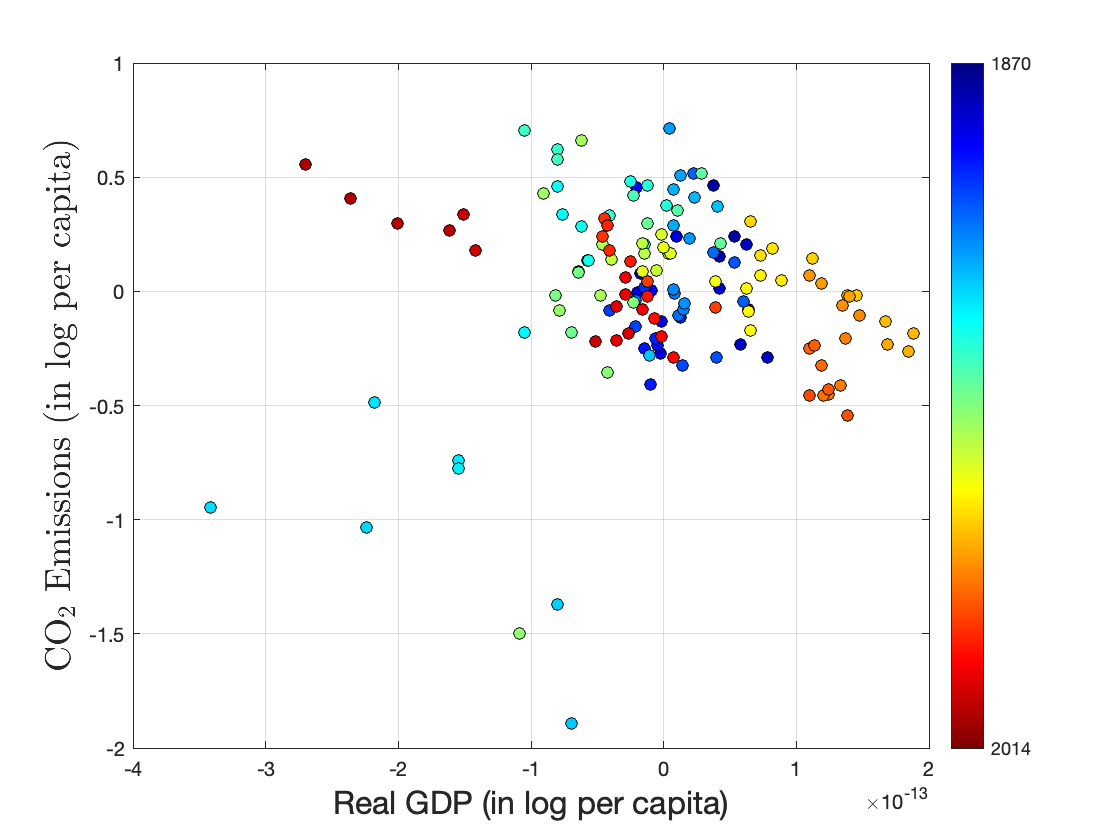}
  \caption{}
\end{subfigure}
\caption{Overview graphs for Finland over 1870-2014. \textbf{(a)} $\log(\text{GDP})$ versus $\log(\text{CO}_2)$ (both per capita). \textbf{(b)} As subfigure (a) but using detrended variables. \textbf{(c)} The log per capita CO\textsubscript{2} emissions time series for Finland. \textbf{(d)} The residual sum of squares (RSS) for the nonlinear model specification $y_t=\tau_1+\tau_2 t + \phi_1 x_t+ \phi_2 x_t^\theta+u_t$ for various values of $\theta$. \textbf{(e)} The RSS as a function of $\theta$ for the flexible nonlinear trend specification $y_t=\tau_1 + \tau_2 t + \tau_3 t^\theta + \phi x_t+u_t$. \textbf{(f)} The relation between $x_t$ and $y_t$ after partialling out the constant, linear trend, and flexible deterministic trend.}
\label{fig:overviewFinland}
\end{figure}

\subsection{RSS($\theta$) for $y_t=\tau_1+\tau_2 t +\phi_1 x_t + \phi_2 x_t^\theta +u_t$} \label{sec:thetagraphs}
\begin{figure}[H]
\begin{subfigure}{.5\textwidth}
  \centering
  \includegraphics[width=.8\linewidth]{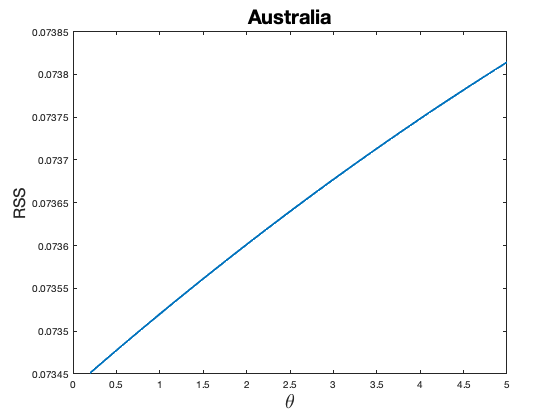}
  \caption{}
\end{subfigure}%
\begin{subfigure}{.5\textwidth}
  \centering
  \includegraphics[width=.8\linewidth]{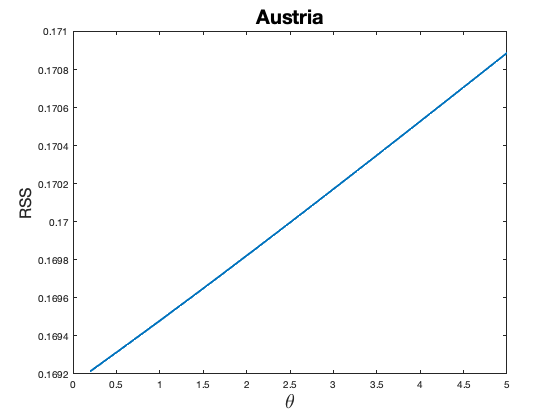}
  \caption{}
\end{subfigure}
\begin{subfigure}{.5\textwidth}
  \centering
  \includegraphics[width=.8\linewidth]{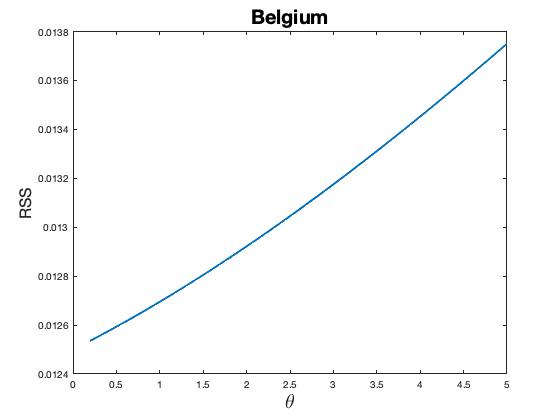}
  \caption{}
\end{subfigure}
\begin{subfigure}{.5\textwidth}
  \centering
  \includegraphics[width=.8\linewidth]{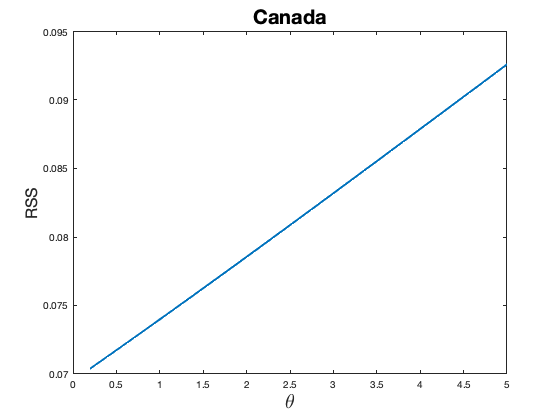}
  \caption{}
\end{subfigure}
\begin{subfigure}{.5\textwidth}
  \centering
  \includegraphics[width=.8\linewidth]{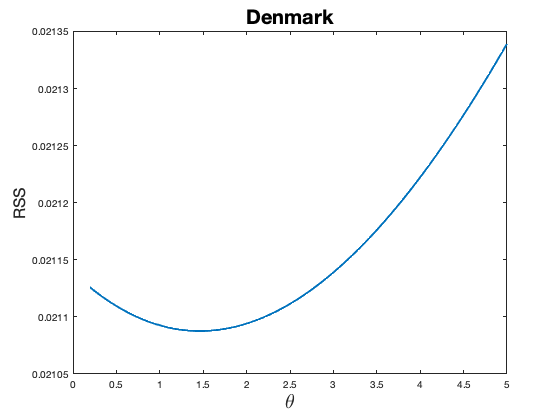}
  \caption{}
\end{subfigure}
\begin{subfigure}{.5\textwidth}
  \centering
  \includegraphics[width=.8\linewidth]{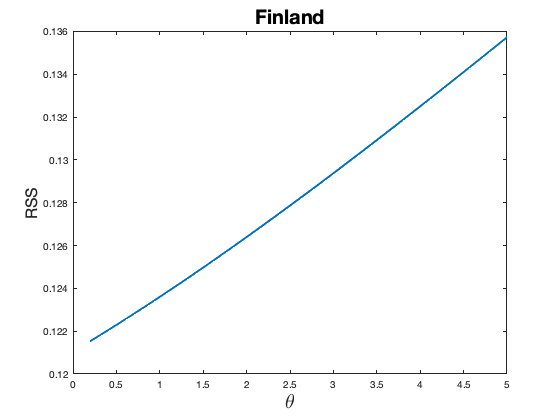}
  \caption{}
\end{subfigure}
\caption{The residual sum of squares (RSS) for the nonlinear specification $y_t=\tau_1+\tau_2 t +\phi_1 x_t + \phi_2 x_t^\theta +u_t$ for various values of $\theta$. This replicates Figure \ref{fig:overviewBelgium}(d) of the main paper for all countries in the data set.}
\label{fig:OverviewRSSvsTheta}
\end{figure}

\begin{figure}[H]
\begin{subfigure}{.5\textwidth}
  \centering
  \includegraphics[width=.8\linewidth]{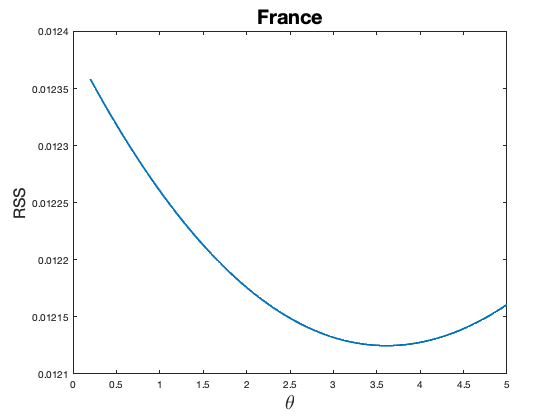}
  \caption*{\textbf{(g)}}
\end{subfigure}%
\begin{subfigure}{.5\textwidth}
  \centering
  \includegraphics[width=.8\linewidth]{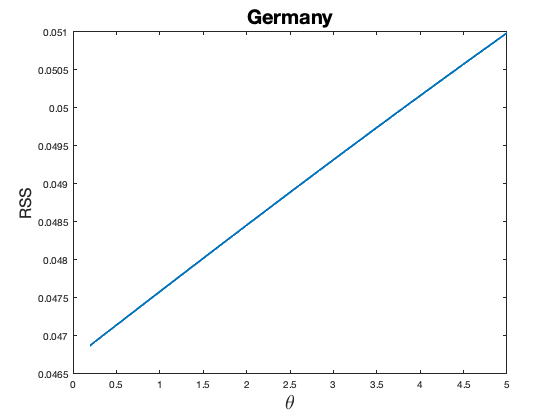}
  \caption*{\textbf{(h)}}
\end{subfigure}
\begin{subfigure}{.5\textwidth}
  \centering
  \includegraphics[width=.8\linewidth]{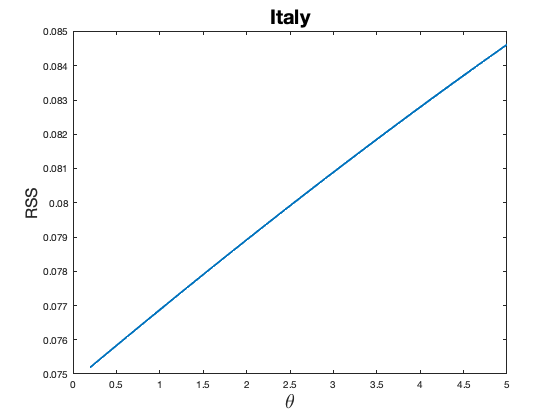}
  \caption*{\textbf{(i)}}
\end{subfigure}
\begin{subfigure}{.5\textwidth}
  \centering
  \includegraphics[width=.8\linewidth]{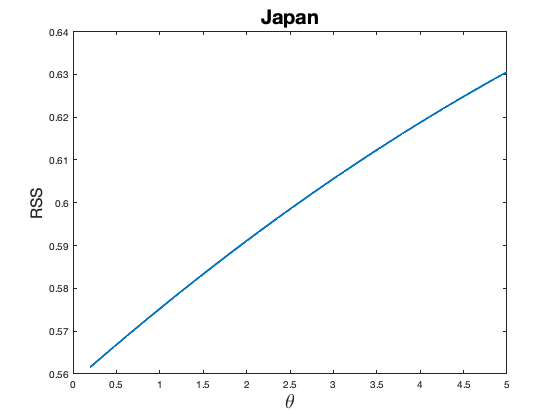}
  \caption*{\textbf{(j)}}
\end{subfigure}
\begin{subfigure}{.5\textwidth}
  \centering
  \includegraphics[width=.8\linewidth]{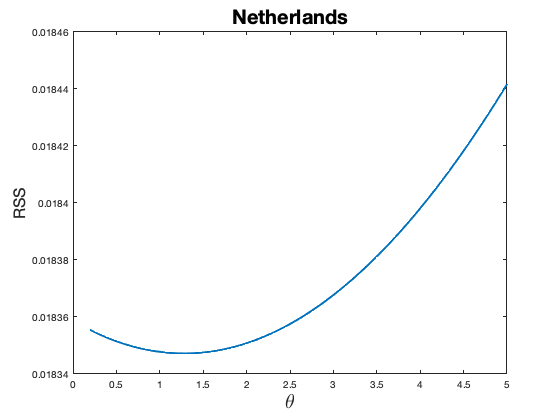}
  \caption*{\textbf{(k)}}
\end{subfigure}
\begin{subfigure}{.5\textwidth}
  \centering
  \includegraphics[width=.8\linewidth]{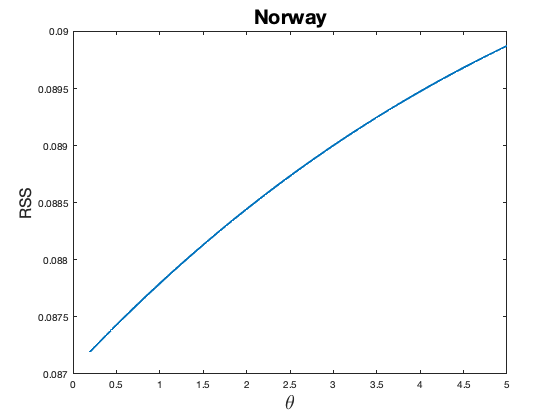}
  \caption*{\textbf{(l)}}
\end{subfigure}
\caption*{Continuation of Figure \ref{fig:OverviewRSSvsTheta}.}
\end{figure}

\begin{figure}[H]
\begin{subfigure}{.5\textwidth}
  \centering
  \includegraphics[width=.8\linewidth]{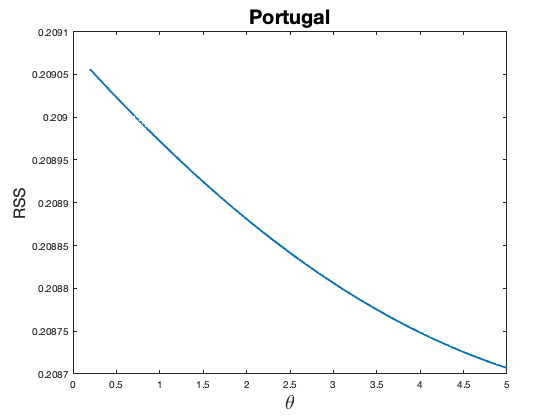}
  \caption*{\textbf{(m)}}
\end{subfigure}%
\begin{subfigure}{.5\textwidth}
  \centering
  \includegraphics[width=.8\linewidth]{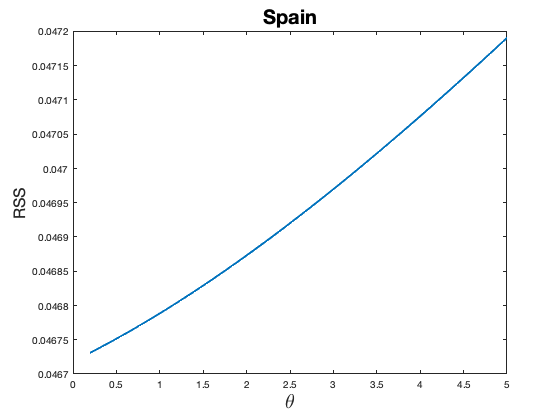}
  \caption*{\textbf{(n)}}
\end{subfigure}
\begin{subfigure}{.5\textwidth}
  \centering
  \includegraphics[width=.8\linewidth]{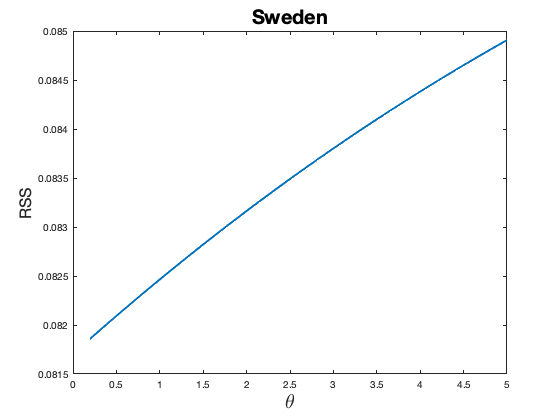}
  \caption*{\textbf{(o)}}
\end{subfigure}
\begin{subfigure}{.5\textwidth}
  \centering
  \includegraphics[width=.8\linewidth]{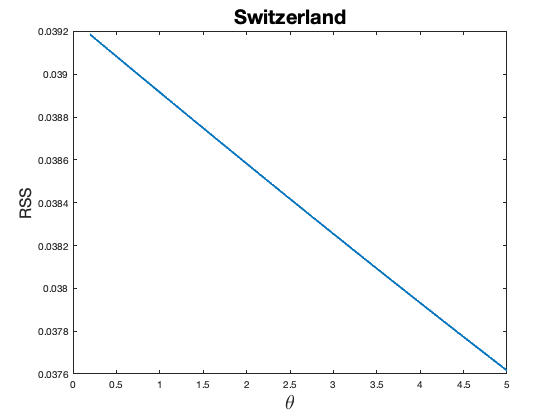}
  \caption*{\textbf{(p)}}
\end{subfigure}
\begin{subfigure}{.5\textwidth}
  \centering
  \includegraphics[width=.8\linewidth]{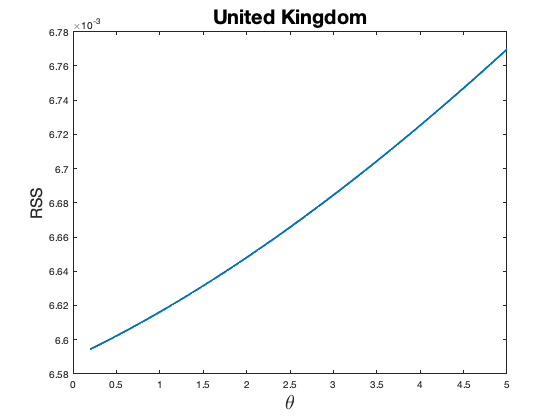}
  \caption*{\textbf{(q)}}
\end{subfigure}
\begin{subfigure}{.5\textwidth}
  \centering
  \includegraphics[width=.8\linewidth]{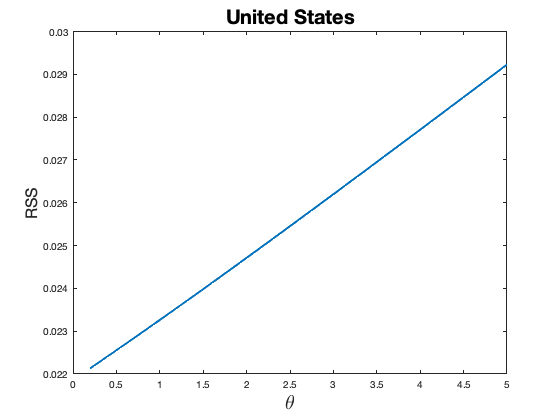}
  \caption*{\textbf{(r)}}
\end{subfigure}
\caption*{Continuation of Figure \ref{fig:OverviewRSSvsTheta}.}
\label{fig:OverviewRSSvsThetaFinal}
\end{figure}

\subsection{Additional results for univariate models}\label{appendix:univeriateanalysis}
Results of a more in-depth univariate analysis are collected in this section. We look at models (M1\textsuperscript{*})--(M4\textsuperscript{*}) as listed in Table \ref{tab:modeloverview}.

\begin{table}[H]
\centering
\caption{An overview of the univariate models.}
\label{tab:modeloverview}
\begin{threeparttable}
	\begin{tabular}{l c l}
	\toprule
 Model	&& Specification\\
 \midrule
 (M1\textsuperscript{*}) 	&& $y_t = \tau_1 + \tau_2 t + \phi_1 x_t + \phi_2 x_t^2 + u_t$ \\
 (M2\textsuperscript{*}) 	&& $y_t = \tau_1 + \tau_2 t + \tau_3 t^2 +\phi_1 x_t + \phi_2 x_t^2 + u_t$\\
 (M3\textsuperscript{*}) 	&& $y_t = \tau_1 + \tau_2 t + \tau_3 t^\theta + \phi_1 x_t + \phi_2 x_t^2 + u_t$ \\
 (M4\textsuperscript{*}) 	&& $y_t = \tau_1 + \tau_2 t + \tau_3 t^\theta + \phi_1 x_t + u_t$\\
	\bottomrule
	\end{tabular}%
    \end{threeparttable}
\end{table}

All three models are of the form:
 \begin{equation}
 y_t	= \tau_1 + \tau_2 t + \tau_3 t^\theta + \phi_1 x_t + \phi_2 x_t^2 + u_t.
 \label{eq:generalspecification}
 \end{equation}
 Model (M1\textsuperscript{*}) is the specification above with $\tau_3=0$ and forces all nonlinearities to be captured through $x_{t}^2$. Specifications (M2\textsuperscript{*}) and (M3\textsuperscript{*}) include deterministic nonlinear time trends. For model (M2\textsuperscript{*}), we allow for $\tau_3\neq 0$ but fix $\theta=2$. Model \eqref{eq:generalspecification} without further restrictions is referred to as (M3\textsuperscript{*}). In the latter model, the NLS estimator for $\theta$ is computed by a grid search over the values $\Theta=[0.05, 0.95]\cup [1.05,10]$ and simulated inference is used (see Section \ref{subsec:sim_inf} of the main paper). Table \ref{tab:ekc_est_kpss} illustrates how increasingly flexible nonlinear deterministic trends affect the parameter estimates for $\phi_1$ and $\phi_2$. Judging exclusively by the signs of $\widehat{\phi}_1$ and $\widehat{\phi}_2$, the EKC exists for 17 out of 18, 9 out of 18, and 8 out of 18 countries for (M1\textsuperscript{*}), (M2\textsuperscript{*}), and (M3\textsuperscript{*}), respectively. Moreover, the significance of squared log per capita GDP (read: $\phi_2$) reduces when nonlinear deterministic time trends are included. For model (M3\textsuperscript{*}), $\phi_2$ is never significantly different from zero at a 10\% level and evidence in favour of EKC becomes rather meagre. The results of the univariate KPSS tests for these models can be found in Table \ref{tab:ekc_est_kpss} under ``Stationarity tests''. In general, the cointegrating relations seem well-specified except maybe for Belgium, Denmark, and UK.

\begin{sidewaystable}
	\begin{table}[H]
	\centering
	\caption{Parameter estimates and output of the KPSS-type of test for stationarity as computed for model specifications (M1\textsuperscript{*})-(M4\textsuperscript{*}). The column $\widehat{KPSS}$ and $M_{opt}$ provide the numerical values of the KPSS tests and the number of chosen residual subblocks, respectively.}
	\label{tab:ekc_est_kpss}
	\resizebox{1\textheight}{!}{%
		\begin{threeparttable}
		\begin{tabular}{l d{2.5} d{2.5} d{3.5} d{2.5} d{3.5} d{2.2} d{1.2} d{1.5} c d{1.5} c d{1.4} c d{1.4} c d{1.4} c}
			\toprule
			& \multicolumn{8}{c}{Parameter estimates} &  & \multicolumn{8}{c}{Stationarity tests} \\
			\cmidrule{2-9}\cmidrule{11-18}
			\multicolumn{1}{c}{} & \multicolumn{2}{c}{(M1\textsuperscript{*})} & \multicolumn{2}{c}{(M2\textsuperscript{*})} & \multicolumn{2}{c}{(M3\textsuperscript{*})} & \multicolumn{2}{c}{(M4\textsuperscript{*})} &  & \multicolumn{2}{c}{(M1\textsuperscript{*})} & \multicolumn{2}{c}{(M2\textsuperscript{*})} & \multicolumn{2}{c}{(M3\textsuperscript{*})} & \multicolumn{2}{c}{(M4\textsuperscript{*})} \\
			\midrule
			\multicolumn{1}{c}{Country} & \multicolumn{1}{c}{$\widehat{\phi}_1$} & \multicolumn{1}{c}{$\widehat{\phi}_2$} & \multicolumn{1}{c}{$\widehat{\phi}_1$} & \multicolumn{1}{c}{$\widehat{\phi}_2$} & \multicolumn{1}{c}{$\widehat{\phi}_1$} & \multicolumn{1}{c}{$\widehat{\phi}_2$} & \multicolumn{1}{c}{$\widehat{\theta}$} & \multicolumn{1}{c}{$\widehat{\phi}_1$} &  & \multicolumn{1}{c}{$\widehat{KPSS}$} & $M_{opt}$ & \multicolumn{1}{c}{$\widehat{KPSS}$} & $M_{opt}$ & \multicolumn{1}{c}{$\widehat{KPSS}$} & $M_{opt}$ & \multicolumn{1}{c}{$\widehat{KPSS}$} & $M_{opt}$ \\
			\midrule
			Australia		& 2.75 		& -0.17       & -23.92^{***}& 1.40^{***} & -12.19^{***} & 0.74 & 0.88 & 1.25^{***} &  & 1.49 & 9 & 1.79 & 9 & 1.80 & 9 & 1.35 & 9 \\
			Austria		& 7.13^{***}	& -0.30^{**}  & 1.25        & 0.03 & 3.75^{***} & -0.12 & 0.88 & 1.55^{***} &  & 1.03 & 7 & 1.07 & 7 & 1.65 & 7 & 1.63 & 7 \\
			Belgium		& 11.45^{***}	& -0.57^{***} & 10.03^{***} & -0.49^{***} & 10.29^{***} & -0.50 & 2.60 & 1.01^{***} &  & 1.91 & 9 & 2.92^{*} & 8 & 2.74^{*} & 8 & 2.28 & 9 \\
			Canada		& 12.72^{***}	& -0.64^{***} & 14.80       & -0.77 & -3.46^{***} & 0.25 & 0.56 & 1.14^{***} &  & 2.83^{*} & 7 & 2.60^{*} & 7 & 1.26 & 9 & 1.28 & 9 \\
			Denmark		& 14.52^{***}	& -0.65^{***} & -2.80       & 0.25 & -5.75^{***} & 0.39 & 2.03 & 1.68^{***} &  & 3.14^{*} & 9 & 3.30^{**} & 9 & 2.98^{*} & 9 & 1.58 & 8 \\
			Finland		& 16.86^{***}	& -0.76^{***} & 16.97^{***} & -0.77^{***} & 22.58^{***} & -1.06 & 1.87 & 3.95^{***} &  & 2.05 & 8 & 2.06 & 8 & 0.71 & 9 & 0.80 & 9 \\
			France		& 10.87^{***} & -0.55^{***} & 3.14^{*}    & -0.12 & 3.31^{***} & -0.13 & 2.09 & 1.00^{***} &  & 1.72 & 9 & 0.69 & 8 & 0.56 & 8 & 2.49 & 9 \\
			Germany		& 6.24^{***}  & -0.31^{***} & -1.82       & 0.13 & -4.42^{***} & 0.29 & 0.59 & 0.89^{***} &  & 2.63^{*} & 7 & 2.1 & 8 & 1.23 & 9 & 2.81^{*} & 7 \\
			Italy 			& 11.76^{***} & -0.55^{***} & 7.31^{**}   & -0.30 & 7.72^{***} & -0.29 & 0.82 & 2.41^{***} &  & 4.18^{**} & 7 & 3.79^{**} & 8 & 0.79 & 7 & 0.78 & 7 \\
			Japan		& 9.86^{***}  & -0.52^{***} & -4.27       & 0.29 & 1.16^{***} & -0.00 & 0.05 & 1.15^{***} &  & 5.17^{***} & 8 & 3.93^{**} & 7 & 1.83 & 9 & 1.84 & 9 \\
			Netherlands	& 8.70^{***} & -0.41^{***} & 1.49        & -0.01 & 0.48^{**} & 0.05 & 1.86 & 1.32^{***} &  & 2.16 & 7 & 0.94 & 7 & 1.2 & 7 & 1.15 & 7 \\
			Norway 		& 3.87        & -0.16       & -9.14^{**}  & 0.51^{**} & -1.10^{**} & 0.16 & 0.46 & 2.05^{***} &  & 2.53^{*} & 7 & 1.06 & 7 & 0.74 & 9 & 1.44 & 8 \\
			Portugal		& 0.09        &  0.04       & -5.86^{***} & 0.42 & -1.11^{**} & 0.15 & 0.05 & 1.69^{***} &  & 6.95^{***} & 8 & 5.28^{**} & 7 & 0.64 & 7 & 1.64 & 7 \\
			Spain		& 7.72^{***}  & -0.37^{***} & 1.98        & -0.01 & 4.31^{***} & -0.16 & 1.55 & 1.52^{***} &  & 2.78^{*} & 7 & 2.03 & 8 & 2.4 & 8 & 2.42 & 8 \\
			Sweden 		& 10.91^{***} & -0.44^{***} & -9.08^{*}   & 0.61^{**} & 0.43 & 0.17 & 0.46 & 3.48^{***} &  & 3.59^{**} & 7 & 1.27 & 7 & 0.73 & 7 & 0.80 & 7 \\
			Switzerland	& 8.57^{***}& -0.29^{***} & -7.86^{**}  & 0.54^{***} & -13.86^{***} & 0.83 & 2.98 & 2.63^{***} &  & 0.80 & 7 & 0.96 & 7 & 0.77 & 7 & 0.75 & 7 \\
			UK 			& 9.32^{***}   & -0.47^{***} & 5.91^{***}  & -0.27^{***} & 4.13^{***} & -0.18 & 3.04 & 0.80^{***} &  & 2.76^{*} & 9 & 4.25^{**} & 9 & 4.30^{**} & 9 & 3.98^{**} & 9 \\
			USA 			& 8.67^{***}   & -0.44^{***} & 0.93        & -0.03 & -5.62^{***} & 0.35 & 0.92 & 0.95^{***} &  & 1.64 & 8 & 1.85 & 8 & 2.25 & 8 & 1.97 & 8 \\
			\bottomrule
		\end{tabular}%
	    \begin{tablenotes}
	    	\footnotesize
	    	\item Note: Asterisks denote rejection of the null hypothesis at the $^{***}1\%$, $^{**}5\%$, and $^{*}10\%$ significance level. Depending on the specific table entry, the null hypothesis refers to either a coefficient being zero or (nonlinear) cointegration.
    	\end{tablenotes}
    \end{threeparttable}
	}
\end{table}
\end{sidewaystable}

The insignificance of $\phi_2$ in model (M3\textsuperscript{*}) suggests a final model specification, namely
\begin{equation}
  y_t	= \tau_1 + \tau_2 t + \tau_3 t^\theta + \phi_1 x_t + u_t. \tag{M4\textsuperscript{*}}
 \label{eq:EKCwithoutxt2}
\end{equation}
Model \eqref{eq:EKCwithoutxt2} specifies a linear cointegrating relation around a flexible time trend and does not incorporate nonlinear effects in log per capita GDP.\footnote{Model specification (M4\textsuperscript{*}) has the additional advantage of being invariant to the possible presence of a drift component in log per capita GDP, also see footnote \ref{footnoteDetTrend} of the main text.} That is, the model specification does not allow for an EKC. As before, we check parameter estimates and test for stationarity of the error terms (the columns labeled ``(M4\textsuperscript{*})'' in Table \ref{tab:ekc_est_kpss}). Some remarks concerning this final model specification are:
\begin{enumerate}
 \item For Belgium, the fitted model reads
 \begin{equation}
  y_t = -0.049+ 0.0063 t - 6.131\times 10^{-6} \; t^{2.603} + 1.006 x_t + \hat{u}_t.
 \label{eq:Belgiumresults}
 \end{equation}
 The flexible power on the linear trend is estimated to be $\widehat{\theta}=2.603$ resulting in nonlinear behaviour over time. Moreover, the negative coefficient in front of $t^{2.603}$ provides a contribution that is sloping down over time. If time effects are ignored, then a 1\% increase in GDP will lead to an estimated 1.006\% increase in fossil-fuel CO\textsubscript{2} emissions.  

 \item The outcomes of the KPSS test do not point towards a misspecified cointegrating relation (Table \ref{tab:ekc_est_kpss}). The flexible deterministic trend is generally sufficient to describe the nonlinear behaviour of the (univariate) log per capita CO\textsubscript{2} emissions over time, that is, \emph{squared log per capita GDP is not needed in the univerariate models}. Visual proof is found in Figures \ref{fig:overviewBelgium}(a), \ref{fig:overviewBelgium}(b) and \ref{fig:overviewBelgium}(f) where the incorporation of  increasingly flexible time effects is seen to remove any apparent nonlinear relationship between log per capita GDP and CO\textsubscript{2} emissions.
\end{enumerate}

Visualisations of the model fits are available in Figures \ref{fig:overallfit}--\ref{fig:overallfit4}.

\begin{sidewaysfigure}[h]
  \centering
  \includegraphics[width=.9\linewidth]{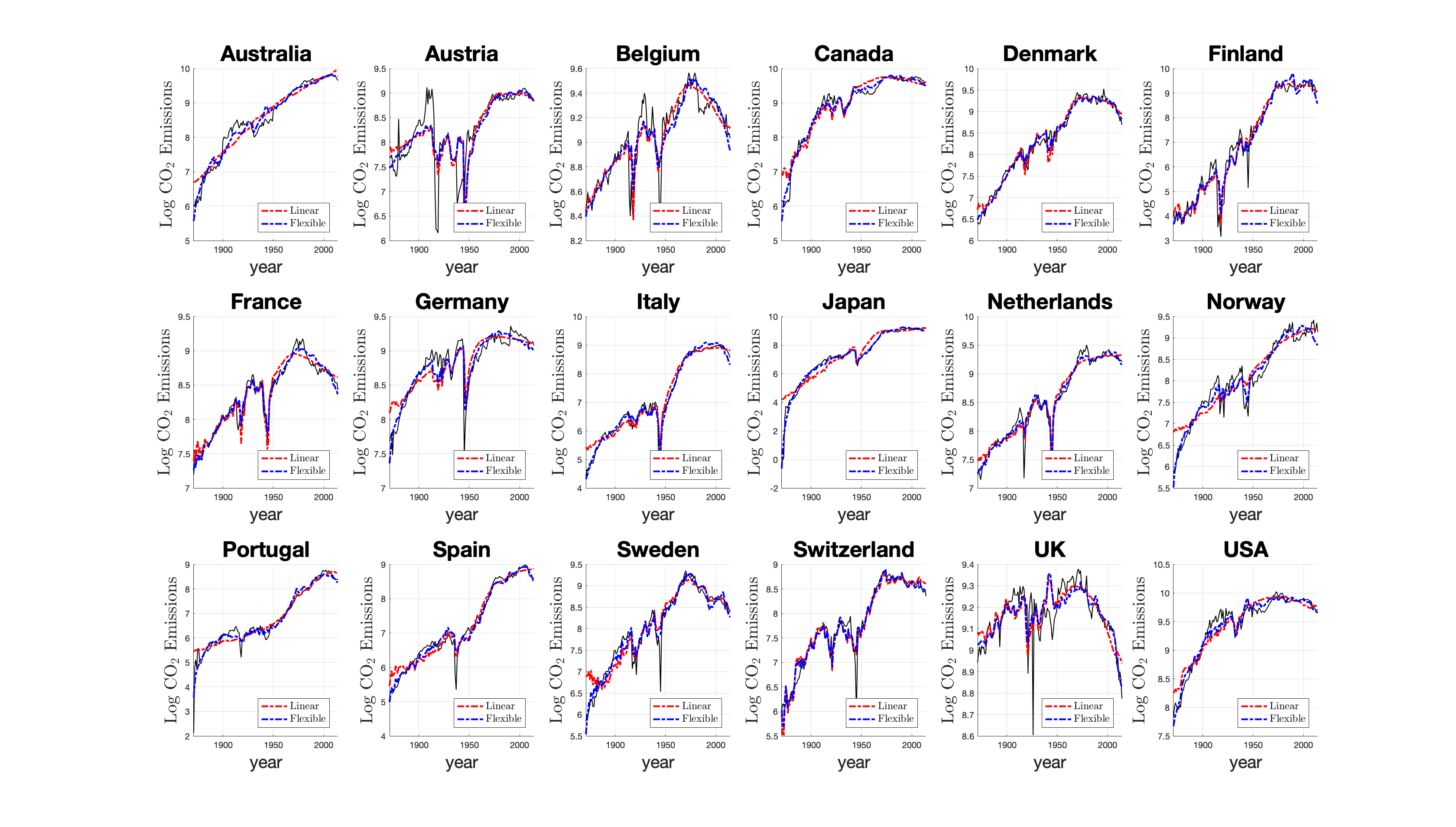}
  \caption{Estimation results for $\COTWO$ emissions: actual values (black), fitted values under the CPR model $y_t = \tau_1 + \tau_2 t+\phi_1 x_t + \phi_2 x_t^2 + u_t$ (red), and fitted values under the GCPR model $y_t = \tau_1 + \tau_2 t + \tau_3 t^\theta+\phi_1 x_t + u_t$ (blue).}
  \label{fig:overallfit}
\end{sidewaysfigure}

\begin{sidewaysfigure}[h]
  \centering
  \includegraphics[width=.9\linewidth]{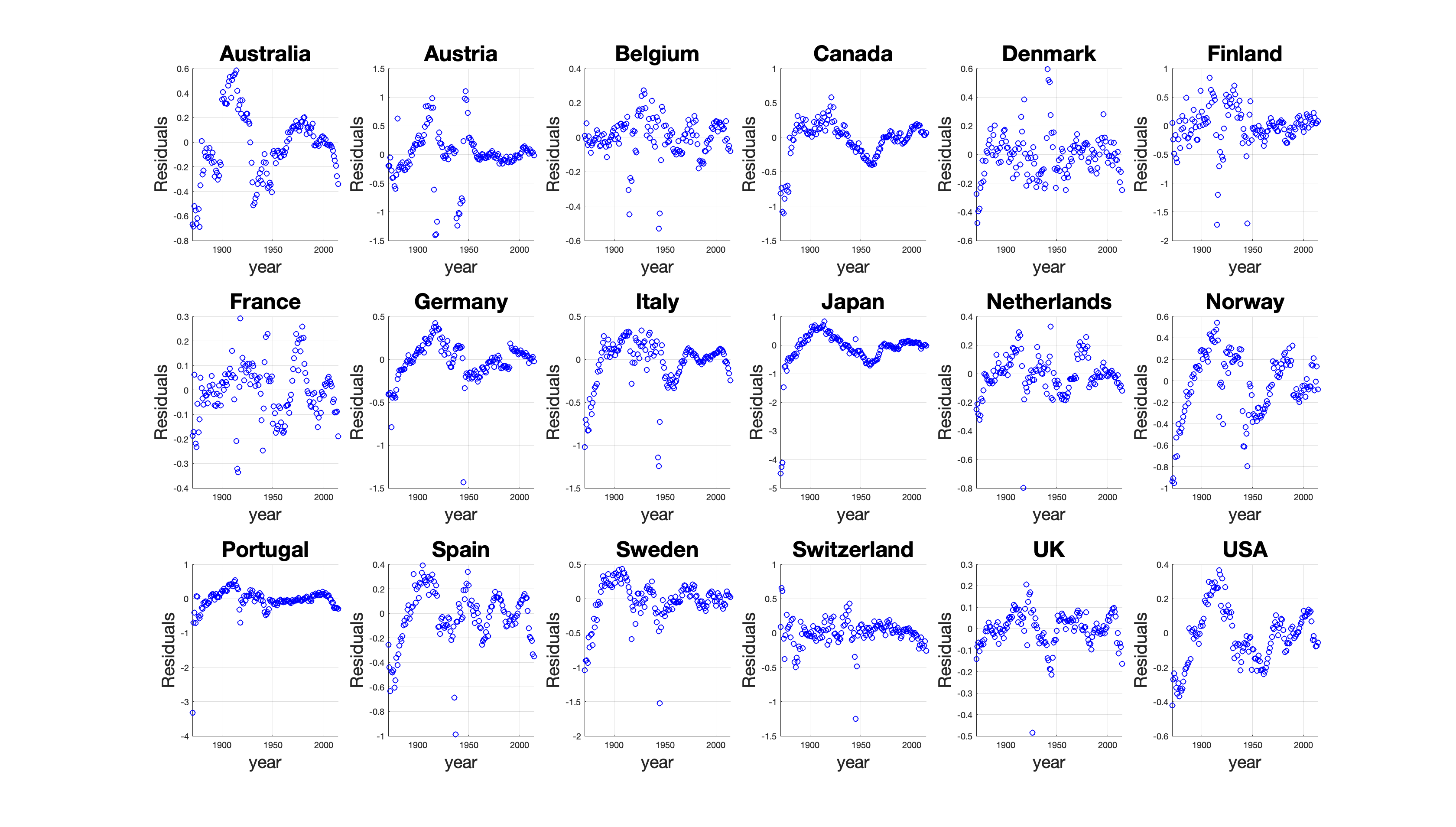}
  \caption{The residual series for each country under model specification (M1): $y_t = \tau_1 + \tau_2 t + \phi_1 x_t + \phi_2 x_t^2 + u_t$.}
  \label{fig:overallfit1}
\end{sidewaysfigure}
\clearpage

\begin{sidewaysfigure}[h]
  \centering
  \includegraphics[width=.9\linewidth]{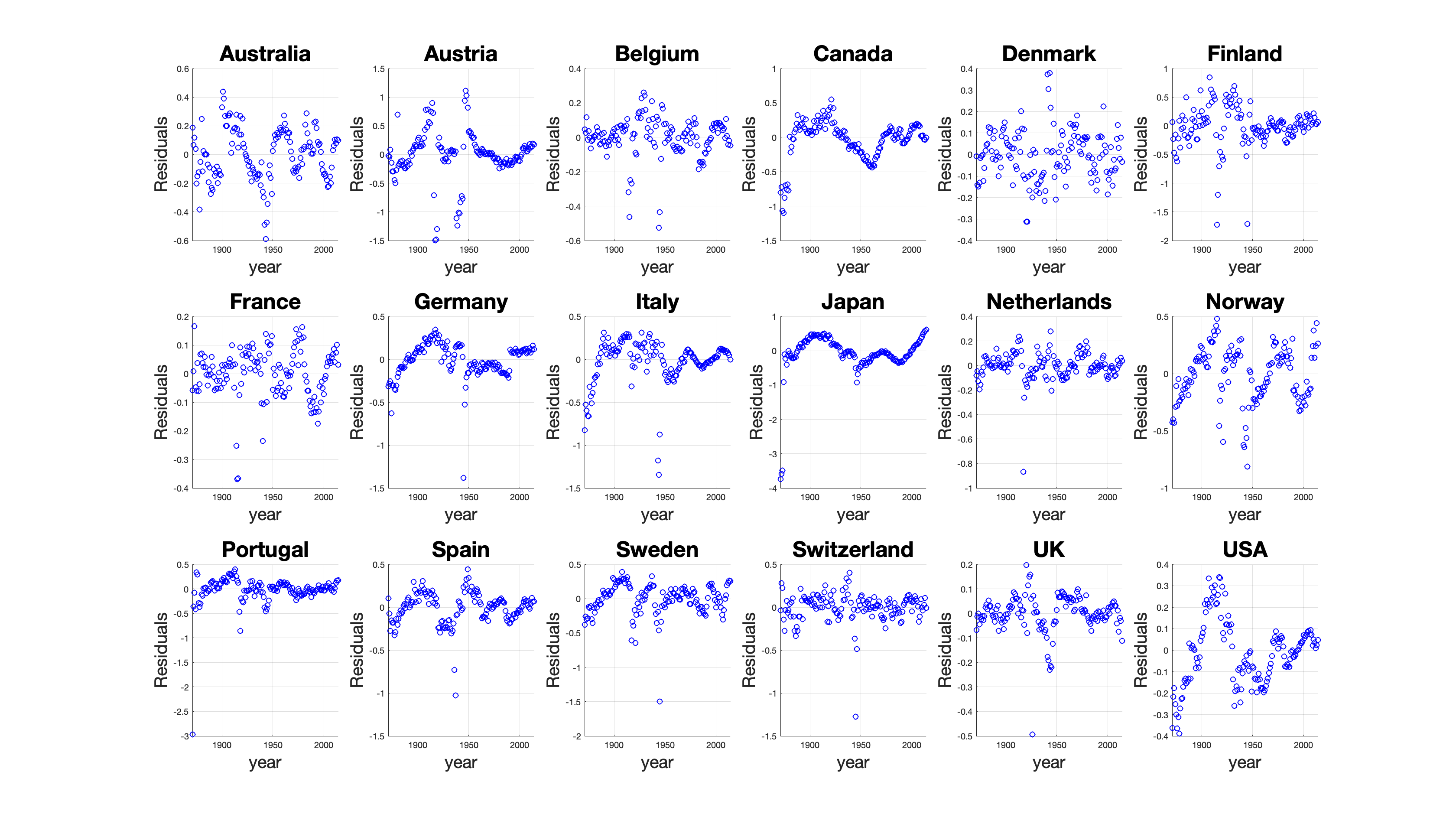}
  \caption{The residual series for each country under model specification (M2): $y_t = \tau_1 + \tau_2 t + \tau_3 t^2+ \phi_1 x_t + \phi_2 x_t^2 + u_t$.}
  \label{fig:overallfit2}
\end{sidewaysfigure}

\begin{sidewaysfigure}[h]
  \centering
  \includegraphics[width=.9\linewidth]{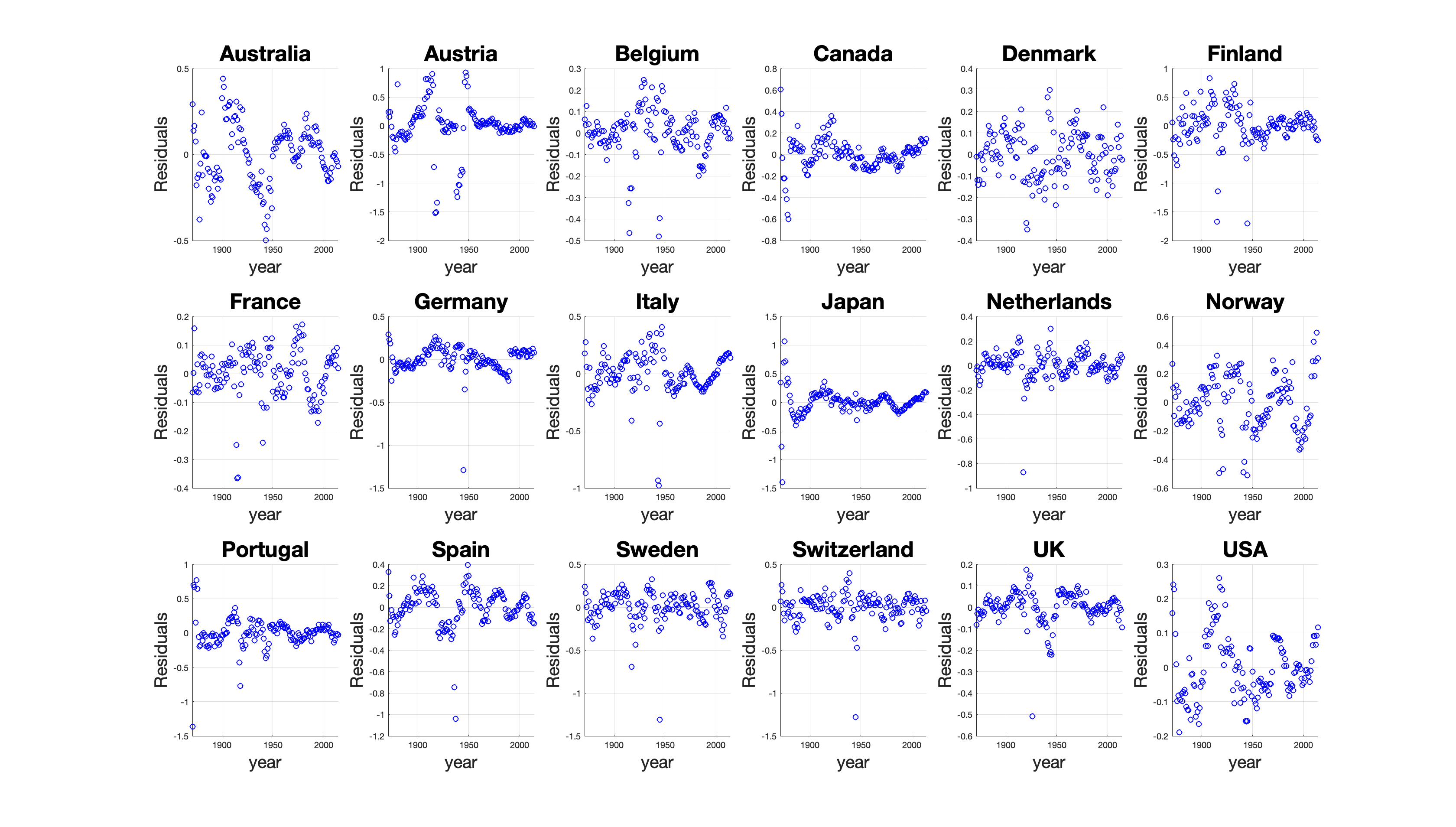}
  \caption{The residual series for each country under model specification (M3): $y_t = \tau_1 + \tau_2 t + \tau_3 t^\theta+ \phi_1 x_t + \phi_2 x_t^2 + u_t$.}
  \label{fig:overallfit3}
\end{sidewaysfigure}

\begin{sidewaysfigure}[h]
  \centering
  \includegraphics[width=.9\linewidth]{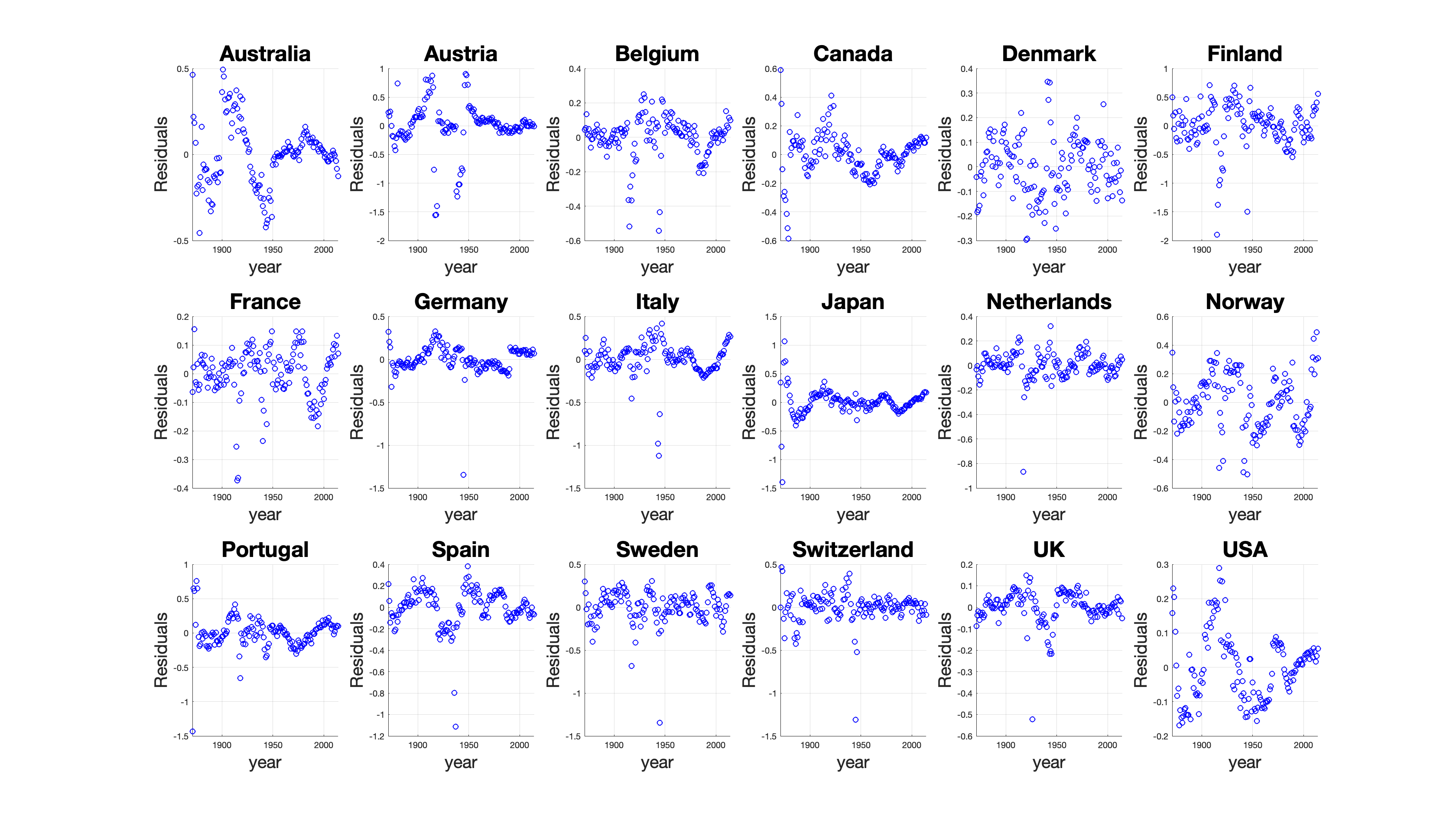}
  \caption{The residual series for each country under model specification (M4): $y_t = \tau_1 + \tau_2 t + \tau_3 t^\theta+ \phi_1 x_t + u_t$.}
  \label{fig:overallfit4}
\end{sidewaysfigure}

\clearpage
\subsection{Nonparametric kernel estimator and linear fit}\label{appendix:graphslinearfit}

\begin{figure}[h]
\centering
\includegraphics[width=.7\linewidth]{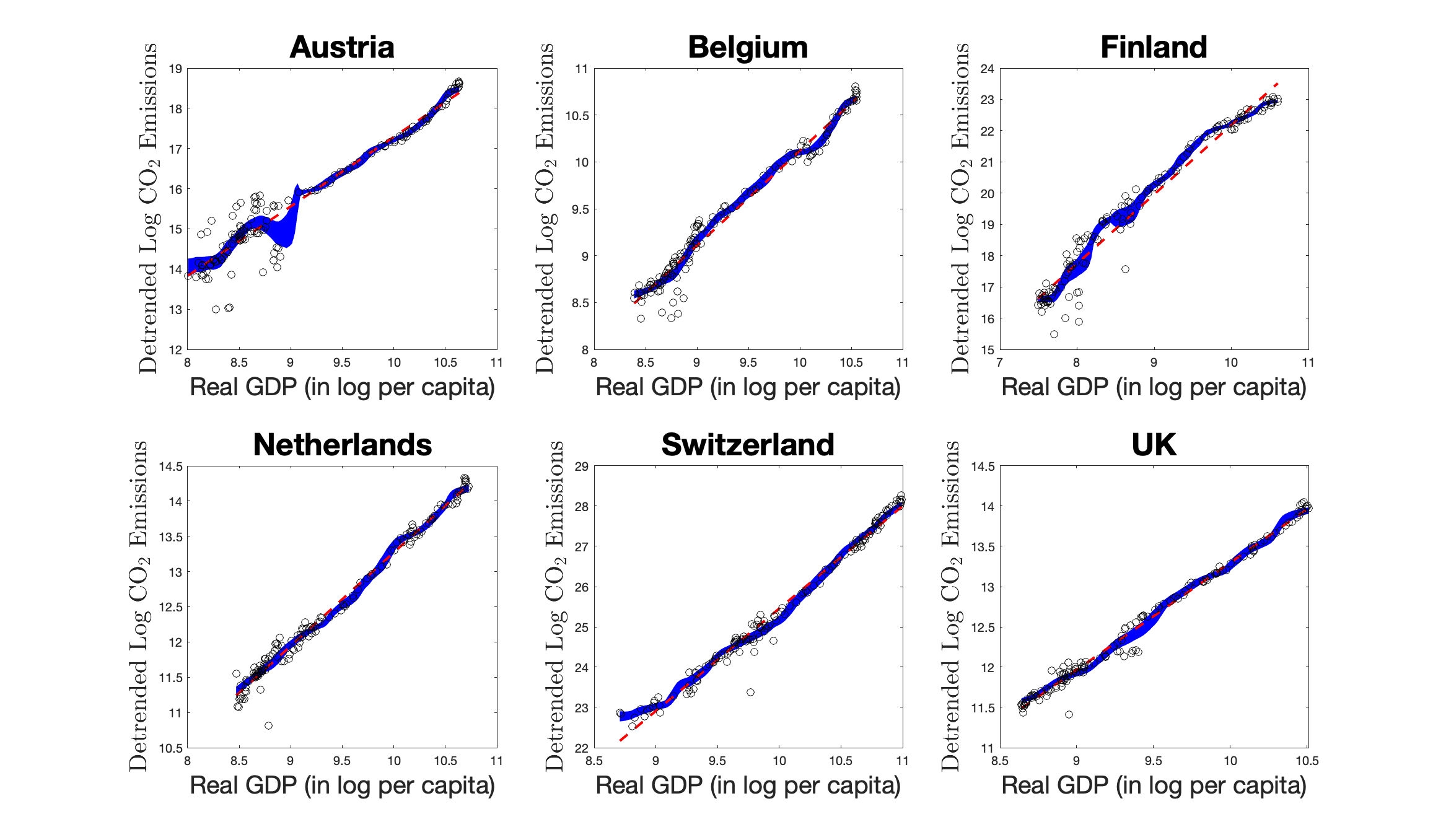}
\caption{The 95\% (point-wise) confidence intervals of the non-parametric kernel estimate for the relationship between GDP and CO\textsubscript{2} emissions (blue) after removal of the country-specific and joint flexible deterministic trends. The red dotted line is the linear fit. Results are based on the \emph{full sample}.}
\end{figure}

\begin{figure}[h]
\centering
\includegraphics[width=.7\linewidth]{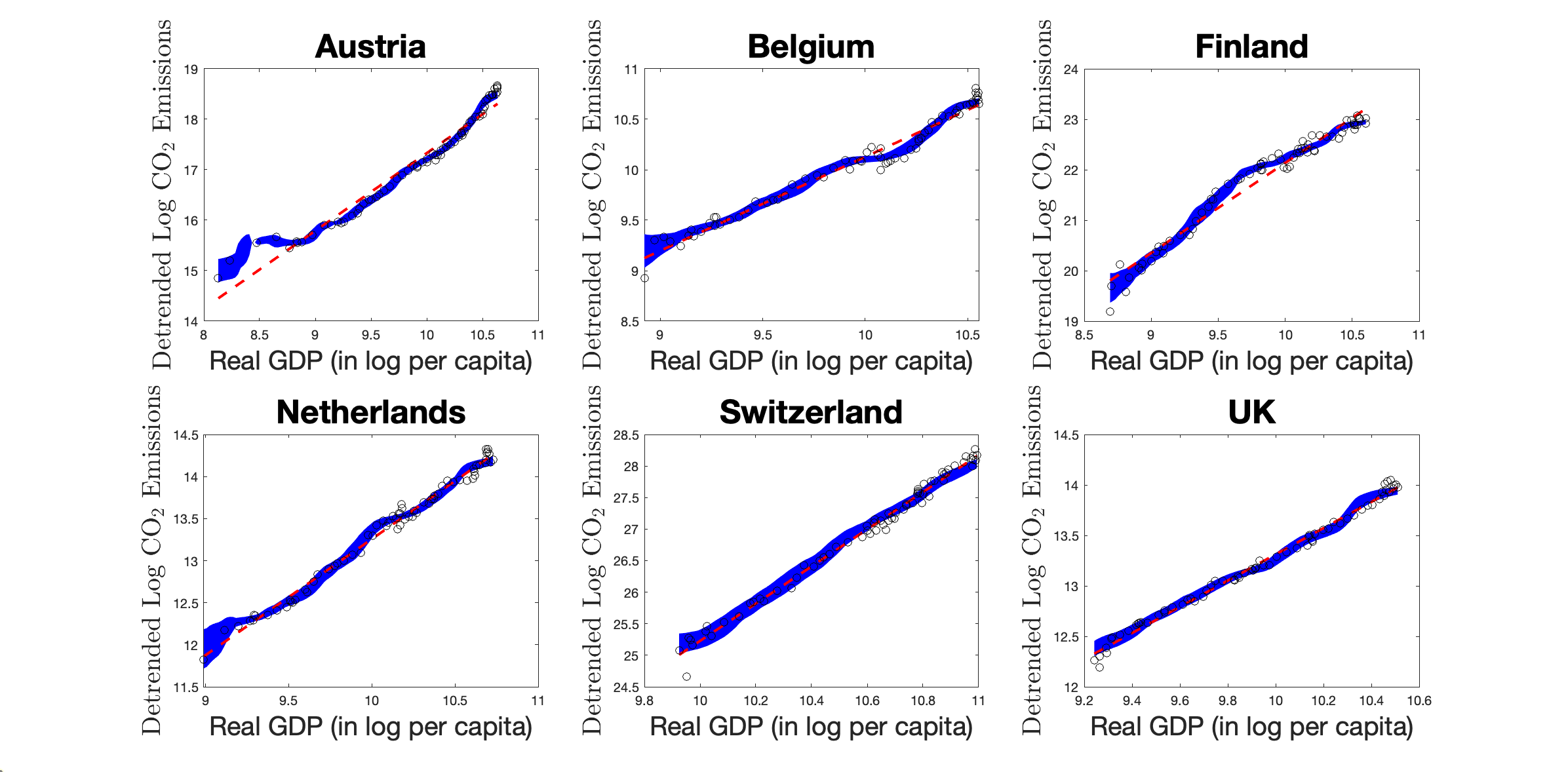}
\caption{The 95\% (point-wise) confidence intervals of the non-parametric kernel estimate for the relationship between GDP and CO\textsubscript{2} emissions (blue) after removal of the country-specific and joint flexible deterministic trends. The red dotted line is the linear fit. Results are based on \emph{observations after World War II}.}
\end{figure}



\end{document}